\documentclass[a4paper,11pt]{article}

\usepackage[linesnumbered,algoruled,boxed,lined]{algorithm2e}
\usepackage{complexity} 
\usepackage{amssymb,amsmath,amsthm}
\usepackage{thm-restate}
\usepackage{mathtools}
\usepackage[a4paper, total={6.5in, 8.5in}]{geometry}
\usepackage{tikz}
\usetikzlibrary{fit}
\usepackage{hyperref}
\usepackage{cleveref}
\usepackage{cite}

\hypersetup{
    breaklinks=true,
    colorlinks=true,
    citecolor=red!70!blue!60!black,
    linkcolor=red!70!blue!90!black,
}

\usepackage[framemethod=tikz]{mdframed}
\mdfsetup{frametitlealignment=\centering}
\usepackage{enumitem}
\setlist[itemize]{noitemsep}
\setlist[itemize]{nosep}
\setlist[enumerate]{noitemsep}
\usepackage{vmargin}
\setmarginsrb{1in}{1in}{1in}{1in}{0pt}{0pt}{0pt}{6mm}

\usepackage{subcaption}
\usepackage{color}

\newcommand{\bd}{\partial}
\newcommand{\inter}{\mathrm{int}}

\newcommand{\calB}{\mathcal{B}}

\newcommand{\calD}{\mathcal{D}}
\newcommand{\calF}{\mathcal{F}}

\newcommand{\calH}{\mathcal{H}}

\newcommand{\calP}{\mathcal{P}}

\newcommand{\tw}{\mathrm{tw}}
\newcommand{\dist}{\mathrm{dist}}
\newcommand{\Skel}{\mathrm{Skel}}

\newcommand{\out}{\texttt{out}}
\newcommand{\inp}{\texttt{in}}

\theoremstyle{plain}
\newtheorem{theorem}{Theorem}[section]

\newtheorem{lemma}[theorem]{Lemma}
\newtheorem{corollary}[theorem]{Corollary}
\newtheorem{observation}[theorem]{Observation}
\newtheorem{claim}[theorem]{Claim}
\newtheorem{identification rule}{Identification Rule}

\newtheorem*{invariants*}{Invariants}

\theoremstyle{definition}
\newtheorem{definition}[theorem]{Definition}

\theoremstyle{remark}

\newcommand{\defparproblem}[4]{
  \vspace{1mm}
\noindent\fbox{
  \begin{minipage}{0.96\textwidth}
  \begin{tabular*}{\textwidth}{@{\extracolsep{\fill}}lr} #1  & 
\\ \end{tabular*}
  {\bf{Input:}} #2  \\
  {\bf{Question:}} #4
  \end{minipage}
  }
  \vspace{1mm}
}

\newenvironment{claimproof}{\begin{proof}\renewcommand{\qedsymbol}{\claimqed}}{\end{proof}\renewcommand{\qedsymbol}{\plainqed}}
\let\plainqed\qedsymbol

\newcommand{\dsn}{{\sc Planar DSN}\xspace}
\newcommand{\scssfull}{{\sc Planar Strongly Connected Steiner Network}\xspace}

\newcommand{\Dsnfull}{{\sc Planar $\calD$-Steiner Network}\xspace}
\newcommand{\Dsn}{{\sc Planar $\calD$-SN}\xspace}
\newcommand{\Dsngenfull}{{\sc $\calD$-Steiner Network}\xspace}
\newcommand{\Dsngen}{{\sc $\calD$-SN}\xspace}

\newcommand{\toughpair}{tough-pair}

\newcommand{\head}{\text{head}}
\newcommand{\tail}{\text{tail}}

\newcommand{\oldC}{\mathcal{A}}

\newcommand{\aux}[1]{\textnormal{\texttt{Aux}}_{#1}}
\newcommand{\col}{\texttt{col}}
\newcommand{\g}{g}

\newcommand{\contr}{\textnormal{\texttt{contr}}}
\newcommand{\Dcontr}{D_{\contr}}

\newcommand{\bip}{\texttt{bip}}

\newcommand{\sss}{\texttt{s}}
\newcommand{\ttt}{\texttt{t}}

\usetikzlibrary{arrows}
\usetikzlibrary{arrows.meta}
\usetikzlibrary{decorations.pathmorphing}
\usetikzlibrary{decorations.markings}
\usetikzlibrary{calc}
\usetikzlibrary{positioning}

\definecolor{green}{RGB}{18, 122, 22}
\definecolor{red}{RGB}{219, 7, 7}
\definecolor{blue}{RGB}{0,0,204}
\tikzset{
  circ/.style = {circle,draw,fill,inner sep=1.3pt},
  mcirc/.style = {circle,draw,fill,inner sep=1pt},
  circR/.style = {circle,draw=red,fill=red,text=red,inner sep=1.3pt},
  circG/.style = {circle,draw=green,fill=green,text=green,inner sep=1.3pt},
  circB/.style = {circle,draw=blue,fill=blue,text=blue,inner sep=1.3pt},
  circb/.style = {circle,draw=blue,fill=blue,text=blue,inner sep=1.1pt},
  circr/.style = {circle,draw=red,fill=red,inner sep=1pt},
  scirc/.style = {circle,draw,fill,inner sep=.8pt},
  invisible/.style = {draw=none,inner sep=0pt,font=\tiny},
  nonedge/.style={decorate,decoration={snake,amplitude=.3mm,segment length=1mm},draw}
}

\newcommand{\executeiffilenewer}[3]{
\ifnum\pdfstrcmp{\pdffilemoddate{#1}}
{\pdffilemoddate{#2}}>0
{\immediate\write18{#3}}\fi
}

\newcommand{\svg}[2]{\def\svgwidth{#1}
\executeiffilenewer{#2.svg}{#2.pdf}
{inkscape -z -D --file=#2.svg 
--export-pdf=#2.pdf --export-latex}
{\input{#2.pdf_tex}}}

\newcommand{\hardbiclique}[1]{$#1$-hard-biclique-pattern}
\newcommand{\hardmatching}[1]{$#1$-hard-matching-pattern}
\newcommand{\hardpattern}[1]{$#1$-hard-pattern}

\newcommand{\htt}{h(t)}
\newcommand{\corr}{\textnormal{\texttt{cor}}}
\newcommand{\cor}[1]{\corr(#1)}

\newcommand{\aaa}{\texttt{a}}
\newcommand{\bbb}{\texttt{b}}

\title{Subexponential Parameterized Directed Steiner Network Problems on Planar Graphs: a Complete Classification}

\author{Esther Galby\thanks{CISPA Helmholtz Center for Information Security, Germany. \texttt{esther.galby@cispa.de}} \and S\' andor Kisfaludi-Bak\thanks{Department of Computer Science, Aalto University, Espoo, Finland. \texttt{sandor.kisfaludi-bak@aalto.fi}} \and D\' aniel Marx\thanks{CISPA Helmholtz Center for Information Security, Germany. \texttt{marx@cispa.de}}\and  Roohani Sharma\thanks{Max Planck Institute for Informatics, Saarland Informatics Campus, Saarbr\"ucken, Germany.  \texttt{rsharma@mpi-inf.mpg.de}}}

\date{}
\begin{document}
 \maketitle

\thispagestyle{empty}
\begin{abstract}
 In the \textsc{Directed Steiner Network} problem, the input is a directed graph $G$, a set $T\subseteq V(G)$ of $k$ terminals, and a demand graph $D$ on $T$. The task is to find a subgraph $H\subseteq G$ with the minimum number of edges such that for every $(s,t)\in E(D)$, the solution $H$ contains a directed $s\to t$ path. 
The goal of this paper is to investigate how the complexity of the problem depends on the demand pattern in \emph{planar graphs.} Formally, if $\mathcal{D}$ is a class of directed graphs closed under identification of vertices, then the \Dsngenfull\ (\Dsngen) problem is the special case where the demand graph $D$ is restricted to be from $\mathcal{D}$.
For general graphs, Feldmann and Marx \cite{FeldmannM16} characterized those families of demand graphs where the problem is fixed-parameter tractable (FPT) parameterized by the number $k$ of terminals. 
They showed that if $\calD$ is a superset of one of five hard families, then \Dsngen is W[1]-hard parameterized by $k$, otherwise it can be solved in time $f(k)\cdot n^{O(1)}$.

For planar graphs, besides the existence of an FPT algorithm, it is also an interesting question whether the W[1]-hard cases can be solved by subexponential parameterized algorithms. 
For example, Chitnis et al.~\cite{ChitnisFHM20} showed that, assuming the Exponential-Time Hypothesis (ETH), there is no $f(k) \cdot n^{o(k)}$ time algorithm 
for the general \Dsngenfull problem on planar graphs, but the special case called \textsc{Strongly Connected Steiner Subgraph} 
(where the demand graph $D$ is a bidirected clique) can be solved in time $f(k) \cdot n^{O(\sqrt{k})}$ on planar graphs.
 We present a far-reaching generalization and unification of these two results: we give a complete characterization of the behavior of every \Dsngen problem on planar graphs. We classify every class $\calD$ closed under identification of vertices into three cases: assuming ETH, either the problem is
\begin{enumerate}
\item solvable in time $2^{O(k)} \cdot n^{O(1)}$, i.e., FPT parameterized by the number $k$ of terminals, but not solvable in time $2^{o(k)} \cdot n^{O(1)}$,
\item solvable in time $f(k) \cdot n^{O(\sqrt{k})}$, but cannot be solved in time $f(k) \cdot n^{o(\sqrt{k})}$, or
\item solvable in time $f(k) \cdot n^{O(k)}$, but cannot be solved in time $f(k) \cdot n^{o({k})}$.
  \end{enumerate}
  We show that the FPT cases (Case 1) are the same as in the case of general graphs: $\calD$ needs to exclude the same five families of hard graphs. We further identify a finite number of hard families that $\calD$ needs to exclude if we want to solve \Dsngen on planar graphs in time $f(k) \cdot n^{O(\sqrt{k})}$ (Case 2). As an important step of our lower bound proof, we discover that, assuming ETH, \Dsngen on planar graphs has no $f(k) \cdot n^{o(k)}$ time algorithm where $\calD$ is the class of all directed bicliques. This corresponds to the following simple problem: given two sets of terminals $S$ and $T$ with $|S|+|T|=k$, find a subgraph with minimum number of edges such that every vertex of $T$ is reachable from every vertex of $S$. Our result gives a rare example of a genuinely planar problem that cannot be solved in time $f(k) \cdot n^{o(k)}$.

\end{abstract}
\newpage
\tableofcontents
\newpage
\setcounter{page}{1}

\section{Introduction}\label{sec:introduction}

Finding Steiner trees and related network design problems were intensively studied in undirected graphs, directed graphs, and planar graphs, from the viewpoint of approximation and parameterized algorithms \cite{ChitnisFHM20,eiben_et_al,DBLP:conf/stacs/DvorakFKMTV18,rajesh-andreas-pasin,khandekar,FeldmannM16,DBLP:journals/iandc/BermanBMRY13,DBLP:journals/siamdm/GuoNS11,feldman-ruhl,DBLP:journals/jal/CharikarCCDGGL99,winter1987steiner,hakimi,DBLP:journals/jacm/BateniHM11,DBLP:conf/soda/BateniCEHKM11,DBLP:conf/focs/MarxPP18,karp1972reducibility,ramanathan1996multicast,salama1997evaluation,Li1992267,Natu1997207,daniel-grid-tiling,levin}. The simplest problem of this type is \textsc{Steiner Tree}, where given a graph $G$ and set $T\subseteq V(G)$ of terminals, the task is to find a tree with smallest number of edges that contains every terminal. This problem models a network-design scenario where the terminals need to be connected to each other with a network of minimum cost. \textsc{Steiner Forest} is the generalization where we do not require connection between every pair of terminals, but have to satisfy a given set of demands. Formally, the input of \textsc{Steiner Forest} is a graph $G$ with pairs of vertices $(s_1,t_1)$, $\dots$, $(s_d,t_d)$, the task is to find a subgraph with the minimum number of edges that satisfies every request, that is, $s_i$ and $t_i$ are in the same component of the solution for every $i\in [d]$. 

On directed graphs, \textsc{Directed Steiner Tree (DST)} is defined by specifying one of the terminals in $T$ to be the root and the task is to find a subgraph with the smallest number of edges such that there is path from the root to every terminal in the solution. This problem models a scenario where we need to construct a network where the root can broadcast to every other terminal. An equally natural network design problem on directed graphs is the \textsc{Strongly Connected Steiner Subgraph (SCSS)} problem, where given a directed graph $G$ and a set $T\subseteq V(G)$ of terminals, the task is to find a subgraph with the smallest number of edges where $T$ is in a single strongly connected component, or in other words, the solution contains a path from every terminal to every other terminal. The directed variant of \textsc{Steiner Forest} generalizes both of these problems: in \textsc{Directed Steiner Network (DSN)}, the
input is a digraph $G$ with pairs of vertices $(s_1,t_1)$, $\dots$, $(s_d,t_d)$, and the task is to find a subgraph with the minimum number of edges that has an $s_i\to t_i$ path for every $i\in [d]$.

\paragraph{Planar graphs.} A well-known phenomenon on planar graphs is that the running time of parameterized algorithms for typical NP-hard problems have exponential dependence on $O(\sqrt{k})$, where $k$ is the parameter, and this dependence is best possible assuming the Exponential-Time Hypothesis (ETH)
\cite{DBLP:conf/focs/MarxPP18,ChitnisFHM20,DBLP:conf/focs/FominLMPPS16,DBLP:journals/talg/MarxP22,DBLP:conf/soda/KleinM14,DBLP:conf/icalp/KleinM12,DBLP:conf/icalp/Marx12,DBLP:conf/fsttcs/LokshtanovSW12,DBLP:journals/algorithmica/Verdiere17,DBLP:conf/stoc/Nederlof20a,10.5555/2815661}. All three of \textsc{Directed Steiner Tree},  \textsc{Strongly Connected Steiner Subgraph}, 
and \textsc{Directed Steiner Network} remain NP-hard on planar graphs. However, they behave very differently from the viewpoint of parameterized complexity: the dependence of the running time on the number $k$ of terminals is very different.
\medskip

\mdfdefinestyle{highlight}{frametitlebackgroundcolor=gray!40,backgroundcolor=gray!20,roundcorner=10pt}
\begin{mdframed}[style=highlight,frametitle={Our starting point}] 
     
    \begin{enumerate}    
\item \textsc{Planar DST} can be solved in time $2^k\cdot n^{O(1)}$ \cite{DBLP:conf/stoc/BjorklundHKK07}, but cannot be solved in time $2^{o(k)}\cdot n^{O(1)}$ \cite{DBLP:conf/focs/MarxPP18}, assuming the ETH.
\item \textsc{Planar SCSS} can be solved in time $2^{O(k\log k)}\cdot n^{O(\sqrt{k})}$ \cite{ChitnisFHM20}, but has no algorithm with running time $f(k) \cdot n^{o(\sqrt{k})}$ for any function $f$, assuming the ETH \cite{ChitnisFHM20}.
\item \textsc{Planar DSN} can be solved in time $f(k) \cdot n^{O({k})}$ \cite{eiben_et_al}, but has no algorithm with running time $f(k) \cdot n^{o({k})}$ for any function $f$, assuming the ETH \cite{ChitnisFHM20}.
\end{enumerate}
\end{mdframed}
\medskip

The goal of this paper is to put these results into the context of a wider landscape of directed network design problems. We systematically explore other special cases of  \textsc{Directed Steiner Network} and determine their behavior on planar graphs. Our main result is showing that every special case behaves similarly to one of these three problems: assuming ETH, the best possible running time is of the form $2^{O(k)}\cdot n^{O(1)}$, $f(k) \cdot n^{O(\sqrt{k})}$, or $f(k)\cdot n^{O(k)}$. Furthermore, we provide an exact combinatorial characterization of the problems belonging to the three classes.
\clearpage

\paragraph{Dichotomy for general graphs.}
We explore the different special cases of \textsc{Directed Steiner Network} on planar graphs in a framework similar to how Feldmann and Marx~\cite{FeldmannM16} treated the problem on general graphs. We can define various special cases of \textsc{Directed Steiner Network} by looking at what kind of graph the connection demands define on the terminals: it is an out-star for \textsc{Directed Steiner Tree}, a bidirected clique for \textsc{Strongly Connected Steiner Subgraph}, and a matching for \textsc{Directed Steiner Network}. More generally, for every class $\calD$ of directed graphs, we investigate the problem where the pattern of demands has to belong to the class $\calD$. Our goal is to understand how the graph-theoretic properties of the members of $\calD$ influence the resulting special case of \textsc{Directed Steiner Tree}.

Formally, for every class $\calD$, Feldmann and Marx~\cite{FeldmannM16} defined the restriction of the problem in the following way.
\medskip

\defparproblem{\Dsngenfull}{Digraph $G$, a set of $k$ terminals $T\subseteq V(G)$, and a demand digraph $D\in \calD$ with vertex set $T$.}{$k$}{What is the minimum number of edges in a subgraph $H$ of $G$ where for each $(u,v)\in E(D)$ there is a $u\rightarrow v$ path in $H$?}
\medskip

Note that only the transitive closure of $D$ matters for the problem: if $D_1$ and $D_2$ have the same transitive closure, then having $D_1$ or $D_2$ in the input results in exactly the same problem. Therefore, it makes sense to consider only classes $\calD$ that are \emph{closed under transitive equivalence,} that is, if $D_1$ and $D_2$ have the same transitive closure and $D_1\in\calD$, then $D_2\in \calD$ as well. Another natural assumption is that $\calD$ is {\em closed under identifying vertices.} That is, if $G\in \calD$ and $G'$ is obtained by merging two vertices $x,y\in V(G)$ to a single vertex whose in- and out-neighbors are the union of the in- and out-neighbors of $x$ and $y$, then $G'$ is also in $\calD$. This closure property models the extension of the problem where we can put multiple terminals at the same vertex and we parameterize by the number of vertices that have terminals.

Feldmann and Marx~\cite{FeldmannM16} characterized those classes $\calD$ closed under transitive equivalence and identifying vertices where \Dsngenfull is fixed-parameter tractable (FPT) parameterized by the number of terminals, that is, can be solved in time $f(k)\cdot n^{O(1)}$. They identified five classes of graphs that prevent the problem from being FPT. A \emph{pure out-diamond} is a complete bipartite graph $K_{2,t}$ directed from the 2-element side to the $t$-element side. A \emph{flawed out-diamond} has in addition  a vertex $v$ and edges going from $v$ to the 2-element side. The pure in-diamond and flawed in-diamond are defined similarly by reversing the orientation of the edges. 
Let us denote by $\oldC_1$, $\oldC_2$, $\dots$, $\oldC_5$ the class of all pure out-diamonds, flawed out-diamonds,
pure in-diamonds, flawed in-diamonds, and directed cycles, respectively.
\begin{theorem}[Feldmann and Marx~\cite{FeldmannM16}]\label{thm:gengraphs}
  Let $\calD$ be a class of graphs closed under transitive equivalence and identifying vertices.
  \begin{enumerate}
\item \textbf{FPT:}   If $\oldC_i\not\subseteq \calD$ for any $i\in [5]$, then \Dsngenfull can be solved in time $2^{O(k)}n^{O(1)}$, where $k$ is the number of terminals.
\item \textbf{Hard:}   If $\oldC_i\subseteq\calD$ for some $i\in [5]$, then \Dsngenfull is W[1]-hard parameterized by the number $k$ of terminals.
  \end{enumerate}
\end{theorem}
The first part of Theorem~\ref{thm:gengraphs} was proved by a combination of an algorithm that solves the problem in time $2^{O(kw\log w)}\cdot n^{O(w)}$ if there is an optimum solution with treewidth $w$ and a combinatorial result showing that if $\calD$ is not the superset of $\oldC_i$ for any $i\in[5]$, then there is a constant bound on the treewidth of optimum solutions. The second part follows from a W[1]-hardness result for each of the five classes $\oldC_i$.

\paragraph{Our result: trichotomy for planar graphs.}
Our main result classifies \Dsnfull into three levels of complexity: 
$2^{O(k)}\cdot n^{O(1)}$, $f(k)\cdot n^{O(\sqrt{k})}$, or $f(k)\cdot n^{O(k)}$ time.
In light of Theorem~\ref{thm:gengraphs} and the earlier results on planar graphs, there are three natural questions that arise:
\begin{enumerate}
\item Are there cases that are FPT on planar graphs, but W[1]-hard on general graphs?
  \item Are there subexponential FPT cases on planar graphs, that is, where the running time is $2^{o(k)}\cdot n^{O(1)}$?
\item Where is the boundary line between the $f(k) \cdot n^{O(\sqrt{k})}$ and $f(k)\cdot n^{O(k)}$ cases?
\end{enumerate}
We answer the first question negatively: the hard cases remain hard on planar graphs. The answer to the second question is also negative: we show that every (nontrivial) case of \Dsn is at least as hard as \textsc{Directed Steiner Tree}, hence a known lower bound \cite{DBLP:conf/focs/MarxPP18} shows that there is no subexponential FPT algorithm, assuming ETH. 
To answer the third question, we define a finite number $\kappa\le 10000$ of classes $\mathcal{C}_i$, $i\in[\kappa]$, and show that these are precisely the classes of patterns that prevent subexponential $f(k) \cdot n^{O(\sqrt{k})}$ time algorithms.
\medskip

\begin{mdframed}[style=highlight,frametitle={Our main result}] 
\begin{theorem}\label{thm:main-intro}
  Let $\calD$ be a class of directed graphs closed under transitive equivalence and identifying vertices where the number of edges is not bounded.
  \begin{enumerate}
  \item \textbf{FPT:}   If $\oldC_i\not\subseteq \calD$ for any $i\in [5]$, then \Dsnfull
    \begin{itemize}
    \item[(i)] can be solved in time $2^{O(k)}\cdot n^{O(1)}$,
    \item[(ii)] but has no $2^{o(k)}\cdot n^{O(1)}$ time algorithm assuming the ETH.
    \end{itemize}
  \item \textbf{Subexponential:}   If $\oldC_i\subseteq\calD$ for some $i\in [5]$, but $\mathcal{C}_i\not\subseteq \calD$ for any $i\in[\kappa]$, then \Dsnfull
    \begin{itemize}
    \item[(iii)] can be solved in time $f(k) \cdot n^{O(\sqrt{k})}$,
      \item[(iv)] but has no
        $f(k) \cdot n^{o(\sqrt{k})}$ time algorithm assuming the ETH.
      \end{itemize}
    \item \textbf{Hard:} If $\mathcal{C}_i\subseteq \calD$ for some $i\in [\kappa]$, then \Dsnfull
      \begin{itemize}
      \item[(v)]  can be solved in time $f(k)\cdot n^{O({k})}$,
        \item[(vi)] but has no $f(k) \cdot n^{o({k})}$ time algorithm assuming the ETH.
        \end{itemize}
      \end{enumerate}
    \end{theorem}
  \end{mdframed}

\paragraph{Hard classes.}
Let us define now the graph classes $\mathcal{C}_i$ representing the hard-patterns. Given a digraph $G$ and a set $X \subseteq V(G)$, an \emph{$X$-source} is a vertex $s \in V(G) \setminus X$ such that $N^+(s) = X$.
Similarly, an \emph{$X$-sink} is a vertex $t \in V(G) \setminus X$ such that $N^-(t) = X$.
The first 4 classes $\mathcal{C}_1$, $\dots$, $\mathcal{C}_4$ are defined by extending a biclique.

\begin{definition}[\hardbiclique{t}]
  \label{def:cleanedbiclique}
A \emph{\hardbiclique{t}} is an (acyclic) digraph $D$ constructed in the following way. We start with two disjoint sets $A$ and $B$ with  $|A| = |B| = t$ and introduce every edge from $A$ to $B$. Furthermore, we introduce into $D$ any combination of the following items (see \Cref{fig:cleanedbiclique}): 
\begin{enumerate}
\item an $A$-source;
\item a $B$-sink.
\end{enumerate}
In particular, there are $2 \cdot 2$ types of $t$-hard-biclique patterns: we let $\mathcal{C}_1,\ldots,\mathcal{C}_4$ be the 4 classes that each contain all the $t$-hard-biclique-patterns of a specific type for every $t$.
\end{definition}

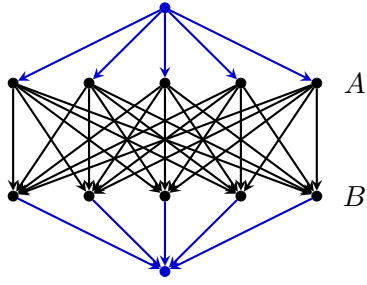
\begin{figure}
\centering
\begin{tikzpicture}
\foreach \i in {1,...,5}
{
\node[circ] (b\i) at (\i,0) {};
}
\node[draw=none] at (5.5,0) {$B$};

\foreach \i in {1,...,5}
{
\node[circ] (a\i) at (\i,1.5) {};
}
\node[draw=none] at (5.5,1.5) {$A$};

\foreach \i in {1,...,5}
\foreach \j in {1,...,5}
{
\draw[->,>=stealth,thick] (a\i) -- (b\j);
}

\node[circB] (t) at (3,-1) {};
\foreach \i in {1,...,5}
{
\draw[->,>=stealth,thick,blue] (b\i) -- (t);
}

\node[circB] (us) at (3,2.5) {};
\foreach \i in {1,...,5}
{
\draw[->,>=stealth,thick,blue] (us) -- (a\i);
}
\end{tikzpicture}
\caption{The 5-hard biclique patterns: each blue vertex may  or may not be present.}
\label{fig:cleanedbiclique}
\end{figure}

The following definition specifies the remaining classes. 
\begin{definition}[\hardmatching{t}]
\label{def:cleanedorderedtoughpair}
A \emph{\hardmatching{t}} is
an (acyclic) digraph $D$ constructed the following way. We start with disjoint vertex sets $W = \{w_1,\ldots,w_t\}$, $X = \{x_1,\ldots,x_t\}$, $Y = \{y_1,\ldots,y_t\}$ and $Z = \{z_1,\ldots,z_t\}$ and introduce the edges $w_ix_i$ and $y_iz_i$ for every $i\in[t]$.
Furthermore, we introduce into $D$ any combination of the following items:
\begin{enumerate}
\item either the directed path $w_1 \rightarrow w_2 \rightarrow \ldots \rightarrow w_t \rightarrow z_1 \rightarrow z_2 \rightarrow \ldots \rightarrow z_t$, or any of the directed paths $w_1 \rightarrow w_2 \rightarrow \ldots \rightarrow w_t$ and 
$z_1 \rightarrow z_2 \rightarrow \ldots \rightarrow z_t$;
\item either the directed path $y_1 \rightarrow y_2 \rightarrow \ldots \rightarrow y_t \rightarrow x_1 \rightarrow x_2 \rightarrow \ldots \rightarrow x_t$, or any of the directed paths $x_1 \rightarrow x_2 \rightarrow \ldots \rightarrow x_t$ and $y_1 \rightarrow y_2 \rightarrow \ldots \rightarrow y_t$;
\item an $S$-source for exactly one $S \in \{W,X,Y,Z,W \cup Y,X \cup Z,X \cup Y,W \cup Z\}$;
\item an $S$-sink for exactly one $S \in \{W,X,Y,Z,W \cup Y,X \cup Z,X \cup Y,W \cup Z\}$;
\item a vertex $r_{WZ}$ such that $N^-(r_{WZ}) = W$ and $N^+(r_{WZ}) = Z$;
\item a vertex $r_{YX}$ such that $N^-(r_{YX}) = Y$ and $N^+(r_{YX}) = X$.
\end{enumerate}
In particular, there are $5 \cdot 5 \cdot 9 \cdot 9 \cdot 2 \cdot 2$ types of $t$-hard matching patterns: we let $\mathcal{C}_5, \ldots, \mathcal{C}_{8104}$ be the 8100 classes that each contain all the $t$-hard-matching-patterns of a specific type for every~$t$.
\end{definition}
Note that some of these classes are isomorphic. For example, adding the path $x_1\to x_t$ or the path $z_1\to z_t$ lead to isomorphic graphs. If we just consider the graph classes where we choose not to add a source, sink, vertex $r_{WZ}$, or vertex $r_{YX}$, then we have $15$ nonisomorphic classes, as shown in Figure~\ref{fig:cleanedorderedtoughpair}. One could think of $t$-hard-matching-patterns as (the transitive closure of) one of these graphs, potentially extended by appropriate sources and sinks.

\newcommand{\twopaths}[8]{
\foreach \i in {1,...,4}
{
\pgfmathsetmacro{\x}{0.75*\i+#1}
\pgfmathsetmacro{\y}{#2+.75}
\node[circ] (a\i) at (\x,\y) {};
\node[circ] (b\i) at (\x,#2) {};
\draw[->,>=stealth,thick] (a\i) -- (b\i);
}

\foreach \i in {1,...,4}
{
\pgfmathsetmacro{\x}{0.75*\i+#1+3.25}
\pgfmathsetmacro{\z}{#2+.75}
\node[circ] (c\i) at (\x,\z) {};
\node[circ] (d\i) at (\x,#2) {};
\draw[->,>=stealth,thick] (c\i) -- (d\i);
}

\ifthenelse{#4=1}{
\foreach \i in {1,...,3}
{
\pgfmathtruncatemacro{\j}{\i+1}
\draw[->,>=stealth,thick] (a\i) -- (a\j);
}
}{}

\ifthenelse{#5=1}{
\foreach \i in {1,...,3}
{
\pgfmathtruncatemacro{\j}{\i+1}
\draw[->,>=stealth,thick] (d\i) -- (d\j);
}
}{}

\ifthenelse{#7=1}{
\foreach \i in {1,...,3}
{
\pgfmathtruncatemacro{\j}{\i+1}
\draw[->,>=stealth,thick] (b\i) -- (b\j);
}
}{}

\ifthenelse{#8=1}{
\foreach \i in {1,...,3}
{
\pgfmathtruncatemacro{\j}{\i+1}
\draw[->,>=stealth,thick] (c\i) -- (c\j);
}
}{}

\ifthenelse{#3=1}{
\draw[->,>=stealth,thick] (a4) -- (d1);
}{}

\ifthenelse{#6=1}{
\draw[->,>=stealth,thick] (c4) -- (#1+6.6,#2+.75) -- (#1+6.6,#2-.35) -- (#1 + .4,#2-.35) -- (#1 + .4,#2) -- (b1);
}{}
}

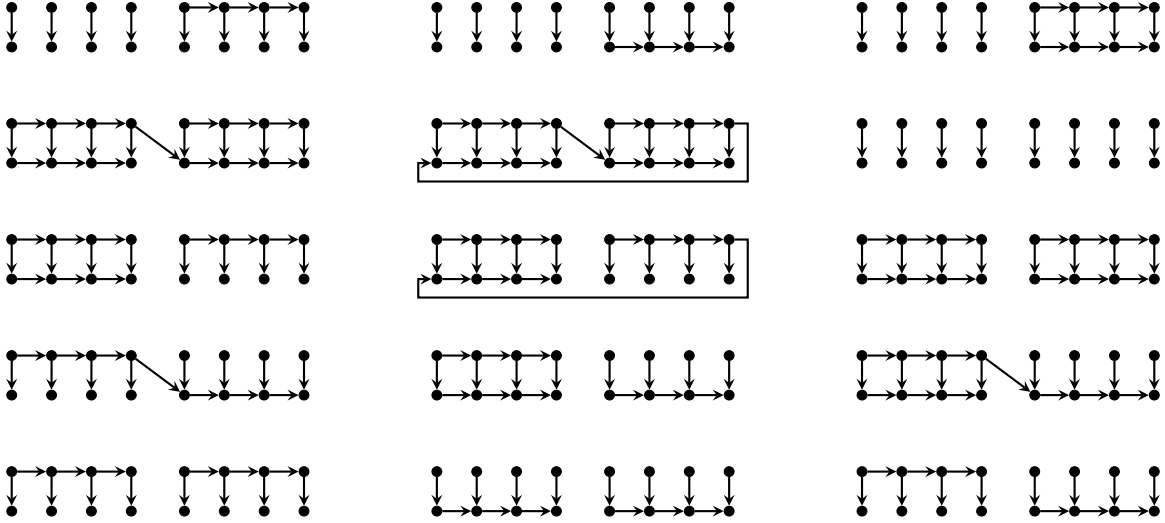
\begin{figure}[t]
\centering
\begin{tikzpicture}[x=0.7cm,y=0.7cm]
\twopaths{0}{0}{0}{1}{0}{0}{0}{1};
\twopaths{8}{0}{0}{0}{1}{0}{1}{0};
\twopaths{16}{0}{0}{1}{1}{0}{0}{0};
\twopaths{0}{2.2}{1}{1}{1}{0}{0}{0};
\twopaths{8}{2.2}{0}{1}{1}{0}{1}{0};
\twopaths{16}{2.2}{1}{1}{1}{0}{1}{0};
\twopaths{0}{4.4}{0}{1}{0}{0}{1}{1};
\twopaths{8}{4.4}{0}{1}{0}{1}{1}{1};
\twopaths{16}{4.4}{0}{1}{1}{0}{1}{1};
\twopaths{0}{6.6}{1}{1}{1}{0}{1}{1};
\twopaths{8}{6.6}{1}{1}{1}{1}{1}{1};
\twopaths{16}{6.6}{0}{0}{0}{0}{0}{0};
\twopaths{0}{8.8}{0}{0}{0}{0}{0}{1}{0};
\twopaths{8}{8.8}{0}{0}{1}{0}{0}{0}{0};
\twopaths{16}{8.8}{0}{0}{1}{0}{0}{1}{0};
\end{tikzpicture}
\caption{The 4-hard matching patterns (without source, sink, $r_{WZ}$, or $r_{YX}$).}
\label{fig:cleanedorderedtoughpair}
\end{figure}

Finally, we define a $t$-hard pattern as any of the patterns defined above.
\begin{definition}[\hardpattern{t}]

\label{def:cleanedtoughpair}
A {\em \hardpattern{t}} is either a \hardbiclique{t} or \hardmatching{t}.
\end{definition}

 \subsection{Overview of our main result}

Observe that Theorem~\ref{thm:main-intro} consists of six statements. Let us briefly discuss how these six statements are proved. Note that some of these statements follow from known results, while for others we need to do a substantial amount of new technical work. The proofs of statments \textit{(iii)} and \textit{(vi)} form the main technical part of the paper (see Figure~\ref{fig:chain}).

\subsubsection*{Statement \textit{(i)}}
  
  The FPT result \textit{(i)} follows directly from Theorem~\ref{thm:gengraphs} (here the surprising aspect is that, by statement \textit{(iv)}, there are no further FPT cases).

\subsubsection*{Statement \textit{(ii)}}

  The lower bound \textit{(ii)} follows by observing that every relevant class contains either all in-stars or all out-stars, hence the lower bound for \textsc{Directed Steiner Tree} \cite{DBLP:conf/focs/MarxPP18} applies. To avoid triviality, we need to assume that the class contains graphs with arbitrarily large number of edges.
  \begin{lemma}\label{lem:findstar}
    Let $\calD$ be a class of graphs closed under identifying vertices and transitive closures where the number of edges of the graphs is not bounded. Then one of the following holds:
\begin{itemize}  \item $\calD$ contains every directed cycle,
  \item $\calD$ contains every out-star, or
  \item $\calD$ contains every in-star.
  \end{itemize}
\end{lemma}
In statement \textit{(ii)} of Theorem~\ref{thm:main-intro}, we assume that $\oldC_i\not\subseteq \calD$, and $\oldC_5$ is the class of all directed cycles. Thus $\calD$ contains either every out-star or every in-star.

\subsubsection*{Statement \textit{(iii)}}

  Our main technical result is proving statement \textit{(iii)}: the existence of a
$f(k)\cdot n^{O(\sqrt{k})}$ time algorithm if
$\mathcal{C}_i\not\subseteq \calD$ for any $i\in[\kappa]$ (in the following subsection, we give a more detailed description of the proof). This
algorithm is obtained by showing that the treewidth of the
optimal solution is always $O(\sqrt{k})$ under these
conditions. Then we can use the following result of Feldmann and Marx \cite{FeldmannM16}.

  \begin{theorem}[Theorem 1.5 of~\cite{FeldmannM16}]\label{thm:algo_tw}
    If an instance $(G,T,D)$ of \textsc{Directed Steiner Network} has an optimum solution $H$ of treewidth $w$, then it can be solved in $2^{O(kw \log w)} \cdot n^{O(w)}$ time. 
  \end{theorem}
  
  Note that this is a slightly weaker form of the statement, with a simplified bound on the running time. With Theorem~\ref{thm:algo_tw} at hand, our main goal is to prove that every optimum solution of \Dsn has treewidth $O(\sqrt{k})$ if $\mathcal{C}_i\not\subseteq \calD$ for any $i\in [\kappa]$.
  
Towards proving this bound, we first translate the
question to a problem on acyclic graphs: it is sufficient to show that
if the solution is acyclic, then the total degree of the branch
vertices (i.e., of degree $>2$) is $O(k)$. More formally,
for a vertex $v$ of a digraph, let $d^*(v)$ denote the \emph{branch degree} of $v$, defined as 
\[d^*(v)=\max(d^+(v)+d^-(v)-2,0),\]
where $d^+(v)$ and $d^-(v)$ denotes the out- and in-degree of $v$, respectively. The total branch degree of a graph $G$ is the sum of all branch degrees of the vertices of $G$.

We say that a feasible solution $H$ of $(G,T,D)$ is \emph{edge-minimal} if for all edges $e\in E(H)$ the graph $H-e$ is not feasible. An edge $e$ is \emph{essential} for some demand edge $tt'\in E(D)$ if there is no $t\rightarrow t'$ path in $H-e$. Note that all edges of an edge-minimal graph $H$ are essential for some demand edge of $D$.
We say that a pattern class $\calD$ is \emph{$c$-bounded} for some $c=O(1)$ if for any instance $(G,T,D)$ \Dsnfull where $G,D$ are acyclic, and any edge-minimal solution $H$, the total branch degree of $H$ is at most $ck$.

The next theorem moves the problem to the domain of acyclic digraphs:
what we need now is a linear bound on the total branch degree of acyclic solutions.

\begin{theorem}\label{thm:subexpalgo}
  If the pattern class $\calD$ is $c$-bounded for some $c=O(1)$, then for any instance of \Dsnfull with $|T|=k$ the solution graph $H$ has treewidth $O(\sqrt{k})$.
\end{theorem}

Applying \Cref{thm:algo_tw} implies that $c$-bounded classes have the desired subexponential algorithm, but we still need to establish a link between non-$c$-bounded classes and $t$-hard-patterns. First, we argue that if the total
branch degree is too large, then a grid-like structure
can be found in the solution. The grid-like structure appears in the
solution to satisfy a set of edges in the demand graph $D$, and this set of demands form
a certain hard structure in the demand pattern that we call the $t$-tough-pair which we define informally here (see Definition~\ref{def:ttough} for a formal definition). We say that two edges $e_1$ and $e_2$ are \emph{weakly independent} if there is no directed path from the head of one to the tail of the other. Edges $e_1$ and $e_2$ are \emph{strongly independent} if, in addition to being weakly independent, there is no directed path containing the heads of both edges and there is no directed path containing the tails of both edges. An edge $e$ is {\em minimal} in a digraph $D$ if there is no path from the head of $e$ to the tail of $e$ avoiding $e$. Let $E_1\cup E_2$ be a vertex-disjoint set of minimal edges with $|E_1|=|E_2|=t$. We say that $(E_1,E_2)$ is a \emph{$t$-\toughpair}~if
\begin{itemize}
\item any two edges in $e,e'\in E_1$ are weakly independent,
\item any two edges in $e,e'\in E_2$ are weakly independent, and
\item any two edges $e_1\in E_1$ and $e_2\in E_2$ are strongly independent.
\end{itemize}
  Observe that in particular the two matchings in a \hardmatching{t} form  (vertical edges in Figure~\ref{fig:cleanedorderedtoughpair}) a $t$-\toughpair.
  Similarly, taking two vertex-disjoint matchings of size $t$ each in a \hardbiclique{t} is also a $t$-\toughpair.
  
Our main structure theorem connects the total branch degree to the existence of these kind of hard structures.

  \begin{restatable}[Structure Theorem]{theorem}{structurethm}\label{thm:maintechnical}
Let $\calD$ be a class of graphs closed under identifying vertices and transitive equivalence. Then either $\calD$ has a pattern with a $t$-\toughpair\ for each positive integer $t$, or it is $c$-bounded for some constant $c$.
\end{restatable}
  
Theorems~\ref{thm:subexpalgo} and Theorem~\ref{thm:maintechnical} show that the existence of arbitrarily large $t$-tough-pairs is the canonical reason why treewidth is not $O(\sqrt{k})$. The lower bounds ruling out $f(k) \cdot n^{o(k)}$ time algorithms essentially rely on the existence of $t$-tough-pairs. However, the existence of a $t$-\toughpair\ in a demand pattern $D\in \calD$ is not sufficient for the lower bound: the $t$-\toughpair\ could be only a small part of the pattern $D$, and hence the lower bounds may not apply. We show, with heavy use of Ramsey's Theorem and other combinatorial arguments, that whenever a large $t$-\toughpair\ appears in a graph, then the graph can be ``cleaned'': we can identify vertices to obtain one of the $t$-hard-patterns. 
  Therefore, if arbitrary large $t$-tough-pairs appear in the members of a class $\calD$ closed under identifying vertices, then the class is a superset of one of the hard classes $\mathcal{C}_i$.

 \begin{restatable}{theorem}{cleaning}\label{thm:cleaning-intro}
   Let $\calD$ be a class of graphs closed under transitive equivalence and identifying vertices. The following two are equivalent:
    \begin{enumerate}
    \item For every $t$, there is a $D\in\calD$ that has a $t$-tough pair.
      \item $\mathcal{C}_i\subseteq \calD$ for some $i\in[\kappa]$.
    \end{enumerate}
\end{restatable}

  We can conclude that if $\calD$ is not the superset of $\mathcal{C}_i$ for any $i\in [\kappa]$, then the treewidth of the optimum solution is $O(\sqrt{k})$, implying that \dsn can be solved in time $f(k) \cdot n^{O(\sqrt{k})}$. 

\subsubsection*{Statement \textit{(iv)}}

  The statement \textit{(iv)} ruling out   $f(k) \cdot n^{o(\sqrt{k})}$ time algorithms follows from the known lower bound for \textsc{Strongly Connected Steiner Subgraph} (i.e., $\calD=\oldC_5$) \cite{ChitnisFHM20} and from reproving the W[1]-hardness of diamonds (i.e., $\calD\in \{\oldC_1,\oldC_2,\oldC_3,\oldC_4\}$) for planar graphs \cite{FeldmannM16}. Compared to the W[1]-hardness on general graphs, the proof for planar graphs is more involved. As it is very usual for planar problems, we establish these lower bounds by reducing from \textsc{$k\times k$-Grid Tiling}, which cannot be solved in time $f(k) \cdot n^{o(k)}$, assuming ETH \cite{10.5555/2815661}. For statement \textit{(iv)}, we need to reduce from \textsc{$\sqrt{k}\times \sqrt{k}$ Grid Tiling} to a \dsn with $O(k)$ terminals forming a pure/flawed in/out-diamond pattern, ruling out $f(k) \cdot n^{o(\sqrt{k})}$ algorithms for such patterns.

  In all these reductions, we are reusing and extending the gadget constructions from earlier work \cite{ChitnisFHM20}. However, the high-level structure of the reduction is substantially different and depends on the pattern class we are considering. In light of Theorem~\ref{thm:algo_tw}, we should first verify, as a sanity check, that the treewidth of the solution can be sufficiently large, that is, it can be $\Omega(\sqrt{k})$ in case of diamonds.
   Typically, one can expect that examples with sufficiently large treewidth shed some light on how the high-level structure of the hardness proof could like.
   Figure~\ref{fig:largetw} shows that treewidth can be indeed sufficiently large: a $\sqrt{k}\times\sqrt{k}$ grid can be obtained from two ``interlocking combs.''

   \subsubsection*{Statement \textit{(v)}}

 The upper bound $f(k) \cdot n^{O(k)}$ (statement \textit{(v)}) follows from the work of Eiben et al.~\cite{eiben_et_al}, who showed that \dsn with $k$ terminals can be always solved within this running time.

 \subsubsection*{Statement \textit{(vi)}}

 To prove statement \textit{(vi)} ruling out $f(k) \cdot n^{o(k)}$ algorithms, we provide such a lower bound for each class $\mathcal{C}_i$ for $i\in [\kappa]$. Analogously to statement \textit{(iv)}, the proof is by reduction from \textsc{$k\times k$ Grid Tiling} to a \dsn instance with a \hardmatching{k}
 or a \hardbiclique{k}, ruling out $f(k) \cdot n^{o(k)}$ algorithms. Again, let us verify that the treewidth can be sufficiently large: Figure~\ref{fig:largetw} shows how a $k\times k$ grid can appear in the solution to an instance with $k$ terminals.
 
For $t$-hard-matching-patterns, the simplest case is when we have two induced matchings of size~$t$. Then a $t\times t$ grid can arise very easily in the solution if the terminals are on the boundary of a grid. The crucial point here is that the \hardmatching{t} was defined in a way that all the additional paths, sources etc. do not interfere with the grid, see the figure for an example. For the \hardbiclique{t}, there is a non-obvious and highly delicate way of constructing an instance with $2t$ terminals where a $t\times t$ grid appears. Combining these constructions gives the lower bound. 

 \begin{theorem}
 \label{thm:hardlowerbound}
Let $\mathcal{D}$ be a class of graphs closed under identifying vertices and transitive equivalence.
If $\mathcal{C}_i \subseteq \mathcal{D}$ for some $i \in[\kappa]$, then \Dsnfull has no $f(k) \cdot n^{o(k)}$ time algorithm assuming the ETH.
 \end{theorem}

 Let us observe that if $\calD$ consists of bicliques directed from one side to the other, then \Dsn corresponds to the following problem: given a planar digraph $G$ with two sets $S,T\subseteq V(G)$ of terminals with $|S|+|T|=k$, find a subgraph with minimum number of edges such that there is a path from every vertex of $S$ to every vertex of $T$. Our result shows that, assuming ETH, this problem has no $f(k) \cdot n^{o(k)}$ time algorithm. This result is surprising, as the problem can be considered to be \emph{genuinely planar} in the sense that the input is a planar graph with $k$ terminals and a single bit of annotation at each terminal. To our knowledge, this is the first example of a relatively natural planar problem where $f(k) \cdot n^{O(k)}$ is best possible and cannot be improved to $f(k) \cdot n^{O(\sqrt{k})}$.

 \begin{figure}[t]
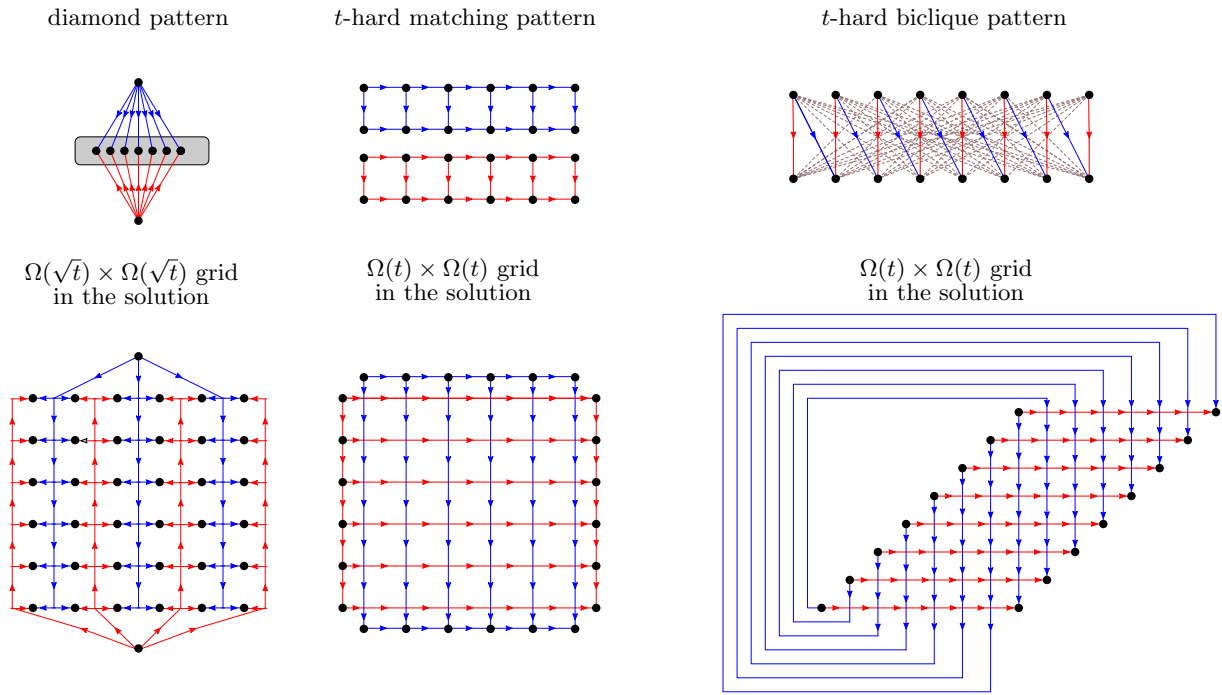

{\footnotesize \svg{\linewidth}{largetw}}
\caption{Pattern graphs (top row) and example minimal solution graphs with large grid patterns and large treewidth (bottom row). The red/blue edges show how (some of the) demands are connected in the solution.}\label{fig:largetw}
\end{figure}   

  \begin{figure}[t]
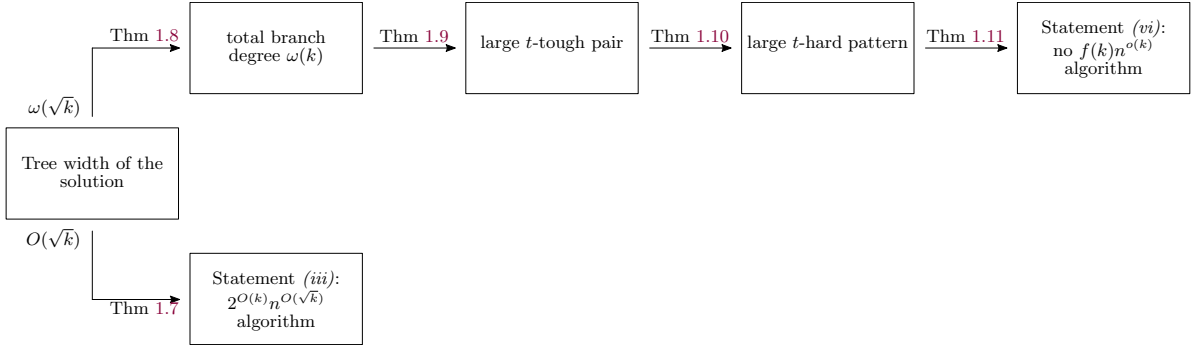

    \begin{center}
\resizebox{\linewidth}{!}{
{\footnotesize \svg{1.3\linewidth}{chain2}}}
\end{center}
\caption{The structure of the proofs of statements \textit{(iii)} and \textit{(vi)}.}\label{fig:chain}
\end{figure}

\subsection{Details of Statement \textit{(iii)}: the $f(k) \cdot n^{O(\sqrt{k})}$ algorithm}
In this section, we give a more detailed overview of the technical steps of the proof of \textit{(iii)} sketched above.

  \paragraph{From treewidth to total branch degree.} Theorem~\ref{thm:subexpalgo} translates the question about the treewidth of the solution in general graphs to a question about the total branch degree of the solution in acyclic graphs. Suppose that we have an edge-minimal solution $H$ in a (not necessarily acyclic) graph $G$ with $k$ terminals. Let us contract the strongly connected components of $H$ in both $G$ and $H$ to obtain $G'$ and $H'$, respectively. We can observe that $H'$ is an acyclic graph that is the optimum solution to an instance in $G'$ with at most $k$ terminals. Our goal is to show that if $H'$ has total branch degree $d$, then $H$ has treewidth $O(\sqrt{d+k})$. Therefore, in the later steps of the proof, we bound the total branch degree of $H'$ by $O(k)$, giving an $O(\sqrt{k})$ bound on the treewidth of $H$.

  We say that a vertex of a strongly connected component of $H$ is a \emph{portal} if it is incident to an edge connecting it to some other component.
  For simplicity of discussion, let us assume here that every strongly connected component of $H$ has at least 3 edges incident to the portals, that is, every vertex of $H'$ has at least 3 incident edges. (If a component has less than 3 such edges and has no terminal, then it consists only of a single vertex and does not affect treewidth anyway; if it has terminals, then it can be taken into account with additional calculations.) By this assumption, the set $P$ of portals have size at most $6d$, where $d$ is the total branch degree of $H'$.

  We want to bound the treewidth of $H$ by showing that there is a set $W$ of $O(d+k)$ vertices such that $H-W$ has treewidth at most $2$. It is known that if removing a set $W$ of vertices from a \emph{planar} graph reduces treewidth to a constant, then the planar graph has treewidth $O(\sqrt{|W|})$. Thus the treewidth bound $O(\sqrt{d+w})$ follows from the existence of such a set $W$.
  
  Let $H[V_i]$ be a strongly connected component of $H$ that has $p_i$ portals and contains $k_i$ terminals. The key observation is that the only role of $H[V_i]$ in the solution is to fully connect the terminals and portals in $H[V_i]$. That is, we can assume that $H[V_i]$ is an optimum solution of a \textsc{Strongly Connected Steiner Subgraph} instance with $p_i+k_i$ terminals. Chitnis et al.~\cite{ChitnisFHM20} showed that we can remove a set $W_i$ of $O(p_i+k_i)$ vertices from such an optimum solution to reduce its treewidth to 2. Therefore, taking the union of $P$ and every $W_i$, we get a set $W$ of size $O(d)+O(\sum (p_i+k_i))=O(d+k)$ whose removal reduces treewidth to 2 (as removing $P$ breaks the graph in a way that each component is a subset of some $V_i$, and the removal of $W_i$ breaks $H[V_i]$ into components of treewidth at most 2).

  \paragraph{Building a skeleton.} Towards the proof of Theorem~\ref{thm:maintechnical}, our goal is to bound the total branch degree by $O(k)$ in an edge-minimal acyclic solution $H$. At some step of the proof, it will be important to assume that $H$ is a triangulated planar graph (every face has exactly three vertices and edges), which is of course not true in general. Therefore, we introduce artificial undirected edges in the graph $H$ to make it triangulated. As these edges do not play any role in the directed problem, it does not change the nature of the solution. Another simplification step is that we assume that there is no vertex $v\not\in T$ with $d^-(v)=d^+(v)=1$. Such a vertex has branch degree 0 and hence suppressing it (i.e., removing it and adding an edge from its in-neighbor to its out-neighbor) has no effect on the total branch degree and on the connectivity of the terminals.

  We start by building a \emph{skeleton} of the solution: a connected subgraph that contains every terminal. The skeleton is composed from {\em segments} of two types. A {\em long segment} is a directed path of $H$ of length at least some constant $L$. A {\em short segment} is any path in the undirected sense of length at most $L$, possibly containing both undirected or directed edges of any orientation. Furthermore, we require that any two long segments in the skeleton are {\em distant,} that is, have distance at least $L$ in the undirected sense.

  A skeleton tree consisting of $O(k)$ terminals and containing all the segments can be built the following way. Initially, we start with an edgeless subgraph $R$ containing only the $k$ terminals. For simplicity of discussion, let us assume that the demand pattern is connected (in the undirected sense). Then there has to be a demand $t_it_j$ such that $t_i$ and $t_j$ are in two different components $C_i$ and $C_j$ of $R$, respectively. This means that $H$ has a directed path $P$ connecting two different components of $R$. If $P$ has length at most $L$, then we can introduce it as short segment to reduce the number of components of $R$. Otherwise, we can shorten $P$ to $P'$ such that every vertex of $P'$ is at distance at least $L$ from $R$ and the two endpoints are at distance exactly $L$ from two different components $C$ and $C'$ of $R$. Then we can reduce the number of components of $R$ by introducing $P'$ as a long segment and two short segments connecting the endpoints of $P'$ to $C$ and $C'$. By repeating these steps, we can reduce the number of components to 1 by introducing $O(k)$ segments in total.  

  \paragraph{Refining the faces.} Our next goal is to further refine the skeleton such that every face of the skeleton has at most $35$ segments on its boundary, and it is still true that the skeleton consists of $O(k)$ segments. We achieve this goal by iteratively dividing a face into two by introducing to the skeleton a new path consisting of at most 5 segments. We argue below that if the division is not very skewed in a certain sense, then the bound $O(k)$ on the number of segments can be achieved even after iterative applications of this step.

Suppose that we have a face $F$ where $x\ge 36$ segments appear on the boundary. Let $P$ be a path between two segments of the boundary and assume that $P$ consists of at most 5 segments. Introducing the path $P$ into the skeleton creates two new faces $F_1$ and $F_2$ that see some number $x_1$ and $x_2$ segments on the boundary of $F$, plus the 5 new segments of $P$. We have $x_1+x_2\le x+2$: if the endpoints of $P$ are internal vertices of segments, then we may have up to 2 segments that are now on the boundary of both $F_1$ and $F_2$.

For a face seeing $x\ge 13$ segments of the skeleton, let us define $x-13\ge 0$ to be the potential of the face. If we chose the path $P$ such that $x_1,x_2\ge 13$, then the potential of the two new faces $F_1$ and $F_2$ are defined. Moreover, the total potential of the two faces is at most
\[
  (x_1+5-13)+(x_2+5-13)\le x-14,
\]
strictly less than the potential of $F$. This means that if we start with a face $F$ that sees $x$ segments of the skeleton, then repeated applications of this step can introduce only $O(x)$ new segments.

\paragraph{Finding a division that is not skewed.} Next we show that if face $F$ sees $x\ge
36$ segments of the skeleton, then we can find a division with $x_1,x_2\ge 13$. Then as we have seen above, repeated applications of this step introduces $O(x)$ segments and divide $F$ into faces that see at most 35 segments each.

Let us divide the boundary of $F$ into three parts, red, green, and blue, each containing at least 12 segments (see Figure~\ref{fig:division}). As every vertex $v$ inside the face $F$ is essential for the solution, there is a directed path $P_v$ from $v$ to some vertex of the boundary; let us fix such a $P_v$ for each~$v$. This defines a color of $v$ according to which of the three parts of the boundary contains the head of $P_v$. Then by Sperner's Lemma and fact that the graph is triangulated, there is a triangle $u_r,u_g,u_b$ inside $F$ where the three vertices have three different colors. From the assumptions that $u_r$, $u_g$, $u_b$ are on three different parts, and each part has length at least 13, it follows that there are two vertices, say $u_r$ and $u_b$, such that both subpaths of the boundary between the heads of $P_{u_r}$ and $P_{u_b}$ have at least 12 segments. Then putting together $P_{u_r}$ and $P_{u_b}$ creates a path $P$ that divides the face $F$ in the required way. This argument needs to be refined a bit further: as we said earlier, we want a skeleton where the long segments are distant, i.e., are at distance at least $L$ from each other. But this can be easily achieved by appropriately shortening the long segments $P_{u_r}$, $P_{u_b}$, and then extending them by three short segments.
\begin{figure}
  \begin{center}
\includegraphics{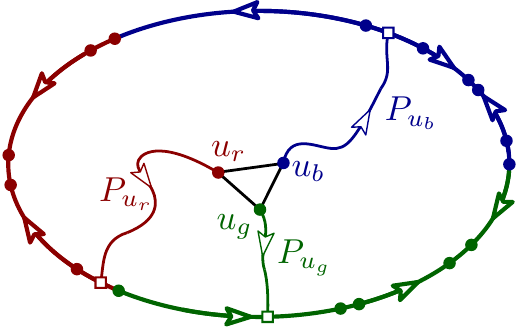}
  \end{center}
  \caption{Finding a division that is not skewed.}\label{fig:division}
\end{figure}  

\paragraph{Many edges incident to a long path.} We assume now that the skeleton has $O(k)$ faces, each seeing at most $35$ segments. If we can show that the total branch degree (of the orginal solution $H$ without the artificial edges) is a constant in each face, then we can bound by $O(k)$ the total branch degree of the solution. We can observe, using the acyclicity of the edges inside the face, that we need to bound only the number of edges incident on the boundary.

Let $e$ be an edge inside the face incident to vertex $v$ of the boundary. We say that $e$ is \emph{essential} for demand $t_it_j$ if removing $e$ breaks every path from $t_i$ to $t_j$. Then we can define a path $P_e$ the following way: let us take any path $P$ from $t_i$ to $t_j$, and let $P_e$ be the subpath of $P$ starting from $e$ (which has to appear on $P$) to the first vertex on the boundary of $F$. Let us consider two edges $e_1$, $e_2$ starting from the same vertex $v$ of the boundary. Let us observe that $P_{e_1}$ and $P_{e_2}$ cannot intersect: then we could bypass e.g. $e_1$ by starting on $P_{e_2}$ and following it until intersection. By a similar argument, $P_{e_1}$ and $P_{e_2}$ cannot go to the same long segment: then one of $P_{e_1}$ and $P_{e_2}$ could be avoided by using the other path and part of the long segment. From these observations, it follows that the only way the boundary can have many edges incident to it is that if there are edges $e_1$, $\dots$, $e_s$ incident to distinct vertices of a long segment $S_1$, with paths $P_{1}$, $\dots$, $P_{s}$ going to distinct vertices of some other long segment $S_2$ (see Figure~\ref{fig:gridintro}).

\begin{figure}
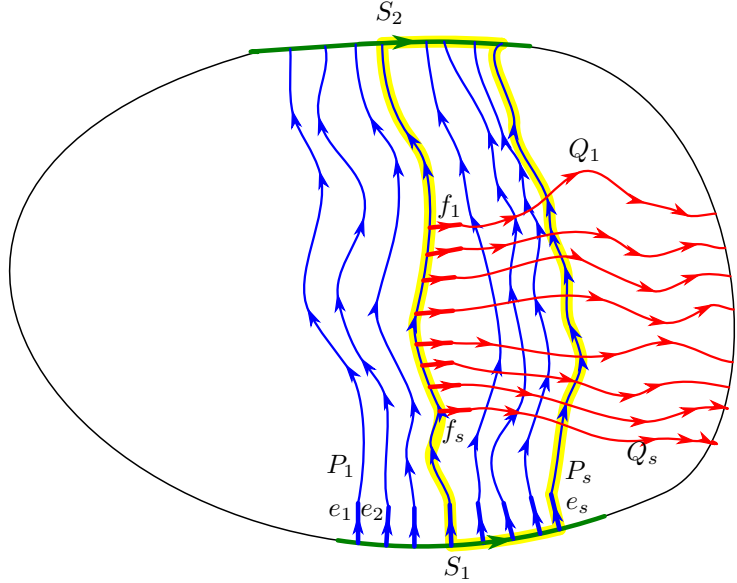

  \begin{center}
    { \small \svg{0.6\linewidth}{findgridintro}}
  \end{center}
  \caption{Finding a grid.}\label{fig:gridintro}
\end{figure}

\paragraph{Finding a grid and a $t$-tough pair.} Now comes the point where we use the assumption that long segments are distant. In particular, this means that the ``middle path'' $P_{e_{s/2}}$ is long. The internal vertices of this path have no terminals (as all the terminals are on the skeleton), hence it is not possible that $d^+(v)=d^-(v)=1$ for any such internal vertex. Thus either there are many vertices on this path that have an edge leaving the path, or many vertices that have an edge entering the path. Assume without loss of generality the former, let $f_1$, $\dots$, $f_{s}$ be these edges. Again, each edge is essential for some demand, hence the path satisfiying the demand has a subpath $Q_{i}$ starting with $f_i$ and going to the boundary. We can observe again that these paths have to be disjoint. Therefore, we can obtain a grid-like structure in the region surrounded by $S_1$, $P_{s/2}$, $S_2$, and $P_s$, see the region highlighted by yellow in Figure~\ref{fig:gridintro}. (There are some other cases to consider, which we ignore here. For example, the paths $Q_i$ may go to $S_1$ or $S_2$.) This region has $s/2-1$ ``vertical'' paths $P_{s/2}$, $\dots$, $P_{s-1}$, intersected by the $s$ ``horizontal paths'' $Q_{1}$, $\dots$, $Q_s$.

We observe that if this grid has $t$ horizontal and vertical paths, then we can use it to discover a $t$-tough pair. Each edge $e_i$ is essential for some minimal demand; let $E_1$ be the set of these $t$ demands. Similarly, we define $E_2$ based on chosing a minimal demand for which $f_i$ is essential. Then we can carefully verify that $(E_1,E_2)$ is a $t$-\toughpair: if there is an edge in the demand graph that is not allowed, then a careful analysis shows that there is a way of bypassing some $e_i$ or $f_i$ in the grid, contradicting the fact that it is essential. This concludes the proof that if we have an upper bound on the size of the largest $t$-tough pair appearing in the graphs of class $\calD$, then we can bound the treewidth of the solution by $O(\sqrt{k})$.

 \paragraph{Cleaning.}
 To prove Theorem~\ref{thm:cleaning-intro}, we need to show that if arbitrary large $t$-tough-pairs appear in the graphs of $\calD$, then $\mathcal{C}_i\subseteq \calD$ for some $i\in[\kappa]$. The proof is a long combinatorial argument to show that we can find $t$-tough-pairs that are canonical in some sense, and then we use the assumption that $\calD$ is closed under identifying vertices to contract the vertices outside the $t$-\toughpair\ into a small constant number of well-behaved vertices.

 Suppose that there is a $t$-\toughpair $(E_1,E_2)$ in a digraph $D$.  The minimality of the edges in $E_1$ and the fact that they do not appear in directed cycles (as they are weakly independent to themselves) imply that for any two edges $x_iy_i, x_jy_j\in E_1$, at least one of the following holds:
 \begin{enumerate}
   \item exactly the edges $x_iy_j, x_jy_i$ appear between $\{x_i,y_i\}$ to $\{x_j,y_j\}$,
   \item there is no edge from $\{x_i,y_i\}$ to $\{x_j,y_j\}$, or
   \item there is no edge from $\{x_j,y_j\}$ to $\{x_i,y_i\}$.
   \end{enumerate}
   Let us consider a complete graph on $t$ vertices $w_1$, $\dots$, $w_t$, and for every $i<j$, color the edge $w_iw_j$ according to which of the three statements hold for the edges $x_iy_i$ and $x_jy_j$ (if more than one statement is true, we can choose arbitrarily). By Ramsey's Theorem, there is a large subset $E'_1\subseteq E_1$ where the same statement holds for any pair of edges. We can find a similar subset $E'_2\subseteq E_2$. We consider two main cases. The first case is when Statement 1 holds either in $E'_1$ or $E'_2$. Then what we have is a matching $x_iy_i$ of minimal edges that is part of a complete bipartite graph, that is, every $x_i$ is adjacent to every $y_j$ (but note that $x_iy_j$ does not have to be a minimal edge). The second case is where we have Statement 2 or 3 in both $E'_1$ and $E'_2$. Then we can reorder $E_1$ and $E_2$ to have a further ordering property: there is no edge from $\{x_i,y_i\}$ to $\{x_j,y_j\}$ for $j<i$. We handle the two cases separately. With further Ramsey arguments and case distinctions, we show that identifications can be used to find a $t'$-hard biclique pattern or a $t'$-hard matching pattern appearing in a graph in $\calD$, where $t'$ is some unbounded function of $t$. It follows that if arbitrarily large $t$-tough pairs appear in $\calD$, then $\calD$ is a superclass of some~$\mathcal{C}_i$.


\section{Formal definition of a $t$-\toughpair}\label{sec:prelims}
In this section we give the formal definition of a $t$-\toughpair.
Further definitions, that are specific to the sections. are defined in the beginning of the respective sections.

Given a digraph $D$ and an edge $e=(u,v) \in E(D)$, we say that $e$ is a {\em minimal} edge of $D$ if $D$ has no $(u,v)$-path of length strictly greater than $1$ in $D$, where the length of the path is the number of edges in it.
We say that a digraph $D$ is {\em reachability-minimal} if each edge of $D$ is minimal.
For an edge $e=(u,v)$ in a directed graph $D$, $v$ is called the {\em head} of $e$ and $u$ is called the {\em tail} of $e$. 
For any $E' \subseteq E(D)$, 
$\head(E')$ (resp.~$\tail(E')$) denotes the set of heads (resp.~tails) of the edges in $E'$.
Next we define weak independence and strong independence that are crucially to define the $t$-\toughpair\ formally.

\begin{definition}[Weakly independent edges]
Given a digraph $D$ and $e_1 = (u_1,v_1), e_2 =(u_2,v_2) \in E(D)$, 
we say that the pair of edges $(e_1,e_2)$ is weakly independent in $D$, 
if $u_1 \neq v_1 \neq u_2 \neq v_2$,  and
$D$ has neither a $(v_1,u_2)$-path nor a $(v_2,u_1)$-path. 
A set of edges $E' \subseteq E(D)$ are weakly independently if every pair of distinct edges in $E'$ are pairwise weakly independent and for each edge $(u_i,v_i) \in E'$, there is no $(v_i,u_i)$-path in $D$.

Informally, a pair of edges is weakly independent, if the head of one cannot reach the tail of the other. Therefore, if a pair of edges are weakly independent, then they cannot lie on a directed path.
\end{definition}

\begin{definition}[Strongly independent edges]
Given a digraph $D$ and $e_1 = (u_1,v_1), e_2 =(u_2,v_2) \in E(D)$, we say that the pair of edges $(e_1,e_2)$ is strongly independent in $D$, if they are weakly independent in $D$, and additionally $D$ has no $(u_1,u_2)$-path, no $(u_2,u_1)$-path, no $(v_1,v_2)$-path and no $(v_2,v_1)$-path. 

Informally, a pair of edges is strongly independent, if they are weakly independent, and the head of one cannot reach the head of the other, and the tail of one cannot reach the tail of the other. That is, the vertices of the heads (resp.~vertices of tails) do not lie on any directed path.
\end{definition}

\begin{definition}[$t$-\toughpair]\label{def:ttough}
Given a digraph $D$, $E_1, E_2 \subseteq E(D)$, we say that $(E_1,E_2)$ is a \toughpair\  in $D$ if:
\begin{enumerate}
\item $|E_1| = |E_2|$,
\item each edge of $E_1 \cup E_2$ is a minimal edge in $D$,
\item all edges in $E_i$ are pairwise weakly independent in $D$, for both $i \in \{1,2\}$, and
\item for each $e_1 \in E_1$ and $e_2 \in E_2$, $(e_1,e_2)$ are strongly independent in $D$.
\end{enumerate}

Further, for a positive integer $t$, we say that $(E_1,E_2)$ is a $t$-\toughpair\ if $|E_1| = |E_2| =t$.
\end{definition}

\section{The structure theorem}

The goal of this section is to prove \Cref{thm:maintechnical} (see Section~\ref{sec:introduction} for the statement). We start by giving the algorithm for the $c$-bounded case, that is the proof of Theorem~\ref{thm:subexpalgo}, and refocusing our efforts on acyclic instances.

\subsection{A subexponential algorithm}
\label{sec:dag}

In order to prove \Cref{thm:subexpalgo}, we will utilize the following result.

\begin{theorem}[Lemma 2.1 of~\cite{ChitnisFHM20}]\label{thm:smalldeletion_scss}
Suppose that $H$ is an edge-minimal solution for the \textsc{Strongly Connected Steiner Network} problem on the instance $(G,T)$, where $|T|=k$. Then there is a set $W\subseteq V(H)$ of $9k$ vertices such that deleting $W$ from $H$ results in a graph of treewidth at most $2$.
\end{theorem}

We start by generalizing the above theorem to our setting.

\begin{lemma}\label{lem:smalldeletion}
Let $H$ be an edge-minimal solution to \Dsnfull where $\calD$ is a $c$-bounded class. Then there is a set of $(30c+29)k$ vertices in $V(H)$ whose deletion from $H$ results in a graph $H'$ of treewidth at most $2$.
\end{lemma}

\begin{proof}
Let $V_1,\dots,V_t$ be the vertex sets of the strongly connected components of $H$ that are not singletons. The degree of $V_i$, denoted by $d(V_i)$, is the number of edges in $H$ with one endpoint in $V_i$ and one endpoint outside $V_i$.

We claim that if $V_i$ has no terminals, then  $d(V_i)\geq 3$. Suppose the contrary: $V_i\cap T=\emptyset$ and $d(V_i)\leq 2$. Notice that $d(V_i)\neq 1$ as in such a case we could remove all edges of $H$ induced by $V_i$ as well as the edge that enters/exits $V_i$. Suppose now that $d(V_i)=2$. Notice that the two edges that have one endpoint in $V_i$ cannot both be entering or both be exiting $V_i$ as again that would be redundant in $H$. Thus $V_i$ has one edge entering at some vertex $u\in V_i$ and one edge exiting from some vertex $v\in V_i$. Let $\pi$ be a path in $H[V_i]$ from $u$ to $v$. Notice that since $|V_i|\geq 2$ and $H[V_i]$ is strongly connected, there must be some edge $e$ in $H[V_i]$ that is not on the path~$\pi$. However all terminal-to-terminal paths passing through $V_i$ can be realized using $\pi$, since if such a path enters $V_i$ then it enters at $u$ and exits at $v$. Consequently, the edge $e$ is redundant, contradicting the minimality of $H$.

We say that a vertex $v$ is a \emph{portal} of $V_i$ if it is an endpoint of some edge $e\in E(H)$ that enters or exits $V_i$. Let $P$ denote the set of all portals, and let $P_i$ be the set of portals in $V_i$. Consider now a strongly connected subgraph $H[V_i]$ that contains at least one terminal, and let $T_i=V_i\cap T$. Since $H[V_i]$ is strongly connected, it is an edge-minimal solution of \scssfull for the terminal set $T_i\cup P_i$: indeed, if some edge $e$ can be removed from $H[V_i]$ to maintain strong connectivity on $T_i\cup P_i$, then $H-e$ would also be feasible for the original problem, which contradicts the edge-minimality of $H$. We can therefore apply \Cref{thm:smalldeletion_scss} for the graph $H[V_i]$ and terminal set $T_i\cup P_i$: there is a set $X_i$ of at most $9|T_i\cup P_i|$ vertices such that $H[V_i\setminus X_i]$ has treewidth at most $2$.

We now delete the following vertices from $H$:
\begin{itemize}
\item the portal set $P$
\item all the sets $X_i$
\item the set $Y$ of vertices outside $\bigcup_i V_i$ that have degree at least $3$.
\end{itemize}
Let $H^*$ be the resulting graph. Each connected component of $H^*$ is either a subgraph of some strongly connected component $H[V_i]$, or it lies outside $\bigcup_i V_i$, thus its vertices have maximum degree at most $2$: in both cases, it has treewidth at most $2$. Consequently, $\tw(H^*)\leq 2$.

In order to bound the number of deleted vertices, let us first introduce an acyclic instance of \Dsnfull based on $(G,T,D)$ and $H$: this is required so that we can use the $c$-boundedness of $\calD$ . Let $G/(\bigcup V_i)$ be the graph obtained by contracting each vertex set $V_i$ into a single vertex (deleting all loops and parallel edges). Notice that the picture $H'$ of $H$ under this contraction is acyclic. In the demand graph $D$, we identify each vertex set $T_i$ individually, resulting in a demand graph $D'$ and terminal set $T'$. Notice that since $\calD$ is closed under identification, we have $D'\in \calD$. 

We claim that $H'$ is an edge-minimal solution to the 
\Dsnfull instance $(G'=H',T',D')$ (where $H'$ and $D'$ are acyclic). First we show that $H'$ is feasible. If $D'$ has an edge $u'v'$, then the edge has an ancestor $uv$ in $D$, so $H$ contains some path $P_{uv}$ connecting the terminals $u$ and $v$ (where $u$ and $v$ cannot be located in the same set $V_i$). When we apply the contraction on $P_{uv}$, we get a path $P'$ connecting $u'$ and $v'$ in $H'$, concluding the proof of feasibility for $H'$.

To prove the edge-minimality of $H'$, suppose the contrary: that there is some edge $a'b'$ in $H'$ whose deletion does not break feasibility. Let $ab$ be an edge of $H$ that contracts to $a'b'$. Notice that $ab$ is not induced by any component $V_i$. We claim that $H-ab$ is feasible for $(G,T,D)$. Consider a demand edge $uv\in E(D)$. If the terminals $u$ and $v$ are in the same $V_i$, then they are connected within $H[V_i]$ and thus they are also connected in $H-ab$. Otherwise these terminals are identified with some distinct terminals $u'$ and $v'$, respectively, and $D'$ has a demand $u'v'$. Thus $H'-a'b'$ connects $u'$ and $v'$ via some path $P'$. Let $u$ and $v$ be arbitrary vertices of $G$ that contract to (or are equal to) $u'$ and $v'$, respectively.

Using $P'$, we will now build a path $P$ connecting $u$ and $v$ in $H-ab$. First, we pick for each edge $x'y' \in P'$ an arbitrary edge whose image after contraction is $x'y'$. Consider now an edge pair $x'y', y'z'$ on $P'$: if $y'$ comes from the contraction of $V_i$, then we may have to select $xy_1$ as the ancestor of $x'y'$ and $y_2z$ as the ancestor of $y'z'$, where $y_1 \neq y_2$ are different portals. But since $(H-ab)[V_i]$ is strongly connected, there is a path $P_i$ connecting the portals $y_1$ to $y_2$ in $(H-ab)[V_i]$. We concatenate this path with the edges $xy_1$ and $y_2z$. Using the same technique, we can build a path from some $u_2$ to some $v_2$ where either $u_2=u$ or $u_2$ and $u$ are in the same component $V_j$ (and the analogous statement holds for $v$ and $v_2$). We can again use the strong connectivity of the components $V_i$ to get a path $P$ from $u$ to $v$. Consequently, $H-ab$ is feasible for the original problem, which contradicts the edge-minimality of $H$.

Finally, we note that for any pair $V_i,V_j$ with $i\neq j$ there can be at most one edge going between $V_i$ and $V_j$. Since these are strongly connected components, there cannot be two edges in different directions, as that would make them into a single strongly connected component. Suppose that $uv$ and $u'v'$ both go from $V_i$ to $V_j$. Then either of these edges can be removed without affecting feasibility, which contradicts the edge-minimality of $H$. This property implies that in the proposed contraction the degree of the vertex that we get from contracting $V_i$ is equal to $d(V_i)$.

We can now bound the size of the deleted set. Since $\calD$ is $c$-bounded, we have that $\sum_{v\in V(H')} d^*(v)\leq ck$. Since each non-0 term in the sum comes from some contraction of $V_i$ or a deleted vertex from $Y$, we obtain the following.
\begin{equation}\label{eq:cbounded}
\sum_{v\in V(H')} d^*(v) = \sum_i d^*(V_i) + \sum_{v\in Y} d^*(v)\leq ck,\end{equation}
where $d^*(V_i)=\max(0,d(V_i)-2)$.
Recall that if $V_i$ has degree at most $2$, then it must have at least one terminal. It follows that it can have at most $2$ portals, and thus $|P_i|\leq 2|T_i|$. If it has degree at least three, then it has at most $d(V_i)= d^*(V_i)+2\leq 3d^*(V_i)$ portals.

Thus the total number of deleted vertices can be bounded as:
\begin{align*}
|Y|+|P|+\Big|\bigcup_i X_i\setminus P\Big|&\leq |Y|+|P|+\sum_i 9|T_i\cup P_i|\\
&\leq |Y|+9|T|+10|P|\\
&= |Y|+9k + 10\cdot \left(\sum_{i\,:\,d(V_i)\geq 3} |P_i| + \sum_{i\,:\,d(V_i)\leq 2} |P_i|\right)\\
&\leq |Y| + 9k + 10\cdot \left(\sum_{i} 3d^*(V_i) + \sum_{i\,:\,d(V_i)\leq 2} 2|T_i|\right)\\
&\leq (30c+29)k,
\end{align*}
where the second inequality uses that the sets $V_i$ are disjoint, and the last inequality uses the bound~\eqref{eq:cbounded}.
Consequently, we have removed at most $(30c+29)k$ vertices from $\bigcup_i V_i$, which concludes the proof.
\end{proof}

The next lemma can essentially be found within the proof of Lemma 2.2 in~\cite{ChitnisFHM20}. 
We reproduce the proof for completeness.

\begin{lemma}\label{lem:twboundfromdeletion}
If $G$ is a planar graph where deleting $k\geq 1$ vertices results in a graph of treewidth $w\geq 1$, then $\tw(G)\leq 3w\sqrt{k}$.
\end{lemma}

\begin{proof}
By the planar grid theorem~\cite{RobertsonST94}, there is a constant $\bar c$ such that any planar graph of treewidth $\bar c \omega$ has a grid minor of size $\omega \times \omega$. If $G$ has treewidth at least $\bar c \cdot \lceil 3w\sqrt{k}\rceil$, then it has a grid minor of size at least $\lceil 3w\sqrt{k}\rceil$. This minor can be decomposed into $\left\lfloor \frac{ \lceil 3w\sqrt{k}\rceil}{w+1}\right\rfloor \cdot \left\lfloor \frac{ \lceil 3w\sqrt{k}\rceil}{w+1}\right\rfloor$ vertex disjoint grid minors, each of size at least $(w+1)\times(w+1)$. Notice that $\left\lfloor \frac{ \lceil 3w\sqrt{k}\rceil}{w+1}\right\rfloor \cdot \left\lfloor \frac{ \lceil 3w\sqrt{k}\rceil}{w+1}\right\rfloor\geq k+1$ for $k,w\geq 1$ integers, thus there is at least one $(w+1)\times(w+1)$ grid minor where no vertex has been deleted. This intact grid minor has treewidth at least $w+1$, which contradicts our assumptions.
\end{proof}

We are now ready to prove \Cref{thm:subexpalgo}.

\begin{proof}[Proof of \Cref{thm:subexpalgo}]
Consider an input $(G,T,D)$. By \Cref{lem:smalldeletion}, we can find a set of $(30c+29)k$ vertices in $H$ whose deletion results in a graph of treewidth at most $2$. \Cref{lem:twboundfromdeletion} then bounds the treewidth of $H$ as $\tw(H)\leq 3\cdot 2 \cdot \sqrt{(30c+29)k}< 47\sqrt{ck}=O(\sqrt{k})$.
This concludes the proof.
\end{proof}
\subsection{A tree of segments}\label{sec:skeleton}

We begin the proof of \Cref{thm:maintechnical} by supposing that $\calD$ is a class that is closed under identification, but it is not $c$-bounded. The goal is now to show that $\calD$ has a pattern with a $t$-tough pair for all positive integers $t$. Consider an instance where $G$ and $D$ are acyclic.

First, we show that it sufficient to consider a weakly connected acyclic solution $H$ where all vertices have undirected degree at least $3$. To show the minimum degree bound, we contract all edges $uv$ of $H$ (and $G$) where $u$ or $v$ has undirected degree at most $2$. If at least one of $u$ and $v$ is a non-terminal, then such contractions do not influence the feasibility and edge minimality of $H$, and it also does not change the total branch degree of $H$. If both $u$ and $v$ are terminals, then let $H'$ denote the new graph after the contraction, and let us identify $u$ and $v$ in $D$, resulting in the graph $D'$. Since $|D'|=k-1$ and $H$ and $H'$ have the same total branch-degree, we have that if $H$ has total branch degree at least $ck$, then also $H'$ has total branch degree at least $c\cdot (k-1)$. It is therefore sufficient to consider graphs $H$ where all vertices of the undirected graph $\bar H$ have degree at least $3$. (Throughout this section, the notation $\bar{X}$ refers to the undirected graph given by the edges of the graph $X$, where $X$ may be a directed graph or mixed graph.) 

Suppose now that $H$ is disconnected, and let $H_i$ be the connected components of $H$ ($i=1,2,\dots$). For each $H_i$ let $D_i$ be the subgraph of $D$ induced by $T\cap V(H_i)$. Notice that the graphs $D_i$ are a partition of the edges of $D$, where each connected component of $D$ is contained in a single graph $D_i$. Observe that each $H_i$ is an edge-minimal solution for the instance $(G,T\cap V(H_i), D_i)$. Moreover, if $H$ has branch-degree at least $ck$, than at least one among the $H_i$ has branch degree at least $c\cdot|T\cap V(H_i)|$. Therefore we can restrict our attention to weakly connected graphs $H$.

\paragraph*{Defining a tree of segments.}
The skeleton $S$ is a $2$-connected mixed graph that we will build based on an optimum solution $H$ of the instance $(G,T,D)$ of \Dsn where $G$ and $D$ are acyclic.
Suppose now that $H$ is a solution to $(G,T,D)$ that is a connected acyclic graph of minimum degree at least $3$.
We fix a plane embedding of $H$, and add undirected edges to $\bar{H}$ in a greedy manner to create a triangulation of the plane with vertex set $V(H)$; let $H_\Delta$ denote the resulting mixed graph, where the edges of $H$ are directed, and the newly added triangulation edges are undirected.

A \emph{long segment} is a \emph{directed} path of length at least $L$ in $H_\Delta$. A \emph{short segment} is a path of $\bar{H}_\Delta$ (i.e., of arbitrary orientation edges in $H_\Delta$) that consists of at most $L$ edges. A pair of segments $A,B$ are \emph{distant} if for any pair of vertices $a\in A$ and $b\in B$ the distance of $a$ and $b$ in $\bar{H}$ is at least $L$. Our goal is to create a skeleton where the boundary of the so-called \emph{relevant} face consists of $O(k)$ (long and short) segments where the long segments are pairwise distant.

Next, we construct a tree $R$ in $H_\Delta$ consisting of $O(k)$ segments that contains all terminals.

\begin{lemma}\label{lem:maketree}
There is a tree $R$ in $H_\Delta$ that contains all terminals and consists of $O(k)$ segments, such that any pair of long segments of $R$ are distant.
\end{lemma}

\begin{proof}
Initially, we set $R$ to be the edgeless forest with vertex set $T$. We add segments to $R$ using the following insertions, until it becomes connected.

\begin{description}
\item[Insertion 1.] Let $C$ and $C'$ be the closest components of the current graph $R$, that is, the pair of components where $\delta(C,C'):=\min_{u\in V(C), v\in V(C')} \dist_{\bar{H}_\Delta} (u,v)$ is minimized. If $\delta(C,C')\leq 2L$, then let $P$ be a shortest path connecting $C$ and $C'$ in $H_\Delta$. Since $|P|\leq 2L$, we can decompose $P$ into at most two segments, which we add to $R$, connecting $C$ and~$C'$.
\item[Insertion 2.]
Suppose that $R$ is not connected, but Insertion 1 can no longer be applied. A vertex $v$ is a \emph{collaborator} of a component $C$ of $R$ if there is an edge $e$ incident to $v$ that is essential for some demand $tt'$ where $t\in V(C)$ or $t'\in V(C)$. 
Let $v$ be a vertex that collaborates with at least two distinct components. Let $P_v$ and $P'_v$ be directed paths connecting $v$ to these two components. We shortcut the loops of $P_v\cup P'_v$, to find a path $P$ whose internal vertices are outside $N_R$ that connects two distinct neighborhoods $N_C$ and $N_{C'}$. Note that $P$ is the concatenation of at most two directed paths, thus it can be decomposed into at most two short or long segments. If we have two long segments, we also need to separate these with a short segment: let $u$ be the first vertex on the first long segment that is within $\bar{H}$-distance $L$ to the other long segment. We connect $u$ to the last vertex $v$ of the other long segment to which it has distance $L$ using a short segment, and remove the parts of the long segments that fall between $u$ and $v$. Let $P$ be the final path, which now consists of at most $3$ segments, and if it has two long segments, then those are distant. Connecting the endpoints of $P$ to $C$ and $C'$ with short segments, we are able to add at most five segments to $R$ that connect $C$ and~$C'$. 
\end{description}

We claim that applying the above insertions exhaustively leads to a tree that contains all terminals. Observe that after each successful insertion the number of connected components of $R$ decreases by at least one (and no cycles can be created), thus after at most $k-1$ insertions we get a connected tree $R$. In order to show this, we need to show that as long as $R$ is disconnected, there is an insertion that we can apply.

\begin{claim}
If $R$ is disconnected, then at least one of the insertions can be applied.
\end{claim}
\begin{claimproof}
Suppose for the sake of contradiction that Insertion 1 and 2 cannot be applied, but $R$ is still disconnected. It follows that each vertex collaborates with at most one component. If some edge $e$ is essential for demand $t_it_j$ and the terminals $t_i,t_j$ lie in different connected components of $R$, then Insertion 2 can be applied, so suppose that no such edge exists, i.e., for each edge the corresponding demand terminals fall into the same component. Since all edges $e$ of $H$ are essential and $H$ has no isolated vertices, we have that all vertices of $H$ collaborate with at least one component of~$R$.

Let $P_0$ be a path in $\bar H$ connecting the distinct components $C$ and $C'$ of $R$ whose internal vertices are disjoint from $V(R)$: such a path exists because $R$ is a disconnected subgraph of the connected graph $H$. The starting vertex of $P_0$ collaborates with $C$, and its ending vertex collaborates with $C'$, so there must be an edge $e$ along this path whose endpoints collaborate with distinct components. But this contradicts the definition of collaboration, which implies that adjacent vertices collaborate with all components for which the connecting edge is essential.
\end{claimproof}

Notice that each insertion adds at most $5$ new segments to $R$, and we can do at most $k-1$ insertions to reach the final tree $R$, thus the final tree $R$ consists of at most $5k-5$ segments. Moreover, long segments can only be added with Insertion 2 and 3, and each of them adds a long segment outside the current $N_R$, thus the newly added long segments are distant from all earlier long segments. Thus insertions preserve the property that long segments are pairwise distant. This concludes the proof.
\end{proof}

\subsection{Region slicing}\label{sec:Sperner}

We will now slice the plane into smaller regions using segments. We say that a closed region (some connected union of faces of $H_\Delta$) in the plane is \emph{relevant} if it is interior-disjoint from the plane tree $R$.  We work towards the following lemma.

\begin{lemma}[Skeleton Lemma]\label{lem:slicedskeleton}
There is a subgraph $\Skel \subset H_\Delta$ consisting of $O(k)$ segments where the faces of $\Skel$ are relevant and they partition the plane, each face of $\Skel$ has at most $35$ segments on its boundary, and within the subgraph of $\bar{H}$ induced by each face of $\Skel$ any pair of long segments on the face boundary are either subpaths of the same directed path of $H$, or they are distant.
\end{lemma}

The proof of the lemma requires a good separation, which we will prove next. The separator that we prove relies on Sperner's lemma, which can be phrased as follows.

\begin{theorem}[Sperner's lemma]
Let $G$ be an undirected simple plane graph where the vertices are assigned one of three colors (red, green, or blue). Suppose that other than the outer face, the rest of the graph is triangulated. Moreover, suppose that the boundary of the outer face has three marked vertices, $v_1,v_2,v_3$ colored with $1,2,3$, and that each vertex on the boundary is assigned one of the colors of the two marked vertices that enclose it on the boundary. Then $G$ has a triangle whose vertices have three different colors.
\end{theorem}

If $\calF$ is a relevant region bounded by some cycle of $\bar{H}_\Delta$, then let $H^\calF$ and $H^\calF_\Delta$ denote the subgraph of $H$ and $H_\Delta$ consisting of the edges in $\calF$ (including edges on the boundary of $\calF$), respectively.

\begin{lemma}[Separator]\label{lem:sperner_sep}
Let $\calF$ be relevant region of $H$ that is bounded by a cycle $F$ of $\bar{H}_\Delta$ which consists of $s\geq 36$ segments, where long segments are distant in $\bar{H}^\calF$. Then there is a path $P$ in $\calF$ that splits $\calF$ into two regions, $\calF_1$ and $\calF_2$ with boundary cycles $F_1$ and $F_2$ consisting of $s_1$ and $s_2$ segments, so that
\begin{itemize}
\item long segments in $F\cup P=F_1\cup F_2$ are distant,
\item $s_1+s_2\leq s+12$,
\item and $13\leq  s_1,s_2 \leq s-4$.
\end{itemize} 
\end{lemma}

\begin{proof}
We will assign colors to the vertices of $H_\Delta$ that fall in $\calF$. First, we split the boundary cycle $\bar{F}$ into three paths $P_1,P_2,P_3$ of almost equal number of segments, i.e., so that the number of segments on each of these paths have a difference of at most 1. We color the endpoints of the paths with $1,2$ and $3$ so that $P_i$ goes from the point of color $i$ to the point of color $i+1$, where indices are defined modulo $3$. Next, all internal vertices of $P_i$ will be colored by $i$ for each $i=1,2,3$, see Figure~\ref{fig:sperner}.

To color the vertices in the interior of $\calF$, let $v$ be an arbitrary vertex there. Since $v$ is not in $R$, it cannot be a terminal, so it has an essential edge incident to it, which is on some terminal-to terminal path. Let $P_v$ be this directed path connecting two terminals that contains~$v$. Since there are no terminals in the interior of $\calF$, going forward on the path we will eventually exit $\calF$. Before this point, there will be a first vertex $w$ of $P$ after $v$ that is of distance at most $L$ to some point $b\in V(F)$, that is, either $\dist{\bar{H}_\Delta^\calF}(v,b)\leq L$ and $v=w$, or $\dist_{\bar{H}^\calF}(w,b)= L$ and $\max_{x\in P[v,w)}\dist_{\bar{H}_\Delta^\calF}(x,V(F))>L$. (If there are multiple vertices $b\in V(F)$ at minimum distance to $w$ in $\bar{H}_\Delta^\calF$, then we choose an arbitrary such vertex $b$). We assign to $v$ the color that we have assigned to $b$. Notice that we have a path $Q$ consisting of a directed path (a subpath of $P_v$) and a short segment that connects $v$ to $b$, and $Q$ is distant from all long segments of $V(F)$.

\begin{figure}[t]
\includegraphics[width=\textwidth]{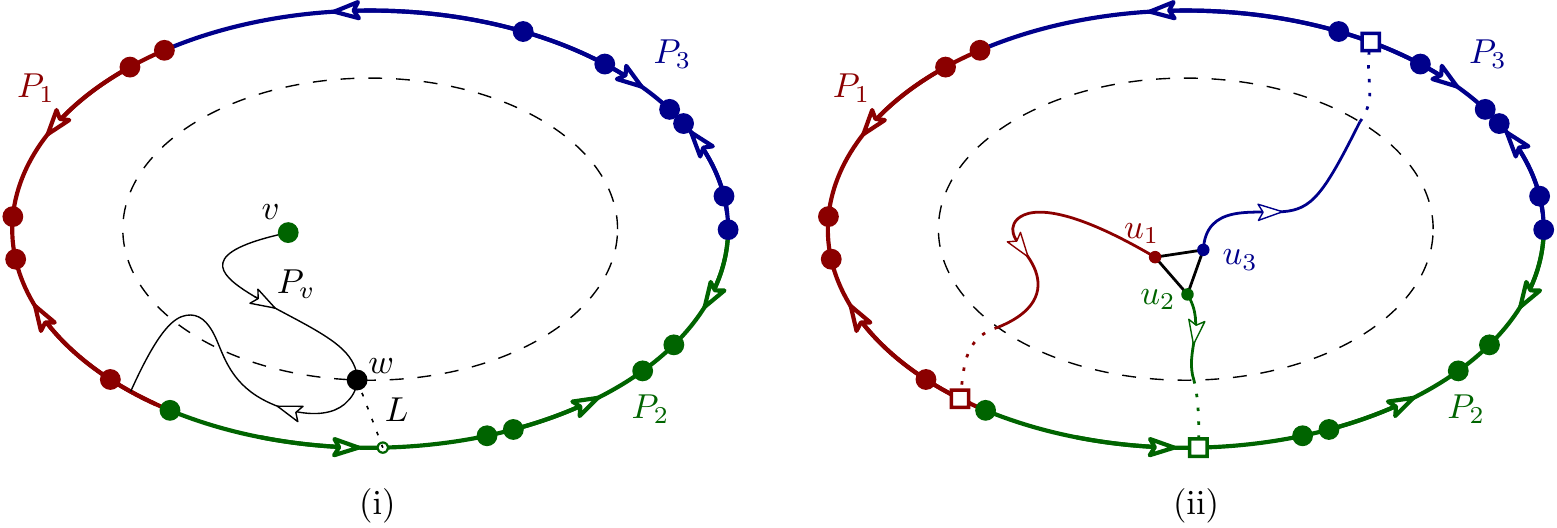}
\caption{(i) Coloring segments on the boundary of the region $\calF$. Long segments are denoted with an arrow. The color of vertex $v$ is defined using a path $P_v$ that arrives in the $L$-neighborhood of $\bd\calF$. (ii) Applying Sperner's lemma in the coloring. Two of the landing vertices in $\bd\calF$ (denoted by empty squares) will be from distant segments along $\bd\calF$.}\label{fig:sperner}
\end{figure}

By Sperner's lemma there is a triangle $u_1u_2u_3$ in this coloring where $u_i$ has color $i$. Let $Q_i$ be the path from $u_i$ to the boundary point $b_i\in V(P_i)$ that we used to assign color $i$ to $u_i$. Since the paths $P_i$ have at least $\lfloor s/3 \rfloor$ segments, we have that there is a pair $b',b''$ among $b_1,b_2,b_3$ such that  along $F$ there are at least $\lfloor s/3 \rfloor$ complete segments between them (and up to two partial segments). Let $u',u''\in \{u_1,u_2,u_3\}$ and $Q',Q''\in \{Q_1,Q_2,Q_3\}$ be the starting points and paths corresponding to $b'$ and $b''$. Consider the path that is the concatenation of $Q'$, the edge $b'b''$, and $Q''$. If both $Q'$ and $Q''$ contain long segments attached to $u'$ and $u''$, then the path can be shortcut in the middle with a short segment to ensure that they remain distant from each other: let $u$ be the first vertex on the first long segment that is within $\bar{H}$-distance $L$ to the other long segment. We connect $u$ to the last vertex $v$ of the other long segment to which it has distance $L$ using a short segment, and remove the parts of the long segments that fall between $u$ and $v$.

Let $P$ be the resulting path connecting $b'$ and $b''$, and let $\calF_1$ and $\calF_2$ be the regions that we get by splitting $\calF$ with $P$. By construction, the resulting path consists of at most $5$ segments, and long segments in $F\cup P$ are pairwise distant. Note that the vertices $b'$ and $b''$ may be internal vertices of some segment of $F$, and in such a case the corresponding segment will be counted both in $F_1$ and $F_2$. We are also counting the at most $5$ segments of $P$ both times. Thus we have that $s_1+s_2\leq s+2+2\cdot 5 = s+12$. Since $P$ has at least one segment, we have that $\min(s_1,s_2)\geq \lfloor s/3 \rfloor + 1\geq 13$ since $s\geq 36$. Taking into account the partial segments, the upper bound can be written as $\max(s_1,s_2)\leq s-\lfloor s/3 \rfloor + 2 + 5 \leq 2s/3 + 8\leq s-4$, since $s\geq 36$.
\end{proof}

We can use this separator to prove Lemma~\ref{lem:slicedskeleton}, but first let us consider a walk along the boundary of the unbounded face of the plane tree $R$. Notice that the walk uses each edge of $R$ twice, once in each direction. Additionally, whenever we encounter a branching vertex (a vertex of degree at least $3$) of $R$ that is an interior vertex of the segment we have been walking along, then we slice the current segment into two smaller segments at $v$, resulting in two (short or long) segments. As a result, we end up with a walk that consists of the original $O(k)$ segments, plus at most the number of branching vertices of $R$, which is at most $O(k)$, as there are at most $O(k)$ leaves in $R$. Consequently, the walk has at most $O(k)$ segments.

\begin{proof}[Proof of Lemma~\ref{lem:slicedskeleton}]
If the walk around $R$ consists of at most $35$ segments, then we set $\Skel=R$; this has all the desired properties. Otherwise, we recursively slice the region outside $R$ into smaller regions using Lemma~\ref{lem:sperner_sep}, until all regions have at most $35$ segments on their boundary, using the following procedure. 

For a partition $\calP=\{\calF_1,\dots,\calF_t\}$ of the plane where each $\calF_i$ is a non-empty union of faces of $H_\Delta$ and where $F_i$ has $s_i$ boundary segments, we define the potential function \[f(\calP)=\sum_{i=1}^t s_i -13t.\]
Notice that the contribution of region $\calF_i$ to $f(\calP)$ is $s_i-13$. By Lemma~\ref{lem:sperner_sep}, for any region $\calF_i$ appearing throughout the slicing we have $s_i\geq 13$, therefore each region contributes a non-negative integer to $f$. It follows that $f(\calP)\geq 0$ for all partitions $\calP$ that appear in the slicing.

If $F$ is the region boundary of an element of $\calP$ with $s\geq 36$ segments on its boundary, then by Lemma~\ref{lem:sperner_sep}, after the slicing it will be replaced by two regions, and the total number of segments on region boundaries increase by at most $12$. Therefore in the potential function $f$ the number $t$ increases by one, while the segment sum increases by at most $12$. Thus if $\calP'$ is the new partition we get by applying Lemma~\ref{lem:sperner_sep} on $F$, then we have
\[f(\calP')\leq f(\calP)+12-13 = f(\calP)-1.\]
Thus the potential is decreasing in each step. Initially we have a singleton partition with value $O(k)$, and we know that the potential remains non-negative, thus there can be at most $O(k)$ steps. It follows that $\Skel$ consists of $O(k)$ regions, each of which has $O(1)$ segments on their boundary. The long segments of any of the faces $\calF$ in $\Skel$ are either on the same directed path of the original solution $H$, or they are distant inside $\bar{H}^\calF_\Delta$. Indeed, this property is true for the initial singleton partition, the distances are preserved by the slicing, and the distance is not shortened when restricting to a subgraph defined by some region $\calF$.
\end{proof}
\subsection{Finding a $t$-tough pair}\label{sec:gridfind}

Let $\calF$ be a relevant face of the skeleton $\Skel$ given by Lemma~\ref{lem:slicedskeleton}. If $F$ is the unbounded face of $\Skel$, then we change the embedding of $H_\Delta$ so that $F$ is not the outer face. Let $H^\calF$ denote the subgraph of $H$ induced by the vertices in $\calF$, and let $B$ denote the vertices that are on the boundary of $\calF$. We denote by $\bd \calF$ the edges of $H^\calF$ that are on the boundary of $\cal F$, and let $\inter \calF$ denote the edge set $E(H^\calF) \setminus \bd \cal F$. (Note that $\calF$ may contain undirected triangulation edges that are not contained in $\bd \calF$ or $\inter \calF$.) By Lemma~\ref{lem:slicedskeleton} we know that $\calF$ consists of $c_\Skel=O(1)$ segments. Observe that if we decompose the short segments on the boundary $\bd\calF$ into length-$0$ paths, then we can think of it as a collection of at most $L\cdot c_\Skel$ directed paths (where the long segments appear as themselves, and vertices only incident to short segments as singletons).

Given an (essential) edge $e$ in $\inter \calF$ (i.e., not on the boundary $\calF$), there exists a directed path in $H^\calF$ connecting two distinct vertices of $B$ using $e$ to satisfy the demand for which $e$ is essential. Note that the edges of this path are all in $\inter\calF$. Such a path is called an \emph{essential path} through $e$.

\begin{lemma}\label{lem:boundary_degree}
For each $v\in B$ we have that $\deg_{H^\calF}(v)\leq 2L\cdot c_\Skel+2$.
\end{lemma}

\begin{proof}
Let $vu$ be an edge where $v\in B$ and $u\not \in B$. Since $vu$ is essential for some demand, there is some essential path $P$ through $vu$; let $v'$ be the other endpoint of $P$. Suppose that $w'\in B$ is reachable from $v'$ along $\bd \calF$. Then there can be no essential path from $v$ to $w'$ that avoids $vu$, as the connection is already established by $P$ and the path along $\bd \calF$ from $v'$ to $w'$. In particular, there can be at most one edge leaving $v$ where the corresponding essential path ends on a given directed path of $\bd\calF$. Since $\bd\calF$ consists of at most $L\cdot c_\Skel$ directed paths, we have that there are at most $L\cdot c_\Skel$ inner edges leaving $v$. An analogous argument for the incoming edges plus the at most two boundary edges proves the desired upper bound on the degree of $v$.
\end{proof}

We will now consider the number of demands for which the edges of $H^\calF$ are essential. In what follows, let $\kappa$ be the number of such demands, i.e., suppose that $H\setminus \inter \calF$ fails to satisfy $\kappa$ of the demands.

\begin{lemma}\label{lem:essentials_intersectioncount}
If $H\setminus \inter \calF$ fails $\kappa$ demands, then the total branching degree of $H^\calF$ is at most $\kappa(\kappa-1)$.
\end{lemma}

\begin{proof}
First we show that the internal vertices of $H^\calF$ have bounded total branching degree. To do so, we first bound the number of intersections between paths satisfying different demands.

We claim that there can be at most one intersection point where there are incoming edges of $H^\calF$ essential for a given pair of demands. Suppose the contrary: that $v,w\in V(H^\calF)$ are distinct vertices such that both of them have an incoming edge essential for demand edge $d$ (denoted by $e(d,v)$ and $e(d,w)$ and an incoming edge essential for demand edge $d'$, denoted by $e(d',v)$ and $e(d',w)$. Note that either $v$ is reachable form $w$ or vice versa, as otherwise we could remove one of $e(d,v)$ and maintain the connection of $d$ via $w$. Assume without loss of generality that there is a path $P$ from $v$ to $w$. Note that $P$ cannot use both $e(d,w)$ and $e(d',w)$; suppose that it does not use $e(d,w)$.  Since $e(d,w)$ is essential for $d$, there is a directed path $Q\subset E(H^\calF)$ through it\footnote{Here $Q$ may use boundary edges of $H^\calF$, i.e., it is not necessarily an essential path.} satisfying the demand $d$. Note that $Q$ must also pass through the edge $e(d,v)$, and in particular, contains vertex $v$.

We can therefore use an initial part of $Q$ to get to $v$: note that this part of $Q$ is disjoint from $P$ as otherwise there would be a closed walk and thus a cycle in $H$. We use $P$ to go from $v$ to $w$, and continue on $Q$ after $w$; as $Q$ goes through $e(d,w)$, it can be continued from $w$. The resulting path $Q'$ essentially replaces $Q[v,w]$ with $P$, and has the same endpoints as $Q$, making $e(d,w)$ non-essential for $d$, a contradiction.

The analogous argument for outgoing edges gives the same bound. Imagine labeling each unordered pair of edges that have the same head with an unordered pair of demands, such that one edge is essential for one demand and the other edge is essential for the other. The above argument implies that these labels must be distinct for all pairs of edges sharing the same head. The analogous can be done for edge pairs that share the same tail. Consequently, the vertices of $H^\calF$ satisfy 
\begin{align*}
\sum_{v\in V(H^\calF)} \binom{\rho(v)}{2} &\leq \binom{\kappa}{2}\\
\sum_{v\in V(H^\calF)} \binom{\delta(v)}{2} &\leq \binom{\kappa}{2},
\end{align*}
where $\rho(v)$ and $\delta(v)$ denote the in- and outdegree of $v$ in $H^\calF$, respectively. On the other hand, observe that $\binom{\rho(v)}{2} + \binom{\delta(v)}{2} \geq d^*(v)$, thus we have that the total branching degree is at most $2\binom{\kappa}{2}=\kappa(\kappa-1)$.
\end{proof}

\begin{lemma}\label{lem:cbound_face_to_full}
If for each relevant face $\calF$ of $\Skel$ we have that $H\setminus \calF$ fails at most $\kappa$ demands, then, then the total branch degree of $H$ is at most $O(\kappa^2)\cdot k$.
\end{lemma}

\begin{proof}
By Lemma~\ref{lem:essentials_intersectioncount}, we have that each relevant face $\calF$ has total branch degree at most $\kappa(\kappa-1)$. Since each vertex of $H$ has degree at least $3$, the same holds for the vertices of $H^\calF$ that are in $\inter \calF$. Thus all of the inner vertices contribute at least $1$ to the branch degree of $H^\calF$; it follows that there are at most $\kappa(\kappa-1)$ inner vertices. It follows that the total degree of the inner vertices is at most $\kappa(\kappa-1)+2\kappa(\kappa-1)$, since the degree and branch degree differs by at most $2$. Let $B_\calF$ denote the vertices of $\bd\calF$, and let $\hat B_\calF$ denote those vertices that are adjacent to some inner vertex of $\calF$ in $H^\calF$. Since the total degree of inner vertices is at most $3\kappa(\kappa-1)$, this also bounds the number of vertices in $\hat B_\calF$: we have $|\hat B_\calF|\leq 3\kappa(\kappa-1)$.

By Lemma~\ref{lem:slicedskeleton} we know that $\Skel$ has $O(k)$ faces, thus the sum of the branch degrees of all relevant faces is $m=O(\kappa^2k)$. We claim that the total branch degree of $H$ is at most $O(k)$ larger than $3m$. Notice that the difference between the two amounts is due to the vertices of $H$ that lie on the shared boundary of some relevant faces of $\Skel$. Suppose now that $v$ has degree $2$ in $\Skel$ with neighboring relevant faces $\calF$ and $\calF'$. Then $d^*_H(v)\leq d^*_{H^\calF}(v)+d^*_{H^{\calF'}}(v) + 2$. Therefore, the total contribution from vertices in $\hat B_\calF$ is at most $2\kappa(\kappa-1)$, and all other vertices have all adjacent edges in the other face, so they satisfy $d^*_{\calF'}(v)=d^*_H(v)$.

Suppose now that $v$ has degree $d_v\geq 3$ in $\Skel$, with neighboring faces $\calF_1,\dots\calF_{d_v}$ and corresponding branch degrees $d^*_1(v),\dots, d^*_{d_v}(v)$. Then we have
\[ d^*_H(v)\leq \sum_{i=1}^{d_v} d^*_i(v) + 2(d_v -1).\]
Consequently, the branch degree of $\calH$ differs by at most $2\sum_v(d_v -1)$, where the sum goes over vertices of $\Skel$ of degree at least $3$. Note that $\Skel$ is a plane graph with $O(k)$ faces, thus the sum is at most $O(k)$, concluding the proof.
\end{proof}

We say that a directed edge $uv$ enters (resp. exits) a path $P$ if $u \not\in V(P)$ and $v\in P$ (resp. $u\in V(P)$ and $v\not\in V(P)$).

\begin{lemma}\label{lem:disjoint_essential}
Suppose that $P$ and $P'$ are essential paths for the edges $e,e'$ respectively, and that the directed paths $Q_s,Q,Q_t$ do not contain these edges. Suppose that $P$ and $P'$ exit the directed path $Q_s$ before $e$ and $e'$ respectively, and the edges $e,e'$ enter the same directed path $Q$. Alternatively, suppose that $e,e'$ exit the same directed path $Q$, and after them $P$ and $P'$ enter the same directed path $Q_t$. Then $P$ and $P'$ are vertex-disjoint, and the paths $Q_s,Q$ (respectively, $Q,Q_t$) have the same ``direction'', that is, their intersections with $P$ and $P'$ appear in the same order.
\end{lemma}

\begin{proof}
Let $e=uv$ and $f=u'v'$, and let $P$ and $P'$ be essential paths through $e$ and $f$ respectively, and suppose that $uv$ and $u'v'$ both enter a directed path. Let $Q$ denote the directed path from $v$ to $v'$. Let $s,t$ and $s',t'$ denote the start- and endpoints of $P$ and $P'$, respectively. We observe that the directed path $Q_s$ must be oriented from $s$ to $s'$: indeed, if it is oriented the other way, then $P'$ cannot be an essential path for $u'v'$, as we can use the path $s'\xrightarrow{Q_s} s\xrightarrow{P} v \xrightarrow{Q} v'$ instead of $P'[s',v']$, and we claim that this path avoids $u'v'$. In case of $Q_s$ and $Q$ do not contain $u'v'$, so one only needs to check $P$. But if $u'v'\in P[s,v]$, then we could create closed walk together with $Q$ and contradict acyclicity.

Now suppose for the sake of contradiction that $P$ and $P'$ intersect at some vertex $p$. We distinguish three cases:
\begin{description}
\item[Case 1:] $p$ occurs after the edge $uv$ on $P$.\\
This contradicts the essentiality of $uv$, as one can use $s\xrightarrow{Q_s} s' \xrightarrow{P'} p$ instead of $P[s,p]$.
\item[Case 2:] $p$ occurs before the edge $uv$ on $P$, and before the edge $u'v'$ on $P'$.\\
This contradicts the essentiality of $u'v'$, as one can use $p\xrightarrow{P} v \xrightarrow{Q} v'$ instead of $P'[p,v']$.
\item[Case 3:] $p$ occurs before the edge $uv$ on $P$, but after the edge $u'v'$ on $P'$.\\
This contradicts the acyclicity of $H$, as $p \xrightarrow{P} v \xrightarrow{Q} v' \xrightarrow{P'} p$ is a closed walk.
\end{description}

Suppose now that $uv$ and $u'v'$ both exit a directed path. Let $Q$ denote the directed path from $u$ to $u'$. We observe that the directed path $Q_t$ must be oriented from $t$ to $t'$: indeed, if it is oriented the other way, then $P$ cannot be an essential path for $uv$, as we can use the path $u \xrightarrow{Q} u'\xrightarrow{P'} t' \xrightarrow{Q_t} t$ instead of $P[u,t]$. Again this path avoids $uv$ as if $uv\in P'[u',t']$ then together with $Q$ we would get a closed walk. 

Now suppose for the sake of contradiction that $P$ and $P'$ intersect at some vertex $p$. We distinguish three cases:
\begin{description}
\item[Case 1:] $p$ occurs before the edge $u'v'$ on $P'$.\\
This contradicts the essentiality of $u'v'$, as one can use $p \xrightarrow{P} t \xrightarrow{Q_t} t'$ instead of $P'[p,t']$.
\item[Case 2:] $p$ occurs after the edge $u'v'$ on $P'$, and after the edge $uv$ on $P$.\\
This contradicts the essentiality of $uv$, as one can use $u \xrightarrow{Q} u' \xrightarrow{P'} p$ instead of $P[u,p]$.
\item[Case 3:] $p$ occurs after the edge $u'v'$ on $P'$, but before the edge $uv$ on $P$.\\
This contradicts the acyclicity of $H$, as $p \xrightarrow{P} u \xrightarrow{Q} u' \xrightarrow{P'} p$ is a closed walk.
\end{description}

This concludes the proof.
\end{proof}

A \emph{bundle} is a collection of pairwise vertex-disjoint essential paths of $\calF$ that exit a directed path $P_s$ of $\bd \calF$ and enter a directed path $P_t$ of $\bd\calF$. For a cycle $C$ in $\bar{H}^\calF_\Delta$ let $H^C$ denote the edges of $H$ that are inside\footnote{That is, in the fixed embedding, these edges lie entirely in the bounded region defined by $C$. Recall that all relevant faces of $\Skel$ are bounded.} $C$ or on $C$.

\begin{definition}[Grid structure]
A grid structure of size $\lambda$ is a cycle $C$ in $\bar{H}^{\calF}_\Delta$ and a pair of directed path sets $\calB_P = \{P_1,\dots,P_{\lambda}\}$ and $\calB_Q = \{Q_1,\dots,Q_\lambda\}$, where the paths are in $H^C$, and the following properties hold. See Figure~\ref{fig:griddef} for an illustration.
\begin{enumerate}[itemsep=1em,label=\textit{(\roman*)}]
\item Path $P_i$ starts at $s^P_i\in V(C)$ and ends at $t^P_i\in V(C)$. Similarly, path $Q_j$ starts at some $s^Q_j\in V(C)$ and ends at $t^Q_j\in V(C)$. Apart from their start- and endpoints, each $P_i$ and $Q_j$ is vertex-disjoint from $C$. The paths of $\calB_P$ are pairwise vertex-disjoint, and the paths of $\calB_Q$ are pairwise vertex disjoint.
\item The cycle $C$ contains the vertices
\[s^P_1,s^P_2,\dots,s^P_\lambda, \quad t^Q_1,t^Q_2,\dots,t^Q_\lambda, \quad t^P_\lambda,t^P_{\lambda-1},\dots,t^P_1, \quad s^Q_\lambda,s^Q_{\lambda-1},\dots,s^Q_1\]
in this cyclic order, or reversed.
\item There is a directed path $Q_s\subset C$ through $s^P_1,s^P_2,\dots,s^P_\lambda$. Similarly, the directed path $Q_t\subset C$ goes through $t^P_1,t^P_2,\dots,t^P_\lambda$, the path $P_s\subset C$ goes through  $s^Q_1,s^Q_2,\dots s^Q_\lambda$, and the path $P_t\subset C$ goes through $t^Q_1,t^Q_2,\dots,t^Q_\lambda$.
\item Each $P_i$ intersects each $Q_j$ in some non-empty connected subpath.
\item Each $P_i$ is a subpath of an essential path for some edge $e_i \in P_i$, and each $Q_j$ is a subpath of an essential path for some edge $e'_j \in Q_j$. Moreover, $e_i\not\in Q_j$ and $e'_j\not\in Q_i$ for any $i,j\in [\lambda]$.
\item For each $1\leq i_1<i_2\leq \lambda$ the head of $e_{i_2}$ is reachable from the tail of $e_{i_1}$ within $H^C$ using a path that avoids  $e_{i_1}$ or $e_{i_2}$. Similarly, the head of $e'_{i_2}$ is reachable from the tail of $e'_{i_1}$ within $H^C$ using a path that avoids  $e'_{i_1}$ or $e'_{i_2}$.
\end{enumerate}
\end{definition}

\begin{figure}[t]
\centering
\includegraphics{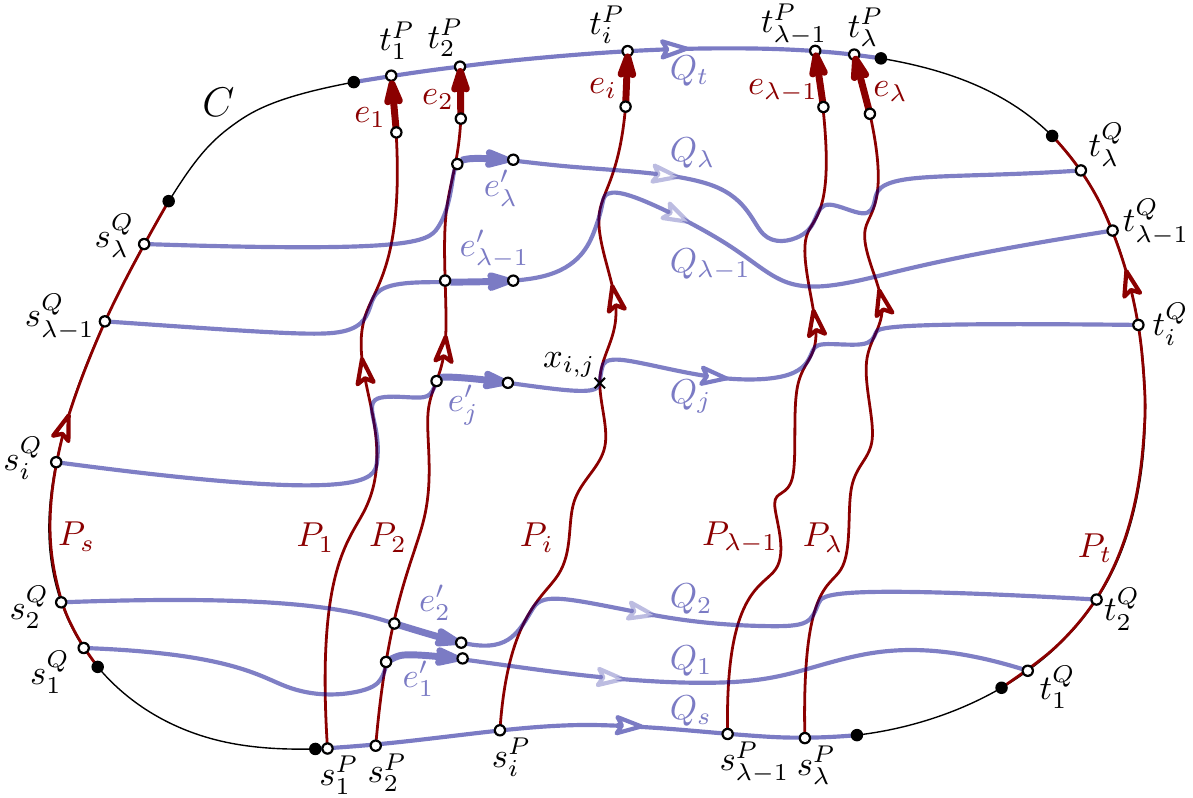}
\caption{A grid structure with cycle $C$ and path sets $\calB_P,\calB_Q$. Thick edges are the defining essential edges. }\label{fig:griddef}
\end{figure}

Let us fix a vertex $x_{i,j}\in P_i\cap Q_j$ for each $i,j\in [\lambda]$. The embedding ensures that $x_{i,1},\dots,x_{i,\lambda}$ appear in this order on $P_i$, and $x_{1,j},\dots,x_{\lambda,j}$ appear in this order on $Q_j$. Observe that contracting the edges of $P_i\cap Q_j$ as well as all edges that have an incident degree-$2$ edge results in a plane grid with grid lines $P_i$ and $Q_j$.

Our task now is to find a large grid structure. Consider a bundle $\calB$ consisting of paths $P_1,\dots,P_t$ from long segment $A$ to long segment $B$ where the indices are according to the order of starting points on $A$. Suppose that $P_{t/4}$ starts at $a_1\in A$ and ends in $b_1\in B$, and that $P_{3t/4}$ starts at $a_2\in A$ and ends in $b_2\in B$. The \emph{spread} of $\calB$ is defined as $\dist_{\bar A}(a_1,a_2)+\dist_{\bar B}(b_1,b_2)$.

\begin{lemma}\label{lem:findgrid}
Suppose that $\kappa>L^4$ and that $L$ is large enough. If $H\setminus \inter \calF$ fails $\kappa$ demands, then $H^\calF$ has a grid structure of size $\Omega(L)$.
\end{lemma}

\begin{proof}
By Lemma~\ref{lem:boundary_degree} we have that each vertex on $\bd\calF$ has degree $O(Lc_\Skel)$. Since there are at most $Lc_\Skel$ vertices on $\bd \calF$ that are on some short segment, we have that at least $\kappa-O(L^2c^2_\Skel)$ demands whose essential paths go through long segments. Consequently, there exists a long segment $S$ that has $\frac{\kappa-O(L^2c^2_\Skel)}{c_\Skel}$ demands going through it. At least half of these demands are exiting $S$ or entering $S$; suppose the former, i.e., there are at least $\frac{\kappa-O(L^2c^2_\Skel)}{2c_\Skel}$ demands exiting $S$ (the entering case can be handled analogously). Since vertices of segment $S$ have degree at most $O(Lc_\Skel)$, we have that among these demands exiting $S$, there must be at least $\frac{\kappa-O(L^2c^2_\Skel)}{2c_\Skel\cdot O(Lc_\Skel)}=\Omega(\kappa/L)$ demands exiting $S$ that have pairwise distinct starting vertices.

Each of these demands have some edge exiting $S$, and for each of these edges let us fix a corresponding essential path. Among these essential paths, at most $O(L^2c^2_\Skel)$ of them end on short segments, and among the rest, at least $\frac{1}{c_\Skel}$ proportion of them end on the same long segment of $\bd\calF$. Thus there are at least 
$\frac{\Omega(\kappa/L)-O(L^2c^2_\Skel)}{c_\Skel}=\Omega(\kappa/L)$
essential paths among them that also end on the same long segment of $\bd \calF$. Among these paths, at least $\frac{1}{O(Lc_\Skel)}$ proportion of them have pairwise distinct ending points. Thus, there exists a bundle of size at least $\Omega(\kappa/L^2)$ where either all defining essential edges exit a long segment of $\bd \calF$, or all defining essential edges enter a long segment of $\bd \calF$. We fix a constant $c^*$ such that there exists a bundle of size at least $c^*\kappa/L^2$. Let $\calB$ be the bundle of minimum spread that has exactly $\lambda:=\lfloor L/(4c^2_\Skel)\rfloor$ and where the defining edges are all entering or all exiting a long segment.  Since $L$ is large enough, we may assume $L>1/c^*$, so $c^*\kappa/L^2>\kappa/L^3>L>\lambda$, so such a bundle $\calB$ exists.

By Lemma~\ref{lem:disjoint_essential}, we have that the paths of $\calB$ are pairwise vertex-disjoint. Moreover, by the same lemma, we can index  the paths of $\calB$ by the order in which their starting points occur in $\bd \calF$ as $P_1,\dots,P_\lambda$. By the essentiality of the defining edges of $\calB$, we know that the target of a path cannot be reachable from its source, thus the starting and ending long segments of $\calB$ are distinct and they cannot be subpaths of the same directed path of $H$. Thus Lemma~\ref{lem:slicedskeleton} implies that the starting and ending long segment of $\calB$ are distant, therefore there are at least $L$ vertices on the middle path $P_m$ of $\calB$, where we set $m=\lfloor \lambda/2 \rfloor$. Note that each vertex of $P_m$ has degree at least $3$, thus each of these vertices has an incident edge that is not on $P_m$. At least half of these edges are on the same side of $P_m$ in the embedding, and at least half of them are all entering or all exiting $P_m$. Consider now the essential paths of the $L/4$ edges selected this way. At least $1/c^2_\Skel$ proportion of them have the same starting and ending segment, thus there are at least $\lambda=\lfloor L/(4c^2_\Skel)\rfloor$ paths all entering or all exiting $P_m$ in the same direction. Let $\calB'$ be this size-$\lambda$ collection of essential paths.

We can index the paths of $\calB'$ by the order in which their essential edge occur on $P_m$ as $P'_1,\dots,P'_\lambda$. Let $e'_j$ the essential edge of $P'_j$.

Let $S_\calB, T_\calB, S'_\calB, T'_\calB$ be starting and ending segments of $\calB$ and $\calB'$, respectively. We denote by $e_i$ and $e'_i$ the essential edge of $P_i$ and $P'_i$. Suppose now that $P_i$ and $P'_j$ intersect more than once, and let $x$ and $x'$ be the first and last intersection along $P_i$. By acyclicity of $H$, we know that $x$ and $x'$ are also the first and last intersection along $P'_j$. Observe that the $e_i$ and $e'_j$ cannot occur between $x$ and $x'$, as that would allow us to circumnavigate an essential edge using a portion of the other path. Consequently, we can change $P'_j$ by exchanging $P'_j[x,x']$ with $P_i[x,x']$; the result is still an essential path for the edge $e'_j$. By making such changes exhaustively, we can ensure that for each $i,j\in [\lambda]$ if $P_i$ intersects $P_j$, then their intersection is a connected (possibly one-vertex) subpath. 

We now distinguish several cases based on what segment $\calB'$ starts and ends on.  Note that $S_\calB\neq T_\calB$ and $S'_\calB\neq T'_\calB$ by essentiality of $e_1$ and $e'_1$, and acyclicity. Assume without loss of generality that $\calB$ goes from bottom to top, with both $S_\calB$ and $T_\calB$ oriented left to right. See Figure~\ref{fig:findgrid}	 for an illustration.

\paragraph*{Case 1.} $S'_\calB\not\in \{S_\calB,T_\calB\}$ (or symmetrically, $T'_\calB\not\in \{S_\calB,T_\calB\}$).
First we show that the paths of $\calB'$ are pairwise disjoint. Suppose that the edges $e'_j$ are entering $P_m$ from the left. Then we can apply Lemma~\ref{lem:disjoint_essential} with $P_1$ playing the role of $Q_s$ and $P_m$ playing the role of $Q$.  If they are exiting $P_m$ on the right, and $T'_\calB \not\in \{S_\calB,T_\calB\}$, then Lemma~\ref{lem:disjoint_essential} is applied with $Q=P_m$ and $Q_t=P_\lambda$. If $T'_\calB \in \{S_\calB,T_\calB\}$, then Lemma~\ref{lem:disjoint_essential} is applied with $Q=P_m$ and $Q_t=T'_\calB$. All remaining cases can be handled with analogous invocations of Lemma~\ref{lem:disjoint_essential}.

Suppose that the edges $e'_i$ enter $P_m$ from the left or exit it on the right. Because of the embedding it follows that all paths of $\calB'$ intersect $P_1,P_2,\dots,P_m$. We claim that if the edges $e'_i$ exit $P_m$, then the paths $P'_i$ also intersect the path $Z=S_\calB[m,m+1]\cup P_{m+1}$, where $S_\calB[m,m+1]$ denotes the portion of $S_\calB$ between the starting point of $P_m$ and $P_{m+1}$. Note that $P'_i$ enters the inside of the closed curve $P_{m} \cup S_\calB[m,m+1] \cup P_{m+1} \cup T_\calB[m,m+1]$, so $P'_i$ must intersect $Z$ or $T_\calB[m,m+1]$ after passing $e'_i$ (because of the earlier simplifaction it cannot intersect $P_m$ again). If $P'_i$ enters $T_\calB[m,m+1]$, then its portion containing $e'_i$ can be curcumnavigated on a directed subpath of $P_m\cup T_\calB[m,m+1]$, contradicting the essentaility of $e'_i$.

If $e'_i$ enters $P_m$, then we set $Z=P_m$.

Consider the cycle $C$ of $\bar H$ formed by $P_1,P'_1,P'_m,Z$ (Since $P_1,Z$ and $P'_1,P'_m$ are vertex-disjoint, and the other pairs have a connected intersection, there is a unique cycle in the union $P_1 \cup P'_1 \cup P'_m \cup Z$.) One can verify that $C$, $\{P_2|_C,\dots,P_{m-1}|_C\}$, and $\{P'_2|_C,\dots,P'_{m-1}|_C\}$ form a grid structure, where $P|_C$ denotes the portion of the path $P$ that falls in the interior of the bounded region of $C$.

The case when the edges $\calB'$  enter $P_m$ from the right or exit it to the left can be handled analogously, using the paths $P_m,\dots,P_\lambda$ instead of $P_1,\dots, P_m$, and setting $Z=P_{m-1}\cup T_\calB[m-1,m]$ or $Z=P_m$ for entering/exiting edges $e'_i$. The cycle $C$ is defined by $Z,P_\lambda,P'_m,P'_\lambda$, and the grid is given by $C$, $\{P_{m+1}|_C,\dots,P_{\lambda-1}|_C\}$, and $\{P'_{m+1}|_C,\dots,P'_{\lambda-1}|_C\}$. In all cases the grid has size at least $\lfloor \lambda/2 \rfloor -2$.

\begin{figure}[t]
\centering
\includegraphics{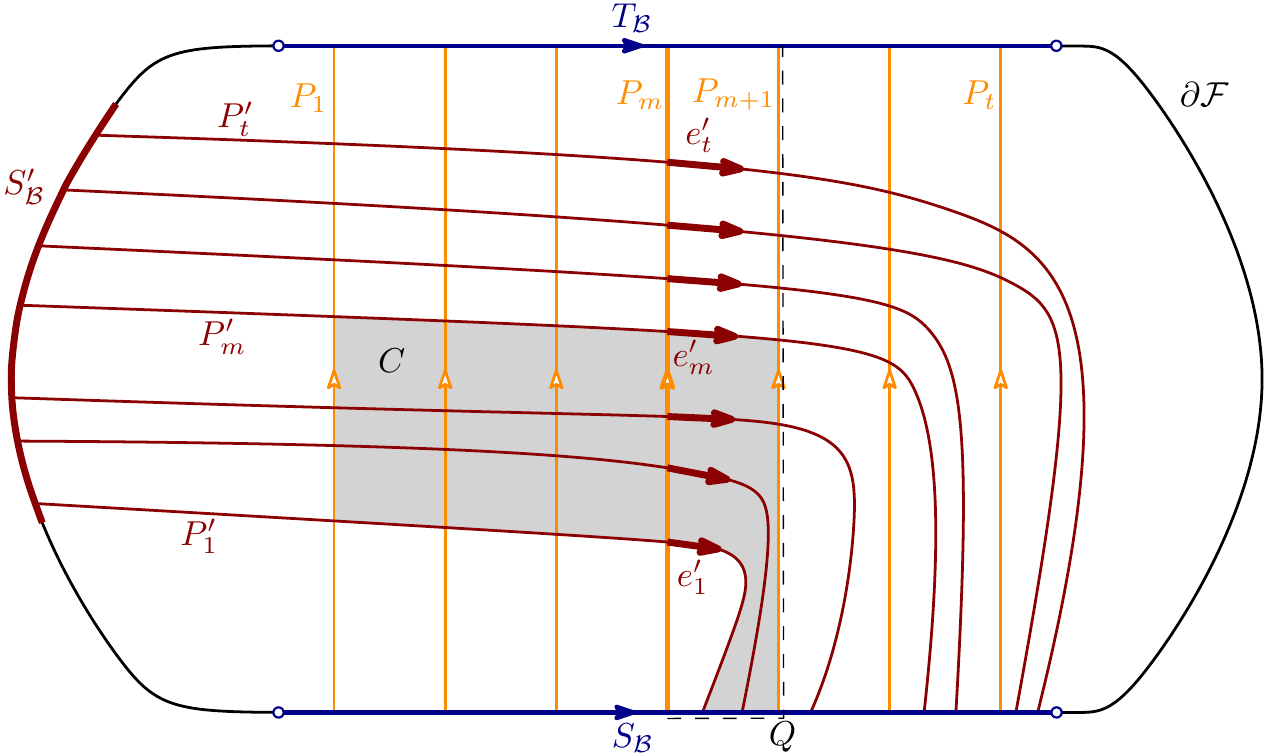}
\caption{Finding a grid structure based on a bundle $\calB$ (orange). The edges $e'_j$ exit the middle path $P_m$ on the right. The bundle $\calB'$ (red) has a distinct starting segment $S'_\calB$, as in Case 1. The directed path $Z$ is depicted with dashed line, and the cycle $C$ around the found grid is the boundary of the gray shaded region.}\label{fig:findgrid}
\end{figure}

\paragraph*{Case 2.} $S_\calB=S'_\calB$ and $T_\calB=T'_\calB$. We claim that this case cannot occur.
If $e'_i$ exits $P_m$ on the right, then we can exchange the portion of $P'_i$ starting at $e'_i$ with a part of $P_m$ and $T_\calB$, contradicting the essentiality of $e'_i$. If $e'_i$ enters $P_m$ from the left, then we can exchange the portion of $P'_i$ ending at $e'_i$ with a part of $S_\calB$ and $P_m$, contradicting the essentiality of $e'_i$.

If $e'_i$ exits $P_m$ on the left, or enters $P_m$ from the right, then we can exchange the portion of $P_m$ starting at $e_m$ with a part of $S_\calB$ and $P'_m$, contradicting the essentiality of $e_m$.

\paragraph{Case 3.} $S_\calB=T'_\calB$ and $T_\calB=S'_\calB$.
First we note that $e'_i$ cannot enter $P_m$ from the right or exit it on the left, as both would create a cycle. We can invoke Lemma~\ref{lem:disjoint_essential} on any pair of paths of $\calB'$ with $Q=P_m$ and either $Q_s=S'_\calB$ or $Q_t=T'_\calB$ to prove that the paths of $\calB'$ are pairwise vertex-disjoint.

Since $\calB$ has minimum spread, we have that $\calB'$ has a spread at least as big. Recall that if $\calB$ goes from $A=S_\calB$ to $B=T_\calB$, then $P_{\lambda/4}$ starts at $a_1\in A$ and ends in $b_1\in B$, while $P_{3\lambda/4}$ starts at $a_2\in A$ and ends in $b_2\in B$. Now $\calB'$ goes from $B$ to $A$, so we set the start and endpoint of $P'_{\lambda/4}$ as $b'_1\in B$ and $a'_1\in A$, and similarly, the start and end of $P'_{3\lambda/4}$ as $b'_2\in B$ and $a'_2\in A$. Now $\dist_{\bar A}(a_1,a_2)+\dist_{\bar B}(b_1,b_2) \leq \dist_{\bar A}(a'_1,a'_2)+\dist_{\bar B}(b'_1,b'_2)$, thus at least one of the inequalities $\dist_{\bar A}(a_1,a_2) \leq \dist_{\bar A}(b'_1,b'_2)$ and $\dist_{\bar B}(b_1,b_2) \leq \dist_{\bar B}(a'_1,a'_2)$ holds.

Suppose that the latter inequality holds. Then because of the embedding we have that $P'_1,\dots,P'_{\lfloor \lambda/4\rfloor}$ all intersect $P_{\lceil \lambda/4 \rceil},\dots, P_{\lfloor \lambda/2 \rfloor}$. We can then define a grid for these two smaller bundles of size at least $\lambda/4-2$ as seen in Case 1 by imagining that the path $T_\calB$ is split into two shorter paths, the first part containing the endpoints of $P'_1,\dots,P'_{\lfloor \lambda/4\rfloor}$, and the second containing the starting points of $P_{\lceil \lambda/4 \rceil},\dots, P_{\lfloor \lambda/2 \rfloor}$. Similarly, if the former inequality ($\dist_{\bar B}(b_1,b_2) \leq \dist_{\bar B}(a'_1,a'_2)$) holds, then $P'_{\lceil 3\lambda/4 \rceil},\dots,P'_{\lambda}$ all intersect $P_{\lceil \lambda/2 \rceil},\dots, P_{\lfloor 3\lambda/4 \rfloor}$, and the grid can again be defined analogously to Case 1.\\

In all possible cases we have shown that a grid of size $\Omega(\lambda)=\Omega(L)$ exists which completes the proof.
\end{proof}

The usefulness of the grid structure is demonstrated by the following lemma.

\begin{lemma}\label{lem:gridsarehard}
If $C,\calB$, and $\calB'$ form a grid structure of size $\lambda$, then the demands corresponding to a subset of their essential paths form a  $(\lambda-2)$-tough pair.
\end{lemma}

\begin{proof}
First we show that the demands corresponding to the essential paths $\calB_P=\{P_1,\dots,P_\lambda\}$ are pairwise weakly independent. For a path $P_i$ let $d^P_i=u^P_iv^P_i$ be the corresponding minimal demand, which is served by some path $P^{ex}_i$ that is an extension of $P_i$. (Similarly, the minimal demand of $Q_j$ is $d^Q_j=u^Q_jv^Q_j$, and it is served by the path $Q^{ex}_j$.) See Figure~\ref{fig:findtough} for an illustration.

Suppose now that $d^P_i$ and $d^P_j$ are not weakly independent, where $i<j$. There cannot be a demand $v^P_ju^P_i$, as a path satisfying this demand together with $P^{ex}_i[u^P_i,x_{i,1}]\cup Q_1[x_{i,1},x_{j,1}] \cup P^{ex}_j[x_{j,1},v^P_j]$ forms a closed walk, i.e., contradicts acyclicity. Thus weak independence must be violated by a demand $v^P_iu^P_j$; let $P_d$ be the path corresponding to this demand. We claim that this contradicts the essentiality of $e'_2$, as $Q^* = Q_2[s^Q_2,t^Q_2]$ can be exchanged with the path
\[P^* = s^Q_2 \xrightarrow{P_s} s^Q_3 \xrightarrow{Q_3} x_{i,3} \xrightarrow{P^{ex}_i} v^P_i \xrightarrow{P_d}
u^P_j \xrightarrow{P^{ex}_j} x_{j,1} \xrightarrow{Q_1}  t^Q_1  \xrightarrow{P_t} t^Q_2.\]

\begin{figure}[t]
\includegraphics[width=\textwidth]{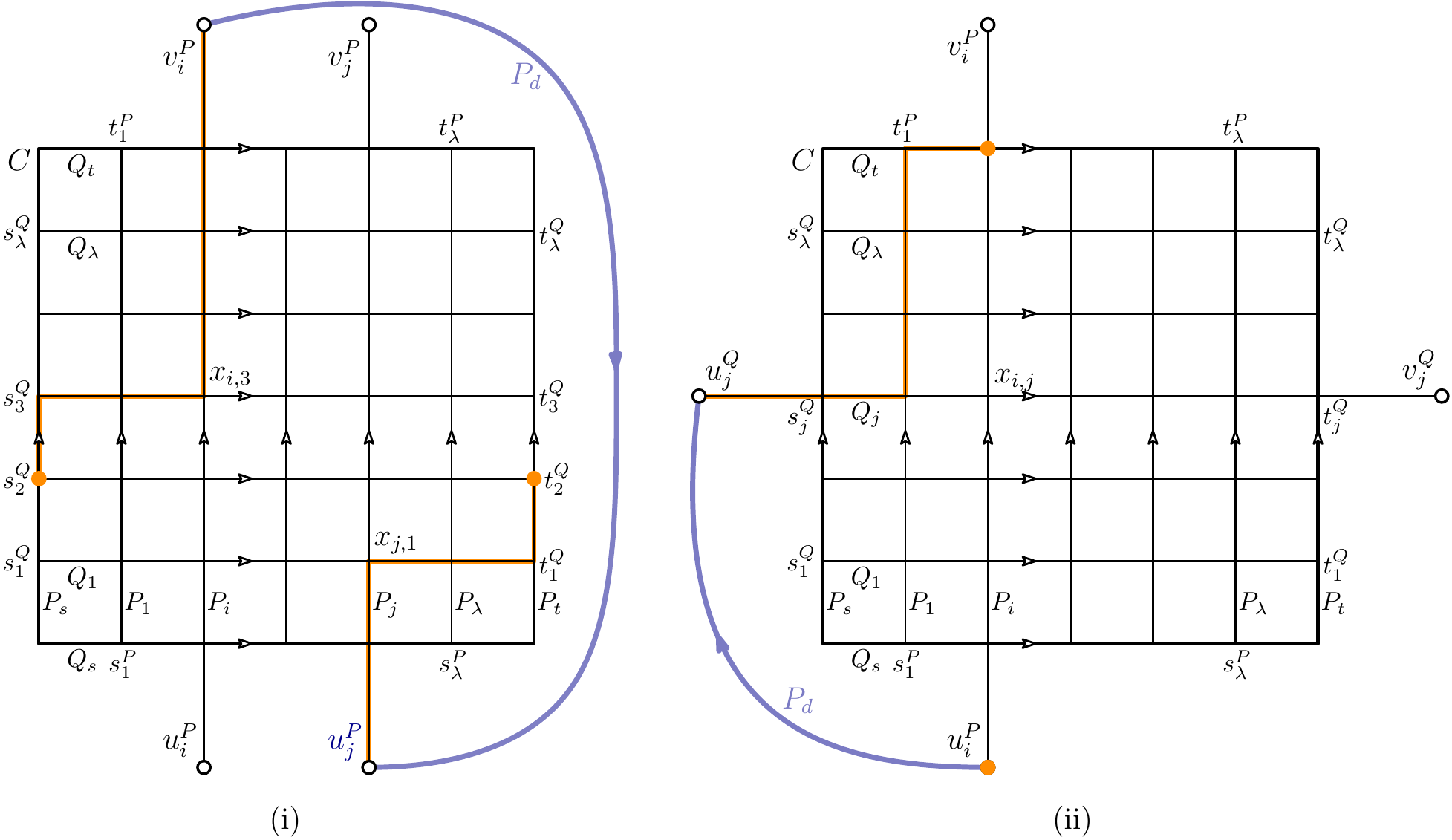}
\caption{(i) Proving weak independence: a demand $v^P_i u^P_j$ leads to the path $Q_2$ being avoidable, contradicting essentiality of $e'_2$. (ii) Proving strong independence. A demand $u^P_iu^Q_j$ leads to the edge $e_i$ being avoidable on $P^{ex}_i$, contradicting its essentiality for $u^P_iv^P_i$. In both sides the path $P^*$ is represented as the concatenation of the path(s) with orange background and the blue path~$P_d$.}\label{fig:findtough}
\end{figure}

Unless $P^*$ also contains $e'_2$, this exchange contradicts the essentiality of $e'_2$.
Suppose now that $P^*$ does contain $e'_2$; by the properties of the grid we know that $e'_2$ cannot lie on any of $P_s,Q_3,P^{ex}_i,P^{ex}_j,Q_1,P_t$, thus it could only be contained in $P_d$. Notice that if $e'_2$ is not essential for the demand of $P_d$, then we can change $P_d$ to exclude $e'_2$ and get a contradiction as above. Thus in what follows, we assume that $e'_2$ is an essential edge for $P_d$.

Notice that the same argument can be repeated for $e'_3$, thus $P_d$ must contain $e'_3$ as an essential edge. By Property (vi) of grid structures, the head of $e'_3$ is reachable from the tail of $e'_{2}$. It follows that $e'_2,e'_3$ must appear in this order on $P_d$, as otherwise we could create a closed walk ($\tail(e'_2) \rightarrow \head(e'_3) \xrightarrow{P_d} \tail(e'_2)$), contradicting acyclicity. Consequently, $P_d$ has a subpath \[\tail(e'_2) \xrightarrow{e'_2} \head(e'_2)\xrightarrow{P_d} \tail(e'_3) \xrightarrow{e'_2} \head(e'_3),\] which we could replace with the guaranteed path $\tail(e'_2) \rightarrow \head(e'_3)$ from the grid structure that avoids $e'_2$ or $e'_3$, contradicting the essentiality of either $e'_2$ or $e'_3$ for $P_d$.
The weak independence of the paths $Q_j$ can be proven symmetrically (by switching the role of $P$ and $Q$).

Next we show strong independence of $d^P_i$ and $d^Q_j$ for all $i,j\in \{2,3,\dots,\lambda-1\}$. Notice that this is sufficient, as it shows that the demands $d^P_2,\dots,d^P_{\lambda-1}$ and $d^Q_2,\dots,d^Q_{\lambda-1}$ form a $(\lambda-2)$-tough pair. Observe that having a demand $v^Q_ju^P_i$ or $v^P_iu^Q_j$ for any $i,j$ would create a closed walk:
\[
x_{i,j} \xrightarrow{Q^{ex}_j} v^Q_j \xrightarrow{} u^P_i \xrightarrow{P^{ex}_i} x_{i,j} \quad
\text{ and } \quad
x_{i,j} \xrightarrow{P^{ex}_i} v^P_i \xrightarrow{} u^Q_j \xrightarrow{Q_j} x_{i,j},\]
respectively, contradicting acyclicity.
Suppose now that there is a demand $u^P_iu^Q_j$ served by a path $P_d$, where $i,j\in \{2,3,\dots,\lambda-1\}$. Then the path $P^{ex}_i[u^P_i, t^P_i]$ can be replaced by
\[P^*=u^P_i \xrightarrow{P_d} u^Q_j \xrightarrow{Q^{ex}_j} x_{1,j} \xrightarrow{P_1} t^P_1 \xrightarrow{Q_t} t^P_i.\]
Similarly to earlier, the existence of such a path contradicts the essentiality of $e_i$, unless $P^*$ passes through $e_i$. The grid properties imply that $e_i$ cannot lie on any of $Q^{ex}_j,P_1,Q_t$, thus it must lie on $P_d$, and moreover, that it must be essential for $P_d$. We now distinguish two cases based on the location of $e_i$ on $P_i$.

\paragraph*{Case 1.} Edge $e_i$ comes before $x_{i,j}$ on $P_i$. Then $e_i$ is inside the bounded region of the non-directed cycle
\[C'= s^P_{i-1} \xrightarrow{Q_s} s^P_{i+1} \xrightarrow{P_{i+1}} x_{i+1,j} \xleftarrow{Q_j} x_{i-1,j} \xleftarrow{P_{i-1}} s^P_{i-1},\]
see Figure~\ref{fig:gridtrap} for an illustration.
We claim that after $P_d$ passes through $e_i$, it is ``trapped'' inside $C'$, i.e., it cannot pass through any vertex of $C'$.

Entering some vertex $x$ of $Q_j[x_{i-1,j},x_{i+1,j}]$ after $e_i$ is not possible since it creates a closed walk $x \xrightarrow{P_d} u^Q_j \xrightarrow{Q^{ex}_j} x$. Entering $Q_s[s^P_{i-1},s^P_i]$ at vertex $x$ would also create a closed walk: $x \xrightarrow{Q_s} s^P_i \xrightarrow{P_i}\head(e_i) \xrightarrow{P_d} x$. Entering  some vertex $x$ of $P^\# := s^P_i \xrightarrow{Q_s} s^P_{i+1} \xrightarrow{ P_{i+1}} x_{i+1,j}$ contradicts the essentiality of $e_i$ for $P_d$, as we can exchange $P_d[u^P_i,x]$ with $u^P_i \xrightarrow{P^{ex}_i} s^P_i \xrightarrow{P^\#} x$. Thus $P_d$ has to enter some vertex $x$ of $P_{i-1}[s^P_{i-1},x_{i-1,j}]$. If $x$ appears before $e_{i-1}$ on $P_{i-1}$, then we get a closed walk $\head{e_i}\xrightarrow{P_d} x \xrightarrow{P_{i-1}} \tail(e_{i-1}) \xrightarrow{} \head(e_i) $, where the last portion of the walk is supplied by Property (vi) of grids. This contradicts acyclicity. If $x$ appears after $e_{i-1}$, then $e_{i-1}$ is non-essential for $P_{i-1}$, as $P_{i-1}[s^P_{i-1},x]$ can be circumnavigated on $s^P_{i-1} \xrightarrow{Q_s} s^P_i \xrightarrow{P_i} \head(e_i) \xrightarrow{P_d} x$. Thus no vertex of $C'$ can be entered by $P_d$ after passing through $e_i$.

\begin{figure}[t]
\centering
\includegraphics[width=\textwidth]{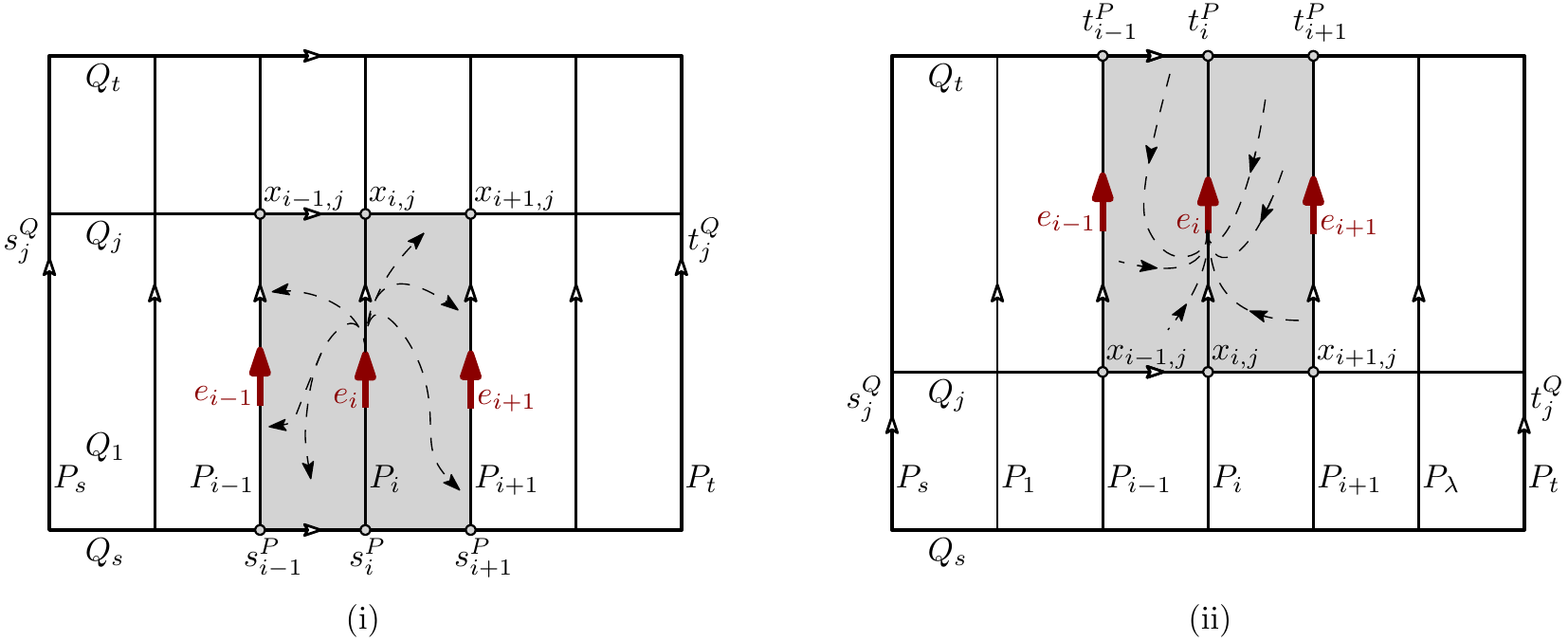}
\caption{(i) The demand path $P_d$ is trapped in the shaded region after passing through $e_i$, as entering its boundary is not possible. (ii) The demand path $P_d$ cannot exit the boundary of the shaded region before it passes through $e_i$.}\label{fig:gridtrap}
\end{figure}

\paragraph*{Case 2.} Edge $e_i$ comes after $x_{i,j}$ on $P_i$.
Then $e_i$ is inside the bounded region of the non-directed cycle
\[C'=  x_{i-1,j} \xrightarrow{Q_j} x_{i+1,j} \xleftarrow{P_{i+1}} t^P_{i+1} \xleftarrow{Q_t} s^P_{i-1} \xleftarrow{P_{i-1}} x_{i-1,j}.\]
We claim that after $P_d$ exits $C'$, it cannot pass through $e_i$.
The case can be handled analogously to Case 1: If $P_d$ exits from $Q_j[x_{i-1,j},x_{i+1,j}]$, $Q^t[t^P_i,t^P_{i+1}]$, or $P_{i+1}$ after $e_{i+1}$, then it creates a closed walk, contradicting acyclicity.
Exiting at vertex $x$ of $P_{i+1}$ before $e_{i+1}$ contradicts the essentiality of $e_{i+1}$ for $P_{i+1}$, as $P_{i+1}[x,t^P_{i+1}]$ can be exchanged with $x \xrightarrow{P_d} \head(e_i) \xrightarrow{P_i} t^P_i \xrightarrow{Q_t} t^P_{i+1}$. Exiting from vertex $x$ of $P^\#:=x_{i-1,j} \xrightarrow{P_{i-1}} t^P_{i-1} \xrightarrow{Q_t} t^P_i$ contradicts essentiality of $e_i$ for $P^{ex}_i$, as $P^{ex}_i[u^P_i,t^P_i]$ can be exchanged with $u^P_i \xrightarrow{P_d} x \xrightarrow{P^\#} t^P_i$. \\

The non-existence of a demand $u^Q_ju^P_i$ can be proven as above by exchanging the role of $P_i$ and $Q_j$. If a demand $v^P_iv^Q_j$ could exist, then in the reversed orientation graph it would be a valid demand $u^P_iu^Q_j$, contradicting the above arguments. The non-existence of demands of the form $v^Q_jv^P_i$ then follows by exchanging the role of $P_i$ and $Q_j$ again. This concludes the proof.
\end{proof}

We are now ready to prove the Structure Theorem (Theorem~\ref{thm:maintechnical}).

\begin{proof}[Proof of Theorem~\ref{thm:maintechnical}]
Suppose that $\calD$ is not $c$-bounded, that is, for any positive real number~$\gamma$, there exists an instance $(G_\gamma,T_\gamma,D_\gamma)$ of \Dsn and an optimum solution $H:=H_\gamma$ such that the total branch degree of $H$ is more than $\gamma k$, where $k = |T_\gamma|$.

We say that a quantity $\mu$ is $\gamma$-tied if it can be lower bounded by $\mu>f(\gamma)$ where $\lim_{\gamma \rightarrow \infty} f(\gamma)=\infty$. We need to show that $H$ contains a tough pair whose size is $\gamma$-tied. This implies that we can find a sequence of patterns that have tough pairs whose size goes to infinity. We can use vertex identifications in these patterns on the tough pairs to get patterns with $t$-tough pairs for all positive integers $t$.

As discussed in the beginning of Section~\ref{sec:skeleton}, we may assume without loss of generality that $H$ is acyclic, weakly connected, and the vertices of $H$ have undirected degree at least $3$.

Note that if for each face $\calF$ of $\Skel$ we have that $H\setminus \inter \calF$ fails at most $\kappa$ demands for some absolute constant $\kappa$, then each $H^\calF$ has total branch degree at most $O(\kappa^2)k$ by Lemma~\ref{lem:cbound_face_to_full}, thus we have a face $\calF$ where $H\setminus \inter \calF$ fails at least $\kappa=\Omega(\gamma^{1/2})$ demands. We set $L=\lfloor \kappa^{1/3}\rfloor =\Omega(\gamma^{1/6})$ (in particular, both $\kappa$ and $L$ are  $\gamma$-tied). Now we invoke Lemma~\ref{lem:findgrid}, which gives us a grid structure of size $\lambda=\Omega(L)$, i.e., a grid structure whose size is $\gamma$-tied. Finally, we use Lemma~\ref{lem:gridsarehard} on this grid: we get a $(\lambda-2)$-tough pair. Since $\lambda$ is $\gamma$-tied, this concludes the proof.
\end{proof}
\section{Cleaning: Identifying to a \hardpattern{t}}\label{sec:cleaning}

The goal of this section is to prove Theorem~\ref{thm:cleaning-intro}.

\cleaning*

Towards the proof of Theorem~\ref{thm:cleaning-intro}, we first prove Lemma~\ref{lem:cleaning-wrap-up}. 

\begin{lemma}[Identifying to a \hardpattern{t}]\label{lem:cleaning-wrap-up}
Let $\mathcal{D}$ be a class of graphs that is closed under transitive equivalence and identifying vertices. Let $D \in \mathcal{D}$ and let $t$ be a positive integer. 
Then there exists $t'$ that depends only on $t$,
such that if $D$ contains a $t'$-\toughpair\ then there exists $D' \in \mathcal{D}$ that is a \hardpattern{t}.
\end{lemma}

From Lemma~\ref{lem:cleaning-wrap-up}, the proof of Theorem~\ref{thm:cleaning-intro} follows easily. We give this proof before proving Lemma~\ref{lem:cleaning-wrap-up}.

\begin{proof}[Proof of Theorem~\ref{thm:cleaning-intro}]

We first prove the forward direction.
From Lemma~\ref{lem:cleaning-wrap-up}, for every positive integer $t$, $\mathcal{D}$ contains a digraph from some \hardpattern{t}, that is, $\mathcal{D}$ contains a digraph from some $\mathcal{C}_i$. Since $\kappa$ is finite, 
there exists $i \in [\kappa]$ such that infinitely many 
digraphs of $\mathcal{C}_i$ belong to $\mathcal{D}$. Further, since $\mathcal{D}$ is closed under vertex identifications and from any member of $\mathcal{C}_i$ every smaller member can be obtained by vertex identifications, we conclude that $\mathcal{C}_i \subseteq \mathcal{D}$.

For the backward direction observe from the definitions of $\{ \mathcal{C}_1, \ldots ,\mathcal{C}_{\kappa}\}$, that each $t$-hard digraph of each $\mathcal{C}_i$ contains in fact a $t$-\toughpair.
\end{proof}

The remainder of this section is dedicated to the proof of Lemma~\ref{lem:cleaning-wrap-up} which is divided into five separate steps presented in Sections~\ref{sec:firstramsey}, \ref{sec:simplifyingbiclique}, \ref{sec:semicleaning}, \ref{sec:semicleantoclean} and \ref{sec:cleaningbiclique}, respectively. Below we give the main results of each of these sections and show how they together imply Lemma~\ref{lem:cleaning-wrap-up}. Before stating the results, we give some important definitions that are used throughout the section.

\paragraph*{Definitions.}

 For any digraph $D$, we denote by $D^{\star}$ the {\em transitive closure} of $D$, that is $D^{\star}$ is obtained from $D$ by repeatedly adding edges $(u,v)$ whenever $(u,v)$ is not already an edge but there is a $(u,v)$-path in $D$.

For a digraph $D$ and two ordered sets $A =(a_1,\ldots, a_t),B=(b_1, \ldots,b_t) \subseteq V(D)$ such that $|A|=|B|$, we say that $D$ has an {\em $(A,B)$-matching} if for each $i \in [t]$, $(a_i,b_i) \in E(D)$. We say that $D$ has an {\em $(A,B)$-induced-matching} if $D$ has an $(A,B)$-matching, for each $i \neq j$, $(a_i,b_j) \not \in E(D)$, $(b_j,a_i) \not \in E(D)$, $(b_i,a_i) \not \in E(D)$ and $A,B$  are independent sets in $D$.

We say that $D$ has an {\em $(A,B)$-biclique} if for each $i,j \in [t]$, $(a_i,b_j) \in E(D)$ and $(b_j,a_i) \not \in E(D)$. We say that $D$ has an {\em $(A,B)$-induced-biclique}
if $D$ has an $(A,B)$-biclique and, $A, B$ are independent sets in $D$. We say that $D$ has a {\em $t$-biclique} if there exists an $(A,B)$-biclique in $D$ for some $A,B \subseteq V(D)$ and $|A|=|B|=t$. Similarly, we say that $D$ has a {\em $t$-induced-biclique} if there exists an $(A,B)$-induced-biclique in $D$ for some $A,B \subseteq V(D)$ and $|A|=|B|=t$. In all these cases, we call the set of edges $\{(a_i,b_i): i \in [t]\}$ as the {\em matching edges} of $(A,B)$.

For positive integers $i,j$, let $\mathcal{R}(i,j)$ denote the minimum number of vertices such that any $\mathcal{R}(i,j)$-vertex complete graph whose edges are colored with $j$ colors, contains a monochromatic clique of size $i$. Let $\mathcal{R}'(i)$ denote the minimum integer such that any bipartite graph that has a matching of size $\mathcal{R}'(i)$, either has an induced matching of size $t$ or a $t$-induced-biclique.
From Ramsey Theorem, such numbers always exists and they depend only on $i,j$, or $i$, respectively.
Recall the definitions of weakly independence, strongly independence and $t$-\toughpair\ from Section~\ref{sec:prelims}.

\begin{definition}[Ordered $t$-\toughpair]
Given a 
digraph $D$, $E_1, E_2 \subseteq E(D^{\star})$, we say that $(E_1,E_2)$ is an ordered $t$-\toughpair\  in $D$ if 
\begin{enumerate}
\item $|E_1| = |E_2|=t$,
\item all edges in $E_i$ are pairwise weakly independent in $D$, for each $i \in \{1,2\}$, and
\item for each $e_1 \in E_1$ and $e_2 \in E_2$, $(e_1,e_2)$ are strongly independent in $D$, and
\end{enumerate}
there exists an ordering of the sets $\head(E_1) = (w_1, \ldots, w_t)$, $\tail(E_1)=(x_1, \ldots, x_t)$, $\head(E_2)=(y_1, \ldots, y_t)$ and $\tail(E_2)=(z_1, \ldots, z_t)$ such that
\begin{enumerate}
\item for any $1\leq i < j \leq t$, there is no $(w_j,x_i)$-path, no $(w_j,w_i)$-path and no $(x_j,x_i)$-path in $D$,
\item for any $1\leq i < j \leq t$, there is no $(y_j,z_i)$-path, no $(y_j,y_i)$-path and no $(z_j,z_i)$-path in $D$.
\end{enumerate}
If $(E_1,E_2)$ is an ordered $t$-\toughpair\ then we treat $E_1$ and $E_2$ as ordered sets such that their head sets and tail sets satisfies the above properties. We say ordered \toughpair\ to mean an ordered $t$-\toughpair\ for some $t$.
\end{definition}

Observe that, unlike the $t$-\toughpair, if $(E_1,E_2)$ is an ordered $t$-\toughpair\ in $D$, then the edges of $E_1$ and $E_2$ may not be minimal in $D$.

\paragraph*{Simplifying the $t$-\toughpair.}
In Section~\ref{sec:firstramsey}, using Ramsey arguments we show that it $D$ contains a $t$-\toughpair\ for a large enough $t$, then it contains one of the three structures described in Lemma~\ref{lem:pre-cleaning}.

\begin{restatable}[Hard sub-structures]{lemma}{firstramsey}\label{lem:pre-cleaning}
Let $D$ be a reachability-minimal digraph and $t$ be a positive integer. 
If $D$ contains an $\mathcal{R}(2t,9)$-\toughpair\
then one of the following holds.
\begin{enumerate}
\item $D$ contains an ordered $t$-\toughpair\, or
\item $D$ contains a $t$-biclique, or
\item there exist ordered sets $A,B \subseteq V(D)$ such that $|A|=|B|=t$, $D$ contains an $(A,B)$-matching and $D^{\star}$ contains an $(A,B)$-induced-biclique. 
\end{enumerate}
\end{restatable}

Note that if $D$ is reachability-minimal and contains a $t$-biclique then it contains a $t$-induced-biclique. 
From Lemma~\ref{lem:pre-cleaning} one concludes that, in order to prove Lemma~\ref{lem:cleaning-wrap-up}, it is enough to identify one of the three structures defined in Lemma~\ref{lem:pre-cleaning} to some \hardpattern{t}. The first outcome of Lemma~\ref{lem:pre-cleaning} is handled in Section~\ref{sec:cleaningordered}, the second outcome is handled in Section~\ref{sec:cleaningbiclique} and the third outcome is handled in Section~\ref{sec:simplifyingbiclique}.

\paragraph*{Cleaning ordered $t$-\toughpair.} In Section~\ref{sec:cleaningordered} we show that if the outcome of Lemma~\ref{lem:pre-cleaning} is an ordered $t$-\toughpair, then one can identify $D$ to either a \hardpattern{t'} or $D^{\star}$ contains a $t'$-induced-biclique whose edges are minimal in $D$. 
This is achieved by two rounds of cleaning: in the first round the digraph is cleaned to an intermediate structure, called a {\em semi-cleaned ordered $t$-\toughpair} (defined below), and in the second round this semi-cleaned ordered $t$-\toughpair\ is either identified to obtain a \hardpattern{t'}, or it can be shown that it contains a $t'$-induced-biclique whose edges are minimal.
This is formalized in Lemmas~\ref{lem:semi-cleaning} and~\ref{lem:semicleantoclean} which are proved in Section~\ref{sec:semicleaning} and~\ref{sec:semicleaning}, respectively.

\begin{definition}[Semi-cleaned ordered $t$-\toughpair]
We say that a digraph $D$ is a semi-cleaned ordered $t$-\toughpair, if it contains an ordered $t$-\toughpair\ $(A,B)$ such that the vertex set of $D$ contains at most two vertices, called a source $\sss$ and a sink $\ttt$, outside of the set $V(A \cup B)$, and $N^-_D(\sss), N^+_D(\ttt) =\emptyset$. (Note that from this definition $(\ttt,\sss) \not \in E(D)$).
\end{definition}

\begin{restatable}[Semi-cleaning the ordered \toughpair]{lemma}{semicleaning}\label{lem:semi-cleaning}
Let $D$ be a
digraph such that $D$ contains an ordered $t^2$-\toughpair, then one can obtain $\widehat{D}$ from $D$ by identification, such that $\widehat{D}$ is a semi-cleaned $t$-ordered \toughpair, for some function $g$ that depends only on $t$.
\end{restatable}

For any positive integer $t$, let $\htt=2\mathcal{R}(2\mathcal{R}(2\mathcal{R}(4t+2,4),4),5)$.

\begin{restatable}[Cleaning the semi-cleaned ordered \toughpair]{lemma}{semicleantoclean}
\label{lem:semicleantoclean}
If $D$ is a semi-cleaned ordered $\htt$-\toughpair\ 
then either, 
\begin{itemize}
\item $D^{\star}$ contains a $t$-induced-biclique whose edges are minimal in $D$, or
\item $D$ can be identified to a digraph $\hat{D}$ such that $\hat{D}$ is transitively equivalent to a \hardmatching{t}.
\end{itemize}
\end{restatable}

For any positive integer $t$, let $\cor{t}= (\htt)^2$.
Corollary~\ref{cor:cleaning-ordered} follows from Lemmas~\ref{lem:semi-cleaning} and~\ref{lem:semicleantoclean}.

\begin{corollary}\label{cor:cleaning-ordered}
If a digraph $D$ contains an ordered $\cor{t}$-\toughpair,
then either $D$ can be identified to a \hardmatching{t}, or $D^{\star}$ contains a $t$-induced-biclique whose edges are minimal in $D$.
\end{corollary}

If the outcome of Corollary~\ref{cor:cleaning-ordered} a \hardmatching{t}, then we are done. Otherwise, we need to clean the $t$-induced-biclique whose edges are minimal,
which is what is done next.
\paragraph*{Cleaning a minimal induced-biclique.} In Section~\ref{sec:cleaningbiclique} we show that if $D$ contains a $9t$-induced-biclique whose edges are minimal, then $D$ can be identified to a \hardbiclique{t}. In this case, Lemma~\ref{lem:cleaning-wrap-up} is proved.

\begin{restatable}[Cleaning minimal biclique]{lemma}{cleaningbiclique}\label{lem:cleaningbiclique}
For a positive integer $t$ and a digraph $D$, if $D^{\star}$ contains a $9t$-induced-biclique whose edges are minimal in $D$, then $D$ can be identified to digraph that is transitively equivalent to a \hardbiclique{t}.
\end{restatable}

\paragraph*{Simplifying the third outcome of Lemma~\ref{lem:pre-cleaning}.} In Section~\ref{sec:simplifyingbiclique}, we show that if Lemma~\ref{lem:pre-cleaning} outputs its third outcome, then one can contract some edges of the input digraph (and hence obtain a digraph in the same pattern class as the input, since the pattern class is closed under identification and contraction is a type of identification) such that the resulting digraph either contains an ordered $t'$-\toughpair\ (which we know how to deal using Corollary~\ref{cor:cleaning-ordered}), or it contains a $t'$-biclique whose edges are minimal (which we also know to deal with because of Lemma~\ref{lem:cleaningbiclique}).

For positive integers $t,p$, define $f(t,p) = 2^{\mathcal{R}'(2t)} \cdot f(t,p-1) + \mathcal{R}'(2t)$ when $p >1$ and $f(t,1)=t$.

\begin{restatable}[Simplifying a biclique]{lemma}{digging}\label{lem:digging}
Let $D$ be a directed graph and $t$ be a positive integer.  
Suppose there exists ordered sets
$A,B \subseteq V(D)$ such that $|A|=|B|\geq f(t,2t+2)$, 
$D^{\star}$ contains an $(A,B)$-induced-biclique and the matching edges of $(A,B)$ are minimal in $D$,.
Then, one can obtain $\widehat{D}$ from $D$ by contraction, such that 
either
\begin{enumerate}
\item $\widehat{D}$ contains an ordered $t$-\toughpair, or
\item $\widehat{D}^*$ contains a $t$-induced-biclique whose edges are minimal in $\widehat{D}$.
\end{enumerate}
\end{restatable}

\begin{proof}[Proof of Lemma~\ref{lem:cleaning-wrap-up}]
We now give a proof of Lemma~\ref{lem:cleaning-wrap-up} using the lemmas stated earlier. The proofs of the those lemmas appear in their respective sections.

Set $t_1 = \max\{\cor{9t},9t\}$, $t_2 =\max \{ f(t_1, 2t_1+2), \cor{9t_1}, 9t_1 \}$ where $f$ is defined as in Lemma~\ref{lem:digging} and $\corr$ is defined in Corollary~\ref{cor:cleaning-ordered}. Then, define $t' =  \mathcal{R}(2t_2,9)$. We will now show that if $D$ contains a $t'$-\toughpair, then there exists $D' \in \mathcal{D}$ that is a \hardpattern{t}.

Given $D$, let $(X,Y)$ be a $t'$-\toughpair\ in $D$. Let $D_1$ be a spanning subgraph of $D$ such that $D_1$ is reachability-minimal, $D_1$ is transitively equivalent to $D$ and $(X,Y)$ is a $t'$-\toughpair\ in $D_1$. Observe that such a graph $D_1$ exists and can be obtained by starting with the edge set $E_0=X \cup Y$ and adding an edge of $D$ to $E_0$ as long as $D[E_0]$ remains reachability-minimal. 

Using Lemma~\ref{lem:pre-cleaning} on $D_1$, we conclude that either $D_1$ contains an ordered $t_2$-\toughpair, or a $t_2$-biclique, or there exists $A,B \subseteq V(D_1)$ such that $|A|=|B|=t_2$ and $D_1$ contains an $(A,B)$-induced-matching and $D_1^{\star}$ contains an $(A,B)$-induced-biclique. 

In the first case, by applying Corollary~\ref{cor:cleaning-ordered} on $D^{\star}_1$, we either get $D_2 \in \mathcal{D}$ such that $D_2$ is some \hardmatching{t}, 
or conclude that $D^{\star}_1$ contains a $9t$-induced-biclique whose edges are minimal in $D_1$. Then applying Lemma~\ref{lem:cleaningbiclique} to $D_1$ gives some \hardbiclique{t} in $\mathcal{D}$.
In the second case, again by applying Lemma~\ref{lem:cleaningbiclique} to $D_1$ gives some \hardbiclique{t} in~$\mathcal{D}$.

In the third case, by applying Lemma~\ref{lem:digging} we conclude that there exists $D_2 \in \mathcal{D}$ such that either $D_2$ contains an ordered $t_1$-\toughpair, in which case a further application of Corollary~\ref{cor:cleaning-ordered} on $D^{\star}_2$ either yields a $D_3 \in \mathcal{D}$ which is a \hardmatching{t}, or we conclude that $D^{\star}_2$ contains a $9t$-induced-biclique whose edges are minimal in $D_2$. In the remaining cases, applying Lemma~\ref{lem:cleaningbiclique} yields a $D_4 \in \mathcal{D}$ which is some \hardbiclique{t}.
\end{proof}

\paragraph*{Basic terminology for the remaining section.}Throughout the remaining section, we use the following basic notation. For two integers $i,j$, $[i]$ denotes the set $\{1, \ldots, i\}$ and $[i,j]$ denotes the set $\{i,i+1, \ldots,j\}$.  For any (di)graph $D$ and sets $S,T \subseteq V(D)$, an $(S,T)$-path in $D$ is a path in $D$ from some vertex of $S$ to some vertex of $T$. If $S$ or $T$ is singleton, say $S=\{v\}$, then we use the notation $(v,T)$-path. Let $e \in E(D)$, let $D/e$ represents the (di)graph obtained after contracting $e$, that is, by deleting the endpoints of $e$ and adding a new vertex $x_e$ such that the set of in-neighbours (resp.~out-neighbours) of $x_e$ is the union of the set of in-neighbours (resp.~out-neighbours) of the end points of $e$ in $D$. For $E' \subseteq E(D)$, $V(E') \subseteq V(G)$ is the set of vertices that are endpoints of some edge in $E'$. If $D$ is a digraph, the for any $v \in V(D)$, $N^-_D(v)$ denotes the set of in-neighbours if $v$ in $D$, $N^+_D(v)$ denotes the set of out-neighbours if $v$ in $D$ and $N_D(v)$ denotes the set of in-neighbours and out-neighbours, called the neighbours, of $v$ in $D$.
We say that a digraph $D$ is connected if its underlying undirected graph is connected. 
\subsection{Simplifying the $t$-\toughpair}\label{sec:firstramsey}

In this section we prove Lemma~\ref{lem:pre-cleaning} restated below.
\firstramsey*

\begin{proof}
Let $(E_1,E_2)$ be an $ \mathcal{R}(2t,9s)$-\toughpair\ in $D$.
Let $t' =\mathcal{R}(2t,9)$.
 Fix an arbitrary ordering of the edges of $E_1$ and $E_2$.  Let $E_1=(e^1_1, \ldots, e^1_{t'})$ and let $E_2=(e^2_1, \ldots, e^2_{t'})$.
Fix $i \in [2]$. We will use Ramsey arguments to prove the lemma. Towards this, construct an auxiliary undirected, complete, edge-colored graph $\aux{i}$ as follows. The vertex set of $\aux{i}$ corresponds to the edges of $E_i$, that is, for each $e \in E_i$ there is a vertex corresponding to $e$.
For the sake of simplicity we denote the vertex of $\aux{i}$ that correspond to the edge $e$ of $E_i$, by $e$ itself.
The coloring function $\col_i$ on the edges of $\aux{i}$ (equivalently on the pair of distinct edges of $E_i$) is defined based on the following. Fix any two distinct edges $e=(u,v),e'=(u',v') \in E_i$. The coloring function $\col_i(e,e')$ is defined based on the existence of the edges $(u,v')$ and $(u',v)$ in $D$ and $D^{\star}$. Below we describe $\col_i(e,e')$.

\begin{enumerate}
\item If $(u,v') \in E(D)$, and
\begin{enumerate}
\item $(u',v) \in E(D)$, then $\col_i(e,e') =(1,1)$,
\item $(u',v) \in E(D^{\star}) \setminus E(D)$, then $\col_i(e,e') =(1,2)$,
\item $(u',v) \not \in E(D)$, $(u',v) \not \in E(D^{\star})$, then $\col_i(e,e') = (1,3)$.
\end{enumerate}

\item If $(u,v') \in E(D^{\star}) \setminus E(D)$, and
\begin{enumerate}
\item $(u',v) \in E(D)$, then $\col_i(e,e') =(2,1)$,
\item $(u',v) \in E(D^{\star}) \setminus E(D)$, then $\col_i(e,e') =(2,2)$,
\item $(u',v) \not \in E(D)$, $(u',v) \not \in E(D^{\star})$, then $\col_i(e,e') = (2,3)$.
\end{enumerate}

\item If $(u,v') \not \in E(D)$, $(u,v') \not \in E(D^{\star})$, and
\begin{enumerate}
\item $(u',v) \in E(D)$, then $\col_i(e,e') =(3,1)$,
\item $(u',v) \in E(D^{\star}) \setminus E(D)$, then $\col_i(e,e') =(3,2)$,
\item $(u',v) \not \in E(D)$, $(u',v) \not \in E(D^{\star})$, then $\col_i(e,e') = (3,3)$.
\end{enumerate}
\end{enumerate}

We now use a Ramsey argument on the $t'$-vertex graph $\aux{i}$ with edge-coloring function $\col_i$ (that uses at most $9$ different colors). Using Ramsey arguments, we conclude that there exists a monochromatic clique in $\aux{i}$ of size $2t$. Below we show how one can get one of the three outcomes in the lemma statement based of the color of the monochromatic clique. Say edges of $E_i$ that correspond to this monochromatic clique in $\aux{i}$ are 
$E'_i=\{e^i_{j_1}, e^i_{j_2}, \ldots , e^i_{j_{t}}\}$ such that $j_1 < \ldots < j_{t}$ (recall we fixed an ordering of the edges of $E_i$). For the ease of notation later, we assume that the sets $E'_i$ are ordered with the ordering as described in the previous line.

\begin{claim}\label{claim:ramsey-orderedpair}
If for each $i \in [2]$, the color of the monochromatic clique in $\aux{i}$ is $(j,3)$ or $(3,j)$ for any $j \in [3]$, then $(E'_1,E'_2)$ form an ordered $2t$-\toughpair\ in $D$.
\end{claim}

\begin{claimproof}
Since $(E_1,E_2)$ is a \toughpair\ and $E'_i \subseteq E_i$, $(E'_1,E'_2)$ is also a \toughpair\ (of size $t$). If $\col_i$ colors the clique corresponding to the edges of $E'_i$ with the color $(j,3)$, then consider the ordering $(e^i_{j_1}, e^i_{j_2}, \ldots , e^i_{j_{t}})$ of $E'_i$, otherwise (when $\col_i$ colors the monochromatic clique with color $(3,j)$), then consider the ordering $(e^i_{j_{t}}, \ldots , e^i_{j_1})$ of $E'_i$. Then from the description of the definition of $\col_i$ corresponding to the case when $\col_i$ takes value $(3,j)$ or $(j,3)$, one concludes that $(E'_1,E'_2)$ is indeed an ordered $t$-\toughpair\ with the ordering described above.
\end{claimproof}
Note that if the graph has an ordered $2t$-\toughpair, then it also has an ordered $t$-\toughpair.

\begin{observation}\label{obs:ramsey-biclique}
It is easy to observe that if the color of the monochromatic clique in $\aux{i}$ is $(1,1)$, then $(\tail(E'_i),\head(E'_i))$ form a $2t$-biclique in $D$. If the color is $(2,2)$, then $D$ contains the $(\tail(E'_i),\head(E'_i))$-induced matching and $D^{\star}$ contains the $(\tail(E'_i),\head(E'_i))$-biclique (thus, the third outcome of the lemma holds). 
\end{observation}

Note that if the graph has $2t$-biclique, then it also has an ordered $t$-biclique.

\begin{claim}\label{claim:ramsey-induced-skew}
If there exists $i \in [2]$ such that the color of the monochromatic clique in $\aux{i}$ is $(1,2)$ or $(2,1)$, then $H$ contains a $t$-biclique.
\end{claim}

\begin{claimproof}
Suppose that the color of the monochromatic clique in $\aux{i}$ is $(1,2)$. Recall that $E'_i =((e^i_{j_1}, e^i_{j_2}, \ldots , e^i_{j_{t}}))$ is the ordered set of edges that correspond to the monochromatic clique. The $(\tail(\{e^i_{j_1}, \ldots , \allowbreak e^i_{j_{t}}\}), \allowbreak \head(\{e^i_{j_{t +1 }}, \ldots, \allowbreak e^i_{j_{2t}}\}))$ form a $t$-biclique in $D$.

Similarly, if the color of the monochromatic clique in $\aux{i}$ is $(2,1)$, then $(\tail(\{e^i_{j_{t+1}}, \ldots, \allowbreak e^i_{j_{2t}}\}), \allowbreak  \head(\{e^i_{j_1}, \ldots, \allowbreak e^i_{j_{t}}\}))$ form a $t$-biclique in $H$.
\end{claimproof}

Thus, from Claims~\ref{claim:ramsey-induced-skew} and~\ref{claim:ramsey-orderedpair}, and Observation~\ref{obs:ramsey-biclique}, the lemma follows.
\end{proof}
\subsection{Simplifying a biclique}\label{sec:simplifyingbiclique}
In this section, we prove the following lemma. Recall that, from Section~\ref{sec:cleaning}, for positive integer $t,p$, $\mathcal{R}'(t)$ is the smallest positive integer such that any bipartite graph with a matching of size $\mathcal{R}'(t)$ either has an induced matching of size $t$ or a $t$-induced-biclique.
Furethermore $f(t,p) = 2^{\mathcal{R}'(2t)} f(t,p-1) + \mathcal{R}'(2t)$ 
when $p>1$ and $f(t,1)=t$. 

\digging*

The remainder of this section is dedicated to the proof of Lemma~\ref{lem:digging}. Recall that $A,B$ are ordered sets. Let $g(t) = f(t, \mathcal{R}(2t,2))$.
Further let $A=(a_1, \ldots, a_{g(t)})$ and $B=(b_1, \ldots, b_{g(t)})$.  

We begin by showing that, one can assume without loss of generality that $D$ is acyclic.
Suppose that $D$ is not acyclic. First observe that for any $a_i,a_j \in A$, such that $i \neq j$, $a_i$ and $a_j$ belong to different strongly connected components of $D$. Similarly, each vertex of $B$ belongs to a distinct strongly connected component of $D$. Further, for any $a_i,b_j$, $a_i$ and $b_j$ do not belong to the same strongly connected component of $D$. Indeed, as otherwise there is a $(b_j,a_i)$-path in $D$ contracting the definition of an $(A,B)$-biclique. Thus, each vertex in $\{a_1, \ldots, a_{g(t)}, b_1 , \ldots, b_{g(t)}\}$ belong to distinct strongly connected components of $D$.
Let $D'$ be obtained from $D$ by contracting each strongly connected component of $D$ into a single vertex. That is, $D'$ has a vertex for each strongly connected component of $D$ and for two vertices $u,v$ in $D'$, there is an edge from $u$ to $v$ in $D'$ if there is vertex in the strongly connected component of $D$ which was contracted onto $u$, that has an edge in $D$ to a vertex of the strongly connected component of $D$ which was contracted to $v$. Let us call a vertex of $D'$ $a_i$ if it is obtained by contracting a strongly connected component containing $a_i$. Similarly, let us call a vertex of $D'$ $b_i$ if it is obtained by contracting a strongly connected component containing $b_i$. It is easy to observe that $D'$ also contains an $(A,B)$-matching and $D'^{\star}$ contains an $(A,B)$-induced-biclique. Indeed, since an $(a_i,b_j)$-path in $D$ implies an $(a_i,b_j)$-path in $D'$ and no $(a_i,a_j)$-path (resp.~no $(b_i,b_j)$-path) in $D$ implies no $(a_i,a_j)$-path (resp.~no $(b_i,b_j)$-path) in $D'$. Further if $D'$ contains an $(a_i,b_i)$-path of length strictly greater than $1$ then so does $D$. Since $D'$ is obtained by contraction operation from $D$, we conclude that without loss of generality, we can assume for the rest of the section that the input graph $D$ in Lemma~\ref{lem:digging} is acyclic.

We say that an edge $e \in E(D) \setminus E(A,B)$ is {\em contraction-redundant} with respect to $(A,B)$ in $D$ if the following holds.

\begin{enumerate}
\item If $D/e$ is acyclic, 
\item all the edges in $\{(a_i,b_i) : i \in [g(t)]\}$ (that is the matching edges of $(A,B)$) are minimal edges in $D/e$, and
\item $(D/e)^{\star}$ has an $(A,B)$-induced-biclique..
\end{enumerate} 

Let $\Dcontr$ be the graph obtained from $D$ by repeatedly contracting the contraction-redundant edges with respect to $(A,B)$ until $\Dcontr$ has no such edge. 
For the ease of notation by $aa$-path we mean a path between two distinct vertices of $A$, by $bb$-path we mean a path between two distinct vertices of $B$ and by $a_{\ast}b_{\ast}$-path we mean a $(a_i,b_i)$-path for some $i \in [g(t)]$.
We divide the proof of Lemma~\ref{lem:digging} into two independent parts based on the length of a longest path in $\Dcontr$. In Section~\ref{sec:biclique:longpath}, we consider the case when the length of a longest path in $D_{contr}$ is at least $2t+2$.
In this case, we show that $\Dcontr$ contains an ordered $t$-\toughpair.
In Section~\ref{sec:biclique:nolongpath}, we consider the case when the length of any longest path in $\Dcontr$ is at most $2t+2$. In this case, we get one of the desired outputs.

\subsubsection{When $\Dcontr$ has a long path}\label{sec:biclique:longpath}

In this section, we prove Lemma~\ref{lem:digging} when the length of a longest path in $\Dcontr$ is at least $2t+2$.
Let $P=(v_1, \ldots,v_p)$
be a directed longest path in $\Dcontr$ (then $p \geq 2t+2$).
Recall that none of the edges of $\Dcontr$, and in particular, none  of the edges of $P$, are {\em contraction-redundant} with respect to $(A,B)$.

\begin{lemma}\label{claim:biclique1}
For each $e \in E(P)$, $\Dcontr/e$ is acyclic.
\end{lemma}

\begin{proof}
Suppose, for the sake of contradiction, that there is a cycle in $\Dcontr/e$. Let $e=(v_i,v_{i+1})$. 
Then, either there is a $(v_i,v_{i+1})$-path in $\Dcontr - \{(v_i,v_{i+1})\}$ or a $(v_{i+1},v_i)$-path in $\Dcontr$. 
If there exists a $(v_{i+1},v_i)$-path in $\Dcontr$ then together with the edge $e=(v_i,v_{i+1})$, it creates a directed closed walk in $\Dcontr$. 
This contradicts that $\Dcontr$ is acyclic. 

In the other case, suppose there exists a $(v_i,v_{i+1})$-path, say $P'$ in $\Dcontr -\{(v_i,v_{i+1})\}$. We first claim that the internal vertices of $P'$ are disjoint from $P$. Suppose not, let $v_j \in V(P)$ be the first internal vertex of $P'$ that belongs to $P$. Then there is a $(v_i,v_j)$-path in $\Dcontr$ and $(v_j,v_{i+1})$-path in $\Dcontr$. If $j<i$, then the $(v_i,v_j)$-path together with the $(v_j,v_i)$-subpath of $P$, gives a closed walk in $\Dcontr$, contradicting its acyclicity. Otherwise, $j>i+1$. In this case, the $(v_j,v_{i+1})$-path together with the $(v_{i+1},v_j)$-subpath of $P$ gives a closed walk in $\Dcontr$, again contradicting its acyclicity.
\end{proof}

Since each edge of $\Dcontr$, in particular, each edge of $P$, is not contraction-redundant, from Lemma~\ref{claim:biclique1} and the definition of contraction-redundant, for each $e \in E(P)$ either
\begin{enumerate}
\item there exists an $aa$-path, or a $bb$-path in $\Dcontr/e$ (that is, $(\Dcontr/e)^{\star}$ does not contain the $(A,B)$ -induced-biclique), or 
\item there exists $i \in [g(t)]$ an $a_{i}b_{i}$-path in $(\Dcontr/e )- \{(a_{i},b_{i})\}$ ($(a_{i},b_{i})$ is not a minimal edge in $\Dcontr/e$).
\end{enumerate}

\begin{lemma}\label{obs:biclique1}
Let $e=(v_i,v_{i+1}) \in E(P)$ such that $H_{contr}/e$ has an $(a_j,b_j)$-path. Then there exists an $(a_j,v_{i+1})$-path and a $(v_i,b_j)$-path in $\Dcontr$.

If $\Dcontr/e$ has a $(a_j,a_{\ell})$-path, for some $a_j,a_{\ell} \in A$, then there exists an $(a_j, v_{i+1})$-path and a $(v_i,a_{\ell})$-path in $\Dcontr$.

If $\Dcontr/e$ has a $(b_j,b_{\ell})$-path, for some $b_j,b_{\ell} \in B$, then there exists an $(b_j, v_{i+1})$-path and $(v_i,b_{\ell})$-path in $\Dcontr$.
\end{lemma}

\begin{proof}
Let $e=(v_i,v_{i+1}) \in E(P)$ such that $D_{contr}/e$ has an $(a_j,b_j)$-path. Then either there exists an $(a_j,v_i)$-path and a $(v_{i+1},b_j)$-path in $\Dcontr$, or there exists an $(a_j,v_{i+1})$-path and $(v_i,b_j)$-path in $\Dcontr$. In the later case, we are done. The former case implies an $(a_j,b_j)$-path in $\Dcontr$ of length at least three (with at least three edges), which contradicts that $(a_j,b_j)$ is a minimal edge in $\Dcontr$.

If $D_{contr}/e$ has an $(a_j,a_{\ell})$-path, then either there exists an $(a_j,v_i)$-path and a $(v_{i+1},a_{\ell})$-path in $\Dcontr$, or there exists an $(a_j,v_{i+1})$-path and $(v_i,a_{\ell})$-path in $\Dcontr$. In the later case, we are done. The former case implies an $(a_j,a_{\ell})$-path in $\Dcontr$ which contradicts that $(A,B)$ is an induced biclique in $\Dcontr^{\star}$ (in particular, that $A$ is an independent set in $\Dcontr^{\star}$).
One can similarly show that if $D_{contr}/e$ has a $(b_j,b_{\ell})$-path, then there exists an $(b_j,v_{i+1})$-path and $(v_i,b_{\ell})$-path in $\Dcontr$.
\end{proof}

For each $e=(v_i,v_{i+1}) \in E(P)$, we say that:
\begin{enumerate}
\item $e$ is {\em $j$-irredundant} in $\Dcontr$, if there is an $(a_j,v_{i+1})$-path and a $(v_i,b_j)$-path in $\Dcontr$.
\item $e$ is {\em $aa$-irredundant} in $\Dcontr$, if there exists $a_j,a_{\ell} \in A$ such that there is an $(a_j,v_{i+1})$-path and a $(v_i,a_{\ell})$-path in $\Dcontr$.
\item  $e$ is {\em $bb$-irredundant} in $\Dcontr$, if there exists $b_j,b_{\ell} \in B$ such that there is an $(b_j,v_{i+1})$-path and $(v_i,b_{\ell})$-path in $\Dcontr$.
\end{enumerate}

From Lemmas~\ref{claim:biclique1} and~\ref{obs:biclique1}, for each $e \in E(P)$, either $e$ is $j$-irredundant, for some $j \in [g(t)]$, or is $aa$-irredundant or is $bb$-irredundant.

\begin{lemma}\label{claim:biclique21}
Let $e=(v_i,v_{i+1}) \in E(P)$ such that $e$ is $aa$-irredundant in $\Dcontr$. Then $i=1$.
\end{lemma}

\begin{proof}
For the sake of contradiction, suppose that $i>1$. Let $e'=(v_{i-1},v_i)$. 
If $e'$ is $aa$-irredundant or $e'$ is $j$-irredundant for some $j \in [g(t)]$, then there exists a path from some vertex of $A$ to $v_i$ in $\Dcontr$. Also since $e$ is $aa$-irredundant, there exists a path from $v_i$ to some vertex of $A$ in $\Dcontr$. This implies either a directed closed walk in $\Dcontr$, contradicting its acyclicity, or a path between two distinct vertices of $A$ in $\Dcontr$, contradicting that $A$ is an independent set in $\Dcontr^{\star}$.

If $e'$ is $bb$-irredundant, then there exists a path from some vertex of $B$ to $v_i$. Also since $e$ is $aa$-irredundant, there exists a path from $v_i$ to some vertex of $A$. This implies a path from some vertex of $B$ to some vertex of $A$. Since $(A,B)$ is an induced biclique in $\Dcontr^{\star}$, that is there is a path from every vertex of $A$ to every vertex of $B$ in $\Dcontr$, this implies a directed closed walk in $\Dcontr$, contradicting its acyclicity.
\end{proof}

\begin{lemma}\label{claim:biclique22}
Let $e=(v_i,v_{i+1}) \in E(P)$ such that $e$ is $bb$-irredundant in $\Dcontr$. Then $i+1=p$.
\end{lemma}

\begin{proof}
For the sake of contradiction, suppose that $i+1<p$. Let $e'=(v_{i+1},v_{i+2})$. 
If $e'$ is $bb$-irredundant or $e'$ is $j$-irredundant for some $j \in [g(t)]$, then there exists a path from $v_{i+1}$ to some vertex of $B$. Also since $e$ is $bb$-irredundant, there exists a path from some vertex of $B$ to $v_{i+1}$. This either implies a directed closed walk in $\Dcontr$, contradicting its acyclicity, or a path between two distinct vertices of $B$ in $\Dcontr$, contradicting that $B$ is an independent set in $\Dcontr^{\star}$.

If $e'$ is $aa$-irredundant, then there exists a path from $v_{i+1}$ to some vertex of $A$ in $\Dcontr$. Also since $e$ is $bb$-irredundant, there exists a path from some vertex of $B$ to $v_{i+1}$ in $\Dcontr$. This implies a path from some vertex of $B$ to some vertex of $A$ in $\Dcontr$. Since 
$(A,B)$ is an induced-biclique in $\Dcontr^{\star}$, this implies a directed closed walk
in $\Dcontr$, contradicting its acyclicity.
\end{proof}

Let $P'$ be a subpath obtained from $P$ by without the first and last edge of $P$, that is $P'=(v_2, \ldots,v_{p-1})$. From Lemmas~\ref{claim:biclique21} and~\ref{claim:biclique22}, each edge $e \in E(P')$ is $j$-irredundant for some $j \in [g(t)]$.

\begin{lemma}\label{claim:biclique3}
Let $e,e'$ be distinct edges of $P'$ such that $e$ is $q$-irredundant in $\Dcontr$ and $e'$ is $r$-irredundant in $\Dcontr$, for some $q,r \in [g(t)]$. Then $q \neq r$.
\end{lemma}

\begin{proof}
For the sake of contradiction, suppose that $q=r$. Without loss of generality, let $e=(v_i,v_{i+1})$ and $e'=(v_j,v_{j+1})$ such that $i<j$. Then there exists an $(a_q,v_{i+1})$-path and $(v_j,b_q)$-path in $\Dcontr$. Since $i < j$, there exists a path from $(v_i,v_j)$ (which is the the $(v_{i+1},v_j)$-subpath of $P'$) in $\Dcontr$. This implies an $(a_q,b_q)$-path of length strictly greater than one in $\Dcontr$, contradicting that $(a_q,b_q)$ is a minimal edge in $\Dcontr$.
\end{proof}

For each $e=(v_i,v_{i+1}) \in E(P')$, fix a $\rho(i) \in [g(t)]$, such $e$ is $\rho(i)$-irredundant in $\Dcontr$. 
The following observation follows from Lemma~\ref{claim:biclique3}.

\begin{observation}\label{obs:biclique1}
For each $v_i \in V(P')$, $i \in [3, p-2]$, there exists a $(a_{\rho(i-1)},v_i)$-path and a $(v_i, b_{\rho(i)})$-path in $\Dcontr$.  Also, for $i,j \in [3,p-2]$, $i \neq j$, $\rho(i) \neq \rho(j)$.
\end{observation}

For each $i \in [3,p-2]$, let $P^{\aaa}_i$ denote an (arbitrarily) fixed $(a_{\rho(i-1)},v_i)$-path in $\Dcontr$ and let $P^{\bbb}_i$ denote an (arbitrarily) fixed $(v_i,b_{\rho(i)})$-path in $\Dcontr$.
Let $E_1 = (v_3, b_{\rho(3)}), \allowbreak  (v_4, b_{\rho(4)}), \ldots, \allowbreak (v_{p/2-2}, b_{\rho(p/2-2)})$ and let 
 $E_2 = ((a_{\rho(p/2-2)},v_{p/2-1}), \allowbreak  (a_{\rho(p/2-1)}, \allowbreak v_{p/2}), \ldots, (a_{\rho(p-1)}, \allowbreak v_{p-2})$. 
 From Observation~\ref{obs:biclique1}, each edge of $E_1 \cup E_2$ is an edge in $D^{\star}_{\contr}$ and the endpoints of the edges in $E_1 \cup E_2$ are distinct. We will now show that $(E_1,E_2)$ is an ordered $(p/2-1)$-\toughpair\ in $D^{\star}_{\contr}$. Since $p \geq 2t+2$, this proves Lemma~\ref{lem:digging} in the case when the length of a longest path in $\Dcontr$ is at least $2t+2$.
 
\begin{lemma}\label{lem:long-path-ordered-tough-pair}
$(E_1,E_2)$ is an ordered $(p/2-1)$-\toughpair\ in $D^{\star}_{\contr}$.
\end{lemma}

\begin{proof}
To prove the lemma, we prove the following claims.

\begin{claim}[Weak independence of $E_1$ (resp.~$E_2$)]\label{claim:biclique-long-path-weak}
The edges in $E_1$ are pairwise weakly independent. Similarly, the edges in $E_2$ are pairwise weakly independent in $D^{\star}_{\contr}$.
\end{claim}

\begin{claimproof}
Let $(a_{\rho(i-1)},{v}_{i}), \allowbreak (a_{\rho(j-1)},{v}_{j}) \in E_1$. For the sake of contradiction, say $(v_j, a_{\rho(i-1)}) \in \Dcontr^{\star}$. Then this implies a $(a_{\rho(j-1},a_{\rho(i-1)})$-path in $\Dcontr^{\star}$ contradicting that $A$ is an independent set in $\Dcontr^{\star}$.
Using symmetric arguments one can show that all edges in $E_2$ are pairwise weakly independent.
\end{claimproof}

\begin{claim}[Orderedness on $E_1$ and $E_2$]\label{claim:biclique-long-path-ordered}
For any $i,j \in [3, p-2]$, $i <j$, $(a_{\rho(j-1)}, v_{i}), \allowbreak (a_{\rho(j-1)}, a_{\rho(i-1)}),  \allowbreak ({v}_{j},  \allowbreak b_{\rho(i)}), \allowbreak (v_j,v_i) \not \in E(\Dcontr^{\star})$.
\end{claim}

\begin{claimproof}
Since $A$ is an independent set in $\Dcontr^{\star}$, $(a_{\rho(j-1)}, a_{\rho(i-1)}) \not \in E(\Dcontr^{\star})$. Since $\Dcontr^{\star}$ is acyclic and there is a $(v_i,v_j)$-path in $\Dcontr$, we conclude that $(v_j,v_i) \not \in E(\Dcontr^{\star})$.
If there is an $(a_{\rho(j-1)}, {v}_{i})$-path in $\Dcontr$, then this together with the $( {v}_{i},v_{j-1})$-subpath of $P'$ and the $(v_{i-1},b_{\rho(i-1)})$ path $P^{\bbb}_{i-1}$, implies a $(a_{\rho(j-1),b_{\rho(j-1)}})$-path of length strictly greater than $1$ in $\Dcontr$, which is a contradiction. Using symmetric arguments, one can show that 
$({v}_{j}, b_{\rho(i)}), (v_j,v_i) \not \in E(\Dcontr^{\star})$.
\end{claimproof}

\begin{claim}[Weak independence between $E_1$ and $E_2$]\label{claim:biclique-long-path-weak-between}
Let $(v_i,b_{\rho(i)}) \in E_1$ and $(a_{\rho(j-1)},v_j) \in E_2$.
There is $(v_j,b_{\rho(i)}), (b_{\rho(i)},a_{\rho(i-1)}) \not \in E(\Dcontr^{\star})$. 
\end{claim}

\begin{claimproof}
First observe from the construction of $E_1,E_2$ that $i <j$.
If $(v_j,b_{\rho(i)}) \in E(\Dcontr^{\star})$, then the $(a_{\rho(i)},v_{i+1})$-path $P^{\aaa}_{i+1}$, together with the $(v_{i+1},v_j)$-subpath of $P'$, together with the $(v_j,b_{\rho(i)})$-path, implies a $(a_{\rho(i)},b_{\rho(i)})$-path of length strictly greater than $1$ in $\Dcontr$, which is a contradiction. 

If $(b_{\rho(i)},a_{\rho(i-1)})  \in E(\Dcontr^{\star})$, then the $(a_{rho(i-1)},b_{\rho(i)})$-path obtained by concatenating $P^{\aaa}_i$ and $P^{\bbb}_i$, together with the $(b_{\rho(i)},a_{\rho(i-1)})$-path, implies a $(a_{\rho(i-1)},a_{\rho(j-1)})$-path in $\Dcontr$, which contradicts that $A$ is an independent set in $\Dcontr$.
\end{claimproof}

\begin{claim}[Strong independence between $E_1$ and $E_2$]\label{claim:biclique-long-path-strong}
Let $(v_i,b_{\rho(i)}) \in E_1$ and $(a_{\rho(j-1)},v_j) \in E_2$.
Then $(v_i,a_{\rho(j-1}), \allowbreak (a_{\rho(j-1)},v_i), \allowbreak (b_{\rho(i)},v_j), \allowbreak (v_j,b_{\rho(i)}) \not \in E(\Dcontr^{\star})$.
\end{claim}

\begin{claimproof}
Recall that $i >j$.
If $(v_i,a_{\rho(j-1)}) \in E(\Dcontr^{\star})$, then together with $P^{\aaa}_i$, it contradicts that $A$ is independent in $\Dcontr^{\star}$. Similarly, if $(b_{\rho(i)},v_j) \in E(\Dcontr^{\star})$, then it contradicts the independence of $B$ in $\Dcontr^{\star}$.

If $(a_{\rho(j-1)},v_i) \in E(\Dcontr^{\star})$, then this together with the $(v_i,v_{j-1})$-subpath of $P'$, together with the path $P^{\bbb}_{j-1}$, implies a $(a_{\rho(j-1)},b_{\rho(j-1)})$-path of length strictly greater than $1$ in $\Dcontr$, which is a contradiction. Similarly if, $(v_j,b_{\rho(i)}) \in E(\Dcontr^{\star})$, then this implies a $(a_{\rho(i)},b_{\rho(i)})$-path of length strictly greater than $1$ in $\Dcontr$, which is a contradiction.
\end{claimproof}
This concludes the proof.
\end{proof}

\subsubsection{No long path in $\Dcontr$}\label{sec:biclique:nolongpath}

In this section, we prove Lemma~\ref{lem:digging} in the case when the length of a longest path in $\Dcontr$ is $p \leq 2t+2$. Recall that $\Dcontr$ contains $(A,B)$-matching and $\Dcontr^{\star}$ contains $(A,B)$-induced-biclique. Also, the edges of $A \cup B$ are minimal in $\Dcontr$. We prove Lemma~\ref{lem:digging} by induction on the length of a longest $(A,B)$-path in $\Dcontr$.
For this purpose we essentially restate Lemma~\ref{lem:digging} in a form that is ``induction-friendly''. Observe, as a base case, that when the length of a longest $(A,B)$-path is $1$, then since $\Dcontr^{\star}$ contains $(A,B)$-induced-biclique, we conclude that $\Dcontr$ contains $(A,B)$-induced-biclique and the edges of the $(A,B)$-induced-biclique are minimal. In this case, we conclude that we get the second output of Lemma~\ref{lem:digging}.

\begin{lemma}\label{lem:biclique:nolongpath}
Let $A,B \subseteq V(\Dcontr)$ be ordered sets such that the length
of a longest $(A,B)$-path in $\Dcontr$ is $p$,
$|A|= |B| \geq f(t,p)$, $\Dcontr^{\star}$ contains an $(A,B)$-induced-biclique and the matching edges of $(A,B)$ are minimal in $\Dcontr$, then either
\begin{enumerate}
\item there exist ordered sets $A',B' \subseteq V(\Dcontr)$ such that the length
of a longest $(A',B')$-path in $\Dcontr$ is at most $p$,
$|A'|= |B'| \geq f(t,p-1)$, $\Dcontr^{\star}$ contains an $(A',B')$-induced-biclique and the matching edges of $(A',B')$ are minimal in $\Dcontr$, or
\item $\Dcontr$ contains an ordered $t$-\toughpair, or
\item $\Dcontr^{\star}$ contains a $t$-induced-biclique whose edges are minimal in $\Dcontr$.
\end{enumerate}
\end{lemma}

\begin{proof}
Let $L$ be the set of vertices of $\Dcontr$ that contains the first internal vertex on every longest $(A,B)$-path in $\Dcontr$.

\begin{claim}\label{claim:biclique6}
There is no $(L,L)$-path in $\Dcontr$.
\end{claim}

\begin{claimproof}
For the sake of contradiction, suppose there exist $x,y \in L$ such that there this an $(x,y)$-path in $\Dcontr$. Note that $x\neq y$, as otherwise, $\Dcontr$ would contain a cycle.
From the definition of $L$, there exists $a_i,a_j \in L$ ($i$ not necessarily distinct from $j$) such that $x$ is the first internal vertex on some $(a_i,B)$-path, say $P_x$, of length $p$ and $y$ is the first internal vertex on some $(a_j,B)$-path, say $P_y$, of length $p$. 
Let $P_{xy}$ be some $(x,y)$-path in $\Dcontr$. We first claim that the set of internal vertices of $P_{x,y}$ is disjoint from $P_y$. Suppose not. Let $z$ be some internal vertex of $P_{xy}$ that is also a vertex of $P_y$. Since $z \in V(P_y)$, there exists a $(y,z)$-path in $\Dcontr$. Also since $z \in V(P_{x,y})$, there exists a $(z,y)$-path in $\Dcontr$. This implies a cycle in $\Dcontr$, which is a contradiction. Thus, we conclude that the internal vertices of $P_{xy}$ are disjoint from that of $P_y$. Then consider the $(x,B)$-path obtained by appending $P_{xy}$ and $P_y$. Note that the length of $P_y$ is $p-1$ and the length of $P_{xy}$ is at least $1$. This $(x,B)$-path, together with the edge $(a_i,x)$, gives an $(a_i,B)$-path of length strictly greater than $p$ in $\Dcontr$, which contradicts that the length of a longest $(A,B)$-path in $\Dcontr$ is $p$.
\end{claimproof}

\begin{claim}\label{claim:biclique7}
For each $(a,x) \in E(\Dcontr)$ such that $a \in A$ and $x \in L$, $(a,x)$ is a minimal edge of $\Dcontr$. 
\end{claim}

\begin{claimproof}
For the sake of contradiction, say there exists an $(a,x)$-path in $\Dcontr$ of length at least $2$. Let this path be $P_{ax}$. Since $x \in L$, there exists an $(A,B)$-path of length $p$ whose first internal vertex is $x$. Let this path be $P_{AB}$. Then the internal vertices of $P_{ax}$ are disjoint from the internal vertices of $P_{AB}$, as otherwise there would be a cycle in $\Dcontr$. By appending the path $P_{sx}$ (which is of length at least $2$) with the subpath of $P_{ABs}$ starting from $x$ (which is of length $p-1$), we get an $(A,B)$-path in $\Dcontr$ of length at least $p+1$, which is a contradiction.
\end{claimproof}

\begin{claim}\label{claim:biclique8}
Let $a \in A$ and $x \in L$, such that $(a,x) \not \in E(\Dcontr)$. Then $(a,x) \not \in E(\Dcontr^{\star})$.
\end{claim}

\begin{claimproof}
For the sake of contradiction, suppose there exists an $(a,x)$-path, say $P_{a,x}$ in $\Dcontr$ of length at least two. Since $x \in L$, there exists $a' \in A$ ($a'$ not necessarily distinct from $a$) such that $x$ is the first internal vertex of a $(a',T)$-path, say $P_{a'B}$, of length $p$. First observe that the internal vertices of $P_{a,x}$ are disjoint from that of $P_{a'B}$, as otherwise there would be a cycle in $\Dcontr$. The path $P_{a,x}$ appended with the $(x,B)$-subpath of $P_{a'B}$ gives an $(A,B)$-path in $\Dcontr$ of length strictly greater than $p$, which is a contradiction.
\end{claimproof}

Consider the bipartite graph $D_{\bip}=\Dcontr[A \cup L]$. Note, from Claim~\ref{claim:biclique7}, that each edge of $D_{\bip}$ is a minimal edge in $\Dcontr$.
We now distinguish into two cases based on the size of the matching in $D_{\bip}$.\\

\noindent{\bf Case 1: The size of a maximum matching in $D_{\bip}$ is at least $\mathcal{R}'(2t)$.} In this case, from Ramsey's Theorem, either there exists an induced matching of size $2t$ in $D_{\bip}$ or a $2t$-induced-biclique. From Claim~\ref{claim:biclique7}, each edge of $D_{\bip}$ is a minimal edge of $D_{contr}$. 

In the first case, we get sets $A^* \subseteq A$ and $B^* \subseteq L$ such that there is an $(A^*,B^*)$-induced-matching in $D_{\bip}$. Since $A$ is an independent set in $\Dcontr^{\star}$, so is $A^*$. Since $B^* \subseteq L$, from Claim~\ref{claim:biclique7}, $B^*$ is an independent set in $\Dcontr^{\star}$. Thus, $(A^*,B^*)$ is an induced-matching in $\Dcontr$. Further, from Claim~\ref{claim:biclique8}, $(A^*,B^*)$ is in fact an induced matching in $\Dcontr^{\star}$. Observe that if a digraph has an induced-matching of size $2t$ in its transitive closure, then it has a ordered $t$-\toughpair. Thus, in this case we get the second outcome of the lemma.

In the second case, when $H_{\bip}$ contains a $2t$-induced-biclique, then this is also a $2t$-induced-biclique in $D_{\contr}$ where all edges of the biclique are minimal edges of $\Dcontr$. Thus, in this case we get the third outcome of the lemma.\\

\noindent{\bf Case 2: The size of a maximum matching in $D_{\bip}$ is at most $\mathcal{R}'(2t)$.} In this case, 
we will find ordered sets $A',B' \subseteq V(\Dcontr)$ satisfies the properties stated in the first outcome of the lemma.

Let $Z$ be a minimum vertex cover of $D_{\bip}$ of size at most $\mathcal{R}'(2t)$. Let $Z_A = Z \cap A$ and 
$Z_L=Z \cap L$. Let $A^*=A \setminus Z_S$. Since $Z_A \cup Z_L$ is a vertex cover of $D_{\bip}$, for any  $a \in A^*$, $N_L(a) \subseteq Z_L$.
For each subset $L' \subseteq Z_L$, let $A_{L'} \subseteq A^*$ such that for each vertex $a \in A_{L'}$, $N_{Z_L}(a)=L'$ (and hence $N_L(a)=L'$).

Fix $L' \subseteq Z_L$ such that $|S_{L'}|$ is maximized. 
Recall that the sets $A =(a_1,\ldots,a_{f(t,p)})$ and $T=(t_1, \ldots, t_{f(t,p)})$ are ordered sets.
Set $A'=A_{L'}$ and $B'$ be the corresponding vertices of $B$, 
that is $b_i \in B'$ if and only if $a_i \in A'$.
From the choice of $L'$, $|A'|=|B'| \geq (|A|-\mathcal{R}'(2t))/2^{\mathcal{R}'(2t)} = f(t,p-1)$.
We will now show that the length of a longest $(A',B')$-path in $\Dcontr$ is at most $p-1$. 

\begin{claim}\label{claim:biclique5}
Let $a_i \in A'$ and $b_j \in B'$ such that $i \neq j$. Let $P=(a_i,x_1,\ldots,x_q,b_j)$ be a longest $(A',B')$-path in $\Dcontr$. Then $x_1 \not \in L$.
\end{claim}
\begin{claimproof}
Since $b_j \in B'$, $a_j \in A'$ (from construction of $B'$).
First observe that since $b_j \in B'$, then $a_j \in A'$.
For the sake of contradiction, suppose $x_1 \in L$. Since $(a_i,x) \in E(\Dcontr)$ and all the vertices of $A'$ have the same neighbourhood in $Z_L$ and $a_j \in A'$, we conclude that $(a_j,x) \in E(\Dcontr)$. Thus, there is an $(a_j,b_j)$-path in $\Dcontr$ of length strictly greater than one, which contradicts that $(a_j,b_j)$ is a minimal edge of $\Dcontr$.
\end{claimproof}

From Claim~\ref{claim:biclique5}, the first internal vertex of every longest $(A',B')$-path is not contained in $L$. Now suppose that there exists an $(A',B')$-path of length $p$. Then its first internal vertex should be in $L$ by the definition of $L$. This is a contradiction.  
\end{proof}

Lemma~\ref{lem:biclique:nolongpath} proves Lemma~\ref{lem:digging} when the length of the longest path $p \leq 2t+2$.
\subsection{Cleaning the ordered \toughpair}
\label{sec:cleaningordered}

The cleaning of ordered \toughpair s is done in two steps:
we first show that an ordered \toughpair\ can be identified to a so-called \emph{semi-cleaned ordered \toughpair} (see \Cref{sec:semicleaning}) and show thereafter how to further identify a semi-cleaned ordered \toughpair\ to a hard matching pattern (see \Cref{sec:semicleantoclean}).

\subsubsection{Semi-cleaning}
\label{sec:semicleaning}

Recall the definition of a semi-cleaned ordered $t$-\toughpair\ from Section~\ref{sec:cleaning}.
The goal of this section is to prove Lemma~\ref{lem:semi-cleaning} restated below.

\semicleaning*

To prove Lemma~\ref{lem:semi-cleaning}, we first define a type of a vertex in $D$ with respect to the ordered $t$-\toughpair\ $(A,B)$. This is based on the neighbourhoods of a vertex in the head sets and tail sets of $A$ and $B$. We then define a notion of a bad type. The vertices having a bad type hinder the semi-cleaning procedure. To overcome this,
we show that if a digraph $D$ contains an ordered $t^2$-\toughpair\ then it also contains an ordered $t$-\toughpair\ $(A',B')$ such that there are no bad vertices with respect to $(A',B')$. We then show that given such an ordered \toughpair\ $(A',B')$ we can identify the remaining vertices to one of the endpoints of $A' \cup B'$.

\paragraph*{Bad vertices with respect to the ordered \toughpair\ $(A,B)$.}
Given an ordered \toughpair\ $(A,B)$ in a digraph $D$,
we say that a vertex $v \in V(D) \setminus V(A\cup B)$ is {\em bad} with respect to $(A,B)$ if $v \in N^+_{D^{\star}}(\tail(A)) \cap N^-_{D^{\star}}(\head(A)) \cap N^+_{D^{\star}}(\tail(B)) \cap N^-_{D^{\star}}(\head(B))$. 

\begin{lemma}[Eliminating bad vertices]\label{lem:badvertices}
Let $D$ be a
digraph such that $D$ contains an ordered $t^2$-\toughpair. Then $D$ contains an ordered $t$-\toughpair\ $(A',B')$ such that there is no bad vertex in $V(D) \setminus V(A' \cup B')$ with respect to $(A',B')$.
\end{lemma}

\begin{proof}
Let $(A,B)$ be an ordered $t^2$-\toughpair\ in $D$ and
let $t'=t^2$.
Let $A=((w_1,x_1), \ldots, (w_{t'},x_{t'}))$ and $B=((y_1,z_1), \ldots, (y_{t'},z_{t'}))$.
For each vertex $v \in V(H) \setminus V(A \cup B)$ which is bad with respect to $(A,B)$, let $W(v)$ be the largest index in $[t']$ such that $(a_{W(v)},v) \in E(D)$, $X(v)$ be the smallest index in $[t']$ such that $(v,b_{X(v)}) \in E(D)$, $Y(v)$ be the largest index in $[t']$ such that $(c_{Y(v)},v) \in E(D)$, and $Z(v)$ be the smallest index in $[t']$ such that $(v,d_{Z(v)}) \in E(D)$.
For every $p,q \in [\sqrt{t'}]$, let $V_{p,q}$ be the set of bad vertices $v$ with respect to $(A,B)$ such that
$W(v),X(v) \in \{(p-1)\sqrt{t'}+1, \ldots, p\sqrt{t'}\}$, and $Y(v),Z(v) \in \{(q-1)\sqrt{t'}, \ldots, q\sqrt{t'}\}$.

If there exists $p,q \in [\sqrt{t'}]$ such that $V_{pq}=\emptyset$, then let $A'=((w_{(p-1)\sqrt{t'}+1},x_{(p-1)\sqrt{t'}}+1), \ldots, (x_{p\sqrt{t'}}, x_{p\sqrt{t'}}))$ and 
$B'=((c_{(q-1)\sqrt{t'}+1},z_{(q-1)\sqrt{t'}+1}),\ldots, (c_{q\sqrt{t'}},z_{q\sqrt{t'}}))$. Then from the definition of $V_{p,q}$, $(A',B')$ is an ordered $\sqrt{t'}$-\toughpair\ such that there are no bad vertices with respect to it.

Without loss of generality, assume that for each $p,q \in [\sqrt{t'}]$, $V_{pq}\neq \emptyset$. Let $v_{p,q}$ denote an arbitrarily fixed vertex in $V_{pq}$. 
Let $A'=((w_{W(v_{2,1}}, v_{2,1}), \ldots, (w_{W(v_{\sqrt{t'},1})}, v_{\sqrt{t'},1}))$ and
 $B'=((y_{Y(v_{1,2})}, v_{1,2}),\ldots, (y_{Y(v_{\sqrt{t'},1})},v_{1,\sqrt{t'}} ))$. We will now show that $(A',B')$ is an ordered \toughpair\ in $D$. Moreover, we will show that there is no $(\tail(B'),\head(A'))$-path in $D$, thereby concluding that there is no bad vertex with respect to $(A',B')$. Note that the edges of $A',B'$ are minimal edges because they belong to $D$ which is a reachability minimal digraph.

\paragraph*{Weak independence of $A'$ (resp.~$B'$).}  We show that  there is no $(\head(A'),\tail(A'))$ path (resp.~no $(\head(B'),\tail(B'))$ path) in $D$. Recall that every vertex in $\head(A')$ is a bad vertex with respect to $(A,B)$.
Thus, each vertex of $\tail(A')$ is an out-neighbour of some vertex of $\tail(B')$. 
Thus, a $(\head(A'),\tail(A'))$-path in $D$ implies a $(\tail(B),\tail(A'))$-path in $D$. Since $\tail(A') \subseteq \tail(A)$, this contradicts that every edge of $A$ is strongly independent with every edge of $B$. 
The other case is symmetric.

\paragraph*{Weak independence between $A'$ and $B'$.} We show that a pair of edges containing an edge of $A'$ and an edge of $B'$ is weakly independent.
To see that there is no $(\head(B'),\tail(A'))$-path, 
recall that every vertex of $\head(B')$ is a bad vertex with respect to $(A,B)$.
Thus, each vertex of $\head(B')$ is an out-neighbour of some vertex of $\tail(B)$. Thus, a $(\head(B'),\tail(A'))$-path in $D$ implies a $(\tail(B),\head(A'))$-path in $D$. Since $\tail(A') \subseteq \tail(A)$, this is a contradiction. One can symmetrically prove that there is no
$(\head(A'),\tail(B'))$-path.

\paragraph*{Strong independence of $A'$ and $B'$.} We prove the following three statements.

{\it There is no $(\tail(A'),\tail(B'))$-path and no $(\tail(B'),\tail(A'))$-path.} Since $\tail(A') \subseteq \tail(A)$ and $\tail(B') \subseteq \tail(B)$, and $(A,B)$ is a  \toughpair, we conclude that there is no $(\tail(A'),\allowbreak
\tail(B'))$-path and no $(\tail(B'),\tail(A'))$-path.

{\it There is no $(\head(A'),\head(B'))$-path.} For the sake of contradiction, suppose there exists $v_{i,1} \in \head(B')$ and $v_{1,j} \in \head(B')$, $i,j \neq 1$, such that there exists a $(v_{i,1}, v_{1,j})$-path in $D$. By definition $w(v_{i,1}) \in \{2\sqrt{t'}+1, \ldots, t'\}$ and $X(v_{1,j}) \in \{1, \ldots, \sqrt{t'}\}$. Also, $(w_{W(v_{i,1})}, v_{i,1}) \in E(D)$ and $(v_{1,j}, x_{X(v_{1,j})}) \in E(D)$. Thus, a $(v_{i,1}, v_{1,j})$-path in $D$ implies a $(w_{W(v_{i,1})}, y_{Y(v_{1,j})})$-path in $D$. Since $X(v_{1,j}) < W(v_{i,1})$, this is a contradiction.

{\it There is no $(\head(B'),\head(A'))$-path.} For the sake of contradiction, suppose there exists $v_{i,1} \in \head(A')$ and $v_{1,j} \in \head(B')$, $i,j \neq 1$, such that there exists a $(v_{1,j}, v_{i,1})$-path in $D$. By definition $Y(v_{1,j}) \in \{2\sqrt{t'}+1, \ldots, t'\}$ and $Z(v_{i,1}) \in \{1, \ldots, \sqrt{t'}\}$. Also, $(y_{Y(v_{1,j})}, v_{1,j}) \in E(D)$ and $(v_{i,1}, z_{Z(v_{i,1})}) \in E(D)$. Thus, a $(v_{1,j}, v_{i,1})$-path in $D$ implies a $(y_{Y(v_{1,j})}, z_{Z(v_{i,1})})$-path in $D$. Since $Z(v_{i,1}) < Y(v_{1,j})$, this is a contradiction.

\paragraph*{Ordered condition on $A'$ (resp.~$B'$).}
Consider vertices $v_{i,1}$ and $v_{j,1}$ such that $i \leq j$. Suppose for the sake of contradiction that there exists a $(v_{j,1},v_{i,1})$-path in $D$. Recall $(w_{W(v_{j,1})}, v_{j,1}) \in E(D)$ and $(v_{i,1},x_{X(v_{i,1})}) \in E(D)$. Thus, there exists a $(w_{W(v_{j,1})}, y_{Y(v_{i,1})})$-path in $D$. Since $i \leq j$, $W(v_{j,1}) \leq X(v_{i,1})$. This is a contradiction. The other case can be proved symmetrically.

\paragraph*{There is no $(\tail(B'),\head(A'))$-path.} For the sake of contradiction, let $v_{i,1} \in \head(A')$, $i\neq 1$, and $y_j \in \tail(B')$ such that there exists a $(y_j, v_{i,1})$-path in $D$. Since $y_j \in \tail(B')$, $j \in \{2\sqrt{t'} +1, \ldots, t'\}$. Also $Z(v_{i,1}) \in \{1, \ldots, \sqrt{t}\}$. Since $(v_{i,1}, z_{Z(v_{i,1})}) \in E(D)$, we conclude that there is a $(y_j, z_{Z(v_{i,1})})$-path. This is a contradiction which concludes the proof.
\end{proof}

\begin{lemma}\label{lem:semin-cleaning-wo-bad}
Let $D$ be a directed 
graph and let $(A,B)$ be an ordered $t$-\toughpair\ in $D$ such that there are no bad vertices with respect to $(A,B)$ in $D$. Then one can obtain a semi-cleaned ordered $t$-\toughpair\ from $D$ by identification.
\end{lemma}

The remainder of this section is devoted to the proof of Lemma~\ref{lem:semin-cleaning-wo-bad}. Note that Lemma~\ref{lem:semin-cleaning-wo-bad}, together with Lemma~\ref{lem:badvertices}, proves Lemma~\ref{lem:semi-cleaning}.

Recall that $A,B$ are ordered sets. Let $\tail(A)=W=(w_1, \ldots,w_t)$, $\head(A)=X=(x_1, \ldots,x_t)$, $\tail(B)=Y=(y_1, \ldots, y_t)$ and $\head(B)=Z=(z_1, \ldots, z_t)$.
We will now describe a procedure of identifying vertices in $V(D) \setminus V(A \cup B)$ onto the vertices of $V(A \cup B)$, as long as there is a vertex that has both an in-neighbour and an out-neighbour in $V(A \cup B) = W \cup X \cup Y \cup Z$ in the graph $D^*$, while maintaining the invariants
that after the identification $(A,B)$ remains an ordered \toughpair\ 
in the new graph, no vertex in the new graph is bad with respect to $(A \cup B)$ and the new graph is acyclic. We will apply these identification rules exhaustively, in order.

\begin{invariants*}\label{inv:ordered-cleaning}
If $\widehat{D}$ is the graph obtained after identification, then $(A,B)$ is an ordered \toughpair\ in $\widehat{D}$, and
there are no bad vertices with respect to $(A,B)$ in $\widehat{D}$.
\end{invariants*}

\begin{identification rule}\label{ir:A}
If there exists $v \in V(D) \setminus (W \cup X \cup Y \cup Z)$ such that $N_{D^{\star}}^+(v) \cap W \neq \emptyset$ and $N_{D^{\star}}^-(v) \cap W \neq \emptyset$, then let $i \in [t]$ be the largest index such that $(w_i,v) \in E(D^{\star})$. Identify $v$ onto $w_i$.
\end{identification rule}

\begin{lemma}\label{lem:irA-safe}
Let $\widehat{D}$ be the graph obtained after the application of Identification Rule~\ref{ir:A}. Then, the invariants are satisfied.
\end{lemma}

\begin{proof}
Since  $N_{D^{\star}}^+(v) \cap W \neq \emptyset$ and $N_{D^{\star}}^-(v) \cap W \neq \emptyset$, 
there is a $(v,W)$-path, say $P_{\out}$, and a $(W,v)$-path say $P_{\inp}$ in $D$.\\

\noindent{\bf Weak independence of $A$ in $\widehat{D}$.} For the sake of contradiction, say there is a $(X,w_i)$-path in $\widehat{D}$. Then there exists a $(X,v)$-path in $D$. This, together with the $(v,W)$-path $P_{\out}$, implies an $(X,W)$-path in $D$, which is a contradiction to the weak independence of the edges of $A$ in $D$.\\

\noindent{\bf Weak independence between $A,B$ in $\widehat{D}$.} For the sake of contradiction, say there is a $(Z,w_i)$-path in $\widehat{D}$. 
Then there exists a $(Z,v)$-path in $D$. This, together with $(v,W)$-path $P_{\out}$, implies an $(Z,W)$-path in $D$, which is a contradiction to the weak independence between the edges of $A, B$ in $D$.\\

\noindent{\bf Strong independence between $A,B$ in $\widehat{D}$.} For the sake of contradiction, 
say there is a $(w_i,Y)$-path in $\widehat{D}$. 
Then there exists a $(v,Y)$-path in $D$. This, together with the $(W,v$-path $P_{\inp}$ in $D$, implies an $(W,Y)$-path in $D$, which is a contradiction to the strong independence between the edges of $A, B$ in $D$.

For the sake of contradiction, 
say there is a $(Y,w_i)$-path in $\widehat{D}$. 
Then there exists a $(Y,v)$-path in $D$. This, together with the $(v,W)$-path $P_{\out}$, implies an $(Y,W)$-path in $D$, which is a contradiction to the strong independence between the edges of $A, B$ in $D$.

\noindent{\bf Orderedness.} Fix $j <i$. Suppose, for the sake of contradiction, that there is a $(w_i,w_j)$-path in $\widehat{D}$. Then, there is a $(v,w_j)$-path in $D$. This, together with the $(w_i,v)$-path in $D$, implies a $(w_i,w_j)$-path in $D$, contradicting the ordered-ness condition on the edges of $A$.

For the sake of contradiction, for $j <i$, say there is a $(w_i,x_j)$-path in $\widehat{D}$. Then there is a $(v,x_j)$-path in $D$, which together with the $(w_i,v)$-path in $D$ contradicts the ordered-ness condition on the edges of $A$.\\

\noindent{\bf No bad vertices.} For the sake of contradiction, say there exists a bad vertex $u$ with respect to $(A,B)$ in $\widehat{D}$, that is, $N_{\widehat{D}^{\star}}^-(u) \cap W, N_{\widehat{D}^{\star}}^+(u) \cap X, N_{\widehat{D}^{\star}}^-(u) \cap Y, N_{\widehat{D}^{\star}}^+(u) \cap Z \neq \emptyset$. 
Then, considering that $\widehat{D}$ is obtained from $D$ by only identifying $v$ onto $w_i$, we get that
$N_{D^{\star}}^+(u) \cap X, N_{D^{\star}}^-(u) \cap Y, N_{D^{\star}}^+(u) \cap Z \neq \emptyset$ and $(v,u) \in E(D^{\star})$.  The $(v,u)$-path, together with the $(W,v)$-path $P_{\in}$, implies a $(W,u)$-path in $D$, that is $N_{D^{\star}}^-(u) \cap W \neq \emptyset$. That is, $u$ is a bad vertex with respect to $(A,B)$ in $D$, which is a contradiction.
\end{proof}

Identification Rule~\ref{ir:B} is the analogue of Identification Rule~\ref{ir:A}.

\begin{identification rule}\label{ir:B}
If there exists $v \in V(D) \setminus (W \cup X \cup Y \cup Z)$ such that $N_{D^{\star}}^+(v) \cap Y \neq \emptyset$ and $N_{D^{\star}}^-(v) \cap Y \neq \emptyset$, then let $i \in [t]$ be the largest index such that $(y_i,v) \in E(D^{\star})$. Identify $v$ onto $y_i$.
\end{identification rule}

Symmetrically to \Cref{lem:irA-safe}, the following holds.

\begin{lemma}\label{lem:irB-safe}
Let $\widehat{D}$ be the graph obtained after the application of Identification Rule~\ref{ir:B}. Then, the invariants are satisfied.
\end{lemma}

\begin{identification rule}\label{ir:C}
If there exists $v \in V(D) \setminus (W \cup X \cup Y \cup Z)$ such that $N_{D^{\star}}^-(v) \cap X \neq \emptyset$ and  $N_{D^{\star}}^+(v) \cap X \neq \emptyset$,
then let $i \in [t]$ be the largest index such that either $(w_i,v) \in E(D^{\star})$ or $(x_i,v) \in E(D^{\star})$. 
Identify $v$ onto $x_i$.
\end{identification rule}

\begin{lemma}\label{lem:irC-safe}
Let $\widehat{D}$ be the graph obtained after the application of Identification Rule~\ref{ir:C}. Then, the invariants are satisfied.
\end{lemma}

\begin{proof}
Since $N_{D^{\star}}^-(v) \cap X \neq \emptyset$ and $N_{D^{\star}}^+(v) \cap X \neq \emptyset$, 
here is a $(X,v)$-path, say $P_{\inp}$, and a $(v,X)$-path say $P_{\out}$ in $D$.\\

\noindent{\bf Weak independence of $A$ in $\widehat{D}$.} For the sake of contradiction, say there is a $(x_i,w_j)$-path in $\widehat{D}$. Then there exists a $(v,w_j)$-path, say $P$, in $D$. 
From the choice of $i$ and since $N_{D^{\star}}^-(v) \cap X \neq \emptyset$, there exists $i' \leq i$, such that there is a $(x_{i'},v)$-path in $D$. This implies a $(x_{i'},w_j)$-path in $D$. Since $i' <i$ and $A$ is weakly independent in $D$, $j < i$. 

From the choice of $i$, either $(w_i,v) \in E(D^*)$, in which case the existence of $P$ implies a $(w_i,w_j)$-path, $j<i$, in $D$, contradicting the weak independence of $A$ in $D$. Or, $(x_i,v) \in E(D^{\star})$, in which case $P$ implies a $(x_i,w_j)$-path, $i <i$ in $D$, again contradicting the weak independence of $A$ in $D$.\\

\noindent{\bf Weak independence between $A,B$ in $\widehat{D}$.} For the sake of contradiction, say there is a $(x_i,Y)$-path in $\widehat{D}$. 
Then there exists a $(v,Y)$-path in $D$. 
This, together with $(X,v)$-path $P_{\inp}$, implies an $(X,Y)$-path in $D$, which is a contradiction to the weak independence between the edges of $A, B$ in $D$.\\

\noindent{\bf Strong independence between $A,B$ in $\widehat{D}$.} For the sake of contradiction, 
say there is a $(Z,x_i)$-path in $\widehat{D}$. 
Then there exists a $(Z,v)$-path in $D$. This, together with the $(v,X)$-path $P_{\out}$ in $D$, implies an $(Z,X)$-path in $D$, which is a contradiction to the strong independence between the edges of $A, B$ in $D$.

For the sake of contradiction, 
say there is a $(x_i,Z)$-path in $\widehat{D}$. 
Then there exists a $(v,Z)$-path in $D$. This, together with the $(X,v)$-path $P_{\inp}$, implies an $(X,Z)$-path in $D$, which is a contradiction to the strong independence between the edges of $A, B$ in $D$.\\

\noindent{\bf Orderedness.} 
Suppose, for the sake of contradiction, that there is a $(w_j,x_i)$-path in $\widehat{D}$ for $j>i$. Then, there is a $(w_j,v)$-path in $D$. Since $j >i$, this contradicts the choice of $i$.

For the sake of contradiction, say there is a $(x_j,x_i)$-path in $\widehat{D}$ for $j >i$. Then, there is a $(x_j,v)$-path in $D$. Since $j >i$, this contradicts the choice of $i$.

For the sake of contradiction, say there is a $(x_i,x_j)$-path in $\widehat{D}$ for $j <i$. Then, there is a $(v,x_j)$-path in $D$. This together with either the $(x_i,v)$-path or the $(w_i,v)$-path, implies either a $(x_i,x_j)$-path or a $(w_i,x_j)$-path in $D$, contradicting the ordered-ness condition on $A$.\\

\noindent{\bf No bad vertices.} For the sake of contradiction, say there exists a bad vertex $u$ with respect to $(A,B)$ in $\widehat{D}$, that is, $N_{\widehat{D}^{\star}}^-(u) \cap W, N_{\widehat{D}^{\star}}^+(u) \cap X, N_{\widehat{D}^{\star}}^-(u) \cap Y, N_{\widehat{D}^{\star}}^+(u) \cap Z \neq \emptyset$. 
Then, considering that $\widehat{D}$ is obtained from $D$ by only identifying $v$ onto $x_i$, we get that
$N_{D^{\star}}^-(u) \cap W, N_{D^{\star}}^-(u) \cap Y, N_{D^{\star}}^+(u) \cap Z \neq \emptyset$ and $(u,v) \in E(D^{\star})$.  The $(u,v)$-path, together with the $(v,X)$-path $P_{\out}$, implies a $(u,X)$-path in $D$, that is $N_{D^{\star}}^+(u) \cap X \neq \emptyset$. That is, $u$ is a bad vertex with respect to $(A,B)$ in $D$, which is a contradiction.
\end{proof}

Identification Rule~\ref{ir:D} is the analogue of Identification Rule~\ref{ir:C}.

\begin{identification rule}\label{ir:D}
If there exists $v \in V(D) \setminus (W \cup X \cup Y \cup Z)$ such that $N_{D^{\star}}^-(v) \cap Z \neq \emptyset$ and  $N_{D^{\star}}^+(v) \cap Z \neq \emptyset$,
then let $i \in [t]$ be the largest index such that either $(y_i,v) \in E(D^{\star})$ or $(z_i,v) \in E(D^{\star})$. 
Identify $v$ onto $z_i$.
\end{identification rule}

Symmetrically to \Cref{lem:irC-safe}, the following holds.

\begin{lemma}\label{lem:irC-safe}
Let $\widehat{D}$ be the graph obtained after the application of Identification Rule~\ref{ir:C}. Then, the invariants are satisfied.
\end{lemma}

The remaining identification rules are applied on a vertex $v$ when the (in and out) neighbourhood of $v$ is empty in one of the sets $W,X,Y$ or $Z$.

\begin{identification rule}\label{ir:emptyA}
If $N_{D^{\star}}(v) \cap W=\emptyset$ and $\emptyset \neq N_{D^{\star}}^-(v) \cap (X \cup W \cup Y \cup Z) \subseteq Y$, then 
let $i \in [t]$ be the largest integer such that $(y_{i},v) \in E(D^*)$.
Identify $v$ onto $y_{i}$.
\end{identification rule}

\begin{lemma}\label{lem:iremptyA-safe}
Let $\widehat{D}$ be the graph obtained after the application of Identification~\ref{ir:emptyA}. Then, $\widehat{D}$ satisfies the invariants.
\end{lemma}

\begin{proof}
\noindent{\bf Weak independence of $B$ in $\widehat{D}$.} For the sake of contradiction, say there is a $(Z,y_i)$-path in $\widehat{D}$. 
Then there exists a $(Z,v)$-path, in $D$, which is a contradiction (to the assumptions of Identification Rule~\ref{ir:emptyA}).\\

\noindent{\bf Weak independence between $A,B$ in $\widehat{D}$.} For the sake of contradiction, say there is a $(X,y_i)$-path in $\widehat{D}$. 
Then there exists a $(X,v)$-path in $D$, which is a contradiction (to the assumptions of Identification Rule~\ref{ir:emptyA}).\\

\noindent{\bf Strong independence between $A,B$ in $\widehat{D}$.} For the sake of contradiction, 
say there is a $(W,y_i)$-path in $\widehat{D}$. 
Then there exists a $(W,v)$-path in $D$, which is a contradiction (to the assumptions of Identification Rule~\ref{ir:emptyA}). Similarly,  if there is a $(y_i,W)$-path in $\widehat{D}$, then there exists a $(v,W)$-path in $D$, which is again a contradiction.\\

\noindent{\bf Orderedness.} 
Fix $j <i$.
For the sake of contradiction,say there is a $(y_i,y_j)$-path in $\widehat{D}$. Then, there is a $(v,y_j)$-path in $D$. This together with the $(y_i,v)$-path implies a $(y_i,y_j)$-path in $D$, which contradicts the ordered-ness condition on $B$.

Similarly, if there is a $(y_i,z_j)$-path in $\widehat{D}$, then there is a $(v,z_j)$-path in $D$. This, together with the $(y_i,v)$-path, implies a $(y_i,z_j)$-path in $D$, again contradicting the ordered-ness condition on $B$.\\

\noindent{\bf No bad vertices.} For the sake of contradiction, say there exists a bad vertex $u$ with respect to $(A,B)$ in $\widehat{D}$, that is, $N_{\widehat{D}^{\star}}^-(u) \cap W, N_{\widehat{D}^{\star}}^+(u) \cap X, N_{\widehat{D}^{\star}}^-(u) \cap Y, N_{\widehat{D}^{\star}}^+(u) \cap Z \neq \emptyset$. 
Then, considering that $\widehat{D}$ is obtained from $D$ by only identifying $v$ onto $y_i$, we get that
$N_{D^{\star}}^-(u) \cap W, N_{D^{\star}}^+(u) \cap X, N_{D^{\star}}^+(u) \cap Z \neq \emptyset$ and $(v,u) \in E(D^{\star})$.  The $(v,u)$-path, together with the $(v,y_i)$-path, implies a $(u,Y)$-path in $D$, that is $N_{D^{\star}}^-(u) \cap Y \neq \emptyset$. That is, $u$ is a bad vertex with respect to $(A,B)$ in $D$, which is a contradiction.
\end{proof}

The following identification rule is symmetric to \Cref{ir:emptyA}.

\begin{identification rule}\label{ir:emptyB}
If $N_{D^{\star}}(v) \cap Y=\emptyset$ and $\emptyset \neq N_{D^{\star}}^-(v) \cap (X \cup W \cup Y \cup Z) \subseteq W$, then 
let $i \in [t]$ be the largest integer such that $(x_{i},v) \in E(D^*)$.
Identify $v$ onto $w_{i}$.
\end{identification rule}

Similarly to \Cref{lem:iremptyA-safe}, we can prove the following.

\begin{lemma}\label{lem:iremptyB-safe}
Let $\widehat{D}$ be the graph obtained after the application of Identification~\ref{ir:emptyB}. Then, $\widehat{D}$ satisfies the invariants.
\end{lemma}

\begin{identification rule}\label{ir:emptyC}
If $N_{D^{\star}}(v) \cap X=\emptyset$ and $\emptyset \neq N_{D^{\star}}^+(v) \cap (X \cup W \cup Y \cup Z) \subseteq Z$, then 
let $i \in [t]$ be the smallest integer such that $(v,z_i) \in E(D^*)$.
Identify $v$ onto $z_{i}$.
\end{identification rule}

\begin{lemma}\label{lem:iremptyC-safe}
Let $\widehat{D}$ be the graph obtained after the application of Identification~\ref{ir:emptyC}. Then, $\widehat{D}$ satisfies the invariants.
\end{lemma}

\begin{proof}
\noindent{\bf Weak independence of $B$ in $\widehat{D}$.} For the sake of contradiction, say there is a $(z_i,Y)$-path in $\widehat{D}$. 
Then there exists a $(v,Y)$-path, in $D$, which is a contradiction (to the assumptions of Identification Rule~\ref{ir:emptyC}).\\

\noindent{\bf Weak independence between $A,B$ in $\widehat{D}$.} For the sake of contradiction, say there is a $(z_i,W)$-path in $\widehat{D}$. 
Then there exists a $(v,W)$-path in $D$, which is a contradiction (to the assumptions of Identification Rule~\ref{ir:emptyC}).\\

\noindent{\bf Strong independence between $A,B$ in $\widehat{D}$.} For the sake of contradiction, 
say there is a $(z_i,X)$-path in $\widehat{D}$. 
Then there exists a $(v,X)$-path in $D$, which is a contradiction (to the assumptions of Identification Rule~\ref{ir:emptyC}). Similarly,  if there is a $(X,z_i)$-path in $\widehat{D}$, then there exists a $(X,v)$-path in $D$, which is again a contradiction.\\

\noindent{\bf Orderedness.} 
For the sake of contradiction,say there is a $(z_i,z_j)$-path in $\widehat{D}$ for $j <i$.. Then, there is a $(v,z_j)$-path in $D$. This together with the $(z_i,v)$-path implies a $(z_i,z_j)$-path in $D$, which contradicts the ordered-ness condition on $B$.

If there is a $(z_i,y_j)$-path in $\widehat{D}$,
then there is a $(v,y_j)$-path in $D$. This, together with the $(y_i,v)$-path, implies a $(y_i,y_j)$-path in $D$, again contradicting the ordered-ness condition on $B$.\\

\noindent{\bf No bad vertices.} For the sake of contradiction, say there exists a bad vertex $u$ with respect to $(A,B)$ in $\widehat{D}$, that is, $N_{\widehat{D}^{\star}}^-(u) \cap W, N_{\widehat{D}^{\star}}^+(u) \cap X, N_{\widehat{D}^{\star}}^-(u) \cap Y, N_{\widehat{D}^{\star}}^+(u) \cap Z \neq \emptyset$. 
Then, considering that $\widehat{D}$ is obtained from $D$ by only identifying $v$ onto $z_i$, we get that
$N_{D^{\star}}^-(u) \cap W, N_{D^{\star}}^+(u) \cap X, N_{D^{\star}}^-(u) \cap Y \neq \emptyset$ and $(u,v) \in E(D^{\star})$.  The $(u,v)$-path, together with the $(v,z_i)$-path, implies a $(u,Z)$-path in $D$, that is $N_{D^{\star}}^+(u) \cap Z \neq \emptyset$. That is, $u$ is a bad vertex with respect to $(A,B)$ in $D$, which is a contradiction.
\end{proof}

The following identification rule is symmetric to \Cref{ir:emptyC}.

\begin{identification rule}\label{ir:emptyD}
If $N_{D^{\star}}(v) \cap Z=\emptyset$ and $\emptyset \neq N_{D^{\star}}^+(v) \cap (X \cup W \cup Y \cup Z) \subseteq X$, then 
let $i \in [t]$ be the smallest integer such that $(v,x_i) \in E(D^*)$.
Identify $v$ onto $x_{i}$.
\end{identification rule}

Similarly to \Cref{lem:iremptyC-safe}, we can prove the following.

\begin{lemma}\label{lem:iremptyD-safe}
Let $\widehat{D}$ be the graph obtained after the application of Identification~\ref{ir:emptyD}. Then, $\widehat{D}$ satisfies the invariants.
\end{lemma}

\begin{lemma}\label{lem:ordered-cleaning-exhaustive}
When none of the Identification Rules~\ref{ir:A}-\ref{ir:emptyD} are applicable, for each vertex of $v \in V(D) \setminus V(A\cup B)$ either $N_D^-(v) =\emptyset$ or $N_D^+(v) = \emptyset$.
\end{lemma}
\begin{proof}
To prove the lemma, we prove the following claims.

\begin{claim}\label{claim:ordered-ir}
If $N_{D^{\star}}^+(v) \cap W \neq \emptyset$, then
$ N_{D^{\star}}^-(v) \cap (W \cup X \cup Y \cup Z) \subseteq W$. 
\end{claim}

\begin{claimproof}
If $x_j \in N_{D^{\star}}^-(v) $ for some $j \in [t]$, then this implies a $(X,W)$-path in $D$, contradicting the weak independence of $A$.
If $y_j \in N_{D^{\star}}^-(v)$ for some $j \in [t]$, then this implies a $(Y,W)$-path in $D$, contradicting the strong independence between $A,B$. 
If $z_j \in N_{D^{\star}}^-(v)$ for some $j \in [t]$, then this implies a $(Z,W)$-path in $D$, contradicting the weak independence between $A,B$. 
\end{claimproof}

Symmetrically to \Cref{claim:ordered-ir}, the following holds.

\begin{claim}\label{claim:ordered-ir2}
If $N_{D^{\star}}^+(v) \cap Y \neq \emptyset$, then
$ N_{D^{\star}}^-(v) \cap (W \cup X \cup Y \cup Z) \subseteq Y$. 
\end{claim}

\begin{claim}\label{claim:ordered-ir3}
If $N_{D^{\star}}^-(v) \cap X \neq \emptyset$, then
$ N_{D^{\star}}^+(v) \cap (W \cup X \cup Y \cup Z) \subseteq X$. 
\end{claim}

\begin{claimproof}
If $w_j \in N_{D^{\star}}^+(v) $ for some $j \in [t]$, then this implies a $(X,W)$-path in $D$, contradicting the weak independence of $A$.
If $y_j \in N_{D^{\star}}^+(v)$ for some $j \in [t]$, then this implies a $(X,Y)$-path in $D$, contradicting the weak independence between $A,B$. 
If $z_j \in N_{D^{\star}}^+(v)$ for some $j \in [t]$, then this implies a $(X,Z)$-path in $D$, contradicting the strong independence between $A,B$. 
\end{claimproof}

Symmetrically to \Cref{claim:ordered-ir3}, the following holds.

\begin{claim}\label{claim:ordered-ir4}
If $N_{D^{\star}}^-(v) \cap Z \neq \emptyset$, then
$ N_{D^{\star}}^+(v) \cap (W \cup X \cup Y \cup Z) \subseteq Z$. 
\end{claim}

From Claims~\ref{claim:ordered-ir} and~\ref{claim:ordered-ir4}, if Identification Rules~\ref{ir:A}-\ref{ir:D} are not applicable and there exists $v \in W \cup X \cup Y \cup Z$ such that $N_{D^{\star}}^+(v) \cap (W \cup X \cup Y \cup Z) \neq \emptyset$ and $N_{D^{\star}}^-(v) \cap (W \cup X \cup Y \cup Z) \neq \emptyset$, then since there are no bad vertices with respect to $(A,B)$, we conclude that either $N_{D^{\star}}(v) \cap W =\emptyset$ or $N_{D^{\star}}(v) \cap X =\emptyset$ or $N_{D^{\star}}(v) \cap Y=\emptyset$ or $N_{D^{\star}}(v) \cap Z =\emptyset$. Thus, Claims~\ref{claim:ordered-ir5}-\ref{claim:ordered-ir8} prove the lemma.

\begin{claim}\label{claim:ordered-ir5}
If Identification Rules~\ref{ir:A}-\ref{ir:D} are no longer applicable, $N_{D^{\star}}(v) \cap W=\emptyset$ and  $N_{D^{\star}}^+(v) \cap (X \cup Y \cup Z) \neq \emptyset$, then $N_{D^{\star}}^-(v) \cap (X \cup Y \cup Z) \subseteq Y$.
\end{claim}

\begin{claimproof}
We first show that that $N_{D^{\star}}^-(v) \cap X = \emptyset$. 
For the sake of contradiction, say $N_{D^{\star}}^-(v) \cap X \neq \emptyset$. 
Then, since Identification Rule~\ref{ir:B} is not applicable, $N_{D^{\star}}^+(v) \cap X = \emptyset$.  Also, $N_{D^{\star}}^+(v) \cap Y = \emptyset$, as otherwise there is a $(X,Y)$-path in $D$, contradicting the weak independence between $A,B$. Further, $N_{D^{\star}}^+(v) \cap Z = \emptyset$, as otherwise there is a $(X,Z)$-path in $D$, contradicting the strong independence between $A,B$.

Similarly, if $N_{D^{\star}}^-(v) \cap Z \neq \emptyset$. 
Then, since Identification Rule~\ref{ir:B} is not applicable, $N_{D^{\star}}^+(v) \cap Z = \emptyset$.  Also, 
$N_{D^{\star}}^+(v) \cap X = \emptyset$, as otherwise there is a $(Z,X)$-path in $D$, contradicting the strong independence between $A,B$. Further, 
$N_{D^{\star}}^+(v) \cap Y = \emptyset$, as otherwise there is a $(Z,Y)$-path in $D$, contradicting the weak independence of $B$.
\end{claimproof}

Similarly to \Cref{claim:ordered-ir5}, we can prove the following.

\begin{claim}\label{claim:ordered-ir6}
If Identification Rules~\ref{ir:A}-\ref{ir:D} are no longer applicable, $N_{D^{\star}}(v) \cap Y=\emptyset$ and  $N_{D^{\star}}^+(v) \cap (X \cup W \cup Z) \neq \emptyset$, then $N_{D^{\star}}^-(v) \cap (W \cup X \cup Z) \subseteq W$.
\end{claim}

\begin{claim}\label{claim:ordered-ir7}
If Identification Rules~\ref{ir:A}-\ref{ir:D} are no longer applicable, $N_{D^{\star}}(v) \cap X=\emptyset$ and $N_{D^{\star}}^-(v) \cap (X \cup Y \cup Z) \neq \emptyset$, then $N_{D^{\star}}^+(v) \cap ( W \cup Y \cup Z) \subseteq Z$.
\end{claim}

\begin{claimproof}
We first show that that $N_{D^{\star}}^+(v) \cap W = \emptyset$. 
For the sake of contradiction, say $N_{D^{\star}}^+(v) \cap W \neq \emptyset$. 
Then, since Identification Rule~\ref{ir:B} is not applicable, $N_{D^{\star}}^-(v) \cap W = \emptyset$. 
Also, $N_{D^{\star}}^-(v) \cap Y = \emptyset$, as otherwise there is a $(Y,W)$-path in $D$, contradicting the strong independence between $A,B$. 
Further, $N_{D^{\star}}^-(v) \cap Z = \emptyset$, as otherwise there is a $(Z,W)$-path in $D$, contradicting the weak independence between $A,B$.

Similarly, if $N_{D^{\star}}^+(v) \cap Y \neq \emptyset$,
then since Identification Rule~\ref{ir:B} is not applicable, $N_{D^{\star}}^-(v) \cap Y = \emptyset$.  
Also,  $N_{D^{\star}}^-(v) \cap W= \emptyset$, as otherwise there is a $(W,Z)$-path in $D$, contradicting the strong independence between $A,B$. 
Further, 
$N_{D^{\star}}^-(v) \cap Z= \emptyset$, as otherwise there is a $(Z,Y)$-path in $D$, contradicting the weak independence of $B$.
\end{claimproof}

Similarly to \Cref{claim:ordered-ir7}, we can prove the following.

\begin{claim}\label{claim:ordered-ir8}
If Identification Rules~\ref{ir:A}-\ref{ir:D} are no longer applicable, $N_{D^{\star}}(v) \cap Z=\emptyset$ and $N_{D^{\star}}^-(v) \cap (X \cup Y \cup Z) \neq \emptyset$, then $N_{D^{\star}}^+(v) \cap ( W \cup X \cup Y) \subseteq X$.
\end{claim}
This concludes the proof of \Cref{lem:ordered-cleaning-exhaustive}.
\end{proof}

Now suppose that Identification Rules~\ref{ir:A}-\ref{ir:emptyD} are no longer applicable. From Lemma~\ref{lem:ordered-cleaning-exhaustive}, each vertex of $V(D) \setminus (W \cup X \cup Y \cup Z)$ either has only in-neighbours or only out-neighbours in $D^{\star}$.
Let $V_{\sss} =\{v : v \in V(D) \setminus (W \cup X \cup Y \cup Z) \text{ and } N_{D^{\star}}^+(v) \cap (W \cup X \cup Y \cup Z) \neq \emptyset \}$ and
$V_{\ttt} = V(D) \setminus (W \cup X \cup Y \cup Z \cup V_{\sss})$. Then identify all the vertices of $V_{\sss}$ onto a new vertex $\sss$ and all the vertices of $V_{\ttt}$ onto a new vertex $\ttt$. The resulting graph is a semi-cleaned ordered $t$-\toughpair.

\subsubsection{From semi-cleaned ordered \toughpair\ to hard-pattern}
\label{sec:semicleantoclean}

The aim of this section is to prove \Cref{lem:semicleantoclean}. We first introduce the \emph{almost $t$-hard matching patterns} which differ from the $t$-hard matching patterns in that the source and sink vertices of items 3 and 4, respectively, in \Cref{def:cleanedorderedtoughpair} may not have "full" neighborhoods (recall that an $S$-source is a source vertex~$s$ such that $N^+(s) = S$ and that an $S$-sink is a sink vertex such that $N^-(s) = S$).

\begin{definition}[almost $t$-hard matching pattern]
\label{def:almost}
An \emph{almost $t$-hard matching pattern} is an (acyclic) digraph $D$ constructed the following way. We start with disjoint vertex sets $W = \{w_1,\ldots,w_t\}$, $X = \{x_1,\ldots,x_t\}$, $Y = \{y_1,\ldots,y_t\}$ and $Z = \{z_1,\ldots,z_t\}$ and introduce the edges $(w_i,x_i)$ and $(y_i,z_i)$ for every $i\in[t]$.
Furthermore, we introduce into $D$ any combination of the following items:
\begin{enumerate}
\item[1.] either the directed path $w_1 \rightarrow w_2 \rightarrow \ldots \rightarrow w_t \rightarrow z_1 \rightarrow z_2 \rightarrow \ldots \rightarrow z_t$, or any of the directed paths $w_1 \rightarrow w_2 \rightarrow \ldots \rightarrow w_t$ and 
$z_1 \rightarrow z_2 \rightarrow \ldots \rightarrow z_t$;
\item[2.] either the directed path $y_1 \rightarrow y_2 \rightarrow \ldots \rightarrow y_t \rightarrow x_1 \rightarrow x_2 \rightarrow \ldots \rightarrow x_t$, or any of the directed paths $x_1 \rightarrow x_2 \rightarrow \ldots \rightarrow x_t$ and $y_1 \rightarrow y_2 \rightarrow \ldots \rightarrow y_t$;
\item[3.] a vertex~$s$ such that $N^-(s) = \emptyset$ and $N^+(s) \subseteq W \cup X \cup Y \cup Z$; 
\item[4.] a vertex $t$ such that $N^+(t) = \emptyset$ and $N^-(t) \subseteq W \cup X \cup Y \cup Z$; 
\item[5.] a vertex $r_{WZ}$ such that $N^-(r_{WZ}) = W$ and $N^+(r_{WZ}) = Z$;
\item[6.] a vertex $r_{YX}$ such that $N^-(r_{YX}) = Y$ and $N^+(r_{YX}) = X$.
\end{enumerate}
\end{definition}

We further introduce the four canonical graphs below corresponding to "half" a hard matching pattern (see \Cref{fig:conf} for an illustration).

\begin{definition}
\label{def:canonical}
Let $U = \{u_1,\ldots,u_t\}$ and $V = \{v_1,\ldots,v_t\}$. We define the following four reachability-minimal acyclic digraphs on vertex set $U \cup V$.
\begin{itemize}
\item The digraph $O_1^t$ is a $(U,V)$-induced-matching.
\item The digraph $O_2^t$ is a $(U,V)$-matching which further contains the directed path on $U$ from $u_1$ to $u_t$.
\item The digraph $O_3^t$ is a $(U,V)$-matching which further contains the directed path on $V$ from $v_1$ to $v_t$.
\item The digraph $O_4^t$  is a $(U,V)$-matching which further contains the directed path on $U$ from $u_1$ to $u_t$ and the directed path on $V$ from $v_1$ to $v_t$.
\end{itemize}
\end{definition}

\newcommand{\onepath}[4]{
\foreach \i in {1,...,4}
{
\pgfmathsetmacro{\x}{0.75*\i+#1}
\pgfmathsetmacro{\y}{#2+.75}
\node[circ] (a\i) at (\x,\y) {};
\node[circ] (b\i) at (\x,#2) {};
\draw[->,>=stealth,thick] (a\i) -- (b\i);
}

\ifthenelse{#3=1}{
\foreach \i in {1,...,3}
{
\pgfmathtruncatemacro{\j}{\i+1}
\draw[->,>=stealth,thick] (a\i) -- (a\j);
}
}{}

\ifthenelse{#4=1}{
\foreach \i in {1,...,3}
{
\pgfmathtruncatemacro{\j}{\i+1}
\draw[->,>=stealth,thick] (b\i) -- (b\j);
}
}{}
}

\begin{figure}
\centering
\begin{subfigure}[h]{.23\textwidth}
\centering
\begin{tikzpicture}
\onepath{0}{0}{0}{0};
\end{tikzpicture}
\caption*{$O_1^4$}
\end{subfigure}
\begin{subfigure}[h]{.23\textwidth}
\centering
\begin{tikzpicture}
\onepath{0}{0}{1}{0};
\end{tikzpicture}
\caption*{$O_2^4$}
\end{subfigure}
\begin{subfigure}[h]{.23\textwidth}
\centering
\begin{tikzpicture}
\onepath{0}{0}{0}{1};
\end{tikzpicture}
\caption*{$O_3^4$}
\end{subfigure}
\begin{subfigure}[h]{.23\textwidth}
\centering
\begin{tikzpicture}
\onepath{0}{0}{1}{1};
\end{tikzpicture}
\caption*{$O_4^4$}
\end{subfigure}
\caption{The graphs of \Cref{def:canonical} for $t=4$.}
\label{fig:conf}
\end{figure}
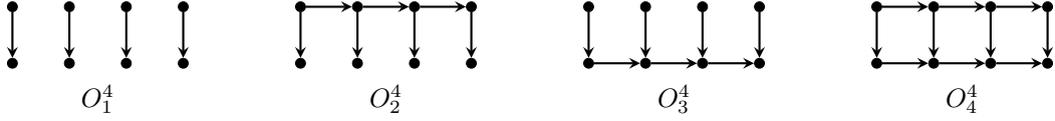

\semicleantoclean*

\begin{proof}
For every $p \geq 1$ and $j \geq 1$, let $\ell_j(p)$ be an integer such that $p \geq \mathcal{R}(\ell_j(p),j)$.
For notational convenience, we show, instead of the above statement, that if $D$ is a semi-cleaned ordered $t$-\toughpair\ then either $D^*$ contains a $g(t)$-induced-biclique whose edges are minimal in $D$, or $D$ can be identified to a digraph $\hat{D}$ such that $\hat{D}$ is transitively equivalent to a $g(t)$-hard matching pattern, where $g(t) = (\ell_4(\ell_4(\ell_5(t))/2)-4)/8$ (equivalently, $h(t) = \mathcal{R}(2\mathcal{R}(2\mathcal{R}(8t+4,4),4),5)$).

Let $D$ be a semi-cleaned ordered $t$-\toughpair\ and let $(A,B)$ be the ordered $t$-\toughpair\ contained in $D$,
where $A = ((w_1,x_1),\ldots,(w_t,x_t))$ and $B = ((y_1,z_1),\ldots,(y_t,z_t))$.
Denote by $W= \{w_i~|~i \in [t]\}$, $X = \{x_i~|~i \in [t]\}$, $Y= \{y_i~|~i \in [t]\}$ and $Z= \{z_i~|~i\in [t]\}$.
Given a digraph $H$, we denote by $\overline{H}$ the underlying undirected graph and by $H^*$ the transitive closure of $H$. Before turning to the proof of the lemma, we first show the following claims. 

\begin{claim}
\label{clm:ramseyA}
One of the following holds.
\begin{itemize}
\item[(1)] There exist subsets $X' \subseteq X$ and $W' \subseteq W$ such that $|X'| = |W'| = \lfloor \ell_5(t)/2 \rfloor$ and $D$ contains a $(W',X')$-induced-biclique.
\item[(2)] There exist a set $I \subseteq [t]$ of size $\ell_5(t)$ such that $D[\{x_i,w_i~|~i \in I\}]$ is transitively equivalent to one of $O_1^{\ell_5(t)},O_2^{\ell_5(t)},O_3^{\ell_5(t)}$ and $O_4^{\ell_5(t)}$.
\end{itemize}
\end{claim}

\begin{claimproof}
Consider the edge-coloring of the complete undirected graph $G$ on vertex set $\{1,\ldots,t\}$ defined as follows: for every $i,j \in [t]$ where $i < j$, the edge $ij \in E(G)$ receives 
\begin{itemize}
\item color 1 if $(w_i,w_j),(x_i,x_j),(w_i,x_j) \notin E(D^*)$;
\item color 2 if $(w_i,w_j) \in E(D)$ and $(x_i,x_j) \notin E(D^*)$;
\item color 3 if $(w_i,w_j) \notin E(D)$ and $(x_i,x_j) \in E(D^*)$;
\item color 4 if $(w_i,w_j),(x_i,x_j) \in E(D^*)$; and
\item color 5 if $(w_i,w_j),(x_i,x_j) \notin E(D^*)$ and $(w_i,x_j) \in E(D^*)$.
\end{itemize}
Note that for any $1 \leq i<j \leq t$, if $D^*$ contains one of $(w_i,w_j)$ and $(x_i,x_j)$ then $D^*$ also contains $(w_i,x_j)$ and thus, each edge of $E(G)$ receives a color in the above coloring.
Now by Ramsey's Theorem, $G$ contains a monochromatic clique of size $\ell_5(t)$: let $1 \leq i_1 < \ldots < i_{\ell_5(t)} \leq t$ be a set of $\ell_5(t)$ vertices inducing a monochromatic clique in $G$. Since $(A,B)$ is an ordered $t$-\toughpair, for every $1 \leq p < q \leq \ell_5(t)$, there is no $(w_{i_q},w_{i_p})$-path, no $(w_{i_q},x_{i_p})$-path and no $(x_{i_q},x_{i_p})$-path. It follows that if the clique in $G$ on vertex set $\{i_1,\ldots,i_{\ell_5(t)}\}$ has 
\begin{itemize}
\item color 1 then $D[\{w_{i_j},x_{i_j}~|~ j \in [\ell_5(t)]\}]$ is transitively equivalent to $O^{\ell_5(t)}_1$;
\item color 2 then $D[\{w_{i_j},x_{i_j}~|~ j \in [\ell_5(t)]\}]$ is transitively equivalent to $O^{\ell_5(t)}_2$;
\item color 3 then $D[\{w_{i_j},x_{i_j}~|~ j \in [\ell_5(t)]\}]$ is transitively equivalent to $O^{\ell_5(t)}_3$;
\item color 4 then $D[\{w_{i_j},x_{i_j}~|~ j \in [\ell_5(t)]\}]$ is transitively equivalent to $O^{\ell_5(t)}_4$;
\item color 5 then $D$ contains a $(\{w_{i_j}~|~j \in [\lfloor \ell_5(t)/2 \rfloor]\},\{x_{i_{\lfloor \ell_5(t)/2 \rfloor + j}}~|~j \in [\lfloor \ell_5(t)/2 \rfloor]\})$-induced-biclique.
\end{itemize}
Thus, if the clique in $G$ on vertex set $\{i_1,\ldots,i_{\ell_5(t)}\}$ does not have color 5 then we may take $I = \{i_1,\ldots,i_{\ell_5(t)}\}$ to prove our claim.  
\end{claimproof}

Symmetrically to \Cref{clm:ramseyA}, we have the following.

\begin{claim}
\label{clm:ramseyB}
One of the following holds.
\begin{itemize}
\item[(1)] There exist subsets $Y' \subseteq Y$ and $Z' \subseteq Z$ such that $|Y'| = |Z'| = \lfloor \ell_5(t)/2 \rfloor$ and $D$ contains a $(Y',Z')$-induced-biclique.
\item[(2)] There exist a set $I \subseteq [t]$ of size $\ell_5(t)$ such that $D[\{y_i,z_i~|~i \in I\}]$ is transitively equivalent to one of $O_1^{\ell_5(t)},O_2^{\ell_5(t)},O_3^{\ell_5(t)}$ and $O_4^{\ell_5(t)}$.
\end{itemize}
\end{claim}

\begin{claim}
\label{clm:ramseyWZ}
For any subsets $W' \subseteq W$ and $Z' \subseteq Z$ where $|W'| = |Z'| = p$, there exist subsets 
$W'' \subseteq W'$ and $Z'' \subseteq Z'$ such that $|W''| = |Z''| = \lfloor \ell_4(p) / 2 \rfloor$ and
\begin{itemize}
\item[$\bullet$] either there is no edge in $D$ from a vertex of $W''$ to a vertex of $Z''$,
\item[$\bullet$] or $D^*$ contains a $(W'',Z'')$-biclique (not necessarily induced).
\end{itemize}
\end{claim}

\begin{claimproof}
Let $W' \subseteq W$ and $Z' \subseteq Z$ be two subsets of size $p$. We distinguish two cases depending on whether the bipartite graph $\overline{D}[W',Z']$ has a matching of size at least $p/2$ or not.

\paragraph{Case 1. The bipartite graph $\overline{D}[W',Z']$ has a matching of size at least $p/2$.} Let $M$ be a matching in $\overline{D}[W',Z']$ of size $\lfloor p/2 \rfloor$ and let $e_1,\ldots,e_{\lfloor p/2 \rfloor}$ be an arbitrary ordering of the edges in $M$. Consider the edge-coloring of the complete undirected graph $G$ on vertex set $\{1,\ldots, \lfloor p/2 \rfloor\}$ defined as follows: for every $1 \leq i < j \leq \lfloor p/2 \rfloor$, the edge $ij \in E(G)$ receives
\begin{itemize}
\item color 1 if $(\tail(e_i),\head(e_j)),(\tail(e_j),\head(e_i)) \notin E(D^*)$;
\item color 2 if $(\tail(e_i),\head(e_j)) \in E(D^*)$ and $(\tail(e_j),\head(e_i)) \notin E(D^*)$;
\item color 3 if $(\tail(e_i),\head(e_j)) \notin E(D^*)$ and $(\tail(e_j),\head(e_i)) \in E(D^*)$; and
\item color 4 if $(\tail(e_i),\head(e_j)),(\tail(e_j),\head(e_i)) \in E(D^*)$;
\end{itemize}
Then by Ramsey's Theorem, $G$ contains a monochromatic clique of size $\ell_4(p)$: let $1 \leq i_1 < \ldots < i_{\ell_4(p)} \leq p/2$ be a set of $\ell_4(p)$ vertices inducing a monochromatic clique in $G$. Now if this clique has
\begin{itemize}
\item color 1 then by taking $W'' = \{\tail(e_{i_j})~|~j \in [\lfloor \ell_4(p)/2 \rfloor]\}$ and $Z'' = \{\head(e_{i_j})~|~j \in [\lfloor \ell_4(p)/2 \rfloor]\}$, there is no edge from $W''$ to $Y''$ in $D^*$ (and a fortiori in $D$);
\item color 2 then by taking $W'' = \{\tail(e_{i_j})~|~j \in [\lfloor \ell_4(p)/2 \rfloor]\}$ and $Z'' = \{\head(e_{i_{\lfloor \ell_4(p)/2 \rfloor + j}})~|~j \in [\lfloor \ell_4(p)/2 \rfloor]\}$, $D^*$ contains a $(W'',Z'')$-biclique;
\item color 3 then by taking $W'' = \{\tail(e_{i_{\lfloor \ell_4(p)/2 \rfloor + j}})~|~j \in [\lfloor \ell_4(p)/2 \rfloor]\}$ and $Z'' = \{\head(e_{i_j})~|~j \in [\lfloor \ell_4(p)/2 \rfloor]\}$, $D^*$ contains a $(W'',Z'')$-biclique;
\item color 4 then by taking $W'' = \{\tail(e_{i_j})~|~j \in [\lfloor \ell_4(p)/2 \rfloor]\}$ and $Z'' = \{\head(e_{i_j})~|~j \in [\lfloor \ell_4(p)/2 \rfloor]\}$, $D^*$ contains a $(W'',Z'')$-biclique.
\end{itemize}

\paragraph{Case 2. The bipartite graph $\overline{D}[W',Z']$ has no matching of size at least $p/2$.} Then by K\"onig's Theorem, $\overline{D}[W',Z']$ has, in this case, a vertex cover at size at most $p/2$: let $V \subseteq W' \cup Z'$ be a minimum vertex cover of $\overline{D}[W',Z']$. Since then, $\min \{ W' \setminus V,Z' \setminus V\} \geq \lfloor \ell_4(p)/ 2 \rfloor \geq \lfloor \ell_4(p)/2 \rfloor$, we may take any subsets $W'' \subseteq W' \setminus V$ and $Z'' \subseteq Z' \setminus V$ of size $\lfloor \ell_4(p)/2 \rfloor$ to prove the claim, as there are surely no edges from $W''$ to $Z''$ in $D$. 
\end{claimproof}

Symmetrically to \Cref{clm:ramseyWZ}, we have the following.

\begin{claim}
\label{clm:ramseyXY}
For any subsets $Y' \subseteq Y$ and $X' \subseteq X$ where $|Y'| = |X'| = p$, there exist subsets 
$Y'' \subseteq Y'$ and $X'' \subseteq X'$ such that $|Y''| = |X''| = \lfloor \ell_4(p)/2 \rfloor$ and
\begin{itemize}
\item[$\bullet$] either there is no edge in $D$ from a vertex of $Y''$ to a vertex of $X''$,
\item[$\bullet$] or $D^*$ contains a $(Y'',X'')$-biclique (not necessarily induced).
\end{itemize}
\end{claim}

We next define partitions with respect to a set and prove thereafter useful properties of these partitions.

\begin{definition}
For any $V \in \{W,X,Y,Z\}$ and any $U \subseteq V$, the \emph{partition of $V$ w.r.t. $U$} is the partition of $V \setminus U$ into four sets $V^\star, V^\circ, V^+$ and $V^-$ such that the following hold.
\begin{itemize}
\item For every $v \in V^\star$, $N^+_{D^*}(v) \cap U \neq \emptyset$ and $N^-_{D^*}(v) \cap U \neq \emptyset$.
\item For every $v \in V^\circ$, $(N^+_{D^*}(v) \cup N^-_{D^*}(v)) \cap U = \emptyset$.
\item For every $v \in V^+$, $N^+_{D^*}(v) \cap U \neq \emptyset$ and $N^-_{D^*}(v) \cap U = \emptyset$.
\item For every $v \in V^-$, $N^+_{D^*}(v) \cap U  = \emptyset$ and $N^-_{D^*}(v) \cap U \neq \emptyset$.
\end{itemize} 
\end{definition}

Note that since $(A,B)$ is an ordered \toughpair, for any $V \in \{W,X,Y,Z\}$ and any $v \in V$, 
\begin{itemize}
\item if $u \in N^-_{D^*}(v)$ then $u$ has a smaller index than that of $v$, and
\item if $u \in N^+_{D^*}(v)$ then $u$ has a larger index than that of $v$. 
\end{itemize}

\begin{claim}
\label{clm:independent}
For any $V \in \{W,X,Y,Z\}$ and any $U \subseteq V$, the following hold.
Let $V^\star, V^\circ, V^+$ and $V^-$ be the partition of $V$ w.r.t $U$. Then there is no edge in $D^*$
\begin{itemize}
\item[(i)] from a vertex of $V^\star$ to a vertex of $V^\circ \cup V^+$, 
\item[(ii)] from a vertex of $V^-$ to a vertex of $V^\star \cup V^\circ \cup V^+$, and 
\item[(iii)] from a vertex of $V^\circ$ to a vertex of $V^+ \cup V^\star$.
\end{itemize}
Furthermore, if $D^*[U]$ has no edge then $V^\star = \emptyset$.
\end{claim}

\begin{claimproof}
Consider $u \in V^\star$. If there exists $v \in V^\circ \cup V^+$ such that $(u,v) \in E(D^*)$ then in particular $\emptyset \neq N^-_{D^*}(u) \cap U \subseteq N^-_{D^*}(v)$, a contradiction to the fact that $v \in V^+ \cup V^\circ$.
Similarly, there is no edge in $D^*$ from a vertex of $V^-$ to a vertex of $V^+ \cup V^\star$: indeed, if there exist $u \in V^-$ and $v\in V^+ \cup V^\star$ such that $(u,v) \in E(D^*)$ then in particular $\emptyset \neq N^-_{D^*}(u) \cap U \subseteq N^-_{D^*}(v)$, a contradiction to the fact that $v \in V^+$. Furthermore, there is no edge from a vertex of $V^-$ to a vertex of $V^\circ$ for if there exist $u \in V^-$ and $v \in V^\circ$ such that $(u,v) \in E(D^*)$, then in particular $\emptyset \neq N^-_{D^*}(u) \cap U \subseteq N^-_{D^*}(v)$, a contradiction to the fact that $v \in V^\circ$. Finally, if there exist $u \in V^\circ$ and $v \in V^+ \cup V^\star$ such that $(u,v) \in E(D^*)$ then in particular $\emptyset \neq N^+_{D^*}(v) \cap U \subseteq N^+_{D^*}(u)$, a contradiction to the fact that $u \in V^\circ$.
Now assume that $D^*[U]$ has no edge and suppose for a contradiction that there exists $u \in V \setminus U$ such that both $N^-_{D^*}(u) \cap U \neq \emptyset$ and $N^+_{D^*}(u) \cap U \neq \emptyset$. Then for any $x \in N^-_{D^*}(u) \cap U$ and $y \in N^+_{D^*}(u) \cap U$, $(x,y) \in E(D^*)$, a contradiction to the fact that $U$ is an independent set of $D^*$. 
\end{claimproof}

\begin{claim}
\label{clm:W+}
For any $V \in \{W,Y\}$ and any $U \subseteq V$, the following hold:
\begin{itemize}
\item if there exist $u \in W \cup X \cup Y \cup Z$ and $v \in V^+$ such that $(u,v) \in E(D^*)$ then $u \in V^+$;
\item if there exist $u \in W \cup X \cup Y \cup Z$ and $v \in V^\circ$ such that $(u,v) \in E(D^*)$ then $u \in V^\circ \cup V^+$,
\end{itemize}
where $V^\star, V^\circ, V^+,V^-$ is the partition of $V$ w.r.t $U$.
\end{claim}

\begin{claimproof}
Suppose that there exist $u \in W \cup X \cup Y \cup Z$ and $v \in V^+$ such that $(u,v) \in E(D^*)$. Observe first that since $V \in \{W,Y\}$ and $(A,B)$ is an ordered \toughpair, necessarily $u \in V$. Now since $v \in V^+$, $u \notin U$ by definition and since by \Cref{clm:independent}, $u \notin V^\star \cup V^\circ \cup V^-$, we conclude that $u \in V^+$.
Suppose next that there exist $u \in W \cup X \cup Y \cup Z$ and $v \in V^\circ$ such that $(u,v) \in E(D^*)$. Then as previously $u \in V$. Now since $v \in V^\circ$, $u \notin U$ by definition and since by \Cref{clm:independent}, $u \notin V^\star \cup V^-$, we conclude that $u \in V^\circ \cup V^+$.
\end{claimproof}

\begin{claim}
\label{clm:X-}
For any $V \in \{X,Z\}$ and any $U \subseteq V$, the following hold:
\begin{itemize}
\item if there exist $u \in V^-$ and $v \in W \cup X \cup Y \cup Z$ such that $(u,v) \in E(D^*)$ then $v \in V^-$;
\item if there exist $v \in V^\circ$ and $v \in W \cup X \cup Y \cup Z$ such that $(u,v) \in E(D^*)$ then $u \in V^\circ \cup V^-$,
\end{itemize}
where $V^\star, V^\circ, V^+,V^-$ is the partition of $V$ w.r.t $U$. 
\end{claim}

\begin{claimproof}
Suppose that there exist $u \in V^-$ and $v \in W \cup X \cup Y \cup Z$ such that $(u,v) \in E(D^*)$. Observe first that since $V \in\{X,Z\}$ and $(A,B)$ is an ordered \toughpair, necessarily $v \in V$. Now since $u \in V^-$, $v \notin U$ by definition and since by \Cref{clm:independent}, $v \notin V^\star \cup V^\circ \cup V^+$, we conclude that $v \in V^-$.
Suppose next that there exist $v \in V^\circ$ and $v \in W \cup X \cup Y \cup Z$ such that $(u,v) \in E(D^*)$ then $u \in V^\circ \cup V^-$. Then as previously $v \in V$. Now since $u \in V^\circ$,  $v \notin U$ by definition and since by \Cref{clm:independent}, $v \notin V^+ \cup V^\star$, we conclude that $v \in V^\circ \cup V^-$.
\end{claimproof}

\newcommand{\sizeWZ}{\ell_4(\ell_5(t))/2}
\newcommand{\sizeXY}{4g(t)+2}
\newcommand{\sizeXYo}{4g(t)+1}
\newcommand{\sizeXYt}{4g(t)}

We now turn to the proof of the lemma. We assume henceforth that $D$ contains no $g(t)$-induced-biclique (we are done otherwise). We first prove that $D$ can be identified to a digraph which is transitively equivalent to an almost $\sizeXYt$-hard matching pattern and then show how to obtain a $g(t)$-hard matching pattern from an almost $4g(t)$-hard matching pattern. 

By \Cref{clm:ramseyA}, there exists a set $I_A \subseteq [t]$ of size $\ell_5(t)$ such that $D[\{w_i,x_i~|~i \in I_A\}]$ is transitively equivalent to one of $O_1^{\ell_5(t)},O_2^{\ell_5(t)},O_3^{\ell_5(t)}$ and $O_4^{\ell_5(t)}$; and similarly, by \Cref{clm:ramseyB}, there exists a set $I_B \subseteq [t]$ of size $\ell_5(t)$ such that $D[\{y_i,z_i~|~i\in I_B\}]$ is transitively equivalent to one of $O_1^{\ell_5(t)},O_2^{\ell_5(t)},O_3^{\ell_5(t)}$ and $O_4^{\ell_5(t)}$. 
Now by \Cref{clm:ramseyWZ} there exist subsets $I_A^1 \subseteq I_A$ and $I_B^1 \subseteq I_B$ such that $|I_A^1| = |I_B^1| = \sizeWZ$ and there is either no edge or every edge in $D$ from $\{w_i~|~i \in I_A^1\}$ to $\{z_i~|~i \in I_B^1\}$. 
By \Cref{clm:ramseyXY}, there further exist subsets $I_A^2 \subseteq I_A^1$ and $I_B^2 \subseteq I_B^1$ such that $|I_A^2| = |I_B^2| = \ell_4(\ell_4(\ell_5(t))/2))/2 = \sizeXY$ and there is either no edge or every edge in $D$ from $\{y_i~|~i \in I_A^2\}$ to $\{x_i~|~i \in I_B^2\}$.
 
Now let $W^\star,W^\circ,W^+,W^-$ be the partition of $W$ w.r.t. $\{w_i~|~i \in I_A^2\}$,
$X^\star,X^\circ,X^+,X^-$ be the partition of $X$ w.r.t $\{x_i~|~i \in I_A^2\}$,
$Y^\star,Y^\circ,Y^+,Y^-$ be the partition of $Y$ w.r.t $\{y_i~|~i \in I_B^2\}$ 
and $Z^\star,Z^\circ,Z^+,Z^-$ be the partition of $Z$ w.r.t $\{z_i~|~i \in I_B^2\}$.
We devise the following identification rules.

\begin{identification rule}
\label{ir:source}
Proceed as follows.
\begin{enumerate}[nosep]
\item Identify every vertex in $W^\circ \cup W^+ \cup Y^\circ \cup Y^+$ (together with the source vertex $\sss$ of $D$ when it exists) to a single vertex $s$.
\item If $D[\{w_i,x_i~|~i \in I_A\}]$ is transitively equivalent to $O_1^{\ell_5(t)}$, further identity every vertex in $X^+$ to $s$.
\item If $D[\{y_i,z_i~|~i \in I_B\}]$ is transitively equivalent to $O_1^{\ell_5(t)}$, further identity every vertex in $Z^+$ to $s$.
\end{enumerate}
\end{identification rule}

\begin{identification rule}
\label{ir:sink}
Proceed as follows.
\begin{enumerate}[nosep]
\item Identify every vertex in $X^\circ \cup X^- \cup Z^\circ \cup Z^-$ (together with the sink vertex $\ttt$ of $D$ when it exists) to a single vertex $t$.
\item If $D[\{w_i,x_i~|~i \in I_A\}]$ is transitively equivalent to $O_1^{\ell_5(t)}$, further identity every vertex in $W^-$ to $t$.
\item If $D[\{y_i,z_i~|~i \in I_B\}]$ is transitively equivalent to $O_1^{\ell_5(t)}$, further identity every vertex in $Y^-$ to $t$.
\end{enumerate}
\end{identification rule}

\begin{identification rule}
\label{ir:WX}
For every consecutive $i < i' \in I_A^2 \cup \{1,t\}$ $($that is, $(I_A^2 \cup \{1,t\}) \cap [i,i'] = \{i,i'\})$ and every $j \in[i,i']$, proceed as follows.
\begin{enumerate}[nosep]
\item $D[\{w_i,x_i~|~i \in I_A\}]$ is transitively equivalent to $O_2^{\ell_5(t)}$: 
if $w_j \in W^\star \cup W^-$ then identify $w_j$ to $w_i$ and if $x_j \in X^+$ then identify $x_j$ to $w_{i'}$.
\item $D[\{w_i,x_i~|~i \in I_A\}]$ is transitively equivalent to $O_3^{\ell_5(t)}$:
if $w_j \in W^-$ then identify $w_j$ to $x_i$ and if $x_j \in X^\star \cup X^+$ then identify $x_j$ to $x_{i'}$.
\item $D[\{w_i,x_i~|~i \in I_A\}]$ is transitively equivalent to $O_4^{\ell_5(t)}$: 
if $w_j \in W^\star \cup W^-$ then identify $w_j$ to $w_i$ and if $x_j \in X^\star \cup X^+$ then identify $x_j$ to $x_{i'}$.
\end{enumerate}
\end{identification rule}

\begin{identification rule}
\label{ir:YZ}
For every consecutive $i < i' \in I_B^2 \cup \{1,t\}$ and every $j \in[i,i']$, proceed as follows.
\begin{enumerate}[nosep]
\item $D[\{y_i,z_i~|~i \in I_B\}]$ is transitively equivalent to $O_2^{\ell_5(t)}$: 
if $y_j \in Y^\star \cup Y^-$ then identify $y_j$ to $y_i$ and if $z_j \in Z^+$ then identify $z_j$ to $y_{i'}$.
\item $D[\{y_i,x_i~|~i \in I_B\}]$ is transitively equivalent to $O_3^{\ell_5(t)}$:
if $y_j \in Y^-$ then identify $y_j$ to $z_i$ and if $z_j \in Z^\star \cup Z^+$ then identify $z_j$ to $z_{i'}$.
\item $D[\{y_i,z_i~|~i \in I_B\}]$ is transitively equivalent to $O_4^{\ell_5(t)}$: 
if $y_j \in Y^\star \cup Y^-$ then identify $y_j$ to $y_i$ and if $z_j \in Z^\star \cup Z^+$ then identify $z_j$ to $z_{i'}$.
\end{enumerate}
\end{identification rule}

Let $D_I$ be the digraph resulting from an exhaustive application of Identification Rules~\ref{ir:source}-\ref{ir:YZ}. 
We aim to show that $D_I$ can be identified to a digraph which is transitively equivalent to an almost $\sizeXYt$-hard matching pattern. To this end, we first prove the following claims and then
distinguish cases depending on the graphs to which $D[\{w_i,x_i~|~i \in I_A\}]$ and $D[\{y_i,z_i~|~i \in I_B\}]$ are transitively equivalent, and whether there is no edge or every edge in $D$ from $\{w_i~|~i \in I_A^2\}$ to $\{z_i~|~i \in I_B^2\}$ (from $\{y_i~|~i \in I_A^2\}$ to $\{x_i~|~i \in I_B^2\}$, respectively).

\begin{claim}
\label{clm:vertexsetDI}
The vertex set of $D_I$ consists of $\{w_i,x_i~|~i \in I_A^2\} \cup \{y_i,z_i~|~i \in I_B^2\}$ together with possibly $s$ or $t$. 
\end{claim}

\begin{claimproof}
We first prove that $V(D_I) \cap W = \{w_i~|~i \in I_A^2\}$. 
Clearly, after \Cref{ir:source}.1 has been applied, there are no more vertices from $W^\circ \cup W^+$ left.
Now if $D[\{w_i,x_i~|~i \in I_A\}]$ is transitively equivalent to $O_1^{\ell_5(t)}$
then $W^\star = \emptyset$ by \Cref{clm:independent} and once \Cref{ir:sink}.2 has been applied, there are no more vertices from $W^-$ left.
Similarly, if $D[\{w_i,x_i~|~i \in I_A\}]$ is transitively equivalent to $O_3^{\ell_5(t)}$
then $W^\star = \emptyset$ by \Cref{clm:independent} and once \Cref{ir:WX}.2 has been exhaustively applied, 
there are no more vertices from $W^-$ left. 
Otherwise, $D[\{w_i,x_i~|~i \in I_A\}]$ is transitively equivalent to either $O_2^{\ell_5(t)}$ or $O_4^{\ell_5(t)}$ in which case, an exhaustive application of \Cref{ir:WX}.1 and \Cref{ir:WX}.3, respectively, leaves no vertex from $W^\star \cup W^-$.
We conclude symmetrically that $V(D_I) \cap Y = \{y_i~|~i \in I_B^2\}$.

Second, let us show that $V(D_I) \cap X = \{x_i~|~i \in I_A^2\}$.
Clearly, after \Cref{ir:sink}.1 has been applied, there are no more vertices from $X^\circ \cup X^-$ left. 
Now if $D[\{w_i,x_i~|~i \in I_A\}]$ is transitively equivalent to $O_1^{\ell_5(t)}$
then $X^\star = \emptyset$ by \Cref{clm:independent} and once \Cref{ir:source}.2 has been applied, there are no more vertices from $X^+$ left.
Similarly, if $D[\{w_i,x_i~|~i \in I_A\}]$ is transitively equivalent to $O_2^{\ell_5(t)}$
then $X^\star = \emptyset$ by \Cref{clm:independent} and once \Cref{ir:WX}.1 has been exhaustively applied, 
there are no more vertices from $X^+$ left. 
Otherwise, $D[\{w_i,x_i~|~i \in I_A\}]$ is transitively equivalent to either $O_3^{\ell_5(t)}$ or $O_4^{\ell_5(t)}$ in which case, an exhaustive application of \Cref{ir:WX}.2 and \Cref{ir:WX}.3, respectively, leaves no vertex from $X^\star \cup X^+$.
We conclude symmetrically that $V(D_I) \cap Z = \{z_i~|~i \in I_B^2\}$.
\end{claimproof}

\begin{claim}
\label{clm:sourcesink}
If $s$ exists then $d^-_{D_I}(s) = 0$ and if $t$ exists then $d^+_{D_I}(t) = 0$.
\end{claim}

\begin{claimproof}
Assume that $s$ exists and suppose to the contrary that there exists $u \in V(D_I)$ such that $u \in N^-_{D_I}(s)$.
Then there exist $v_s,v_u \in V(D)$ such that $v_s$ has been identified to $s$, $v_u$ has been identified to $u$ and $(v_u,v_s) \in E(D)$. 
Since the source vertex $\sss$ of $D$ (if any) has in-degree zero in $D$, necessarily $v_s \neq \sss$.
Furthermore, $v_s$ cannot belong to $W^\circ \cup W^+ \cup Y^\circ \cup Y^+$: indeed, if, say, $v_s \in W^\circ \cup W^+$ (the case where $v_s \in Y^\circ \cup Y^+$ is symmetric) then by \Cref{clm:W+}, $v_u \in W^\circ \cup W^+$, a contradiction as $v_u$ is not identified to $s$.
Thus, it must be that $D[\{w_i,x_i~|~i \in I_A\}]$ is transitively equivalent to $O_1^{\ell_5(t)}$ and $v_s \in X^+$,
or $D[\{y_i,x_i~|~i \in I_B\}]$ is transitively equivalent to $O_1^{\ell_5(t)}$ and $v_s \in Z^+$.
Assume without loss of generality that the latter holds (the other case is symmetric).
Then since $(A,B)$ is an ordered \toughpair, necessarily $v_u \in Y \cup Z$.
However, since $v_u$ is not identified to $s$, $v_u \notin Y^\circ \cup Y^+$ 
and since $D*[\{y_i|~i \in I_B\}]$ has no edge, $Y^\star = \emptyset$ by \Cref{clm:independent}.
Furthermore, $v_u$ cannot belong to $\{y_i~|~i \in I_A^2\} \cup Y^-$ for otherwise, since $v_s \in Z^+$ and $(A,B)$ is an ordered \toughpair, there exist $i,j \in I_A^2$ such that $i < j$ and $(y_i,z_j) \in E(D^*)$, a contradiction to the fact that $D[\{y_i,x_i~|~i \in I_B\}]$ is transitively equivalent to $O_1^{\ell_5(t)}$.
Thus, it must be that $v_u \in Z$.
However, since $v_s \in Z^+$, $v_u \notin Z^\star \cup Z^- \cup Z^\circ$ by \Cref{clm:independent}
and $v_u \notin \{z_i~|~i \in I_B^2\}$ by definition of $Z^+$. 
Thus, we conclude that $v_u \in Z^+$, a contradiction since $v_u$ is not identified to~$s$.

Assume now that $t$ exists and suppose to the contrary that there exists $u \in V(D_I)$ such that $u \in N^+_{D_I}(t)$.
Then there exist $v_t,v_u \in V(D)$ such that $v_t$ is identified to $t$, $v_u$ is identified to $u$ and $(v_t,v_u) \in E(D)$.
Since the sink vertex $\ttt$ of $D$ (if any) has out0degree zero in $D$, necessarily $v_t \neq \ttt$.
Furthermore, $v_t$ cannot belong to $X^\circ \cup X^- \cup Z^\circ \cup Z^-$: indeed, if, say, $v_t \in X^\circ \cup X^-$ (the case where $v_t \in Z^\circ \cup Z^-$ is symmetric) then by \Cref{clm:X-}, $v_u \in X^\circ \cup X^-$,
a contradiction since $v_u$ is not identified to $t$.
Thus, it must be that $D[\{w_i,x_i~|~i \in I_A\}]$ is transitively equivalent to $O_1^{\ell_5(t)}$ and $v_t \in W^-$, or
$D[\{y_i,z_i~|~i \in I_B\}]$ is transitively equivalent to $O_1^{\ell_5(t)}$ and $v_t \in Y^-$.
Assume without loss of generality that the latter holds (the other case is symmetric).
Then since $(A,B)$ is an ordered \toughpair, necessarily $v_u \in Y \cup Z$.
However, since $v_u$ is not identified to $t$, $v_u \notin Z^\circ \cup Z^-$
and since $D*[\{z_i|~i \in I_B\}]$ has no edge, $Z^\star = \emptyset$ by \Cref{clm:independent}.
Furthermore, $v_u$ cannot belong to $\{z_i~|~i \in I_B^2\} \cup Z^+$ for otherwise, since $v_t \in Y^-$ and $(A,B)$ is an ordered \toughpair, there exist $i,j \in I_B^2$ such that $i < j$ and $(y_i,z_j) \in E(D^*)$, a contradiction to the fact that $D[\{y_i,x_i~|~i \in I_B\}]$ is transitively equivalent to $O_1^{\ell_5(t)}$.
Thus, it must be that $v_u \in Y$.
However, since $v_t \in Y^-$, $v_u \notin Y^\star \cup Y^\circ \cup Y^+$ by \Cref{clm:independent}
and $v_u \notin \{y_i~|~i \in I_B^2\}$ by definition of $Y^-$.
Thus, we conclude that $v_u \in Y^-$, a contradiction since $v_u$ is not identified to $t$. 
\end{claimproof}

\begin{claim}
\label{clm:sameinDIWX}
For any $i \in [4]$, if $D[\{w_j,x_j~|~j\in I_A\}]$ is transitively equivalent to $O_i^{\ell_5(t)}$ then
$D_I[\{w_j,x_j~|~j \in I_A^2\}]$ is transitively equivalent to $O_i^{4g(t)+2}$.
\end{claim}

\begin{claimproof}
Let us first show that if $D[\{w_i~|~i \in I_A\}]$ has no edge then $D_I[\{w_i~|~i \in I_A^2\}]$ has no edge as well. 
Assume that $D[\{w_i~|~i \in I_A\}]$ has no edge, that is, $D[\{w_i,x_i~|~i \in I_A\}]$ is transitively equivalent to $O_1^{\ell_5(t)}$ or $O_3^{\ell_5(t)}$.
Then by \Cref{clm:independent}, $W^\star = \emptyset$.
Suppose first that $D[\{w_i,x_i~|~i \in I_A\}]$ is transitively equivalent to $O_1^{\ell_5(t)}$. 
Then every vertex in $W^\circ \cup W^+$ is identified to $s$ and every vertex in $W^-$ is identified to $t$ 
where $d^-_{D_I}(s) = d^+_{D_I}(t) = 0$ by \Cref{clm:sourcesink}.
Since, in this case, every vertex in $X$ is identified either to $s$, $t$ or itself 
and no other identification may create an edge between two vertices of $\{w_i~|~i \in I_A^2\}$,
we conclude that $D_I[\{w_i~|~i \in I_A^2\}]$ has no edge.
Suppose second that $D[\{w_i,x_i~|~i \in I_A\}]$ is transitively equivalent to $O_3^{\ell_5(t)}$.
Then every vertex in $W^\circ \cup W^+$ is identified to $s$ and every vertex in $X^\circ \cup X^-$ is identified to $t$
where $d^-_{D_I}(s) = d^+_{D_I}(t) = 0$ by \Cref{clm:sourcesink}.
Since every vertex in $W^- \cup X^\star \cup X^+$ is identified to a vertex in $\{x_i~|~i \in I_A^2\}$ 
and no other identification may create an edge between two vertices of $\{w_i~|~i \in I_A^2\}$,
we conclude that $D_I[\{w_i~|~i \in I_A^2\}]$ has no edge.

Second, let us show that if $D[\{w_i~|~i \in I_A\}]$ is transitively equivalent to the directed path 
$w_{\min I_A} \rightarrow w_{\min I_A +1} \rightarrow \ldots \rightarrow w_{\max I_A}$ then
$D_I[\{w_i~|~i \in I_A^2\}]$ is transitively equivalent to the directed path 
$w_{\min I_A^2} \rightarrow w_{\min I_A^2 +1} \rightarrow \ldots \rightarrow w_{\max I_A^2}$.
Assume that $D[\{w_i~|~i \in I_A\}]$ is transitively equivalent to this directed path, 
that is, $D[\{w_i,x_i~|~i \in I_A\}]$ is transitively equivalent to $O_2^{\ell_5(t)}$ or $O_4^{\ell_5(t)}$.
Note that since $I_A^2 \subseteq I_A$, it is enough to show that no edge from a vertex in $\{w_i~|~i \in I_A^2\}$ to a vertex in $\{w_i~|~i \in I_A^2\}$ of smaller index is created (we refer to such an edge as a \emph{bad edge} in the following).
First, since every vertex $w_j \in W^\star \cup W^-$ is identified to the vertex $w_{j'} \in \{w_i~|~i \in I_A^2\}$ such that $j' = \max \{i \in I_A^2~:~ i \leq j\}$, no bad edge is created here.
Now if $D[\{w_i,x_i~|~i \in I_A\}]$ is transitively equivalent to $O_2^{\ell_5(t)}$ (note that, in this case, $X^\star = \emptyset$ by \Cref{clm:independent}), then every vertex $x_j \in X^+$ is identified to the vertex $w_{j'} \in \{w_i~|~i \in I_A^2\}$ such that $j' = \min \{i \in I_A^2~:~j \leq i\}$ and so, no bad edge is created here; and if $D[\{w_i,x_i~|~i \in I_A\}]$ is transitively equivalent to $O_4^{\ell_5(t)}$, then every vertex in $X^\star \cup X^+$ is identified to a vertex in $\{x_i~|~i \in I_A^2\}$ and thus, no bad edge is created here as well. 
Since every vertex in $W^\circ \cup W+$ is identified to $s$ and every vertex in $X^\circ \cup X^-$ is identified to $t$ where $d^-_{D_I}(s) = d^+_{D_I}(t) = 0$ by \Cref{clm:sourcesink}, no bad edge is created with these identifications;
and since no other identification may create a bad edge, we conclude that $D_I[\{w_i~|~i \in I_A^2\}]$ is indeed transitively equivalent to the directed path $w_{\min I_A^2} \rightarrow \ldots \rightarrow w_{\max I_A^2}$.

Consider next $X$. As in the case of $W$, let us first show that if $D[\{x_i~|~i \in I_A\}]$ has no edge then $D_I[\{x_i~|~i \in I_A^2\}]$ has no edge as well. 
Assume that $D[\{x_i~|~i \in I_A\}]$ has no edge, that is, $D[\{w_i,x_i~|~i \in I_A\}]$ is transitively equivalent to $O_1^{\ell_5(t)}$ or $O_2^{\ell_5(t)}$.
Then by \Cref{clm:independent}, $X^\star = \emptyset$.
Now if $D[\{w_i,x_i~|~i \in I_A\}]$ is transitively equivalent to $O_1^{\ell_5(t)}$,
we conclude, symmetrically to the case of $W$, that $D_I[\{x_i~|~i \in I_A^2\}]$ has no edge.
Suppose therefore that $D[\{w_i,x_i~|~i \in I_A\}]$ is transitively equivalent to $O_2^{\ell_5(t)}$.
Since then, every vertex in $X^\circ \cup X^-$ is identified to $t$, every vertex in $W^\circ \cup W^+$ is identified to $s$  where $d^-_{D_I}(s) = d^+_{D_I}(t) = 0$ by \Cref{clm:sourcesink},
and every vertex in $X^+ \cup W^\star \cup W^-$ is identified to a vertex in $\{w_i~|~i \in I_A^2\}$, 
no edge between two vertices in $\{x_i~|~i \in I_A^2\}$ is created with these identifications.
Since no other identification may create an edge between two vertices in $\{x_i~|~i \in I_A^2\}$, 
we conclude that $D_I[\{x_i~|~i \in I_A^2\}]$ has no edge in this case as well.

Let us next show that if $D[\{x_i~|~i \in I_A\}]$ is transitively equivalent to the directed path 
$x_{\min I_A} \rightarrow x_{\min I_A +1} \rightarrow \ldots \rightarrow x_{\max I_A}$ then
$D_I[\{x_i~|~i \in I_A^2\}]$ is transitively equivalent to the directed path 
$x_{\min I_A^2} \rightarrow x_{\min I_A^2 +1} \rightarrow \ldots \rightarrow x_{\max I_A^2}$.
Assume that $D[\{x_i~|~i \in I_A\}]$ is transitively equivalent to this directed path, that is,
$D[\{w_i,x_i~|~i \in I_A\}]$ is transitively equivalent to $O_3^{\ell_5(t)}$ or $O_4^{\ell_5(t)}$.
Note that since $I_A^2 \subseteq I_A$, it is enough to show that no edge from a vertex in $\{x_i~|~i \in I_A^2\}$ to a vertex in $\{x_i~|~i \in I_A^2\}$ of smaller index is created (we refer to such an edge as a \emph{bad edge} in the following).
First, since every vertex $x_j \in X^\star \cup X^+$ is identified to the vertex $x_{j'} \in \{w_i~|~i \in I_A^2\}$ such that $j' = \min \{i \in I_A^2~:~ j \leq i\}$, no bad edge is created here. 
Now if $D[\{w_i,x_i~|~i \in I_A\}]$ is transitively equivalent to $O_3^{\ell_5(t)}$ (note that, in this case, $W^\star = \emptyset$ by \Cref{clm:independent}), then every vertex $w_j \in W^-$ is identified to the vertex $x_{j'} \in \{x_i~|~i \in I_A^2\}$ such that $j' = \max \{i \in I_A^2~:~i \leq j\}$ and so, no bad edge is created here; and if $D[\{w_i,x_i~|~i \in I_A\}]$ is transitively equivalent to $O_4^{\ell_5(t)}$, then every vertex in $W^\star \cup W^-$ is identified to a vertex in $\{w_i~|~i \in I_A^2\}$ and thus, no bad edge is created here as well. 
Since every vertex in $X^\circ \cup X^-$ is identified to $t$ and every vertex in $W^\circ \cup W+$ is identified to $s$ where $d^-_{D_I}(s) = d^+_{D_I}(t) = 0$ by \Cref{clm:sourcesink}, no bad edge is created with these identifications;
and since no other identification may create a bad edge, we conclude that $D_I[\{x_i~|~i \in I_A^2\}]$ is indeed transitively equivalent to the directed path $x_{\min I_A^2} \rightarrow \ldots \rightarrow x_{\max I_A^2}$.

By combining the above four properties, we conclude that if $D[\{w_j,x_j~|~j\in I_A\}]$ is transitively equivalent to $O_i^{\ell_5(t)}$ for some $i \in [4]$, then
$D_I[\{w_j,x_j~|~j \in I_A^2\}]$ is transitively equivalent to $O_i^{|I_A^2|}$; and since $|I_A^2| = 4g(t) + 2$, the claims follows.
\end{claimproof}

Symmetrically to \Cref{clm:sameinDIWX}, we have the following.

\begin{claim}
\label{clm:sameinDIYZ}
For any $i \in [4]$, if $D[\{y_j,z_j~|~j\in I_B\}]$ is transitively equivalent to $O_i^{\ell_5(t)}$ then
$D_I[\{y_j,z_j~|~j \in I_B^2\}]$ is transitively equivalent to $O_i^{4g(t)+2}$.
\end{claim}

Now note that since every vertex in $A$ ($B$, respectively) is identified either to $s$, $t$ or a vertex in $A$ ($B$, respectively) and $d^-_{D_I}(s) = d^+_{D_I}(t) = 0$ by \Cref{clm:sourcesink}, 
no edge between $A$ and $B$ is created by these identifications that didn't already exist.
This implies in particular that if there is no edge in $D$ from $\{w_i~|~i \in I_A^2\}$ to $\{z_i~|~i \in I_B^2\}$ 
and from $\{y_i~|~i\in I_B^2\}$ to $\{x_i~|~i \in I_A^2\}$ then by Claims~\ref{clm:sameinDIWX} and \ref{clm:sameinDIYZ}, $D_I - \{s,t\}$ is transitively equivalent to the disjoint union of $O_i^{4g(t)+2}$ and $O_j^{4g(t)+2}$, 
where $O_i^{4g(t)+2}$ ($O_j^{4g(t)+2}$, respectively) is the digraph to which $D[\{w_a,x_a~|~a \in I_A\}]$ ($D[\{y_a,z_a~|~a \in I_B^2\}]$, respectively) is transitively equivalent.
Thus, since we may safely identify, e.g., the two first edges of $D[\{w_i,x_i~|~i \in I_A^2\}]$ 
(that is, $(w_{\min I_A^2},x_{\min I_A^2})$ and $(w_{\min I_A^2+1},x_{\min I_A^2+1})$) to the third edge of
$D[\{w_i,x_i~|~i \in I_A^2\}]$ and proceed similarly with $D[\{y_i,z_i~|~i \in I_B^2\}]$,
we conclude that, in this case, $D_I$ can be identified to an almost $4g(t)$-hard matching pattern.
Note that, more generally, if there is no edge in $D$ from $\{w_i~|~i \in I_A^2\}$ to $\{z_i~|~i \in I_B^2\}$
then $D_I[\{w_i~|~i \in I_A^2\} \cup \{z_i~|~i \in I_B^2\}]$ is transitively equivalent to the disjoint union of $D_I[\{w_i~|~i \in I_A^2\}]$ and $D_I[\{z_i~|~i \in I_B^2\}]$ (in particular, no vertex $r_{WZ}$ may be created by identifications); 
and the same holds for $\{y_i~|~i\in I_B^2\}$ and $\{x_i~|~i \in I_A^2\}$.

Assume henceforth that there is every edge either from $\{w_i~|~i \in I_A^2\}$ to $\{z_i~|~i \in I_B^2\}$,
from $\{y_i~|~i\in I_B^2\}$ to $\{x_i~|~i \in I_A^2\}$, or both.
Note that since we assume $D$ not to contain a $g(t)$-induced-biclique,
if there is every edge from $\{w_i~|~i \in I_A^2\}$ to $\{z_i~|~i \in I_B^2\}$
then at least one of $D_I[\{w_i~|~i \in I_A^2\}]$ and $D_I[\{z_i~|~i \in I_B^2\}]$ has edges 
(in fact, is transitively equivalent to a directed path by Claims~\ref{clm:sameinDIWX} and \ref{clm:sameinDIYZ});
and the same holds for $\{y_i~|~i\in I_B^2\}$ and $\{x_i~|~i \in I_A^2\}$.
Now if there is every edge in $D$ from $\{w_i~|~i \in I_A^2\}$ to $\{z_i~|~i \in I_B^2\}$
and both $D_I[\{w_i~|~i \in I_A^2\}]$ and $D_I[\{z_i~|~i \in I_B^2\}]$ are transitively equivalent to a directed path,
then $D_I\{w_i~|~i \in I_A^2\} \cup \{z_i~|~i \in I_B^2\}]$ is transitively equivalent to the directed path
$w_{\min I_A^2} \rightarrow \ldots \rightarrow w_{\max I_A^2} \rightarrow z_{\min I_B^2} \rightarrow \ldots \rightarrow z_{\max I_B^2}$ in which case we need not introduce the vertex $r_{WZ}$;
and we conclude symmetrically with $\{y_i~|~i\in I_B^2\}$ and $\{x_i~|~i \in I_A^2\}$.
 
By the above discussion, there remains to consider the case where every edge in $D$ from $\{w_i~|~i \in I_A^2\}$ to $\{z_i~|~i \in I_B^2\}$ (or from $\{y_i~|~i\in I_B^2\}$ to $\{x_i~|~i \in I_A^2\}$) exist and exactly one of $D_I[\{w_i~|~i \in I_A^2\}]$ and $D_I[\{z_i~|~i \in I_B^2\}]$ ($D_I[\{y_i~|~i\in I_B^2\}]$ and $D_I[\{x_i~|~i \in I_A^2\}]$, respectively) has no edge. In this case, we devise the following identification rules (for simplicity, we still call $D_I$ the graph resulting from these identifications).

\begin{identification rule}
\label{rr:WZ}
If there is every edge from $\{w_i~|~i \in I_A^2\}$ to $\{z_i~|~i \in I_B^2\}$ in $D_I$ 
and exactly one of $D_I^*[\{w_i~|~i \in I_A^2\}]$ and $D_I^*[\{z_i~|~i \in I_B^2\}]$ has no edge, then proceed as follows.
\begin{enumerate}[nosep]
\item If $D_I^*[\{w_i~|~i \in I_A^2\}]$ has no edge then identify $y_{\min I_B^2}$ to $s$ and
rename $z_{\min I_B^2}$ as $r_{WZ}$.
\item If $D_I^*[\{z_i~|~i \in I_B^2\}]$ has no edge then identify $x_{\max I_A^2}$ to $t$ and
rename $w_{\max I_A^2}$ as $r_{WZ}$.
\end{enumerate}
\end{identification rule}

\begin{claim}
\label{clm:renamingWZ}
If \Cref{rr:WZ} has been applied then it still holds that $d^-_{D_I}(s) = d^+_{D_I}(t) = 0$.
Furthermore, $N^-_{D_I}(r_{WZ}) \setminus \{s\} = \{w_i~|~i \in I_A^2\}$ and $N^+_{D_I}(r_{WZ}) \setminus \{t\} = \{z_i~|~i \in I_B^2 \setminus \{\min I_B^2\}\}$.
\end{claim}

\begin{claimproof}
Assume that \Cref{rr:WZ} has been applied and that $D_I^*[\{w_i~|~i \in I_A^2\}]$ has no edge (the case where $D_I^*[\{z_i~|~i \in I_B^2\}]$ has no edge is symmetric).
Then, by assumption, $D_I^*[\{z_i~|~i \in I_B^2\}]$ must have edges;
and in fact, by \Cref{clm:sameinDIYZ} and definition of $I_B^2$, 
$D_I^*[\{z_i~|~i \in I_B^2\}]$ must be transitively equivalent to the directed path 
$z_{\min I_B^2} \rightarrow z_{\min I_B^2 +1} \rightarrow \ldots \rightarrow z_{\max I_B^2}$.
It follows that $N^+_{D_I}(r_{WZ}) \setminus \{t\} = N^+_{D_I}(z_{\min I_B^2}) \setminus \{t\} = \{z_i~|~i \in I_B^2 \setminus \{\min I_B^2\}\}$
and since, by assumption, there is every from $\{w_i~|~i \in I_A^2\}$ to $\{z_i~|~i \in I_B^2\}$ in $D_I$,
we conclude that $N^-_{D_I}(r_{WZ}) \setminus \{s\} = \{w_i~|~i \in I_A^2\}$.
Finally to see that $d^-_{D_I}(s) = d^+_{D_I}(t) = 0$, note that by construction, the only possible in-neighbor for $y_{\min I_B^2}$ is $s$ and the only possible out-neighbor for $x_{\max I_A^2}$ is $t$.
\end{claimproof}

\begin{identification rule}
\label{rr:YX}
If there is every edge from $\{y_i~|~i \in I_B^2\}$ to $\{x_i~|~i \in I_A^2\}$ in $D_I$ 
and exactly one of $D_I^*[\{y_i~|~i \in I_B^2\}]$ and $D_I^*[\{x_i~|~i \in I_A^2\}]$ has no edge, then proceed as follows.
\begin{enumerate}
\item If $D_I^*[\{y_i~|~i \in I_B^2\}]$ has no edge then identify $w_{\min I_A^2}$ to $s$ and
rename $x_{\min I_A^2}$ as $r_{YX}$.
\item If $D_I^*[\{x_i~|~i \in I_A^2\}]$ has no edge then identify $z_{\max I_B^2}$ to $t$ and 
rename $y_{\max I_B^2}$ as $r_{YX}$.
\end{enumerate}
\end{identification rule}

Symmetrically to \Cref{clm:renamingWZ}, we have the following.

\begin{claim}
\label{clm:renamingYX}
If \Cref{rr:YX} has been applied then it still holds that $d^-_{D_I}(s) = d^+_{D_I}(t) = 0$.
Furthermore, $N^-_{D_I}(r_{YX}) \setminus \{s\} = \{y_i~|~i \in I_B^2 \setminus \{\max I_B^2\}\}$ and $N^+_{D_I}(r_{YX}) \setminus \{t\} = \{x_i~|~i \in I_A^2\}$.
\end{claim}

Note finally that \Cref{rr:WZ} reduces the size of $E(D_I - \{s,t\})$ by one as does \Cref{rr:YX} (it may be that the size of $\{(w_i,x_i)~|~i \in I_A^2\}$ reduces by two while the size of $\{(y_i,z_i)~|~i \in I_B^2\}$ remains the same if \Cref{rr:WZ}.2 and \Cref{rr:YX}.1 are applied, and vice-versa). Since we may always safely identify consecutive edges in $\{(w_i,x_i)~|~i \in I_A^2\}$ or $\{(y_i,z_i)~|~i \in I_B^2\}$ to reduce their size to $4g(t)$ (as done above), we conclude that $D_I$ can be identified to a digraph transitively equivalent to an almost $4g(t)$-hard matching pattern.\\

There remains to show how to obtain a $g(t)$-hard matching pattern from an almost $4g(t)$-hard matching pattern. To this end, let $H$ be an almost $\sizeXYt$-hard matching pattern. We first show how to handle the source vertex from item 3 in \Cref{def:almost} and then how to handle the sink vertex from item 4 in \Cref{def:almost}.

Let $s$ be the vertex of $H$ such that $N^-(s) = \emptyset$ and $N^+(s) \subseteq W \cup X \cup Y \cup Z$.
Suppose first that $N^+(s) \cap W \neq \emptyset$ and let $I_1 \subseteq [\sizeXYt]$ be the set of indices such that $N^+(s) \cap W = \{w_i~|~i \in I_1\}$. By the pigeonhole principle, either (1) $|I_1| \geq \sizeXYt/2$ or (2) $|[\sizeXYt] \setminus I_1| \geq \sizeXYt/2$. If (1) holds then we proceed as follows.
\begin{itemize}
\item For every $j < \min I_1$, we identify $w_j$ and $x_j$ to $w_{\min I_1}$ and $x_{\min I_1}$, respectively.
\item For every consecutive $i < i' \in I_1$ and for every $j \in ]i,i'[$, we identify $w_j$ and $x_j$ with $w_i$ and $x_i$, respectively.
\item For every $j > \max I_1$, we identify $w_j$ and $x_j$ to $w_{\max I_1}$ and $x_{\max I_1}$, respectively. 
\end{itemize}
Suppose next that (2) holds. If $H$ contains the path $w_1 \rightarrow w_2 \rightarrow \ldots \rightarrow w_{\sizeXYt}$ then we simply identify $s$ to $w_1$. Let us therefore assume that $H$ does not contain this path. If $H$ contains the path $x_1 \rightarrow x_2 \rightarrow \ldots \rightarrow x_{\sizeXYt}$ then for every $i \in I_1$, we identify $w_i$ to $x_i$ and further identify $s$ to $x_1$.
Finally, if $H$ does not contain this path then we identify every vertex in $\{w_i,x_i~|~i \in I_1\} \cup \{s\}$ to $w_1$.

Now suppose that $N^+(s) \cap W = \emptyset$ and $N^+(s) \cap X \neq \emptyset$. Let $I_1 \subseteq [\sizeXYt]$ be the set of indices such that $N^+(s) \cap W = \{w_i~|~i \in I_1\}$. By the pigeonhole principle, either (1) $|I_1| \geq \sizeXYt/2$ or (2) $|[\sizeXYt] \setminus I_1| \geq \sizeXYt/2$. If (1) holds then we proceed as follows.
\begin{itemize}
\item For every $j < \min I_1$, we identify $w_j$ and $x_j$ to $w_{\min I_1}$ and $x_{\min I_1}$, respectively.
\item For every consecutive $i < i' \in I_1$ (that is, $I_1 \cap [i,'i]= \{i,i'\}$) and for every $j \in ]i,i'[$, we identify $w_j$ and $x_j$ with $w_i$ and $x_i$, respectively.
\item For every $j > \max I_1$, we identify $w_j$ and $x_j$ to $w_{\max I_1}$ and $x_{\max I_1}$, respectively. 
\end{itemize}
Suppose next that (2) holds. If $H$ contains the path $x_1 \rightarrow x_2 \rightarrow \ldots \rightarrow x_{\sizeXYt}$ then we simply identify $s$ to $x_1$. Let us therefore assume that $H$ does not contain this path. If $H$ contains the path $w_1 \rightarrow w_2 \rightarrow \ldots \rightarrow w_{\sizeXYt}$ then for every $i \in I_1$, we identify $x_i$ to $w_i$ and further identify $s$ to $w_1$. Finally, if $H$ does not contain this path then we identify every vertex in $\{w_i,x_i~|~i \in I_1\} \cup \{s\}$ to $w_1$.

By proceeding symmetrically with $Y$ and $Z$, we may identify $H$ to an almost $2g(t)$-hard matching pattern $H_1$ in which $s$ satisfies item 3 of \Cref{def:cleanedtoughpair}. 
Now let $t$ be the vertex of $H$ (and of $H_1$) such that $N^+(t) = \emptyset$ and $N^-(t) \subseteq W \cup X \cup Y \cup Z$. Then using similar arguments, we may further identify $H_1$ to an almost $g(t)$-hard matching pattern $H_2$ in which $s$ still satisfies item 3 of \Cref{def:cleanedtoughpair} and $t$ satisfies item 4 of \Cref{def:cleanedtoughpair}, that is, $H_2$ is a $g(t)$-hard matching pattern which concludes the proof. 
\end{proof}
\newcommand{\secondhard}{second-hard-structure}
\subsection{Cleaning the induced-biclique with minimal edges}\label{sec:cleaningbiclique}

In this section, we prove Lemma~\ref{lem:cleaningbiclique} restated below.

\cleaningbiclique*

\begin{proof}
Let $(A,B)$ be a $9t$-biclique in $D$ such the edges of $A \cup B$ are minimal edges of $D$, and such that $(A,B)$ is a $9t$-induced-biclique in $D^{\star}$. Then $A,B$ are independent sets in $D^{\star}$.

We first claim that for any $v \in V(D) \setminus (A \cup B)$, it is not the case that 
$N^-_{D^{\star}}(v) \cap (A \cup B) \neq \emptyset$ and $N^+_{D^{\star}}(v) \cap (A \cup B)\neq \emptyset$. For the sake of contradiction, say it is the case. Let $x \in N^-_{D^{\star}}(v) \cap (A \cup B)$ and let $y \in N^+_{D^{\star}}(v) \cap (A \cup B)$. Then there exists an $(x,y)$-path in $D$ of length strictly greater than $1$. If $x,y \in A$ (resp.~$x,y \in B$), then this contradicts that $A$ (resp.~$B$) is an independent set in $D^*$. If $x \in A$ and $y \in B$, then this contradicts the minimality of the edges of the $(A,B)$-biclique. If $x \in B$ and $y \in A$, then this contradicts the definition of an $(A,B)$-biclique.

Let $V_{\sss} \subseteq V(D) \setminus (A \cup B)$ such that for each $v \in V_{\sss}$, $N^+_{D^{\star}}(v) \neq \emptyset$ and $N^-_{D^{\star}}(v)  \cap (A \cup B)=\emptyset$.
Let $V_{\ttt} = V(D) \setminus (A \cup B \cup V_{\sss})$. Then from the claim in the above paragraph, 
for each $v \in V_{\ttt}$, $N^+_{D^{\star}}(v) \cap (A \cup B)=\emptyset$. 
Observe that $V_{\sss} \uplus V_{\ttt}$ partition the vertex set $V(D) \setminus (A \cup B)$. If $V_{\sss} \neq  \emptyset$, then identify the vertices of $V_{\sss}$ into $\sss$, if $V_{\ttt} \neq \emptyset$, then identify the vertices of $V_{\ttt}$ into $\ttt$. Let $\widehat{D}$ be the resulting digraph. 
We now show that $E(V_{\ttt}, V_{\sss}) =\emptyset$. For the sake of contradiction, say $u \in V_{\ttt}$ and $v \in V_{\sss}$ such that $(u,v) \in E(D^{\star})$. Then $N^+_{D^{\star}}(u) \cap (A \cup B) \neq \emptyset$.
Since $u \not \in V_{\sss}$, $N^-_{D^{\star}}(u) \neq \emptyset$. But this contradicts the claim in the second paragraph. Thus, we conclude that $N^-_{\widehat{D}}(\sss), N^+_{\widehat{D}}(\ttt) = \emptyset$.

Now let $t' = 9t$.
Recall that $A,B$ are ordered sets, say $A=(a_1, \ldots, a_{t'})$ and $B=(b_1, \ldots, b_{t'})$.
Let $I_1 \subseteq [t']$ such that for $i \in I_1$, $(\sss,a_i) \in E(\widehat{D})$. Let $I_2 \subseteq [t']$ such that for each $i \in I_2$, $(\sss,a_i) \not \in  E(\widehat{D})$ and $(\sss,b_i) \in E(\widehat{D})$. Let $I_3 \subseteq [t']$ such that for each $i \in I_3$, $(\sss,a_i), (\sss,b_i) \not \in E(\widehat{D})$. Observe that $[t']=I_1 \uplus I_2 \uplus I_3$. 
Let $J_1 \subseteq [t']$ such that for each $j \in J_1$, $(b_j,\ttt) \in E(\widehat{D})$. Let $J_2 \subseteq [t']$ such that for each $j \in J_2$, $(b_j , \ttt) \not \in E(\widehat{D})$ and $(a_j,\ttt) \in E(\widehat{D})$. Let $J_3 \subseteq [t']$ such that for each $j \in J_3$, $(b_j , \ttt), (a_j,\ttt) \not \in E(\widehat{D})$. Observe that $[t'] = J_1 \uplus J_2 \uplus J_3$.

By the pigeon-hole principle, either $|I_1| \geq t'/3$ or $|I_2| \geq t'/3$ or $|I_3| \geq t'/3$. We consider these three cases separately.\\

\noindent
{\bf Case 1. $|I_1| \geq t'/3$.} Suppose first that $|J_1 \cap I_1| \geq |I_1|/3$. Let $A' =\{a_i: i \in J_1 \cap I_1\}$ and $B' =\{b_i: i \in J_1 \cap I_1\}$. Observe that $|A'|=|B'| \geq t'/3$. Identify the vertices of $(A\setminus A') \cup (B \setminus B')$ with $\sss$. Let the resulting digraph be  $\widehat{D'}$. Observe that $A' \subseteq N^+_{\widehat{D'}}(\sss), B' \subseteq N^-_{\widehat{D'}}(\ttt)$ and $(\ttt, \sss) \not \in E(\widehat{D'})$.

Suppose next that $|J_2 \cap I_1| \geq |I_1|/3$. Let $A' =\{a_i:  i \in J_2 \cap I_1\}$ and $B' =\{b_i:  i \in J_2 \cap I_1\}$. Observe that $|A'|=|B'| \geq t'/3$. Identify the vertices of $(A\setminus A') \cup (B \setminus B')$ with $\sss$. Let the resulting digraph be  $\widehat{D'}$. Observe that $A' \subseteq N^+_{\widehat{D'}}(\sss), A' = N^-_{\widehat{D'}}(\ttt)$ and $(\ttt,\sss) \not \in E(\widehat{D'})$.

Suppose finally that $|J_3 \cap I_1| \geq |I_1|/3$. Let $A' =\{a_i:  i \in J_3 \cap I_1\}$ and $B' =\{b_i:  i \in J_3 \cap I_1\}$. Observe that $|A'|=|B'| \geq t'/3$.  Identify the vertices of $(A\setminus A') \cup (B \setminus B') \cup \ttt$ with $\sss$. Let the resulting digraph be  $\widehat{D'}$. Observe that $A' \subseteq N^+_{\widehat{D'}}(\sss)$.\\

\noindent
{\bf Case 2. $|I_2| \geq t'/3$.} Suppose first that $|J_1 \cap I_2| \geq |I_2|/3$. Let $A' =\{a_i:  i \in J_1 \cap I_2\}$ and $B' =\{b_i: i \in J_1 \cap I_2\}$. Observe that $|A'|=|B'| \geq t'/3$.
Identify the vertices of $(A\setminus A') \cup (B \setminus B')$ with $\sss$. Let the resulting digraph be  $\widehat{D'}$. Observe that $B' = N^+_{\widehat{D'}}(\sss), B' \subseteq N^-_{\widehat{D'}}(\ttt)$ and $(\ttt,\sss) \not \in E(\widehat{D'})$.

Suppose next that $|J_2 \cap I_2| \geq |I_2|/3$. Let $A' =\{a_i: i \in J_2 \cap I_2\}$ and $B' =\{b_i:  i \in J_2 \cap I_2\}$. Observe that $|A'|=|B'| \geq t'/3$.
Identify the vertices of $(A\setminus A') \cup (B \setminus B')$ with $\sss$. Let the resulting digraph be  $\widehat{D'}$. Observe that  $B' = N^+_{\widehat{D'}}(\sss), A' = N^-_{\widehat{D'}}(\ttt)$ and $(\ttt,\sss) \not \in E(\widehat{D'})$.

Suppose finally that $|J_3 \cap I_2| \geq |I_2|/3$. Let $A' =\{a_i:  i \in J_3 \cap I_2\}$ and $B' =\{b_i:  i \in J_3 \cap I_2\}$. Observe that $|A'|=|B'| \geq t'/3$.
Identify the vertices of $(A\setminus A') \cup (B \setminus B') \cup \ttt$ with $\sss$. Let the resulting digraph be $\widehat{D'}$. Observe that $B' = N^+_{\widehat{D'}}(\sss)$. \\

\noindent
{\bf Case 3. $|I_3| \geq t'/3$.} Suppose first that $|J_1 \cap I_3| \geq |I_3|/3$. Let $A' =\{a_i:  i \in J_1 \cap I_3\}$ and $B' =\{b_i:  i \in J_1 \cap I_3\}$. Observe that $|A'|=|B'| \geq t'/3$.
Identify the vertices of $(A\setminus A') \cup (B \setminus B')$, if any, with $\sss$. Let the resulting digraph be  $\widehat{D'}$. Observe that if $(A\setminus A') \cup (B \setminus B') =\emptyset$, then then $\sss \not \in V(\widehat{D'})$ and $B \subseteq N^-_{\widehat{D'}}(\ttt)$. Otherwise, additionally, $\sss \in V(D)$ and since $(A,B)$ is an induced biclique, $B'=N^+_{\widehat{D'}}(\sss)$. Moreover, $(\ttt,\sss) \not \in E(\widehat{D'})$.

Suppose next that $|J_2 \cap I_3| \geq |I_3|/3$. Let $A' =\{a_i:  i \in J_2 \cap I_3\}$ and $B' =\{b_i:  i \in J_2 \cap I_3\}$. Observe that $|A'|=|B'| \geq t'/3$.
Identify the vertices of $(A\setminus A') \cup (B \setminus B')$, if any,  with $\sss$. Let the resulting digraph be  $\widehat{D'}$. If $(A\setminus A') \cup (B \setminus B') =\emptyset$, then $\sss \not \in V(\widehat{D'})$ and $A = N^-_{\widehat{D'}}(\ttt)$. Otherwise, additionally $\sss \in V(\widehat{D'})$ and since $(A,B)$ is an induced biclique, $B= N^+_{\widehat{D'}}(\sss)$. Moreover, $(\ttt,\sss) \not \in E(\widehat{D'})$.
  
Suppose finally that $|J_3 \cap I_3| \geq |I_3|/3$. Let $A' =\{a_i: i \in J_3 \cap I_3\}$ and $B' =\{b_i: i \in J_3 \cap I_3\}$. Observe that $|A'|=|B'| \geq t'/3$ and $(A\setminus A') \cup (B \setminus B') \neq \emptyset$ since $D$ is connected (without loss of generality).
Identify the vertices of $(A\setminus A') \cup (B \setminus B') \cup \ttt$ with $\sss$ . Let the resulting digraph be $\widehat{D'}$. Since $(A,B)$ is an induced biclique and $(A\setminus A') \cup (B \setminus B') \neq \emptyset$, 
 $B' = N^+_{\widehat{D'}}(\sss)$. 
\end{proof}
\section{Hardness proofs and lower bounds}\label{sec:lower}

This section is devoted to presenting the lower bounds on the complexity of \Dsn for various classes $\calD$. First we prove the combinatorial result Lemma~\ref{lem:findstar} that implies that there are no nontrivial class with subexponential FPT algorithms on planar graphs (Section~\ref{sec:findstar}). Then we present reductions from \textsc{Grid Tiling} to \Dsn where $\calD$ is the class of diamonds (Section~\ref{sec:diamonds}), hard matching patterns (Section~\ref{sec:orderedtough}), or hard biclique patterns (Section~\ref{sec:biclique}).

\subsection{Finding a star}
We present here a short proof that every nontrivial class $\calD$ in our setting contains either all cycles, all in-stars, or all out-stars.
\label{sec:findstar}
\begin{proof}[Proof (of Lemma~\ref{lem:findstar})]
  Suppose that there is an integer $\alpha$ such that $\calD$ does not contain a cycle, in-star, or out-star on $\alpha$ vertices. This means in particular that a graph $\calD$ cannot have a strongly connected component with $\alpha$ vertices. Indeed, by selecting $\alpha$ vertices of the component and identifying every other vertex of the graph with one of these vertices, we would obtain a strongly connected graph on $\alpha$ vertices, which is transitively equivalent to a cycle on $\alpha$ vertices.
It also follows that a graph $D\in\calD$ has less than $\alpha$ strongly connected components of size larger than 1: otherwise, identifying one vertex from each such component to a single vertex would create a strongly connected component of size larger than $\alpha$.
  
This bound on the size of the strongly connected components implies that identifying the vertex set of each strongly connected component of some $D\in\calD$ to a single vertex results in an acyclic graph whose size is at most a factor of $\alpha$ smaller. Therefore, the fact that $\calD$ contains arbitrarily large graphs implies that $\calD$ contains arbitrary large acyclic graphs.
Furthermore, $\calD$ contains arbitrarily large graphs without isolated vertices: if we have a graph with $c$ edges and identify every isolated vertex with some vertex $v$, then we get a graph without isolated vertices and exactly $c$ edges. If we identify the strongly connected components of such a graph to single vertices, then less than $\alpha$ isolated vertices can be created (as there are less than $\alpha$ components of size larger than 1). It follows that $\calD$ contains arbitrarily large acyclic graphs with less than $\alpha$ isolated vertices. 

Observe that $D\in\calD$ cannot have a path longer than $\alpha+1$: otherwise, identifying the first and last vertices of the path would create a member of $D\in\calD$ that has a strongly connected component larger than $\alpha$. This means that every acyclic $D\in \calD$ can be partitioned into $\alpha+2$ levels, that is, $D$ has a topological ordering with blocks $B_1$, $\dots$, $B_\ell$ with $\ell\le \alpha+2$ such that each block $B_i$ is an independent set. If an acyclic graph $D_c$ has at least $(\alpha+2)c$ vertices, then one such block $B_i$ has size at least $c$. Let $D^*_c$ be obtained by identifying the blocks $B_1$, $\dots$, $B_{i-1}$ to a single vertex $x$, and the blocks $B_{i+1}$, $\dots$, $B_\ell$ to a single vertex $y$. If there are $\alpha$ vertices in $B_i$ adjacent to both $x$ and $y$, then identifying $x$ and $y$ creates a strongly connected component, a contradiction. If there is a set $I\subseteq B_i$ of  $\alpha$ vertices that is adjacent only to $x$, then identifying $(B_i\setminus I)\cup \{x,y\}$ to a single vertex results in an out-star on $\alpha+1$ vertices. Similarly, if there is a set $I\subseteq B_i$ of  $\alpha$ vertices that is adjacent only to $y$, then identifying $(B_i\setminus I)\cup \{x,y\}$ to a single vertex results in an in-star on $\alpha+1$ vertices. Finally, $I$ cannot have $\alpha$ vertices that are adjacent to neither $x$ or $y$, because they would be isolated vertices. Therefore, if $c\ge 4\alpha$, then we arrive to a contradiction in one of the three cases. 
\end{proof}
\subsection{Diamonds}
\label{sec:diamonds}

The aim of this section is to prove the following.

\begin{theorem}
\label{thm:diamond}
For any $\ell \in [4]$, {\sc Planar $\oldC_\ell$-Steiner Network} is $\mathsf{W}[1]$-hard parameterized by the number $k$ of terminals and 
does not admit an $f(k) \cdot n^{o(\sqrt{k})}$ algorithm for any computable function $f$,
unless $\mathsf{ETH}$ fails.
\end{theorem}

We only formally prove the statement for pure out-diamonds as it will become clear from the proof that
to handle 
\begin{enumerate}
\item flawed out-diamonds, it suffices to add a vertex $s$ and edges $(s,r_1)$ and $(s,r_2)$ in the construction below;
\item pure in-diamonds, it suffices to reverse the direction of every edge in the construction below; and
\item flawed in-diamonds, it suffices to add a vertex $t$ and edges $(r_1,t)$ and $(r_2,t)$, and additionally reverse the direction of every edge in the construction below.
\end{enumerate}
In the following, we assume that the class of all pure out-diamonds is $\oldC_1$.
To prove the statement, we first introduce the {\sc $k \times k$-Grid Tiling} problem, formally defined below.\\

\defparproblem{{\sc $k \times k$-Grid Tiling}}{Integers $k,n$ and $k^2$ nonempty set $S_{i,j} \subseteq [n] \times [n]$ where $i,j \in [k]$.}{$k$}{For each $i,j \in [k]$, does there exist an entry $(x_{i,j},y_{i,j}) \in S_{i,j}$ such that
\begin{itemize}
\item for every $i \in [k]$, $x_{i,1} = x_{i,2} = \ldots = x_{i,k}$ and
\item for every $j \in [k]$, $y_{1,j} = y_{2,j} = \ldots = y_{k,j}$?
\end{itemize}}

\bigskip

Throughout the paper, we assume that whenever given a set $S \subseteq [n] \times [n]$, $1 < x,y < n$ holds for every $(x,y) \in S$: it suffices to increase $n$ by two and replace $(x,y)$ by $(x+1,y+1)$ otherwise.

It was shown \cite[Theorem 14.28]{10.5555/2815661} that, under $\mathsf{ETH}$, {\sc $k \times k$-Grid Tiling} does not admit an algorithm running in time $f(k)\cdot n^{o(k)}$ for any computable function $f$. 
To prove \Cref{thm:diamond} for pure out-diamonds, 
we give a reduction which transforms an instance of {\sc $k \times k$-Grid Tiling} 
into an instance of (edge-weighted)\footnote{We argue at the end of the section that with polynomially bounded integer weights, the edge-weighted version of the problem reduces to its unweighted version.} {\sc Planar $\oldC_1$-Steiner Network} with $O(k^2)$ terminals.
To this end, we design three types of gadgets: 
the \emph{connector gadget}, the \emph{down main gadget} and the \emph{up main gadget}.
The reduction represents each set of the {\sc $k \times k$-Grid Tiling} instance with a copy of the up or down main gadget,
and uses the connector gadgets to further connect these gadgets (see \Cref{fig:reduction}).
The remainder of this section is organized as follows.
We first introduce the connector gadget and prove in \Cref{lem:cg} that it satisfies several desired properties.
We then introduce the down main gadget and prove in \Cref{lem:dmg} that it also satisfies several desired properties.
The up main gadget is introduced thereafter together with its symmetrical \Cref{lem:umg}.
We end the section with the precise description of the reduction and a proof of its correctness.\\

\noindent
\textbf{Connector gadget.} Given an integer $n > 0$, the connector gadget $CG_n$ is an edge-weighted planar digraph consisting of $n \times n$ grid, where the vertex lying at the intersection of column $i$ and row $j$ is denoted by $x_{i,j}$, and $2n+2$ additional vertices $p_1,\ldots,p_n,q_1,\ldots,q_n,p,q$. The adjacencies and edge weights are defined as follows (we fix $N = 3n$).
\begin{itemize}
\item \emph{Left source edges:} for every $j \in [n]$, there is an edge $(p_j,x_{1,j})$. Together these edges are called left source edges. The weight of each such edge is set to $N^2$.
\item \emph{Right source edges:} for every $j \in [n]$, there is an edge $(q_j,x_{n,j})$. Together these edges are called right source edges. The weight of each such edge is set to $N^2$.
\item \emph{Top sink edges:} for every $i \in [n]$, there is an edge $(x_{i,n},p)$. Together these edges are called top sink edges. The weight of each such edge is set to $N^2$.
\item \emph{Bottom sink edges:} for every $i \in [n]$, there is an edge $(x_{i,1},q)$. Together these edges are called bottom sink edges. The weight of each such edge is set to $N^2$.
\item The $n \times n$ grid is divided in two according to the diagonal $j = n+1-i$: the vertices on the top-right part (that is, the vertices $x_{i,j}$ such that $j+i \geq n+1$) are connected with $\leftarrow$ and $\uparrow$ edges, and the vertices of the bottom-left part (that is, the vertices $x_{i,j}$ such that $j+i \leq n+1$) are connected with $\rightarrow$ and $\downarrow$ edges. More specifically, we define the following edges.
\begin{itemize}
\item \emph{Up edges}: for every $i \in [n]$ and every $j \in [n-1]$ such that $j+i \geq n+1$, there is an edge $(x_{i,j},x_{i,j+1})$. Together these edges are called up edges.
\item \emph{Down edges:} for every $i \in [n]$ and every $j \in [n] \setminus \{1\}$ such that $j+i \leq n+1$, there is an edge $(x_{i,j},x_{i,j-1})$. Together these edges are called down edges.
\end{itemize}
We set the weight of every up/down edge to $N$.
\begin{itemize}
\item \emph{Left edges:} for every $i \in [n] \setminus \{1\}$ and every $j \in [n]$ such that $j+i > n+1$, there is an edge $(x_{i,j},x_{i-1,j})$. Together these edges are called left edges.
\item \emph{Right edges:} for every $i \in [n-1]$ and every $j \in [n]$ such that $j+i < n+1$, there is an edge $(x_{i,j},x_{i+1,j})$. Together these edges are called right edges.
\end{itemize}
We set the weight of every left/right edge to $1$.
\end{itemize}  
This concludes the construction of the connector gadget $CG_n$  (see \Cref{fig:connectorgad} for an illustration of the connector gadget for $n = 4$). In the following, we call the vertices $p_1,\ldots,p_n$ the \emph{left vertices}, the vertices $q_1,\ldots,q_n$ the \emph{right vertices} and the vertices $p,q$ the \emph{terminal vertices}. We further set 
\[
C_n^* = 4N^2+(n-1)N+n-1.
\]
A set $E \subseteq E(CG_n)$ satisfies the \emph{connectedness} property if the following hold in $E$:
\begin{itemize}
\item $p$ can be reached from some right vertex and from some left vertex;
\item $q$ can be reached from some right vertex and from some left vertex.
\end{itemize}
A set $E \subseteq E(CG_n)$ satisfying the connectedness property \emph{represents} an integer $j \in [n]$ if the only left source edge in $E$ is the one incident to $p_j$ and the only right source edge in $E$ is the one incident to $q_j$ (see \Cref{fig:connectorgad} for a set of edges representing 2).

\begin{figure}
\centering
\begin{tikzpicture}
\node[rectangle, fill=yellow, rounded corners,minimum width=5.3cm,minimum height=.25cm] at (3.5,2) {};
\node[rectangle, fill=yellow, rounded corners,minimum width=.25cm,minimum height=3.2cm] at (4,2.5) {};
\node[rectangle, fill=yellow, rounded corners,minimum width=.25cm,minimum height=1.35cm,rotate=27.5] at (3.75,4.5) {};
\node[rectangle, fill=yellow, rounded corners,minimum width=.25cm,minimum height=1.35cm,rotate=-27.5] at (3.75,.5) {};

\foreach \i in {2,...,5}
\foreach \j in {1,...,4}
{\node[circ] (\i\j) at (\i,\j) {};} 

\foreach \i in {2,...,4}
{\pgfmathtruncatemacro{\y}{6-\i}
  \foreach \j in {\y,...,2}
  {\pgfmathtruncatemacro{\z}{\j-1}
    \draw[->,>=stealth,thick] (\i\j) -- (\i\z);}}
    
\foreach \i in {5,...,3}
{\pgfmathtruncatemacro{\y}{6-\i}
  \foreach \j in {\y,...,3}
  {\pgfmathtruncatemacro{\z}{\j+1}
    \draw[->,>=stealth,thick] (\i\j) -- (\i\z);}}    
    
\foreach \j in {1,...,3}
{\pgfmathtruncatemacro{\y}{5-\j}
  \foreach \i in {2,...,\y}
  {\pgfmathtruncatemacro{\z}{\i+1}
   \draw[->,>=stealth,thick] (\i\j) -- (\z\j);}} 

\foreach \j in {4,...,2}
{\pgfmathtruncatemacro{\y}{7-\j}
  \foreach \i in {5,...,\y}
  {\pgfmathtruncatemacro{\z}{\i-1}
   \draw[->,>=stealth,thick] (\i\j) -- (\z\j);}} 

\foreach \j in {1,...,4}
{\node[circ,label=left:{\small $p_{\j}$}] (p\j) at (1,\j) {};}

\foreach \j in {1,...,4}
{\draw[->,>=stealth,thick] (p\j) -- (2\j);}

\foreach \j in {1,...,4}
{\node[circ,label=right:{\small $q_{\j}$}] (q\j) at (6,\j) {};}

\foreach \j in {1,...,4}
{\draw[->,>=stealth,thick] (q\j) -- (5\j);}

\node[circ,label=below:{\small $q$}] (q) at (3.5,0) {};

\foreach \i in {2,...,5}
{\draw[->,>=stealth,thick] (\i1) -- (q);}

\node[circ,label=above:{\small $p$}] (p) at (3.5,5) {};

\foreach \i in {2,...,5}
{\draw[->,>=stealth,thick] (\i4) -- (p);}
\end{tikzpicture}
\caption{The connector gadget for $n = 4$. A set of edges representing 2 is highlighted.}
\label{fig:connectorgad}
\end{figure}
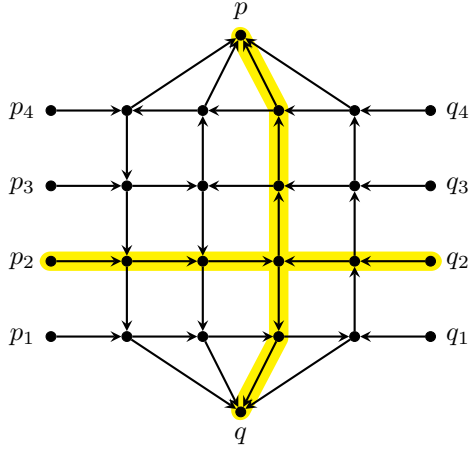

\begin{lemma}
\label{lem:cg}
For any integer $n > 0$, the connector gadget $CG_n$ satisfies the following properties.
\begin{itemize}
\item[(1)] For every $j \in [n]$, there exists a set $E_j \subseteq E(CG_n)$ of weight $C_n^*$ representing $j$.
\item[(2)] If there exists a set $E \subseteq E(CG_n)$ of weight at most $C_n^*$ satisfying the connectedness property, then $E$ has weight exactly $C_n^*$ and represents some integer $j \in [n]$.
\end{itemize}
\end{lemma}

\begin{proof}
To prove (1), it suffices to take $E_j$ to be the union of the following sets of edges:
\begin{itemize}
\item $\{(p_j,x_{1,j}),(q_j,x_{n,j}),(x_{n+1-j,n},p),(x_{n+1-j,1},q)\}$; 
\item $\{(x_{i,j},x_{i+1,j})~|~1 \leq i \leq n-j\}$; 
\item $\{(x_{i,j},x_{i-1,j})~|~n+2-j \leq i \leq n\}$; 
\item $\{(x_{n+1-j,\ell}, x_{n+1-j,\ell+1})~|~j \leq \ell \leq n-1\}$; and
\item $\{(x_{n+1-j,\ell},x_{n+1-j,\ell -1})~|~2 \leq \ell \leq j\}$.
\end{itemize}
It is not difficult to see that $E_j$ represents $j$ and has weight $C_n^*$.\\

Next, suppose that $E \subseteq E(CG_n)$ is a set of weight at most $C_n^*$ satisfying the connectedness property. Let us show that $E$ has weight exactly $C_n^*$ and represents some integer $j \in [n]$. To this end, we prove the following claims.

\begin{claim}
\label{clm:sourcesinkedges}
$E$ contains exactly one left source edge, one right source edge, one top sink edge and one bottom sink edge.
\end{claim}

\begin{claimproof}
Since $E$ satisfies the connectedness property, it contains at least one left source edge, one right source edge, one top sink edge and one bottom sink edge. Now if $E$ contains at least two left source edges (the other cases are symmetric), then the weight of $E$ is at least $5N^2$; however, by definition,
\[
C_n^*  = 4N^2 + (n-1)N+ n-1 < 4N^2 + nN + nN < 5N^2
\]
as $N = 3n > 2n$, a contradiction.
\end{claimproof}

In the following, we let $p_{j_1}$ be the left vertex incident to the left source edge in $E$ and $q_{j_2}$ be the right vertex incident to the right source edge in $E$.

\begin{claim}
\label{clm:updownedges}
For every $j \in [n-1]$, $E$ contains exactly one (up or down) edge with one endvertex on row $j$ and one endvertex one row $j+1$. In particular, $E$ contains exactly $n-1$ up/down edges.
\end{claim}

\begin{claimproof}
First note that if there exists $j \in [n-1]$ such that $E$ contains no (up or down) edge with one endvertex on row $j$ and the other endvertex on row $j+1$, then either $j_1 \leq j$ in which case $p_{j_1}$ cannot reach $p$ in $E$, or $j_1 > j$ in which case $p_{j_1}$ cannot reach $q$ in $E$, a contradiction in both cases to the connectedness of $E$. Thus, for every $j \in [n-1]$, $E$ contains at least one (up or down) edge with one endvertex on row $j$ and the other endvertex on row $j+1$; in particular, $E$ contains at least $n-1$ up/down edges. Now if $E$ contains at least $n$ up/down edges then by \Cref{clm:sourcesinkedges}, the weight of $E$ is at least $4N^2 + nN$; however, by definition, 
\[
C_n^* = 4N^2 + (n-1)N + n-1 < 4N^2 + (n-1)N + n < 4N^2 + nN
\]
as $N = 3n > n$, a contradiction.
\end{claimproof}

\begin{claim}
\label{clm:leftrightedges}
$E$ contains exactly $n-1$ left/right edges.
\end{claim}

\begin{claimproof}
Observe first that for every $i \in [n-1]$, $E$ contains at least one (left or right) edge with one endvertex on column $i$ and the other endvertex on column $i+1$: indeed, if this were not the case for some $i \in [n-1]$, then every path in $CG_n[E]$ from $p_{j_1}$ to $p$ would be edge-disjoint from every path in $CG_n[E]$ from $q_{j_2}$ to $p$ and thus, $E$ would contain at least two top sink edges, a contradiction to \Cref{clm:sourcesinkedges}. Hence, for every $i \in [n-1]$, $E$ contains at least one (left or right) edge with one endvertex on column $i$ and the other endvertex on column $i+1$; in particular, $E$ contains at least $n-1$ left/right edges. Now if $E$ contains at least $n$ left/right edges then by Claims~\ref{clm:sourcesinkedges} and \ref{clm:updownedges}, the weight of $E$ is at least $4N^2 + (n-1)N + n > C_n^*$, a contradiction.
\end{claimproof}

Now by Claims~\ref{clm:sourcesinkedges}, \ref{clm:updownedges} and \ref{clm:leftrightedges}, $E$ has weight exactly $C_n^*$; we next show that $j_1 = j_2$ which would imply that $E$ represents $j_1$. Suppose to the contrary that $j_2 < j_1$ (the case where $j_2 > j_1$ is symmetric). By \Cref{clm:updownedges}, $E$ contains exactly one edge $e$ with one endvertex on row $j_2$ and one endvertex on row $j_2+1$; but then, either $e$ is an up edge in which case $p_{j_1}$ cannot reach $q$ in $E$, or $e$ is a down edge and $q_{j_2}$ cannot reach $p$ in $E$, a contradiction in both cases to the connectedness of $E$ which concludes the proof. 
\end{proof}


\noindent
\textbf{Down Main Gadget.} Given an integer $n > 0$, the down main gadget $dMG_S$ represents a set $S \subseteq [n] \times [n]$\footnote{Recall that, by assumption, $1 < x,y< n$ holds for every $(x,y) \in S$.} and is constructed as follows. It is an edge-weighted planar digraph consisting of a $2n \times n^2$ grid, where the vertex lying at the intersection of column $i$ and row $j$ is denoted by $x_{i,j}$, and $6n$ additional vertices $\ell_1,\ldots,\ell_n,\ell'_1,\ldots,\ell'_n,r_1,\ldots,r_n,r'_1,\ldots,r'_n,t_1,\ldots,t_n$, $b_1,\ldots,b_n$. The adjacencies and edge weights are defined as follows (we fix $M=13n^2$). 

\begin{itemize}
\item \emph{Source edges:} for every $i \in [n]$, there is an edge $(t_i,x_{i,n^2})$. Together these edges are called source edges. The weight of each such edge is set to $M^5$.
\item \emph{Bottom sink edges:} for every $i \in [n]$, there is an edge $(x_{n+i,1},b_i)$. Together these edges are called bottom sink edges. The weight of each such edge is set to $M^5$.
\item \emph{Left sink edges:} for every $j \in [n]$, there is an edge $(\ell'_j,\ell_j)$ whose weight is set to $Mj$. Together these edges are called left sink edges.
\item  \emph{Right sink edges:} for every $j \in [n]$, there is an edge $(r'_j,r_j)$ whose weight is set to $M^2-Mj$. Together these edges are called right sink edges.
\item \emph{Left internal sink edges:} for every $j \in [n]$ and every $(j-1)n+1 \leq p \leq jn$, there is an edge $(x_{1,p},\ell'_j)$ whose weight is set to $p-(j-1)n$. For every fixed $j \in [n]$, these edges together are called the $j^{th}$ left internal sink edges.
\item \emph{Right internal sink edges:} for every $j \in [n]$ and every $(j-1)n + 1 \leq p \leq jn$, there is an edge $(x_{2n,p},r'_j)$ whose weight is set to $n+1-(p-(j-1)n)$. For every fixed $j \in [n]$, these edges together are called the $j^{th}$ right internal sink edges.
\item \emph{Right bridge edges:} for every $j \in [n]$ and every $(j-1)n + 1 \leq p \leq jn$, there is an edge $(x_{n,p},x_{n+1,p})$. For every fixed $j \in [n]$, these edges together are called the $j^{th}$ right bridge edges. The weight of each $j^{th}$ right bridge edge is set to $M^4$.
\item \emph{Downward bridge edges:} for every $j \in \{pn+1~|~p \in [n-1]\}$ and every $i \in [2n]$, there is an edge $(x_{i,j},x_{i,j-1})$. For every fixed $j \in \{pn+1~|~p \in [n-1]\}$, these edges together are called the $j^{th}$ downward bridge edges. The weight of each $j^{th}$ downward bridge edge is set to $M^3$.
\item \emph{Down edges:} for every $i \in [2n]$, every $j \in [n]$ and every $(j-1)n + 2 \leq p \leq jn$, there is an edge $(x_{i,p},x_{i,p-1})$. Together these edges are called down edges. The weight of each such edge is set to $4$.
\item \emph{Right edges:} for every $j \in [n]$, every $(j-1)n+1 \leq p \leq jn$ and every $ jn+1-p \leq q \leq 2n-1$ such that $q \neq n$, there is an edge $(x_{q,p},x_{q+1,p})$. Together these edges are called right edges. The weight of each such edge is set to $4$.
\item \emph{Left edges:} for every $j \in [n]$, every $(j-1)n+1 \leq p \leq jn-1$ and every $2 \leq q \leq jn+1-p$, there is an edge $(x_{q,p},x_{q-1,p})$. Together these edges are called left edges. The weight of each such edge is set to $4$.
\item \emph{Shortcut edges:} for every $s=(i,j) \in S$, we introduce two shortcut edges $e^\ell_s,e^r_s$ as follows. Set $p = jn+1$, then
\begin{itemize}
\item subdivide the edge $(x_{i,p+1-i},x_{i,p-i})$ by adding a vertex $y_{i,j}$ and the edge $(x_{i,p+1-i},y_{i,j})$ (of weight $3$) with the edge $(y_{i,j},x_{i,p-i})$ (of weight $1$);
\item subdivide the edge $(x_{n+i,p-i},x_{n+i,p-1-i})$ by adding a vertex $z_{i,j}$ and the edge $(x_{n+i,p-i},\allowbreak
z_{i,j})$ (of weight $3$) with the edge $(z_{i,j},x_{n+i,p-1-i})$ (of weight $1$). 
\end{itemize}
The introduced edges are called \emph{down subdivided edges}. Then
\begin{itemize}
\item $e^\ell_s= (y_{i,j},x_{i-1,p-i})$ and its weight is set to $2$;
\item $e^r_s = (z_{i,j},x_{n+i+1,p-i})$ and its weight is set to $2$.
\end{itemize}
The edges $e^\ell_s$ are called \emph{left} shortcut edges and the edges $e^r_s$ are called \emph{right} shortcut edges.  
\end{itemize}
This concludes the construction of the down main gadget $dMG_S$ (see \Cref{fig:downmaingad} for an illustration of the down main gadget with $n=4$ representing $S= \{(2,2),(2,3),(3,2)\}$). In the following, we call the vertices $\ell_1,\ldots,\ell_n$ the \emph{left vertices}, the vertices $r_1,\ldots,r_n$ the \emph{right vertices}, the vertices $t_1,\ldots,t_n$ the \emph{top vertices} and the vertices $b_1,\ldots,b_n$ the \emph{bottom vertices}. Furthermore, we set
\[
M_n^* = 2M^5 + M^4 + (n-1)M^3 + M^2 + (4n+1)(n+1) - 12.
\] 
We further let $V_\ell = \{x_{i,j}~|~1 \leq i \leq n \text{ and } 1 \leq j \leq n^2\} \cup \{y_{i,j}~|~(i,j) \in S\}$, $V_r = \{x_{i,j}~|~n+1 \leq i \leq 2n \text{ and } 1 \leq j \leq n^2\} \cup \{z_{i,j}~|~(i,j) \in S\}$ and for every $j \in [n]$, $V_j = \{x_{i,p}~|~1 \leq i \leq 2n \text{ and } (j-1)n + 1 \leq p \leq jn\}$. Note that, by construction, the following holds.

\begin{observation}
\label{obs:reach}
For every $j \in [n]$ and every $x_{i,p} \in V_j \cap V_\ell$, the following holds.
\begin{itemize}
\item If $n+1-i > p - (j-1)n$ then $x_{i,p}$ cannot reach a vertex $x_{i',p'} \in V_j$ such that $n+1-i' \leq p' - (j-1)n$.
\item If $n+1-i < p - (j-1)n$ then $x_{i,p}$ cannot reach a vertex $x_{i',p'} \in V_j$ such that $i' < i$ and $p' - (j-1)n > n+1-i$.
\end{itemize}
\end{observation}

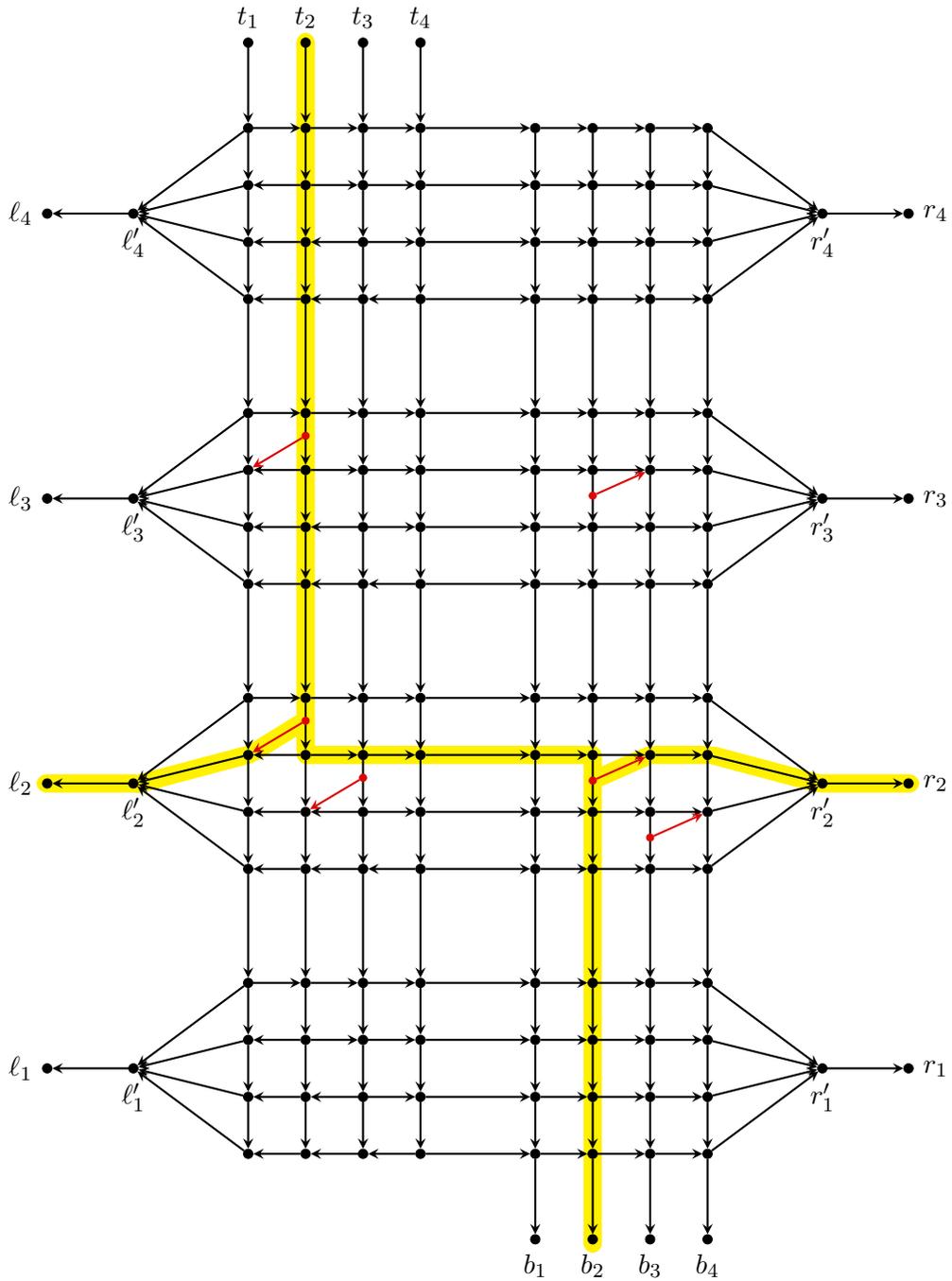
\begin{figure}
\centering
\begin{tikzpicture}[scale=.8]
\node[rectangle, fill=yellow, rounded corners,minimum width=.25cm,minimum height=10.2cm] at (2,14.3) {};
\node[rectangle, fill=yellow, rounded corners,minimum width=4.25cm,minimum height=.25cm] at (4.5,8) {};
\node[rectangle, fill=yellow, rounded corners,minimum width=.25cm,minimum height=7.1cm] at (7,3.7) {};
\node[rectangle, fill=yellow, rounded corners,minimum width=1.5cm,minimum height=.25cm] at (-1.75,7.5) {};
\node[rectangle, fill=yellow, rounded corners,minimum width=1.5cm,minimum height=.25cm] at (11.75,7.5) {};
\node[rectangle, fill=yellow, rounded corners,minimum width=1.88cm,minimum height=.25cm,rotate=15] at (.05,7.75) {};
\node[rectangle, fill=yellow, rounded corners,minimum width=1.74cm,minimum height=.25cm,rotate=-15] at (10,7.75) {};
\node[rectangle, fill=yellow, rounded corners,minimum width=1cm,minimum height=.25cm] at (8.55,8) {};
\node[rectangle, fill=yellow, rounded corners,minimum width=1cm,minimum height=.25cm,rotate=32] at (1.5,8.3) {};
\node[rectangle, fill=yellow, rounded corners,minimum width=1cm,minimum height=.25cm,rotate=24] at (7.55,7.8) {};

\foreach \i in {1,...,4}
\foreach \j in {1,...,4}
{\node[circ] (\i\j) at (\i,\j) {};}

\foreach \i in {1,...,4}
\foreach \j in {6,...,9}
{\pgfmathtruncatemacro{\p}{\j-1}
\node[circ] (\i\p) at (\i,\j) {};}

\foreach \i in {1,...,4}
\foreach \j in {11,...,14}
{\pgfmathtruncatemacro{\p}{\j-2}
\node[circ] (\i\p) at (\i,\j) {};}

\foreach \i in {1,...,4}
\foreach \j in {16,...,19}
{\pgfmathtruncatemacro{\p}{\j-3}
\node[circ] (\i\p) at (\i,\j) {};}

\foreach \i in {6,...,9}
\foreach \j in {1,...,4}
{\pgfmathtruncatemacro{\p}{\i-1}
\node[circ] (\p\j) at (\i,\j) {};}

\foreach \i in {6,...,9}
\foreach \j in {6,...,9}
{\pgfmathtruncatemacro{\p}{\i-1}
\pgfmathtruncatemacro{\q}{\j-1}
\node[circ] (\p\q) at (\i,\j) {};}

\foreach \i in {6,...,9}
\foreach \j in {11,...,14}
{\pgfmathtruncatemacro{\p}{\i-1}
\pgfmathtruncatemacro{\q}{\j-2}
\node[circ] (\p\q) at (\i,\j) {};}

\foreach \i in {6,...,9}
\foreach \j in {16,...,19}
{\pgfmathtruncatemacro{\p}{\i-1}
\pgfmathtruncatemacro{\q}{\j-3}
\node[circ] (\p\q) at (\i,\j) {};}

\foreach \i in {6,...,9}
{\pgfmathtruncatemacro{\p}{\i-5}
\node[circ,label=below:{$b_\p$}] (b\p) at (\i,-.5) {};}

\foreach \i in {1,...,4}
{\node[circ,label=above:{$t_\i$}] (t\i) at (\i,20.5) {};}

\foreach \j in {0,...,3}
{\pgfmathsetmacro{\p}{\j*5+2.5}
\pgfmathtruncatemacro{\t}{\j+1}
\node[circ,label=below:{$\ell'_\t$}] (l'\t) at (-1,\p){};
\node[circ,label=left:{$\ell_\t$}] (l\t) at (-2.5,\p){};
\draw[->,>=stealth,thick] (l'\t) -- (l\t);}

\foreach \j in {0,...,3}
{\pgfmathsetmacro{\p}{\j*5+2.5}
\pgfmathtruncatemacro{\t}{\j+1}
\node[circ,label=below:{$r'_\t$}] (r'\t) at (11,\p){};
\node[circ,label=right:{$r_\t$}] (r\t) at (12.5,\p){};
\draw[->,>=stealth,thick] (r'\t) -- (r\t);}

\foreach \i in {1,...,8} 
\foreach \j in {2,...,16}
{\pgfmathtruncatemacro{\t}{\j-1}
\draw[->,>=stealth,thick] (\i\j) -- (\i\t);}

\foreach \j in {1,...,4}
{\pgfmathtruncatemacro{\a}{(\j-1)*4+1}
\pgfmathtruncatemacro{\b}{\j*4-1}
\foreach \p in {\a,...,\b}
{\pgfmathtruncatemacro{\d}{\p-(\j-1)*4}
\pgfmathtruncatemacro{\c}{5-\d}
\foreach \q in {2,...,\c}
{\pgfmathtruncatemacro{\t}{\q-1}
\draw[->,>=stealth,thick] (\q\p) -- (\t\p);}}}

\foreach \j in {1,...,4}
{\pgfmathtruncatemacro{\a}{(\j-1)*4+2}
\pgfmathtruncatemacro{\b}{\j*4}
\foreach \p in {\a,...,\b}
{\pgfmathtruncatemacro{\d}{\p-(\j-1)*4}
\pgfmathtruncatemacro{\c}{5-\d}
\foreach \q in {\c,...,3}
{\pgfmathtruncatemacro{\t}{\q+1}
\draw[->,>=stealth,thick] (\q\p) -- (\t\p);}}}

\foreach \j in {1,...,16}
\foreach \i in {4,...,7}
{\pgfmathtruncatemacro{\p}{\i+1}
\draw[->,>=stealth,thick] (\i\j) -- (\p\j);}

\foreach \i in {1,...,4}
{\pgfmathtruncatemacro{\t}{\i+4}
\draw[->,>=stealth,thick] (t\i) -- (\i16);
\draw[->,>=stealth,thick] (\t1) -- (b\i);}

\foreach \j in {1,...,4}
{\pgfmathtruncatemacro{\a}{(\j-1)*4+1}
\pgfmathtruncatemacro{\b}{\j*4}
\foreach \p in {\a,...,\b}
{\draw[->,>=stealth,thick] (1\p) -- (l'\j);
\draw[->,>=stealth,thick] (8\p) -- (r'\j);}}

\node[circr] (22l) at (2,8.6) {};
\draw[->,>=stealth,thick,red] (22l) --(17);
\node[circr] (22r) at (7,7.55) {};
\draw[->,>=stealth,thick,red] (22r) -- (77);

\node[circr] (23l) at (2,13.6) {};
\draw[->,>=stealth,thick,red] (23l) --(111);
\node[circr] (23r) at (7,12.55) {};
\draw[->,>=stealth,thick,red] (23r) -- (711);

\node[circr] (32l) at (3,7.6) {};
\draw[->,>=stealth,thick,red] (32l) --(26);
\node[circr] (32r) at (8,6.55) {};
\draw[->,>=stealth,thick,red] (32r) -- (86);
\end{tikzpicture}
\caption{The down main gadget $dMG_S$ with $n=4$ representing $S=\{(2,2),(2,3),(3,2)\}$ (the red edges are the shortcut edges). A set of edges representing $(2,2)$ is highlighted.}
\label{fig:downmaingad}
\end{figure}

\noindent
A set $E \subseteq E(dMG_S)$ satisfies the \emph{connectedness} property if the following hold in $E$: 
\begin{itemize}
\item a top vertex can reach a bottom vertex;
\item a top vertex can reach a right vertex; 
\item a top vertex can reach a left vertex.
\end{itemize}
A set $E \subseteq E(dMG_S)$ satisfying the connectedness property \emph{represents} a pair $(i,j) \in [n] \times [n]$ if the only source edge in $E$ is the one incident to $t_i$, the only bottom sink edge in $E$ is the one incident to $b_i$, the only left sink edge in $E$ is the one incident to $\ell_j$ and the only right sink edge in $E$ is the one incident to $r_j$ (see \Cref{fig:downmaingad} for a set of edges representing $(2,2)$).

\begin{lemma}
\label{lem:dmg}
For any $n > 0$ and any $S \subseteq [n] \times [n]$, the down main gadget $dMG_S$ satisfies the following properties.
\begin{itemize}
\item[(1)] For every $(i,j) \in S$, there exists a set $E_{i,j} \subseteq E(dMG_S)$ of weight $M_n^*$ representing $(i,j)$.
\item[(2)] If there exists a set $E \subseteq E(dMG_S)$ of weight at most $M_n^*$ satisfying the connectedness property, then $E$ has weight exactly $M_n^*$ and represents a pair $(i,j) \in S$.
\end{itemize}
\end{lemma}

\begin{proof}
To prove (1), it suffices to take $E_{i,j}$ to be the union of the following set of edges:
\begin{itemize}
\item $\{(t_i,x_{i,n^2}),(x_{n+i,1},b_i),(x_{0,jn+1-i},\ell'_j),(\ell'_j,\ell_j),(x_{2n+1,jn+1-i},r'_j),(r'_j,r_j)\}$;
\item $\{(x_{i,p},x_{i,p-1})~|~ jn+2-i\leq p \leq n^2\} \cup \{(x_{n+i,p},x_{n+i,p-1})~|~2 \leq p \leq jn+1-i\}$;
\item $\{(x_{p,jn+1-i},x_{p-1,jn+1-i})~|~ 1 \leq p \leq i-1\} \cup \{e^\ell_{(i,j)}\}$;
\item $\{(x_{p,jn+1-i},x_{p+1,jn+1-i})~|~i \leq p \leq 2n \text { and } p \neq n+i\} \cup \{e^r_{(i,j)}\}$.
\end{itemize}
It is not difficult to see that $E_{i,j}$ represents $(i,j)$ and has weight $M_n^*$.\\

Next, suppose that $E \subseteq E(dMG_S)$ is a set of weight at most $M_n^*$ satisfying the connectedness property. Let us show that $E$ has weight exactly $M_n^*$ and represents some $(i,j) \in S$. To this end, we first prove the following claims.

\begin{claim}
\label{clm:sourcesinkedges2}
$E$ contains exactly one source edge and one bottom sink edge.
\end{claim}

\begin{claimproof}
Since $E$ satisfies the connectedness property, it contains at least one source edge and one bottom sink edge. Now if $E$ contains at least two source edges, then the weight of $E$ is at least $3M^5$; however, by definition,
\begin{equation*}
\begin{split}
M_n^* &=  2M^5 + M^4 + (n-1)M^3 + M^2 + (4n+1)(n+1) - 12\\ 
& < 2M^5 + n^2M^4 + n^2M^4 + n^2M^4 + 9n^2M^4 < 3M^5
\end{split}
\end{equation*}
as $M = 13n^2 > 12n^2$, a contradiction. We conclude similarly if $E$ contains at least two bottom sink edges.
\end{claimproof}

In the following, we let $t_{i_1}$ be the top vertex incident to the source edge in $E$ and $b_{i_2}$ be the bottom vertex incident to the bottom sink edge in $E$. 

\begin{claim}
\label{clm:rightbridge}
$E$ contains exactly one right bridge edge.
\end{claim}

\begin{claimproof}
Since $t_{i_1}$ can reach $b_{i_2}$ in $E$, $E$ contains at least one right bridge edge. Now if $E$ contains at least two right bridge edges then by \Cref{clm:sourcesinkedges2}, the weight of $E$ is at least $2M^5 + 2M^4$; however, by definition,
\begin{equation*}
\begin{split}
M_n^* &= 2M^5 + M^4 + (n-1)M^3 + M^2 + (4n+1)(n+1) - 12\\
& < 2M^5 + M^4 + n^2M^4 + n^2M^4 + 9n^2M^4 \\
& < 2M^5 + 2M^4
\end{split}
\end{equation*}
as $M = 13n^2 > 11n^2$, a contradiction.
\end{claimproof}

\begin{claim}
\label{clm:downwardbridges}
For every $j \in [n-1]$, $E$ contains exactly one $j^{th}$ downward bridge edge. In particular, $E$ contains exactly $n-1$ downward bridge edges.
\end{claim}

\begin{claimproof}
Since $t_{i_1}$ can reach $b_{i_2}$ in $E$, $E$ contains at least one $j^{th}$ downward bridge edge for every $j \in [n-1]$. Now if $E$ contains at least $n$ downward bridge edges then by Claims~\ref{clm:sourcesinkedges2} and \ref{clm:rightbridge}, the weight of $E$ is at least $2M^5 + M^4 + nM^3$; however, by definition,
\begin{equation*}
\begin{split}
M_n^* &=  2M^5 + M^4 + (n-1)M^3 + M^2 + (4n+1)(n+1) - 12\\
& < 2M^5 + M^4 + (n-1)M^3 + n^2M^2 + 9n^2M^2 \\
& < 2M^5 + M^4 + nM^3
\end{split}
\end{equation*}
as $M = 13n^2 > 10n^2$, a contradiction.
\end{claimproof}

Since $E$ satisfies the connectedness property, it contains at least one left sink edge $(\ell'_{j_1},\ell_{j_1})$ and at least one right sink edge $(r'_{j_2},r_{j_2})$; we next show that $E$ contains in fact no other left or right sink edge.

\begin{claim}
\label{clm:j1=j2}
It holds that $j_1 = j_2$. Furthermore, $E$ contains a $j_1^{th}$ right bridge edge.
\end{claim}

\begin{claimproof}
Towards a contradiction, suppose that $j_1 \neq j_2$.
Assume first that $j_1 < j_2$ and let $P$ be a path from $t_{i_1}$ to $r_{j_2}$ in $dMG_S[E]$. Then $P$ must contain a $j^{th}$ right bridge edge $e$ for some $j \geq j_2$, and since $E$ contains exactly one right bridge edge by \Cref{clm:rightbridge}, it follows that every path in $dMG_S[E]$ from $t_{i_1}$ to $b_{i_2}$ contains $e$. But $t_{i_1}$ can reach $\ell_{j_1}$ in $E$ as well and so, $E$ must contain at least $n-j$ downward bridge edges for $t_{i_1}$ to reach the tail of $e$, plus $j -1$ downward bridge edges for the head of $e$ to then reach $b_{i_2}$, plus $j-j_1$ downward bridge edges for $t_{i_1}$ to also reach $\ell_{j_1}$, a contradiction to \Cref{clm:downwardbridges}. Suppose next that $j_2 < j_1$. Then since $E$ contains the edge $(\ell'_{j_1},\ell_{j_1})$ and the edge $(r'_{j_2},r_{j_2})$, it follows from Claims~\ref{clm:sourcesinkedges2}, \ref{clm:rightbridge} and \ref{clm:downwardbridges} that the weight of $E$ is at least $2M^5 + M^4 + (n-1)M^3 + Mj_1 + M^2 - Mj_2 \geq 2M^5 + M^4 + (n-1)M^3 + M^2 + M$. However, by definition,
\begin{equation*}
\begin{split}
M_n^* &= 2M^5 + M^4 + (n-1)M^3 + M^2 + (4n+1)(n+1) - 12\\
& <  2M^5 + M^4 + (n-1)M^3 + M^2 + 9n^2\\
& < 2M^5 + M^4 + (n-1)M^3 + M^2 + M
\end{split}
\end{equation*}
as $M = 13n^2 > 9n^2$, a contradiction. Thus, $j_1 = j_2$ and by arguing as above, we can show that $E$ must then contain a $j_1^{th}$ right bridge edge.
\end{claimproof}

In the following, we say that a path $P$ in $dMG_S$ \emph{crosses a row} $j$ if $V(P) \cap \{x_{i,j}~|~i \in [2n]\} \neq \emptyset$ and that $P$ \emph{crosses a column} $i$ if $V(P) \cap \{x_{i,j}~|~j \in [n^2]\} \neq \emptyset$. Furthermore, we define
\begin{itemize}
\item a \emph{$2d$-move} to be the union of two down subdivided edges $e,e'$ such that the head of $e$ coincides with the tail of $e'$;
\item a \emph{$d\ell$-move} to be the union of a down subdivided edge $e$ and a left shortcut edge $e'$ such that the head of $e$ coincides with the tail of $e'$;
\item a \emph{$dr$-move} to be the union of a down subdivided edge $e$ and a right shortcut edge $e'$ such that the head of $e$ coincides with the tail of $e'$. 
\end{itemize}
For convenience, we also call
\begin{itemize}
\item a right edge an \emph{$r$-move},
\item a left edge an \emph{$\ell$-move},
\item a down edge a \emph{$d$-move} and
\item a downward bridge edge a \emph{$Db$-move}.
\end{itemize}
Note that no $dr$-move is possible in $dMG_S[V_\ell]$ and that no $x$-move with $x \in \{d\ell,\ell\}$ is possible in $dMG_S[V_r]$. Furthermore, for any $x \in \{2d,r,\ell,d\}$, the weight of an $x$-move is 4 and for any $x \in \{dr,d\ell\}$, the weight of an $x$-move is 5. Given a path $P$ of $dMG_S$ and $x \in \{2d,d\ell,dr,r,\ell,d,Db\}$, we let $m(P,x)$ be the number of $x$-moves in $P$.\\

Let $(x_{1,j_0},\ell'_{j_1})$ be a $j_1^{th}$ left internal sink edge contained in $E$ (recall that $\ell_{j_1} \in V(E)$). Further let $x_{i_0,j_1n}$ be the head of the $j_1^{th}$ downward bridge in $E$ if $j_1 \leq n-1$ (recall that by \Cref{clm:downwardbridges}, $E$ contains exactly one such edge) and the head of the source edge in $E$ otherwise (note that in this case $i_0 = i_1$). Since $x_{i_0,j_1n}$ can reach $x_{1,j_0}$, the following holds.

\begin{observation}
\label{obs:i0j0}
$j_0 - (j_1-1)n \leq n+1-i_0$.
\end{observation}

Now by connectedness of $E$, there exists a path $P^t$ in $dMG_S[E]$ from the head of the source edge in $E$ to $x_{i_0,j_1n}$; we next lower bound the weight of $P^t$.

\begin{claim}
\label{clm:weightPt}
The weight of $P^t$ is at least $(n-j_1)M^3 + 4(n-1)(n-j_1) + |i_1 - i_0|$.
\end{claim}

\begin{claimproof}
Observe first that since by Claims~\ref{clm:rightbridge} and \ref{clm:j1=j2}, $E$ contains one $j_1^{th}$ right bridge edge and no other right bridge edge, necessarily $i_0 \leq n$. Now $P^t$ must cross every row $j_1n \leq j \leq n^2$ and thus, 
\[
m(P^t,Db) + m(P^t,d) + m(P^t,2d) + m(P^t,d\ell) = n^2 - j_1n.
\]
By \Cref{clm:downwardbridges}, $P^t$ contains exactly $n-j_1$ downward bridge edges, that is, $m(P^t,Db) = n - j_1$. We next distinguish cases depending on whether $i_1 \leq i_0$ or $i_1 > i_0$. Suppose first that $i_1 \leq i_0$. Then since $P^t$ must cross every column $i_1 \leq i \leq i_0$, $m(P^t,r) \geq i_0 - i_1$. It follows that the weight of $P^t$ is at least 
\begin{equation*}
\begin{split}
&~M^3m(P^t,Db) + 4m(P^t,d) + 4m(P^t,2d) + 5m(P^t,d\ell) + 4m(P^t,r) \\
& = (n-j_1)M^3 + 4m(P^t,d) + 4m(P^t,2d) + 5m(P^t,d\ell) + 4m(P^t,r) \\
& \geq (n-j_1)M^3 + 4(m(P^t,d) + m(P^t,2d) + m(P^t,d\ell)) + 4(i_0 - i_1) \\
& \geq (n-j_1)M^3 + 4(n^2 - j_1n - (n-j_1)) + 4(i_0 - i_1)\\
& \geq (n-j_1)M^3 + 4(n-1)(n-j_1) + (i_0 - i_1)
\end{split}
\end{equation*}
Second, suppose that $i_1 > i_0$. Then since $P^t$ must cross every column $i_0 \leq i \leq i_1$,
\[
m(P^t,\ell) + m(P^t,d\ell) = m(P^t,r) + i_1 - i_0
\]
and so, the weight of $P^t$ is at least
\begin{equation*}
\begin{split}
&~M^3m(P^t,Db) + 4m(P^t,d) + 4m(P^t,2d) + 5m(P^t,d\ell) + 4m(P^t,\ell) + 4m(P^t,r) \\
& = (n-j_1)M^3 + 4m(P^t,d) + 4m(P^t,2d) + 5m(P^t,d\ell) + 4m(P^t,\ell) + 4m(P^t,r) \\
& \geq (n-j_1)M^3 + 4(n^2 - j_1n - (n-j_1)) + m(P^t,d\ell) + 4m(P^t,\ell) + 4(m(P^t,\ell)\\
&+m(P^t,d\ell) - (i_1 - i_0))\\
& \geq (n-j_1)M^3 + 4(n-1)(n-j_1) + 5m(P^t,d\ell) + 8m(P^t,\ell) - 4(i_1 - i_0)\\
& \geq (n-j_1)M^3 + 4(n-1)(n-j_1) + 5(m(P^t,d\ell) + m(P^t,\ell)) - 4(i_1 - i_0)\\
& \geq (n-j_1)M^3 + 4(n-1)(n-j_1) + (i_1 - i_0)\\
\end{split}
\end{equation*}
as $m(P^t,d\ell) + m(P^t,\ell) \geq i_1 - i_0$ which proves our claim.
\end{claimproof}

Let $(x_{2n,j'_0},r'_{j_1})$ be a $j_1^{th}$ right internal sink edge contained in $E$ (recall that $r_{j_2} \in V(E)$ and $j_2 = j_1$ by \Cref{clm:j1=j2}). Further let $x_{i'_0,(j_1-1)n+1}$ be the tail of the $(j_1-1)^{th}$ downward bridge in $E$ if $j_1 >1$ (recall that by \Cref{clm:downwardbridges}, $E$ contains exactly one such edge) and the tail of the bottom sink edge contained in $E$ otherwise (note that in this case $i'_0 = n + i_2$). Then by connectedness of $E$, there exists a path $P^b$ in $dMG_S[E]$ from $x_{i'_0,(j_1-1)n+1}$ to the tail of the bottom sink edge in $E$; we next lower bound the weight of $P^b$.

\begin{claim}
\label{clm:weightPb}
The weight of $P^b$ is at least $(j_1-1)M^3 + 4(n -1)(j_1-1) + 4(n + i_2 - i'_0)$.
\end{claim}

\begin{claimproof}
Observe first that since by Claims~\ref{clm:rightbridge} and \ref{clm:j1=j2}, $E$ contains a $j_1^{th}$ right bridge edge and no other right bridge edge, necessarily $i'_0 \geq n+1$. Furthermore, since $dMG_S[V_r]$ contains no edge $r$-move, the following holds.

\begin{observation}
\label{obs:i'0}
$i'_0 \leq n+i_2$. 
\end{observation}

\noindent
Now $P^b$ must cross every column $i'_0 \leq i \leq n + i_2$ and every row $1 \leq j \leq (j_1-1)n+1$; thus 
\[
m(P^b,Db) + m(P^b,d) + m(P^b,2d) = (j_1-1)n \text{ and } m(P^b,r) + m(P^b,dr) = n + i_2 - i'_0.
\]
Since by \Cref{clm:downwardbridges}, $P^b$ contains exactly $j_1-1$ downward bridge edges, it follows that the weight of $P^b$ is at least 
\begin{equation*}
\begin{split}
&~M^3m(P^b,Db) + 4m(P^b,d) + 4m(P^b,2d) + 4m(P^b,r) + 5m(P^b,dr)\\
&= (j_1-1)M^3 + 4(m(P^b,d) + m(P^b,2d)) + 4m(P^b,r) + 5m(P^b,dr)\\
& \geq (j_1-1)M^3 + 4((j_1-1)n - (j_1-1)) + 4(m(P^b,r) + m(P^b,dr))\\
& \geq (j_1-1)M^3 + 4(j_1-1)(n - 1) + 4(n + i_2 - i'_0)
\end{split}
\end{equation*}
which proves our claim.
\end{claimproof}

Let $x_{n+1,t}$ be the head of the $j_1^{th}$ right bridge edge in $E$ (recall that by Claims~\ref{clm:rightbridge} and \ref{clm:j1=j2}, $E$ contains exactly one $j_1^{th}$ right bridge edge and no other right bridge edge). Then by connectedness of $E$, there exist in $dMG_S[E]$ a path $P$ and a path $P'$ from $x_{i_0,j_1n}$ to $x_{1,j_0}$ and from $x_{i_0,j_1n}$ to $x_{n,t}$ respectively; we next lower bound the weight of $P \cup P'$. To this end, denote by $x_{i',t'}$ the last vertex in $V(P)$ belonging to $V(P') \cap V_{j_1}$. 

\begin{observation}
\label{obs:i't'}
$\max\{t,j_0\} \leq t'$ and $n+1-i' \leq t' - (j_1-1)n$.
\end{observation}

Indeed, since $x_{i',t'}$ can reach both $x_{1,j_0}$ and $x_{n,t}$, $t' \geq j_0$ and $t' \geq t$. Now if $n+1-i' > t' - (j_1-1)n$ then by \Cref{obs:reach}, $x_{i',t'}$ cannot reach $x_{n,t}$ as $t - (j_1-1)n \geq 1 = n+1 -n$, a contradiction. 

\begin{observation}
\label{obs:i0i'}
$i_0 \leq i'$.
\end{observation}

Indeed, if $i_0 > i'$ then, in particular, $i_0 > 1$ and so, $n+1-i_0 < n = j_1n - (j_1-1)n$; but then, by \Cref{obs:reach}, $x_{i_0,j_1n}$ cannot reach $x_{i',t'}$ as $n+1-i_0 < n+1 - i' \leq t' - (j_1-1)n$ by \Cref{obs:i't'}, a contradiction. 

\begin{claim}
\label{clm:weightPP'}
The weight of $P \cup P'$ is at least
\[
4(n-1) + 4(t'-t) + 4(j_1n-j_0) + 4(i'-i_0) - 3(p_1 + m(P,d\ell))
\]
where 
$$
p_1 = \left\{
    \begin{array}{ll}
        1 & \mbox{if } (i',j_1) \in S, ~e^\ell_{i',j_1} \in E(P) \mbox{ and } x_{i',t'-1} \in V(P') \\
        0 & \mbox{otherwise.}
    \end{array}
\right.
$$
\end{claim}

\begin{claimproof}
Denote by $e$ the down or down subdivided edge with tail $x_{i',t'}$ (note that $e$ is a down subdivided edge if and only if $(i',j_1) \in S$). Then by definition of $x_{i',t'}$, the weight of $P \cup P'$ is at least the sum of the weights of $P'$ and $P[x_{i',t'},x_{1,j_0}]$ minus the weight of $e$ whenever $(i',j_1) \in S$ and $e \in E(P') \cap E(P[x_{i',t'},x_{1,j_0}])$. Let us therefore lower bound the weights of $P'$ and $P[x_{i',t'},x_{1,j_0}]$.

First note that since for any left shortcut edge $f$ with head $x_{p,q} \in V_j$, it holds that $n+1-p > q-(j-1)n$, $x_{p,q}$ cannot reach any vertex $x_{p',q'} \in V_j$ such that $n+1-p' \leq q' - (j-1)n$ by \Cref{obs:reach}. Further note that since the tail of $f$ can only be reached through $x_{p-1,q-1}$ where $n+1 - (p-1) < (q-1) - (j-1)n$, no vertex $x_{p',q'} \in V_j$ such that $n+1 - p' > q' - (j-1)n$ can reach the tail of $f$ by \Cref{obs:reach}. It follows that $P'$ and $P[x_{i_0,j_1n},x_{i',t'}]$ contain no left shortcut edge (recall that $n+1-i' \leq t' - (j_1-1)n$ by \Cref{obs:i't'}) and that $P[x_{i',t'},x_{1,j_0}]$ contains at most one left shortcut edge, that is, the following holds.

\begin{observation}
\label{obs:Pdl}
$m(P',d\ell) = m(P[x_{i_0,j_1n},x_{i',t'}], d\ell) = 0$ and $m(P,d\ell) \in \{0,1\}$. 
\end{observation}

\noindent
Now $P'$ must cross every row $t \leq j \leq j_1n$ and $P[x_{i',t'},x_{1,j_0}]$ must cross every row $j_0 \leq j \leq t'$ (recall that $\max\{t,j_0\} \leq t'$ by \Cref{obs:i't'}); thus
\begin{equation*}
\begin{split}
& m(P',d) + m(P',2d) = j_1n-t \text{ and } \\
& m(P[x_{i',t'}.x_{1,j_0}],d) + m(P[x_{i',t'}.x_{1,j_0}],2d) + m(P[x_{i',t'}.x_{1,j_0}],d\ell) = t'-j_0.
\end{split}
\end{equation*}
Similarly, $P'$ must cross every column $i_0 \leq i \leq n$ and $P[x_{i',t'},x_{1,j_0}]$ must cross every column $1 \leq i \leq i'$ (recall that $i_0 \leq i'$ by \Cref{obs:i0i'}); thus
\begin{equation*}
\begin{split}
& m(P',\ell) + m(P',r) \geq n-i_0 \text{ and }\\
& m(P[x_{i',t'}.x_{1,j_0}],\ell) + m(P[x_{i',t'}.x_{1,j_0}],d\ell) \geq i' - 1.
\end{split}
\end{equation*}
Now note that $p_1=1$ if and only if $(i',j_1) \in S$ and $e$ belongs to both $P$ and $P[x_{i',t'},x_{1,j_0}]$: indeed, if $p_1 =1$ then clearly $(i',j_1) \in S$ and $e \in E(P) \cap E(P[x_{i',t'},x_{1,j_0}])$. Conversely, if $(i',j_1) \in S$ and $e$ belongs to both $P$ and $P[x_{i',t'},x_{1,j_0}]$ then since by definition of $x_{i',t'}$, $x_{i',t'-1}$ cannot belong to both $P'$ and $P[x_{i',t'},x_{1,j_0}]$, and $P'$ contains no left shortcut edge by \Cref{obs:Pdl}, it must be that $P'$ contains $x_{i',t'-1}$ while $P[x_{i',t'},x_{1,j_0}]$ contains the left shortcut edge $e^\ell_{(i',j_1)}$, that is, $p_1=1$. It follows that the weight of $P \cup P'$ is at least
\begin{equation*}
\begin{split}
& 4(m(P',d) + m(P',2d)) + 4(m(P',\ell) + m(P',r)) + 4(m(P[x_{i',t'}.x_{1,j_0}],d) + m(P[x_{i',t'}.x_{1,j_0}],2d)) \\
+&~5(m(P[x_{i',t'}.x_{1,j_0}],d\ell) - p_1) + 2p_1 + 4m(P[x_{i',t'}.x_{1,j_0}],\ell)\\
\geq & ~4(j_1n-t) + 4(n-i_0) + 4(t'-j_0) + m(P[x_{i',t'}.x_{1,j_0}],d\ell) - 3p_1 \\
+&~4(i'-1 - m(P[x_{i',t'}.x_{1,j_0}],d\ell)) \\
\geq & ~4(n-1) + 4(t'-t) + 4(j_1n- j_0) + 4(i'-i_0) - 3(p_1+m(P[x_{i',t'}.x_{1,j_0}],d\ell))
\end{split}
\end{equation*}
which proves our claim as $m(P[x_{i',t'}.x_{1,j_0}],d\ell)=m(P,d\ell)$ by \Cref{obs:Pdl}.
\end{claimproof}

Note that if $p_1 = 1$ then, since $x_{i'-1,t'-1} \in V(P)$ can reach $x_{1,j_0}$ and $x_{i',t'-1} \in V(P')$ can reach $x_{n,t}$, the following holds.

\begin{observation}
\label{obs:p1}
If $p_1 = 1$ then $j_0 \leq t'-1$ and $t \leq t'-1$.
\end{observation} 

Now by connectedness of $E$, there exist in $dMG_S[E]$ a path $Q$ and a path $Q'$ from $x_{n+1,t}$ to $x_{i'_0,(j_1-1)n+1}$ and from $x_{n+1,t}$ to $x_{2n,j'_0}$, respectively; we next lower bound the weight of $Q \cup Q'$. To this end, denote by $x_{i'',t''}$ the last vertex in $V(Q')$ belonging to $V(Q) \cap V_{j_1}$. 

\begin{claim}
\label{clm:weightQQ'}
The weight of $Q \cup Q'$ is at least 
\[
4(n-2) + 4(i'_0-i'') + 4(t''-j'_0) + 4(t-(j_1-1)n) - 2p_2
\]
where 
$$
p_2 = \left\{
    \begin{array}{ll}
        1 & \mbox{if } (i''-n,j_1) \in S \mbox{ and } e^r_{i''-n,j_1} \in E(Q'),\\
        0 & \mbox{otherwise.}
    \end{array}
\right.
$$
\end{claim}

\begin{claimproof}
First note that since $x_{n+1,t}$ can reach $x_{i'',t''}$, and $x_{i'',t''}$ can reach both $x_{2n,j'_0}$ and $x_{i'_0,(j_1-1)n+1}$, the following holds.

\begin{observation}
\label{obs:t''}
$j'_0 \leq t'' \leq t$ and $i'' \leq i'_0$. 
\end{observation}

Now denote by $e$ the down or down subdivided edge with tail $x_{i'',t''}$ (note that $e$ is a down subdivided edge if and only if $(i''-n,j_1) \in S$). Then by definition of $x_{i'',t''}$, the weight of $Q \cup Q'$ is at least the sum of the weights of $Q$ and $Q'[x_{i'',t''},x_{2n,j'_0}]$ minus the weight of $e$ whenever $(i''-n,j_1) \in S$ and $e \in E(Q) \cap E(Q'[x_{i'',t''},x_{2n,j'_0}])$. Let us therefore lower bound the weights of $Q$ and $Q'[x_{i'',t''},x_{2n,j'_0}]$. 

Since $Q$ must cross every row $(j_1-1)n+1 \leq j \leq t$ and $Q'[x_{i'',t''},x_{2n,j'_0}]$ must cross every row $j'_0 \leq j \leq t''$ (recall that $j'_0 \leq t'' \leq t$ by \Cref{obs:t''}), it follows that
\begin{equation*}
\begin{split}
& m(Q,d) + m(Q,2d) = t-((j_1-1)n+1) \text{ and}\\
& m(Q'[x_{i'',t''},x_{2n,j'_0}],d) + m(Q'[x_{i'',t''},x_{2n,j'_0}],2d) = t'' - j'_0.
\end{split}
\end{equation*} 
Similarly, $Q$ must cross every column $n+1 \leq i \leq i'_0$ and $Q'[x_{i'',t''},x_{2n,j'_0}]$ must cross every column $i'' \leq i \leq 2n$ (recall that $i'' \leq i'_0$ by \Cref{obs:t''}); thus
\begin{equation*}
\begin{split}
& m(Q,r) + m(Q,dr) = i'_0 - (n+1) \text{ and}\\
& m(Q'[x_{i'',t''},x_{2n,j'_0}],r) + m(Q'[x_{i'',t''},x_{2n,j'_0}],dr) = 2n-i''.
\end{split}
\end{equation*}
Now note that $p_2 =1$ if and only if $(i''-n,j_1) \in S$ and $e$ belongs to both $Q$ and $Q'$: indeed, if $p_2=1$ then by definition $(i''-n,j_1) \in S$, $e \in E(Q')$ and $x_{i''+1,t''} \in V(Q')$ which implies, by definition of $x_{i'',t''}$, that $x_{i''+1,t''} \notin V(Q)$ and so, $e \in E(Q)$. Conversely, if $(i''-n,j_1) \in S$ and $e$ belongs to both $Q$ and $Q'$ then, by definition of $x_{i'',t''}$, it must be that $x_{i'',t''-1}$ belongs to one of $Q$ and $Q'$ while $x_{i''+1,t''}$ belongs to the other. However, if $x_{i''+1,t''}$ belongs to $Q$ (and so, $x_{i'',t''-1} \in V(Q')$) then, since $Q$ crosses row $(j_1-1)n+1$ and $Q'$ crosses column $2n$, necessarily $Q[x_{i''+1,t''},x_{i'_0,(j_1-1)n+1}] \cap Q'[x_{i'',t''-1},x_{2n,j'_0}] \cap V_{j_1} \neq \emptyset$, a contradiction to the definition of $x_{i'',t''}$. Thus, $x_{i''+1,t''} \in V(Q')$ and so, $e^r_{(i''-n,t'')} \in E(Q')$, that is, $p_2=1$. It follows that the weight of $Q \cup Q'$ is at least
\begin{equation*}
\begin{split}
&  4(m(Q,d) + m(Q,2d)) + 4m(Q,r) + 5m(Q,dr) + 4(m(Q'[x_{i'',t''},x_{2n,j'_0}],d) \\ 
+&~m(Q'[x_{i'',t''},x_{2n,j'_0}],2d))+4m(Q'[x_{i'',t''},x_{2n,j'_0}],r) + 5(m(Q'[x_{i'',t''},x_{2n,j'_0}],dr) - p_2) + 2p_2\\
\geq &~4(t-((j_1-1)n+1)) + 4(m(Q,r) + m(Q,dr)) + 4(t''-j'_0) + 4(2n-i'')\\ 
+&~ m(Q'[x_{i'',t''},x_{2n,j'_0}],dr) - 3p_2\\
\geq &~4(t-((j_1-1)n+1)) + 4(i'_0-(n+1)) + 4(t'' - j'_0) + 4(2n - i'')\\
+&~m(Q'[x_{i'',t''},x_{2n,j'_0}],dr) - 3p_2\\
\geq & ~4(n-2) + 4(i'_0-i'') + 4(t''-j'_0) + 4(t-(j_1-1)n) - 2p_2
\end{split}
\end{equation*}
which proves our claim.
\end{claimproof}

\noindent
In the following, we write $j_0 = (j_1-1)n+p_0$ and $j'_0 = (j_1-1)n+p'_0$ and further let 
\[
W = |i_1-i_0| + 4(n+i_2-i'') + 4(i'-i_0) + 4(t'-j_0) + 4(t''-j'_0) - (8 + 3(p_1 + m(P,d\ell)) + 2p_2).
\]
Observe that by Claims~\ref{clm:weightPt}, \ref{clm:weightPb}, \ref{clm:weightPP'} and \ref{clm:weightQQ'}, the weight of the union of $P^t,P^b,P,P',Q$ and $Q'$ is at least
\[
(n-1)M^3 + 4n(n+1) + W
\]
and so, by Claims~\ref{clm:sourcesinkedges2}, \ref{clm:rightbridge}, \ref{clm:downwardbridges} and \ref{clm:j1=j2}, the weight of $E$ is at least
\begin{equation}
\label{eq:weightE}
2M^5 + M^4 + (n-1)M^3 + M^2 + (4n+1)(n+1) + p_0 - p'_0 + W.
\end{equation}

\begin{claim}
\label{clm:i1=i2}
The following hold.
\begin{itemize}
\item[(1)] $i_1 = i_2$ and $(i_1,j_1) \in S$. 
\item[(2)] $W=-12$ and $p_0=p'_0$.
\end{itemize}
\end{claim}

\begin{claimproof}
Since the weight of $E$ is at most $M_n^*$, it follows from \Cref{eq:weightE} that
\begin{equation*}
\begin{split} 
M_n^* & = 2M^5 + M^4 + (n-1)M^3 + M^2 + (4n+1)(n+1) - 12\\
& \geq 2M^5 + M^4 + (n-1)M^3 + M^2 + (4n+1)(n+1) + p_0-p'_0 + W.
\end{split}
\end{equation*}
Thus, by definition of $W$,
\begin{equation}
\label{eq:upperbound}
\begin{split}
& p_0 - p'_0 + |i_1-i_0| + 4(n+i_2-i'') + 4(i'-i_0) + 4(t'-j_0) + 4(t''-j'_0)\\
&\leq -4 + 3(p_1 + m(P,d\ell)) + 2p_2 \\
& \leq 4
\end{split}
\end{equation}
\noindent
as $p_1,p_2 \in \{0,1\}$ by definition and $m(D,d\ell) \in \{0,1\}$ by \Cref{obs:Pdl}. Now suppose to the contrary that $p'_0 > p_0$ or, equivalently, that $j'_0 > j_0$. Then since $t' \geq t \geq t'' \geq j'_0$ by Observations~\ref{obs:i't'} and \ref{obs:t''}, $p'_0 - p_0 = j'_0 - j_0 \leq t' - j_0$ and so,
\begin{equation*}
\begin{split}
& p_0 - p'_0 + 4(t'-j_0) + |i_1-i_0| + 4(n+i_2-i'') + 4(i'-i_0)+ 4(t''-j'_0)\\
&\geq 3(t'-j_0) + |i_1-i_0| + 4(n+i_2-i'') + 4(i'-i_0) + 4(t''-j'_0).
\end{split}
\end{equation*} 
But $n+i_2 \geq i''$, $i' \geq i_0$ and $t' \geq t'' \geq j'_0 > j_0$ by Observations~\ref{obs:i'0}, \ref{obs:i't'}, \ref{obs:i0i'} and \ref{obs:t''}, hence
\begin{equation*}
3(t'-j_0) + |i_1-i_0| + 4(n+i_2-i'') + 4(i'-i_0) + 4(t''-j'_0) \geq 3.
\end{equation*} 
It then follows from \Cref{eq:upperbound} and the above that $p_1 = m(P,d\ell) = p_2=1$ and $t' = j_0 + 1$; however, by \Cref{obs:p1}, if $p_1 = 1$ then $t' \geq t+1$ and so, by \Cref{obs:i't'}, $j'_0 \leq t \leq t'-1 = j_0$, a contradiction. Thus, $p_0 \geq p'_0$ and so,
\begin{equation}
\label{eq:positif}
 p_0 - p'_0 + |i_1-i_0| + 4(n+i_2-i'') + 4(i'-i_0) + 4(t'-j_0) + 4(t''-j'_0) \geq 0
\end{equation}
\noindent 
by Observations~\ref{obs:i'0}, \ref{obs:i't'}, \ref{obs:i0i'}  and \ref{obs:t''}. Now suppose for a contradiction that $p_1=0$. Then by \Cref{eq:upperbound},
\[
 p_0 - p'_0 + |i_1-i_0| + 4(n+i_2-i'') + 4(i'-i_0) + 4(t'-j_0) + 4(t''-j'_0) \leq 1
\]
and so, by Observations~\ref{obs:i'0}, \ref{obs:i't'}, \ref{obs:i0i'} and \ref{obs:t''}, we must have $n+i_2 = i''$, $t' = j_0$, $t'' = j'_0$ and $i_0 = i'$. But then, by Observations~\ref{obs:i0j0} and \ref{obs:i't'}, $p_0 = n+1-i_0$ which implies that $P$ contains no left shortcut edge as $t'=j_0$; thus, by \Cref{eq:upperbound},  
\[
 p_0 - p'_0 + |i_1-i_0| + 4(n+i_2-i'') + 4(i'-i_0) + 4(t'-j_0) + 4(t''-j'_0) \leq -2,
\]
as $m(P,d\ell) = 0$, a contradiction to \Cref{eq:positif}. Hence, $p_1 = 1$ and so, $m(D,d\ell) = 1$ as well. Now by \Cref{obs:p1}, $t' \geq j_0+1$ and so, by \Cref{eq:upperbound},
\[
4 \leq 4(t'-j_0) \leq p_0 - p'_0 + |i_1-i_0| + 4(n+i_2-i'') + 4(i'-i_0)+ 4(t'-j_0) + 4(t''-j'_0) \leq 2 + 2p_2.
\]
Therefore, $p_2 = 1$ and $t'-1=j_0=j'_0=t''$, $i_1=i_0=i'$ and $n+i_2=i''$ by Observations~\ref{obs:i'0}, \ref{obs:i't'}, \ref{obs:i0i'} and \ref{obs:t''}; in particular, $W = -12$. Now since $p_1 = 1$, $(i',j_1) \in S$ by definition and so, $t' = j_1n+2 - i'$ by construction; and since $p_2 = 1$, $(i''-n,j_1) \in S$ by definition and so, $t''= j_1n+1-(i''-n)$ by construction. As $t' -1 = t''$, $i' = i_1$ and $i_2 = i''-n$, we conclude that $i_1 = i_2$ and $(i_1,j_1) \in S$. 
\end{claimproof}

It now follows from Claims~\ref{clm:j1=j2} and \ref{clm:i1=i2}(1) that $E$ represents the pair $(i_1,j_1) \in S$; and by combining \Cref{eq:weightE} with \Cref{clm:i1=i2}(2), we conclude that the weight of $E$ is at least $M_n^*$ and thus exactly $M_n^*$ which completes the proof.
\end{proof}


\noindent
\textbf{Up Main Gadget.} Given an integer $n > 0$, the up main gadget $uMG_S$ represents a set $S \subseteq [n] \times [n]$\footnote{Recall that, by assumption, $1 < x,y< n$ holds for every $(x,y) \in S$.} and is constructed as follows. It is an edge-weighted planar digraph consisting of a $2n \times n^2$ grid, where the vertex lying at the intersection of column $i$ and row $j$ is denoted by $x_{i,j}$, and $6n$ additional vertices $\ell_1,\ldots,\ell_n,\ell'_1,\ldots,\ell'_n,r_1,\ldots,r_n,$ $r'_1,\ldots,r'_n,t_1,\ldots,t_n$, $b_1,\ldots,b_n$. The adjacencies and edge weights are defined as follows. 

\begin{itemize}
\item \emph{Source edges:} for every $i \in [n]$, there is an edge $(b_i,x_{n+i,1})$. Together these edges are called source edges. The weight of each such edge is set to $M^5$.
\item \emph{Top sink edges:} for every $i \in [n]$, there is an edge $(x_{i,n^2},t_i)$. Together these edges are called top sink edges. The weight of each such edge is set to $M^5$.
\item \emph{Left sink edges:} for every $j \in [n]$, there is an edge $(\ell'_j,\ell_j)$ whose weight is set to $Mj$. Together these edges are called left sink edges.
\item  \emph{Right sink edges:} for every $j \in [n]$, there is an edge $(r'_j,r_j)$ whose weight is set to $M^2-Mj$. Together these edges are called right sink edges. 
\item \emph{Left internal sink edges:} for every $j \in [n]$ and every $(j-1)n+1 \leq p \leq jn$, there is an edge $(x_{1,p},\ell''_j)$ whose weight is set to $p-(j-1)n$. For every fixed $j \in [n]$, these edges together are called the $j^{th}$ left internal sink edges.
\item \emph{Right internal sink edges:} for every $j \in [n]$ and every $(j-1)n + 1 \leq p \leq jn$, there is an edge $(x_{2n,p},r''_j)$ whose weight is set to $n+1-(p-(j-1)n)$. For every fixed $j \in [n]$, these edges together are called the $j^{th}$ right internal sink edges.
\item \emph{Left bridge edges:} for every $j \in [n]$, and every $(j-1)n + 1 \leq p \leq jn$, there is an edge $(x_{n+1,p},x_{n,p})$. For every fixed $j \in [n]$, these edges together are called the $j^{th}$ right bridge edges. The weight of each $j^{th}$ left bridge edge is set to $M^4$.
\item \emph{Upward bridge edges:} for every $j \in \{pn~|~p \in [n-1]\}$ and every $i \in [2n]$, there is an edge $(x_{i,j},x_{i,j+1})$. For every fixed $j \in \{pn~|~p \in [n-1]\}$, these edges together are called the $j^{th}$ upward bridge edges. The weight of each $j^{th}$ downward bridge edge is set to $M^3$.
\item \emph{Up edges:} for every $i \in [2n]$, every $j \in [n]$ and every $(j-1)n + 1 \leq p \leq jn-1$, there is an edge $(x_{i,p},x_{i,p+1})$. Together these edges are called up edges. The weight of each such edge is set to~$4$.
\item \emph{Left edges:} for every $j \in [n]$, every $(j-1)n+1 \leq p \leq jn-1$ and every $2 \leq q \leq jn+1-p$, there is an edge $(x_{q,p},x_{q-1,p})$; and for every $j \in [n]$, every $(j-1)n+1 \leq p \leq jn-1$ and every $2 \leq i \leq n$, there is an edge $(x_{i,p},x_{i-1,p}$. Together these edges are called left edges. The weight of each such edge is set to $4$.
\item \emph{Right edges:} for every $j \in [n]$, every $(j-1)n+2 \leq p \leq jn$ and every $jn+1-p \leq q \leq n-1$, there is an edge $(x_{n+q,p},x_{n+q+1,p})$. Together these edges are called right edges. The weight of each such edge is set to $4$.
\item \emph{Shortcut edges:} for every $s=(i,j) \in S$, we introduce two shortcut edges $e^\ell_s,e^r_s$ as follows. Set $p = jn+1$ then
\begin{itemize}
\item subdivide the edge $(x_{i,p-i},x_{i,p+1-i})$ by adding a vertex $y_{i,j}$ and the edge $(x_{i,p-i},y_{i,j})$ (of weight $3$) with the edge $(y_{i,j},x_{i,p+1-i})$ (of weight $1$);
\item subdivide the edge $(x_{n+i,p-1-i},x_{n+i,p-i})$ by adding a vertex $z_{i,j}$ and the edge 

$(x_{n+i,p-1-i},z_{i,j})$ (of weight $3$) with the edge $(z_{i,j},x_{n+i,p-i})$ (of weight $1$). 
\end{itemize}
The introduced edges are called \emph{down subdivided edges}. Then
\begin{itemize}
\item $e^\ell_s= (y_{i,j},x_{i-1,p-i})$ and its weight set to $2$;
\item $e^r_s = (z_{i,j},x_{n+i+1,p-i})$ and its weight set to $2$.
\end{itemize}
The edges $e^\ell_s$ are called \emph{left} shortcut edges and the edges $e^r_s$ are called \emph{right} shortcut edges.  
\end{itemize}
This concludes the construction of the up main gadget $uMG_S$ (see \Cref{fig:upmaingad} for an illustration of the up main gadget $uMG_S$ with $n=4$ representing $S= \{(2,2),(2,3),(3,2)\}$). We call the vertices $\ell_1,\ldots,\ell_n$ the \emph{left vertices}, the vertices $r_1,\ldots,r_n$ the \emph{right vertices}, the vertices $t_1,\ldots,t_n$ the \emph{top vertices} and the vertices $b_1,\ldots,b_n$ the \emph{bottom vertices}.\\

\begin{figure}
\centering
\begin{tikzpicture}[scale=.8]
\node[rectangle, fill=yellow, rounded corners,minimum width=.25cm,minimum height=10.2cm] at (2,14.3) {};
\node[rectangle, fill=yellow, rounded corners,minimum width=4.25cm,minimum height=.25cm] at (4.5,8) {};
\node[rectangle, fill=yellow, rounded corners,minimum width=.25cm,minimum height=7.1cm] at (7,3.7) {};
\node[rectangle, fill=yellow, rounded corners,minimum width=1.5cm,minimum height=.25cm] at (-1.75,7.5) {};
\node[rectangle, fill=yellow, rounded corners,minimum width=1.5cm,minimum height=.25cm] at (11.75,7.5) {};
\node[rectangle, fill=yellow, rounded corners,minimum width=1.88cm,minimum height=.25cm,rotate=15] at (.05,7.75) {};
\node[rectangle, fill=yellow, rounded corners,minimum width=1.74cm,minimum height=.25cm,rotate=-15] at (10,7.75) {};
\node[rectangle, fill=yellow, rounded corners,minimum width=1cm,minimum height=.25cm] at (8.55,8) {};
\node[rectangle, fill=yellow, rounded corners,minimum width=1cm,minimum height=.25cm,rotate=24] at (1.5,8.22) {};
\node[rectangle, fill=yellow, rounded corners,minimum width=1.15cm,minimum height=.25cm,rotate=29.4] at (7.55,7.73) {};

\foreach \i in {1,...,4}
\foreach \j in {1,...,4}
{\node[circ] (\i\j) at (\i,\j) {};}

\foreach \i in {1,...,4}
\foreach \j in {6,...,9}
{\pgfmathtruncatemacro{\p}{\j-1}
\node[circ] (\i\p) at (\i,\j) {};}

\foreach \i in {1,...,4}
\foreach \j in {11,...,14}
{\pgfmathtruncatemacro{\p}{\j-2}
\node[circ] (\i\p) at (\i,\j) {};}

\foreach \i in {1,...,4}
\foreach \j in {16,...,19}
{\pgfmathtruncatemacro{\p}{\j-3}
\node[circ] (\i\p) at (\i,\j) {};}

\foreach \i in {6,...,9}
\foreach \j in {1,...,4}
{\pgfmathtruncatemacro{\p}{\i-1}
\node[circ] (\p\j) at (\i,\j) {};}

\foreach \i in {6,...,9}
\foreach \j in {6,...,9}
{\pgfmathtruncatemacro{\p}{\i-1}
\pgfmathtruncatemacro{\q}{\j-1}
\node[circ] (\p\q) at (\i,\j) {};}

\foreach \i in {6,...,9}
\foreach \j in {11,...,14}
{\pgfmathtruncatemacro{\p}{\i-1}
\pgfmathtruncatemacro{\q}{\j-2}
\node[circ] (\p\q) at (\i,\j) {};}

\foreach \i in {6,...,9}
\foreach \j in {16,...,19}
{\pgfmathtruncatemacro{\p}{\i-1}
\pgfmathtruncatemacro{\q}{\j-3}
\node[circ] (\p\q) at (\i,\j) {};}

\foreach \i in {6,...,9}
{\pgfmathtruncatemacro{\p}{\i-5}
\node[circ,label=below:{$b_\p$}] (b\p) at (\i,-.5) {};}

\foreach \i in {1,...,4}
{\node[circ,label=above:{$t_\i$}] (t\i) at (\i,20.5) {};}

\foreach \j in {0,...,3}
{\pgfmathsetmacro{\p}{\j*5+2.5}
\pgfmathtruncatemacro{\t}{\j+1}
\node[circ,label=below:{$\ell'_\t$}] (l'\t) at (-1,\p){};
\node[circ,label=left:{$\ell_\t$}] (l\t) at (-2.5,\p){};
\draw[->,thick,>=stealth] (l'\t) -- (l\t);}

\foreach \j in {0,...,3}
{\pgfmathsetmacro{\p}{\j*5+2.5}
\pgfmathtruncatemacro{\t}{\j+1}
\node[circ,label=below:{$r'_\t$}] (r'\t) at (11,\p){};
\node[circ,label=right:{$r_\t$}] (r\t) at (12.5,\p){};
\draw[->,thick,>=stealth] (r'\t) -- (r\t);}

\foreach \i in {1,...,8} 
\foreach \j in {1,...,15}
{\pgfmathtruncatemacro{\t}{\j+1}
\draw[->,>=stealth,thick] (\i\j) -- (\i\t);}

\foreach \j in {1,...,4}
{\pgfmathtruncatemacro{\a}{(\j-1)*4+1}
\pgfmathtruncatemacro{\b}{\j*4-1}
\foreach \p in {\a,...,\b}
{\pgfmathtruncatemacro{\d}{\p-(\j-1)*4}
\pgfmathtruncatemacro{\c}{5-\d}
\foreach \q in {2,...,\c}
{\pgfmathtruncatemacro{\e}{\q+4}
\pgfmathtruncatemacro{\t}{\e-1}
\draw[->,>=stealth,thick] (\e\p) -- (\t\p);}}}

\foreach \j in {1,...,4}
{\pgfmathtruncatemacro{\a}{(\j-1)*4+2}
\pgfmathtruncatemacro{\b}{\j*4}
\foreach \p in {\a,...,\b}
{\pgfmathtruncatemacro{\d}{\p-(\j-1)*4}
\pgfmathtruncatemacro{\c}{5-\d}
\foreach \q in {\c,...,3}
{\pgfmathtruncatemacro{\d}{\q+4}
\pgfmathtruncatemacro{\t}{\d+1}
\draw[->,>=stealth,thick] (\d\p) -- (\t\p);}}}

\foreach \j in {1,...,16}
\foreach \i in {1,...,4}
{\pgfmathtruncatemacro{\p}{\i+1}
\draw[<-,>=stealth,thick] (\i\j) -- (\p\j);}

\foreach \i in {1,...,4}
{\pgfmathtruncatemacro{\t}{\i+4}
\draw[<-,>=stealth,thick] (t\i) -- (\i16);
\draw[<-,>=stealth,thick] (\t1) -- (b\i);}

\foreach \j in {1,...,4}
{\pgfmathtruncatemacro{\a}{(\j-1)*4+1}
\pgfmathtruncatemacro{\b}{\j*4}
\foreach \p in {\a,...,\b}
{\draw[->,>=stealth,thick] (1\p) -- (l'\j);
\draw[->,>=stealth,thick] (8\p) -- (r'\j);}}

\node[circr] (22l) at (2,8.45) {};
\draw[->,>=stealth,thick,red] (22l) --(17);
\node[circr] (22r) at (7,7.4) {};
\draw[->,>=stealth,thick,red] (22r) -- (77);

\node[circr] (23l) at (2,13.45) {};
\draw[->,>=stealth,thick,red] (23l) --(111);
\node[circr] (23r) at (7,12.4) {};
\draw[->,>=stealth,thick,red] (23r) -- (711);

\node[circr] (32l) at (3,7.45) {};
\draw[->,>=stealth,thick,red] (32l) --(26);
\node[circr] (32r) at (8,6.4) {};
\draw[->,>=stealth,thick,red] (32r) -- (86);
\end{tikzpicture}
\caption{The up main gadget $uMG_S$ with $n=4$ representing $S=\{(2,2),(2,3),(3,2)\}$ (the red edges are the shortcut edges). A set of edges representing $(2,2)$ is highlighted.}
\label{fig:upmaingad}
\end{figure}
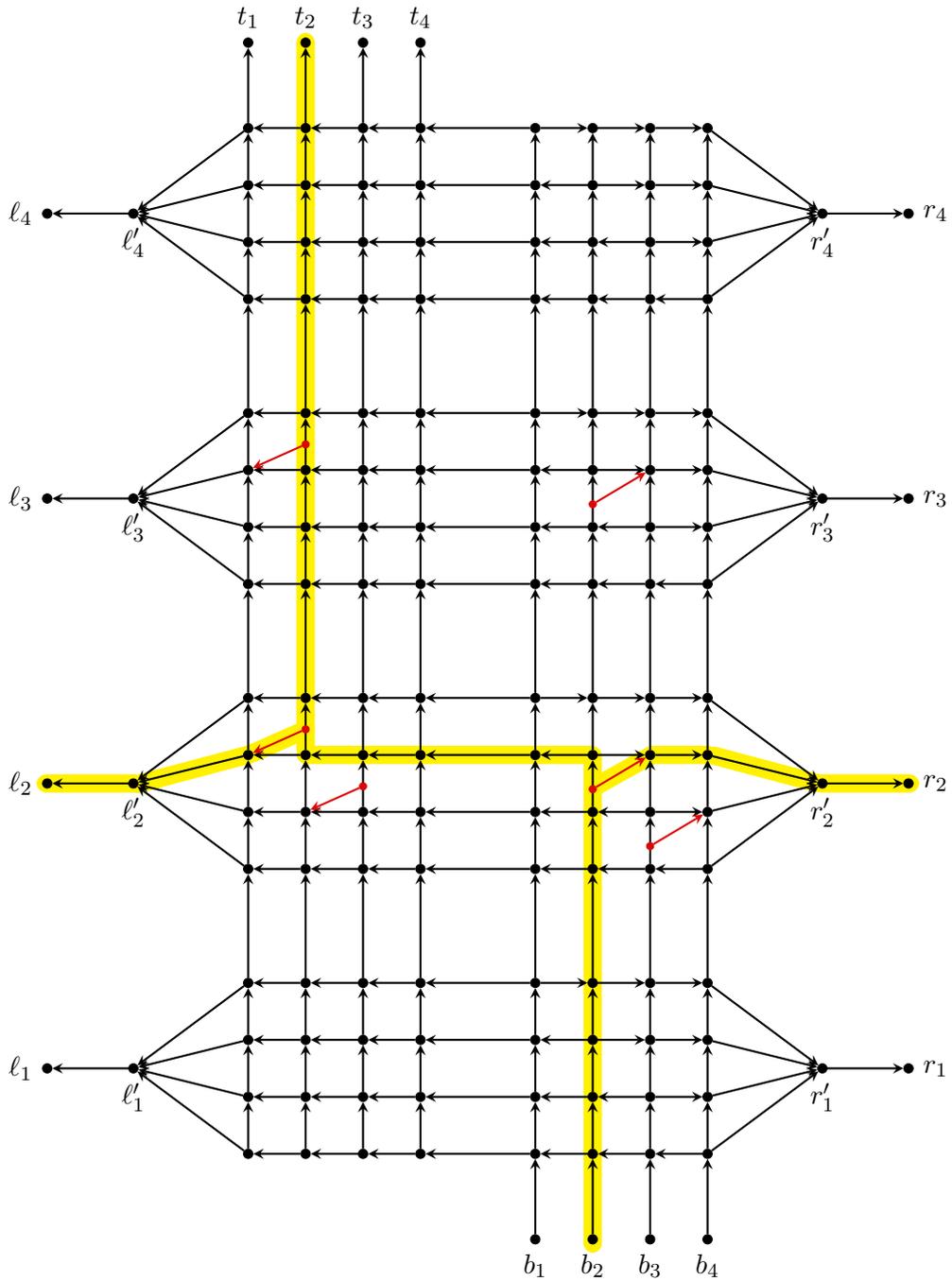

\noindent
A set $E \subseteq E(uMG_S)$ satisfies the \emph{connectedness} property if the following hold in $E$: 
\begin{itemize}
\item a bottom vertex can reach a top vertex;
\item a bottom vertex can reach a left vertex;
\item a bottom vertex can reach a right vertex.
\end{itemize}
A set $E \subseteq E(uMG_S)$ satisfying the connectedness property \emph{represents} a pair $(i,j) \in [n] \times [n]$ if the only source edge in $E$ is the one incident to $b_i$, the only top sink edge in $E$ is the one incident to $t_i$, the only left sink edge in $E$ is the one incident to $\ell_j$ and the only right sink edge in $E$ is the one incident to $r_j$ (see \Cref{fig:upmaingad} for a set of edges representing $(2,2)$). Symmetrical to \Cref{lem:dmg}, we have the following.

\begin{lemma}
\label{lem:umg}
For any $n > 0$ and any $S \subseteq [n] \times [n]$, the up main gadget $uMG_S$ satisfies the following properties.
\begin{itemize}
\item[(1)] For every $(i,j) \in S$, there exists a set $E_{i,j} \subseteq E(uMG_S)$ of weight $M_n^*$ representing $(i,j)$.
\item[(2)] If there exists a set $E \subseteq E(uMG_S)$ of weight at most $M_n^*$ satisfying the connectedness property then $E$ has weight exactly $M_n^*$ and represents a pair $(i,j) \in S$.
\end{itemize}
\end{lemma}


\noindent
\textbf{Reduction.} Given an instance $(k,n,\{S_{i,j}~|~i,j \in [k]\})$ of \textsc{Grid Tiling}, we construct an equivalent instance $(G,T,D)$ of (edge-weighted) \textsc{Planar $\oldC_1$-Steiner Network} where $G$ is defined as follows (see \Cref{fig:reduction}). 
\begin{itemize}
\item We introduce a total of $k^2$ (down/up) main gadgets and $k(k+1)$ connector gadgets.
\item For every set $S_{i,j}$ of the \textsc{Grid Tiling} instance such that $i$ is odd, we introduce an up main gadget $uMG_{i,j}$ representing $S_{i,j}$. The up main gadget $uMG_{i,j}$ is surrounded by two connector gadgets: $CG_{i,j}$ lying to its left and $CG_{i+1,j}$ lying to its right. We identify each right vertex of $CG_{i,j}$ with the left vertex of $uMG_{i,j}$ of the same index, and each left vertex of $CG_{i+1,j}$ with the right vertex of $uMG_{i,j}$ of the same index. Furthermore, for every $j \in [k-1]$, $uMG_{i,j}$ lies above $uMG_{i,j+1}$ and we identify each bottom vertex of $uMG_{i,j}$ with the top vertex of $uMG_{i,j+1}$ of the same index.
\item For every set $S_{i,j}$ of the \textsc{Grid Tiling} instance such that $i$ is even, we introduce a down main gadget $dMG_{i,j}$ representing $S_{i,j}$. The down main gadget $dMG_{i,j}$ is surrounded by two connector gadgets: $CG_{i,j}$ lying to its left and $CG_{i+1,j}$ lying to its right. We identify each right vertex of $CG_{i,j}$ with the left vertex of $dMG_{i,j}$ of the same index, and each left vertex of $CG_{i+1,j}$ with the right vertex of $dMG_{i,j}$ of the same index. Furthermore, for every $j \in [k-1]$, $dMG_{i,j}$ lies above $dMG_{i,j+1}$ and we identify each bottom vertex of $dMG_{i,j}$ with the top vertex of $dMG_{i,j+1}$ of the same index.
\item We introduce $k$ terminals $t_1,\ldots,t_k$ and add the following edges of weight 0: for every odd $i \in [k]$, we add an edge from each top vertex of $uMG_{i,1}$ to $t_i$ and for every even $i \in [k]$, we add an edge from each bottom vertex of $dMG_{i,k}$ to $t_i$.
\item We introduce two vertices $r_1$ and $r_2$ and add the following edges of weight 0: there is an edge from $r_1$ to every terminal $t_i$ such that $i \in [k]$ is odd; for every even $i \in [k]$, there is an edge from $r_1$ to every top vertex of $dMG_{i,1}$; for every $j \in [k]$, there is an edge from $r_1$ to every left vertex of $CG_{1,j}$; and if $k$ is odd then for every $j \in [k]$, there is an edge from $r_1$ to every right vertex of $CG_{k+1,j}$. Similarly, there is an edge from $r_2$ to every terminal $t_i$ such that $i$ is even; for every odd $i \in [k]$, there is an edge from $r_1$ to every bottom vertex of $dMG_{i,k}$; and if $k$ is even then for every $j \in [k]$, there is an edge from $r_1$ to every right vertex of $CG_{k+1,j}$.
\end{itemize} 
This concludes the construction of $G$. The set $T$ of terminals consists of the union of the two terminal vertices in each connector gadget and $\{r_1,r_2\} \cup \{t_i~|~i \in [k]\}$ (note that $|T| = 2k(k+1) + k + 2$). The demand graph $D$ is the pure out-diamond on vertex set $T$ where $r_1$ and $r_2$ are the two vertices of in-degree 0. In the following, we let 
\[
W_n^* = k^2M_n^* + k(k+1)C_n^*.   
\]

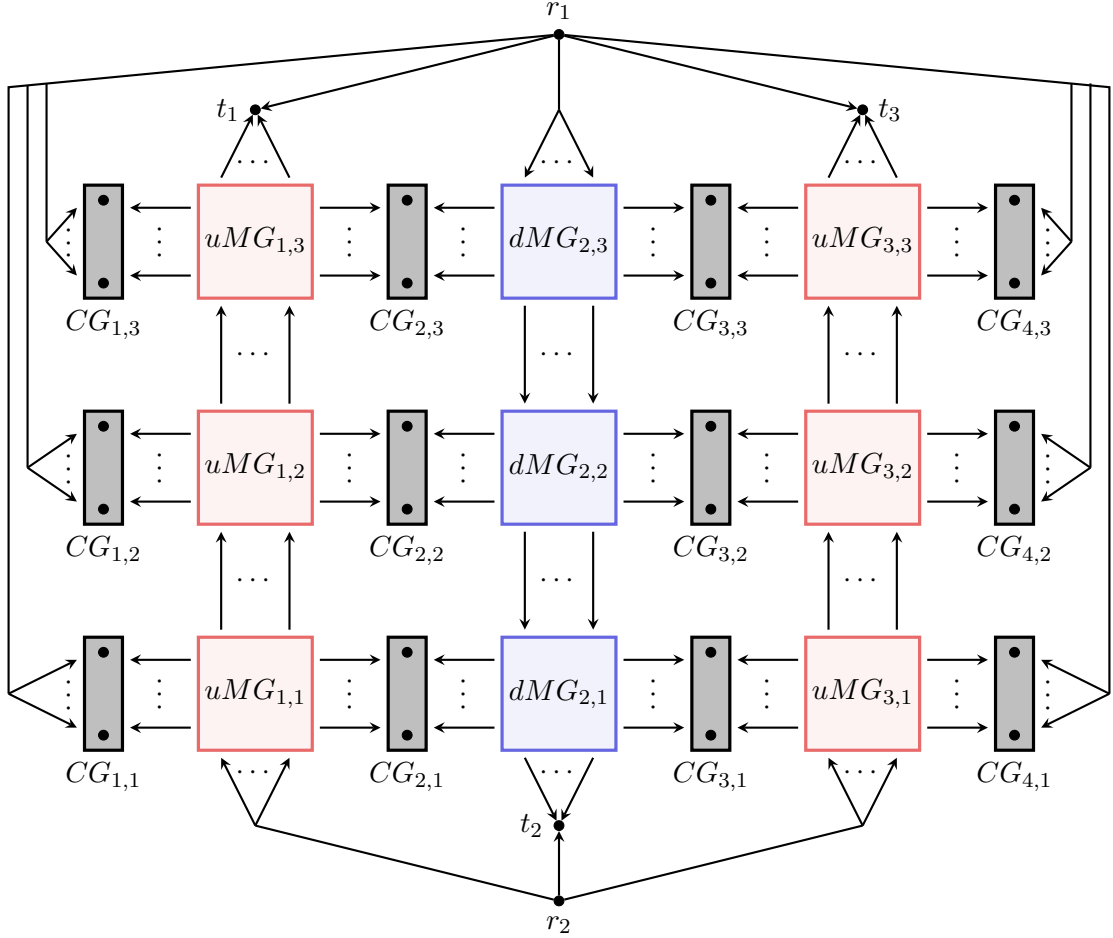
\begin{figure}
\centering
\begin{tikzpicture}
\foreach \i in {0,...,3}
\foreach \j in {0,...,2}
{\pgfmathsetmacro{\a}{\i*4}
\pgfmathsetmacro{\b}{\a+.5}
\pgfmathsetmacro{\c}{\j*3}
\pgfmathsetmacro{\d}{\c+1.5}
\filldraw[fill=lightgray,very thick] (\a,\c) rectangle (\b,\d);
\pgfmathsetmacro{\e}{\a+.25}
\pgfmathsetmacro{\f}{\c-.35}
\pgfmathtruncatemacro{\p}{\i+1}
\pgfmathtruncatemacro{\q}{\j+1}
\node[draw=none] at (\e,\f) {$CG_{\p,\q}$};
\pgfmathsetmacro{\f}{\c+.2}
\node[circ] at (\e,\f) {};
\pgfmathsetmacro{\g}{\c+1.3}
\node[circ] at (\e,\g) {};
}

\foreach \j in {0,...,2}
{\pgfmathsetmacro{\c}{\j*3}
\pgfmathsetmacro{\d}{\c+1.5}
\filldraw[color=red!60,fill=red!5,very thick] (1.5,\c) rectangle (3,\d);
\pgfmathsetmacro{\e}{\c+.75}
\pgfmathtruncatemacro{\q}{\j+1}
\node[draw=none] at (2.25,\e) {$uMG_{1,\q}$};
\pgfmathsetmacro{\f}{\c+.3}
\draw[->,thick,>=stealth] (1.4,\f) -- (.6,\f);
\pgfmathsetmacro{\g}{\c+.78}
\node[draw=none,rotate=90] at (1,\g) {$\cdots$};
\pgfmathsetmacro{\h}{\c+1.2}
\draw[->,thick,>=stealth] (1.4,\h) -- (.6,\h);
\draw[->,thick,>=stealth] (3.1,\f) -- (3.9,\f);
\node[draw=none,rotate=90] at (3.5,\g) {$\cdots$};
\draw[->,thick,>=stealth] (3.1,\h) -- (3.9,\h);
}

\foreach \i in {0,2}
\foreach \j in {0,...,1}
{\pgfmathsetmacro{\a}{\i*4}
\pgfmathsetmacro{\b}{\a+1.8}
\pgfmathsetmacro{\c}{\j*3+1.6}
\pgfmathsetmacro{\d}{\c+1.3}
\draw[->,thick,>=stealth] (\b,\c) -- (\b,\d);
\pgfmathsetmacro{\e}{\b+.9}
\draw[->,thick,>=stealth] (\e,\c) -- (\e,\d);
\pgfmathsetmacro{\f}{\b+.45}
\pgfmathsetmacro{\g}{\c+.65}
\node[draw=none] at (\f,\g) {$\cdots$};}

\foreach \i in {1}
\foreach \j in {0,...,1}
{\pgfmathsetmacro{\a}{\i*4}
\pgfmathsetmacro{\b}{\a+1.8}
\pgfmathsetmacro{\c}{\j*3+1.6}
\pgfmathsetmacro{\d}{\c+1.3}
\draw[<-,thick,>=stealth] (\b,\c) -- (\b,\d);
\pgfmathsetmacro{\e}{\b+.9}
\draw[<-,thick,>=stealth] (\e,\c) -- (\e,\d);
\pgfmathsetmacro{\f}{\b+.45}
\pgfmathsetmacro{\g}{\c+.65}
\node[draw=none] at (\f,\g) {$\cdots$};}

\foreach \j in {0,...,2}
{\pgfmathsetmacro{\c}{\j*3}
\pgfmathsetmacro{\d}{\c+1.5}
\filldraw[color=blue!60,fill=blue!5,very thick] (5.5,\c) rectangle (7,\d);
\pgfmathsetmacro{\e}{\c+.75}
\pgfmathtruncatemacro{\q}{\j+1}
\node[draw=none] at (6.25,\e) {$dMG_{2,\q}$};
\pgfmathsetmacro{\f}{\c+.3}
\draw[->,thick,>=stealth] (5.4,\f) -- (4.6,\f);
\pgfmathsetmacro{\g}{\c+.78}
\node[draw=none,rotate=90] at (5,\g) {$\cdots$};
\pgfmathsetmacro{\h}{\c+1.2}
\draw[->,thick,>=stealth] (5.4,\h) -- (4.6,\h);
\draw[->,thick,>=stealth] (7.1,\f) -- (7.9,\f);
\node[draw=none,rotate=90] at (7.5,\g) {$\cdots$};
\draw[->,thick,>=stealth] (7.1,\h) -- (7.9,\h);}

\foreach \j in {0,...,2}
{\pgfmathsetmacro{\c}{\j*3}
\pgfmathsetmacro{\d}{\c+1.5}
\filldraw[color=red!60,fill=red!5,very thick] (9.5,\c) rectangle (11,\d);
\pgfmathsetmacro{\e}{\c+.75}
\pgfmathtruncatemacro{\q}{\j+1}
\node[draw=none] at (10.25,\e) {$uMG_{3,\q}$};
\pgfmathsetmacro{\f}{\c+.3}
\draw[->,thick,>=stealth] (9.4,\f) -- (8.6,\f);
\pgfmathsetmacro{\g}{\c+.78}
\node[draw=none,rotate=90] at (9,\g) {$\cdots$};
\pgfmathsetmacro{\h}{\c+1.2}
\draw[->,thick,>=stealth] (9.4,\h) -- (8.6,\h);
\draw[->,thick,>=stealth] (11.1,\f) -- (11.9,\f);
\node[draw=none,rotate=90] at (11.5,\g) {$\cdots$};
\draw[->,thick,>=stealth] (11.1,\h) -- (11.9,\h);}

\node[circ,label=left:{$t_1$}] (t1) at (2.25,8.5) {};
\draw[->,thick,>=stealth] (1.8,7.6) -- (t1); 
\draw[->,thick,>=stealth] (2.7,7.6) -- (t1); 
\node[draw=none] at (2.25,7.8) {$\cdots$};
\node[circ,label=left:{$t_2$}] (t2) at (6.25,-1) {};
\draw[->,thick,>=stealth] (5.8,-.1) -- (t2); 
\draw[->,thick,>=stealth] (6.7,-.1) -- (t2); 
\node[draw=none] at (6.25,-.3) {$\cdots$};
\node[circ,label=right:{$t_3$}] (t3) at (10.25,8.5) {};
\draw[->,thick,>=stealth] (9.8,7.6) -- (t3); 
\draw[->,thick,>=stealth] (10.7,7.6) -- (t3); 
\node[draw=none] at (10.25,7.8) {$\cdots$};

\node[circ,label=above:{$r_1$}] (r1) at (6.25,9.5) {};
\draw[->,thick,>=stealth] (r1) -- (t1);
\draw[->,thick,>=stealth] (r1) -- (t3);
\draw[<-,thick,>=stealth] (5.8,7.6) -- (6.25,8.5); 
\draw[<-,thick,>=stealth] (6.7,7.6) -- (6.25,8.5); 
\node[draw=none] at (6.25,7.8) {$\cdots$};
\draw[thick] (r1) -- (6.25,8.5);

\draw[->,thick,>=stealth] (-.5,6.75) -- (-.1,7.2);
\draw[->,thick,>=stealth] (-.5,6.75) -- (-.1,6.3);
\node[draw=none,rotate=90] at (-.2,6.77) {$\cdots$};
\draw[thick] (-.5,8.85) -- (-.5,6.75);
\draw[thick] (13,8.85) -- (13,6.75);

\draw[->,thick,>=stealth] (-.75,3.75) -- (-.1,4.2);
\draw[->,thick,>=stealth] (-.75,3.75) -- (-.1,3.3);
\node[draw=none,rotate=90] at (-.2,3.77) {$\cdots$};
\draw[thick] (-.75,8.825) -- (-.75,3.75);
\draw[thick] (13.25,8.85) -- (13.25,3.75);

\draw[->,thick,>=stealth] (-1,.75) -- (-.1,1.2);
\draw[->,thick,>=stealth] (-1,.75) -- (-.1,.3);
\node[draw=none,rotate=90] at (-.2,.77) {$\cdots$};
\draw[thick] (r1) -- (-1,8.8) -- (-1,.75);
\draw[thick] (r1) -- (13.5,8.8) -- (13.5,.75);

\draw[->,thick,>=stealth] (13,6.75) -- (12.6,7.2);
\draw[->,thick,>=stealth] (13,6.75) -- (12.6,6.3);
\node[draw=none,rotate=90] at (12.7,6.77) {$\cdots$};

\draw[->,thick,>=stealth] (13.25,3.75) -- (12.6,4.2);
\draw[->,thick,>=stealth] (13.25,3.75) -- (12.6,3.3);
\node[draw=none,rotate=90] at (12.7,3.77) {$\cdots$};

\draw[->,thick,>=stealth] (13.5,.75) -- (12.6,1.2);
\draw[->,thick,>=stealth] (13.5,.75) -- (12.6,.3);
\node[draw=none,rotate=90] at (12.7,.77) {$\cdots$};

\node[circ,label=below:{$r_2$}] (r2) at (6.25,-2) {};
\draw[->,thick,>=stealth] (r2) -- (t2);
\draw[<-,thick,>=stealth] (1.8,-.1) -- (2.25,-1); 
\draw[<-,thick,>=stealth] (2.7,-.1) -- (2.25,-1);
\node[draw=none] at (2.25,-.3) {$\cdots$};
\draw[thick] (r2) -- (2.25,-1); 
\draw[<-,thick,>=stealth] (9.8,-.1) -- (10.25,-1); 
\draw[<-,thick,>=stealth] (10.7,-.1) -- (10.25,-1);
\node[draw=none] at (10.25,-.3) {$\cdots$};
\draw[thick] (r2) -- (10.25,-1); 
\end{tikzpicture}
\caption{An illustration of the reduction from \textsc{Grid Tiling} to  \textsc{Planar $\oldC_1$-Steiner Network} with $k=3$ (the black vertices are the terminals).}
\label{fig:reduction}
\end{figure}

\begin{lemma}
The \textsc{Grid Tiling} instance $(k,n,\{S_{i,j}~|~i,j \in [k]\})$ has a solution if and only if the (edge-weighted) \textsc{Planar $\oldC_1$-Steiner Network} instance $(G,T,D)$ has a solution of weight at most $W_n^*$.
\end{lemma}

\begin{proof}
Assume first that the instance $(k,n,\{S_{i,j}~|~i,j \in [k]\})$ of \textsc{Grid Tiling} has a solution, that is, for every $i,j \in [k]$, there is an entry $(x_{i,j},y_{i,j}) \in S_{i,j}$ such that
\begin{itemize}
\item for every $i \in [k]$, $x_{i,1} = x_{i,2} = \ldots = x_{i,k} = \alpha_i$ and
\item for every $j \in [k]$, $y_{1,j} = y_{2,j} = \ldots = y_{k,j} = \beta_j$.
\end{itemize}
We construct a solution $E$ for the \textsc{Planar $\oldC_1$-Steiner Network} instance $(G,T,D)$ of weight at most $W_n^*$ as follows. Include in $E$
\begin{itemize}
\item for every odd $i \in [k]$, the edge $(r_1,t_i)$ of weight 0 and for every even $i \in [k]$, the edge $(r_2,t_i)$ of weight 0;
\item for every even $i \in [k]$, the edge from $r_1$ to the top vertex of index $\alpha_i$ in $dMG_{i,1}$ of weight 0 and for every odd $i \in [k]$, the edge from $r_2$ to the bottom vertex of index $\alpha_i$ in $dMG_{i,k}$ of weight 0;  
\item for every $j \in [k]$, the edge from $r_1$ to the left vertex of index $\beta_j$ in $CG_{1,j}$ of weight 0;
\item if $k$ is odd then for every $j \in [k]$, the edge from $r_1$ to the right vertex of index $\beta_j$ in $CG_{k+1,j}$ of weight 0 and if $k$ is even then for every $j \in [k]$, the edge from $r_2$ to the right vertex of $\beta_j$ in $CG_{k+1,j}$ of weight 0;
\item for every $j \in [k]$ and every $i \in [k]$, the set $E^C_{i,j} \subseteq E(CG_{i,j})$ of weight $C_n^*$ representing $\beta_j$ whose existence is guaranteed by \Cref{lem:cg}(1);
\item for every $j \in [k]$ and every odd $i \in [k]$, the set $E^M_{i,j} \subseteq E(uMG_{i,j})$ of weight $M_n^*$ representing $(\alpha_i,\beta_j)$ whose existence is guaranteed by \Cref{lem:umg}(1); and
\item for every $j \in [k]$ and every even $i \in [k]$, the set $E^M_{i,j} \subseteq E(dMG_{i,j})$ of weight $M_n^*$ representing $(\alpha_i,\beta_j)$ whose existence is guaranteed by \Cref{lem:dmg}(1).
\end{itemize}
It is not difficult to see that the weight of $E$ is $k^2M_n^* + k(k+1)C_n^* = W_n^*$ and that by the connectedness of the sets $E^C_{p,j}$ and $E^M_{i,j}$ for every $i,j \in [k]$ and $p \in [k+1]$, $r_1$ and $r_2$ can reach every terminal in $T \setminus \{r_1,r_2\}$ in $E$.\\ 

Conversely, assume that the \textsc{Planar $\oldC_1$-Steiner Network} instance $(G,T,D)$ has a solution $E$ of weight at most $W_n^*$. We contend that for every $i \in [k+1]$ and every $j \in [k]$, $E \cap E(CG_{i,j})$ satisfies the connectedness property. Indeed, if this is not the case for some $i \in [k+1]$ and $j\in [k]$, say no left vertex of $V(CG_{i,j})$ can reach the terminal $p \in V(CG_{i,j})$ in $E$ (the other cases are symmetric) then either $i$ is even in which case $r_2$ cannot reach $p$ in $E$, or $i$ is odd in which case $r_1$ cannot reach $p$ in $E$, a contradiction in both cases. Similarly, the restriction of $E$ to any (down/up) main gadget satisfies the connectedness property: indeed, if for some $j \in [k]$ and for some even $i \in [k]$ (we argue similarly if $i$ is odd), $E \cap E(dMG_{i,j})$ does not satisfy the connectedness property, then either no top vertex of $dMG_{i,j}$ can reach a bottom vertex of $dMG_{i,j}$ in $E$ in which case $r_2$ cannot reach $t_i$ in $E$; or no top vertex of $dMG_{i,j}$ can reach a left (or right) vertex of $dMG_{i,j}$ in $E$ in which case $r_2$ (or $r_1$) cannot reach the terminal vertices of $CG_{i,j}$ in $E$, a contradiction in both cases.

Next, we argue that the weight of the restriction of $E$ to any connector gadget is $C_n^*$ and that the weight of the restriction of $E$ to any (down/up) main gadget is $M_n^*$. To this end, let $c$ and $C$ be the number of connector gadgets whose weight in $E$ is at most $C_n^*$ and greater than $C_n^*$, respectively. Then $c+C=k(k+1)$ and by \Cref{lem:cg}(2), any connector gadget whose weight in $E$ is at most $C_n^*$ has in fact a weight of exactly $C_n^*$ in $E$. Similarly, let $m$ and $M$ be the number of (down/up) main gadgets whose weight in $E$ is at most $M_n^*$ and greater than $M_n^*$ respectively. Then $m + M = k^2$ and by Lemmas~\ref{lem:dmg} and \ref{lem:umg}, any (down/up) main gadget whose weight in $E$ is at most $M_n^*$ has in fact a weight of exactly $M_n^*$ in $E$. Now by definition of $W_n^*$, 
\begin{equation*}
\begin{split}
W_n^*& = k^2M_n^* + k(k+1)C_n^*\\
& \geq mM_n^* + M(M_n^*+1) + cC_n^* + C(C_n^*+1)\\
& = k^2M_n^* + M + k(k+1)C_n^* + C
\end{split}
\end{equation*}  
which implies that $M = C = 0$. Thus, every connector gadget has weight $C_n^*$ in $E$ and every (down/up) main gadget has weight $M_n^*$ in $E$. From Lemmas~\ref{lem:cg}(2), \ref{lem:dmg}(2) and \ref{lem:umg}(2), it then follows that
\begin{itemize}
\item for every $j \in [k]$ and every $i \in [k+1]$, the restriction of $E$ to the connector gadget $CG_{i,j}$ represents an integer $\beta'_{i,j} \in [n]$;
\item for every $j \in [k]$ and every even $i \in [k]$, the restriction of $E$ to the down main gadget $dMG_{i,j}$ represents a pair $(\alpha_{i,j},\beta_{i,j}) \in [n] \times [n]$; and
\item for every $j \in [k]$ and every odd $i \in [k]$, the restriction of $E$ to the up main gadget $uMG_{i,j}$ represents a pair $(\alpha_{i,j},\beta_{i,j}) \in [n] \times [n]$.
\end{itemize}
Let us show that for every $i,j \in [k]$ the entries $(\alpha_{i,j},\beta_{i,j}) \in S_{i,j}$ form a solution to the \textsc{Grid Tiling} instance $(k,n,\{S_{i,j}~|~i,j \in [k]\})$ which if true, would conclude the proof. To this end, we first prove that for every $i,j \in [k]$, $\beta'_{i,j} = \beta_{i,j}$. Consider an even $i \in [k]$. Then by \Cref{lem:dmg}(2), the only left sink edge in $E$ incident to a left vertex of $dMG_{i,j}$ is the one incident to the left vertex of index $\beta_{i,j}$; and by \Cref{lem:cg}(2), the only right source edge in $E$ incident to a right vertex of $CG_{i,j}$ is the one incident to the right vertex of index $\beta'_{i,j}$. Thus, if $\beta'_{i,j} \neq \beta_{i,j}$ then $r_1$ cannot reach the terminal vertices of $CG_{i,j}$ in $E$, a contradiction. We conclude similarly if $i$ is odd. 

Second, we show that for every $i,j \in [k]$, $\beta'_{i+1,j} = \beta_{i,j}$. Consider an even $i \in [k]$. Then by \Cref{lem:dmg}(2), the only right sink edge in $E$ incident to a right vertex of $dMG_{i,j}$ is the one incident to the right vertex of index $\beta_{i,j}$; and by \Cref{lem:cg}(2), the only left source edge in $E$ incident to a left vertex of $CG_{i+1,j}$ is the one incident to the left vertex of index $\beta'_{i,j}$. Thus, if $\beta'_{i,j} \neq \beta_{i,j}$ then $r_2$ cannot reach the terminal vertices of $CG_{i+1,j}$ in $E$, a contradiction. We conclude similarly if $i$ is odd. 

It follows from the above that for every $i \in [k-1]$ and every $j \in [k]$, $\beta_{i+1,j} = \beta'_{i+1,j} = \beta_{i,j}$; we next show that for every $i \in [k]$ and every $j \in [k-1]$, $\alpha_{i,j} = \alpha_{i,j+1}$. Consider an even $i \in [k]$. Then by \Cref{lem:dmg}(2), the only bottom sink edge in $E$ incident to a bottom vertex of $dMG_{i,j}$ is the one incident to the bottom vertex of index $\alpha_{i,j}$; and by \Cref{lem:dmg}(2), the only source edge in $E$ incident to a top vertex of $dMG_{i,j+1}$ is the one incident to the top vertex of index $\alpha_{i,j+1}$. Thus, if $\alpha_{i,j} \neq \alpha_{i,j+1}$ then $r_1$ cannot reach the terminal $t_i$ in $E$, a contradiction. We conclude similarly if $i$ is odd. Therefore, the entries $(\alpha_{i,j},\beta_{i,j}) \in S_{i,j}$ form a solution to the \textsc{Grid Tiling} instance $(k,n,\{S_{i,j}~|~i,j \in [k]\})$ as claimed.
\end{proof}

Let us finally explain how to get rid of the edge-weights (we use the same trick as in \cite{ChitnisFHM20}).
We replace every edge $(x,y)$ of weight $w$ in the instance $(G,T,D)$ of (edge-weighted) {\sc Planar $\mathcal{C}_1$-Steiner Network} constructed above,
with a directed path from $x$ to $y$ of length $w\cdot n + 1$ where $n = |V(G)|$. 
We let $G'$ be the resulting graph.
Then similarly to \cite[Theorem A.1]{ChitnisFHM20}, we can show that the instance $(G,T,D)$ of edge-weighted {\sc Planar $\mathcal{C}_1$-Steiner Network} has a solution of weight at most $W$ if and only if the instance $(G',T,D)$ of {\sc Planar $\mathcal{C}_1$-Steiner Network} has a solution of size at most $Wn + n$.

\subsection{Hard patterns}

The aim of this section is to prove that for every $\ell \in [8104]$, {\sc Planar $\mathcal{C}_\ell$-Steiner Network} 
is $\mathsf{W}[1]$-hard parameterized by the number $k$ of terminals and does not admit a
$f(k) \cdot n^{o(k)}$ algorithm for any computable function $f$, unless $\mathsf{ETH}$ fails.
For each $\ell \in [8104]$, we give a reduction which transforms an instance of {\sc $k \times k$-Grid Tiling} (see \Cref{sec:diamonds} for a definition of this problem) into an instance of (edge-weighted)\footnote{We then use the same trick as in \Cref{sec:diamonds} to get rid of the edge-weights.} {\sc Planar $\mathcal{C}_\ell$-Steiner Network} with $O(k)$ terminals.
The constructed instances in each case are very similar and are based on a construction developed in the proof of \cite[Theorem 1.4]{ChitnisFHM20}, which we describe below as the \emph{main gadget}. We then show how to built upon this construction to handle the hard matching patterns (see \Cref{sec:orderedtough}) and the hard biclique patterns (see \Cref{sec:biclique}).\\

\noindent
\textbf{Main Gadget.} As mentioned above, we use the same construction as in the proof of \cite[Theorem 1.4]{ChitnisFHM20}. More precisely, given an integer $n > 0$ and a subset $S \subseteq [n] \times [n]$\footnote{Recall that, by assumption, $1 < x,y< n$ holds for every $(x,y) \in S$.}, we first construct an edge-weighted planar digraph $G(S)$ as follows. The graph $G(S)$ consists of an $n \times n$ grid where the horizontal edges are oriented towards the right and the vertical edges are oriented towards the bottom, that is, denoting by $x_{i,j}$ the vertex lying at the intersection of column $i$ and row $j$, there is an edge 
\begin{itemize}
\item $(x_{i+1,j},x_{i,j})$ (of weight 2) for every $i \in [n-1]$ and $j \in [n]$, and 
\item $(x_{i,j},x_{i,j+1})$ (of weight 2) for every $i \in [n]$ and $j \in [n-1]$. 
\end{itemize}
Then for every $(i,j) \in S$, we subdivide the edge $(x_{i-1,j},x_{i,j})$, by adding a vertex $y_{i,j}$ and the edges $(x_{i-1,j},y_{i,j})$ and $(y_{i,j},x_{i,j})$ (both of weight 1), and further add the edge $(x_{i,j-1},y_{i,j})$ (of weight 1). 
This concludes the construction of $G(S)$.
In the following, we call the vertices $x_{1,1},x_{1,2},\ldots,x_{1,n}$ the \emph{left vertices}, the vertices $x_{n,1}, x_{n,2}\ldots,x_{n,n}$ the \emph{right vertices}, the vertices $x_{1,n},x_{2,n}, \ldots,x_{n,n}$ the \emph{top vertices} and the vertices $x_{1,1},x_{2,1},\ldots, x_{n,1}$ the \emph{bottom vertices}. \\

Now given a collection $\mathcal{S} = \{S_{i,j}~|~i,j \in [k]\}$ of $k^2$ subsets of $[n] \times [n]$, the \emph{main gadget $MG(\mathcal{S})$} for $\mathcal{S}$ is constructed as follows.
\begin{itemize}
\item For every set $S_{i,j} \in \mathcal{S}$, we introduce a copy of the graph $G(S_{i,j})$ as constructed above. For every $i \in [n]$ and $j \in [n-1]$, the graph $G(S_{i,j+1})$ lies below the $G(S_{i,j})$; we add an edge (of weight 2) from each top vertex of $G(S_{i,j+1})$ to the bottom vertex of $G(S_{i,j})$ of the same index. Similarly, for every $i \in [n-1]$ and $j \in [n]$, the graph $G(S_{i,j})$ lies to the left of the graph $G(S_{i+1,j})$; we add an edge (of weight 2) from each right vertex of $G(S_{i,j})$ to the left vertex of $G(S_{i+1,j})$ of the same index.
\item We introduce $4k$ additional vertices $a_1,\ldots,a_k,b_1,\ldots,b_k,c_1,\ldots,c_k,d_1,\ldots,d_k$ and the following edges (we fix $\Delta = 5n^2$).
\begin{itemize}
\item For every $j \in [k]$ and every $i \in [n]$, there is an edge from $a_j$ to the left vertex of $G(S_{1,j})$ of index $i$ of weight $\Delta(n+1-i)$.
\item For every $j \in [k]$ and every $i \in [n]$, there is an edge from the right vertex of $G(S_{n,j})$ of index $i$ to $b_j$ of weight $\Delta i$.
\item For every $j \in [k]$ and every $i \in [n]$, there is an edge from $c_j$ to the top vertex of $G(S_{j,1})$ of index $i$ of weight $\Delta(n+1-i)$.
\item For every $j \in [k]$ and every $i \in [n]$, there is an edge from the bottom vertex of $G(S_{j,n})$ of index $i$ to $d_j$ of weight $\Delta i$.
\end{itemize}
\end{itemize}
This concludes the construction of the main gadget $MG(\mathcal{S})$ for $\mathcal{S}$ (see \Cref{fig:bimaingad} for an illustration of the main gadget $MG(\mathcal{S})$ with $n=4$ for $\mathcal{S} = \{S_{i,j}~|~i,j \in [3]\}$ where $S_{1,1} =\{(2,2),(2,3),(3,3)\}$ and $S_{i,j} = \emptyset$ for every $(i,j) \neq (1,1)$). We now set $I = \{a_i,b_i,c_i,d_i~|~i \in [k]\}$ and let $M$ be the induced matching $\{(a_i,b_i),(c_i,d_i)~|~i\in [k]\}$. We further set
\[
B^* = 2k(\Delta(n+1) + 2(k+1) +2k(n-1)).
\]

\begin{lemma}[\cite{ChitnisFHM20}]
\label{lem:grid-dsn}
The {\sc $k \times k$-Grid Tiling} instance $(k,n,\mathcal{S})$ has a solution if and only if the (edge-weighted) {\sc Planar $M$-Steiner Network} instance $(MG(\mathcal{S}),I,M)$ has a solution of weight at most $B^*-k^2$.
\end{lemma}

\begin{figure}
\centering
\begin{tikzpicture}[scale=.7]
\foreach \x in {0,5,10}
\foreach \y in {0,5,10}
{\pgfmathtruncatemacro{\a}{\x+1}
\pgfmathtruncatemacro{\e}{\x+3}
\pgfmathtruncatemacro{\b}{\x+4}
\pgfmathtruncatemacro{\c}{\y+1}
\pgfmathtruncatemacro{\f}{\y+2}
\pgfmathtruncatemacro{\d}{\y+4}

\pgfmathsetmacro{\u}{\x+.67}
\pgfmathsetmacro{\v}{\x+4.33}
\pgfmathsetmacro{\w}{\y+.67}
\pgfmathsetmacro{\z}{\y+4.33}
\filldraw[fill=lightgray] (\u,\w) rectangle (\v,\z);
\ifthenelse{\x=0}{\pgfmathtruncatemacro{\s}{\x+1}}{\ifthenelse{\x=5}{\pgfmathtruncatemacro{\s}{\x-3}}{\pgfmathtruncatemacro{\s}{\x-7}}}
\ifthenelse{\y=0}{\pgfmathtruncatemacro{\t}{\y+3}}{\ifthenelse{\y=5}{\pgfmathtruncatemacro{\t}{\y-3}}{\pgfmathtruncatemacro{\t}{\y-9}}}
\node[draw=none,label={[label distance=-.35cm]10:\footnotesize $G(S_{\s,\t})$}] at (\v,\z) {};

\foreach \i in {\c,...,\d}
\foreach \j in {\a,...,\b}
{
\node[mcirc] (\j\i) at (\j,\i) {};
}

\ifthenelse{\x=0}{
\pgfmathsetmacro{\r}{\y+2.5}
\node[circ,label=left:{\small $a_\t$}] (a\t) at (-.4,\r) {};
\foreach \i in {\c,...,\d}
{\pgfmathtruncatemacro{\o}{5-\i+\y}
\ifthenelse{\o=1}{
\draw[->,thick,>=stealth] (a\t) -- (\a\i) node[pos=.86,left] {{\fontsize{4}{10}\selectfont $\Delta$}};}
{\ifthenelse{\o=4}{\draw[->,thick,>=stealth] (a\t) -- (\a\i) node[pos=.33,right] {{\fontsize{4}{10}\selectfont $\o \Delta$}};}{\draw[->,thick,>=stealth] (a\t) -- (\a\i) node[pos=.63,above] {{\fontsize{4}{10}\selectfont $\o \Delta$}};}}}}{}

\ifthenelse{\x=10}{
\pgfmathsetmacro{\r}{\y+2.5}
\node[circ,label=right:{\small $b_\t$}] (b\t) at (15.4,\r) {};
\foreach \i in {\c,...,\d}
{\pgfmathtruncatemacro{\o}{\i-\y}
\ifthenelse{\o=1}{
\draw[->,thick,>=stealth] (\b\i) -- (b\t) node[pos=.59,left] {{\fontsize{4}{10}\selectfont $\Delta$}};}
{\ifthenelse{\o=4}{\draw[->,thick,>=stealth] (\b\i) -- (b\t) node[pos=.11,right] {{\fontsize{4}{10}\selectfont $\o \Delta$}};}{\draw[->,thick,>=stealth] (\b\i) -- (b\t) node[pos=.4,above] {{\fontsize{4}{10}\selectfont $\o \Delta$}};}}}}{}

\ifthenelse{\y=0}{
\pgfmathsetmacro{\r}{\x+2.5}
\node[circ,label=below:{\small $d_\s$}] (d\s) at (\r,-.4) {};
\foreach \j in {\a,...,\b}
{\pgfmathtruncatemacro{\o}{\j-\x}
\ifthenelse{\o=1}{
\draw[->,thick,>=stealth] (\j\c) -- (d\s) node[pos=.4,right] {{\fontsize{4}{10}\selectfont $\Delta$}};}
{\ifthenelse{\o=4}{\draw[->,thick,>=stealth] (\j\c) -- (d\s) node[pos=.4,label={[label distance=-.2cm]0:{\fontsize{4}{10}\selectfont $\o \Delta$}}] {};}{\draw[->,thick,>=stealth] (\j\c) -- (d\s) node[pos=.4,label={[label distance=-.2cm]0:{\fontsize{4}{10}\selectfont $\o \Delta$}}] {};}}}}{}

\ifthenelse{\y=10}{
\pgfmathsetmacro{\r}{\x+2.5}
\node[circ,label=above:{\small $c_\s$}] (c\s) at (\r,15.4) {};
\foreach \j in {\a,...,\b}
{\pgfmathtruncatemacro{\o}{5-\j+\x}
\ifthenelse{\o=1}{
\draw[->,thick,>=stealth] (c\s) -- (\j\d) node[pos=.6,right] {{\fontsize{4}{10}\selectfont $\Delta$}};}
{\ifthenelse{\o=4}{\draw[->,thick,>=stealth] (c\s) -- (\j\d) node[pos=.6,label={[label distance=-.2cm]0:{\fontsize{4}{10}\selectfont $\o \Delta$}}] {};}{\draw[->,thick,>=stealth] (c\s) -- (\j\d) node[pos=.6,label={[label distance=-.2cm]0:{\fontsize{4}{10}\selectfont $\o \Delta$}}] {};}}}}{}

\foreach \i in {\c,...,\d}
\foreach \j in {\a,...,\e}
{
\pgfmathtruncatemacro{\k}{\j+1}
\draw[->,thick,>=stealth] (\j\i) -- (\k\i);
}

\foreach \j in {\a,...,\b}
\foreach \i in {\d,...,\f}
{
\pgfmathtruncatemacro{\k}{\i-1}
\draw[->,thick,>=stealth] (\j\i) -- (\j\k);
}

\foreach \j in {\a,...,\b}
{\ifthenelse{\y=0}{}{
\pgfmathtruncatemacro{\k}{\y-1}
\draw[->,thick,>=stealth] (\j\c) -- (\j\k);}
}

\foreach \i in {\c,...,\d}
{\ifthenelse{\x=0}{}{
\pgfmathtruncatemacro{\k}{\x-1}
\draw[->,thick,>=stealth] (\k\i) -- (\a\i);}
}
}

\node[circr] (h) at (2,2.6) {};
\node[circr] (g) at (2,3.6) {};
\node[circr] (q) at (3,3.6) {};
\draw[->,thick,>=stealth,red] (12) -- (h);
\draw[->,thick,>=stealth,red] (13) -- (g);
\draw[->,thick,>=stealth,red] (23) -- (q);
\end{tikzpicture}
\caption{The main gadget $MG(\mathcal{S})$ with $n=4$ for $\mathcal{S} = \{S_{i,j}~|~i,j \in [3]\}$ where $S_{1,1} =\{(2,2),(2,3),(3,3)\}$ and $S_{i,j} = \emptyset$ for every $(i,j) \neq (1,1)$.}
\label{fig:bimaingad}
\end{figure}
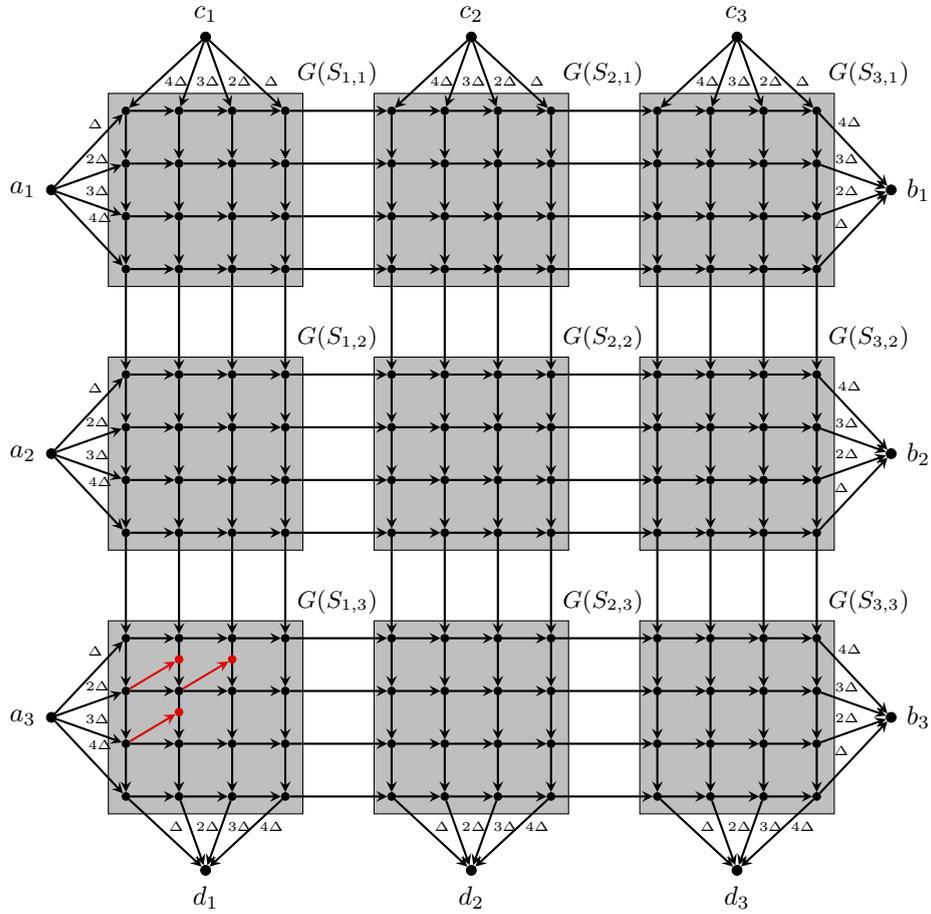

\subsubsection{Hard matching patterns}
\label{sec:orderedtough}  

The aim of this section is to prove hardness for the hard matching patterns. We restate here the definition of these graphs for the reader's convenience.

\begingroup
\def\thetheorem{\ref{def:cleanedorderedtoughpair}}
\begin{definition}[$t$-hard matching pattern]
A \emph{$t$-hard matching pattern} is
an (acyclic) digraph $G$ constructed the following way. We start with disjoint vertex sets $A= \{a_1,\ldots,a_t\}$, $B = \{b_1,\ldots,b_t\}$, $C = \{c_1,\ldots,c_t\}$ and $D = \{d_1,\ldots,d_t\}$ and introduce the edges $(a_i,b_i)$ and $(c_i,d_i)$ for every $i\in[t]$.
Furthermore, we introduce into $G$ any combination of the following items:
\begin{enumerate}
\item either the directed path $a_1 \rightarrow a_2 \rightarrow \ldots \rightarrow a_t \rightarrow d_1 \rightarrow d_2 \rightarrow \ldots \rightarrow d_t$, or any of the directed paths $a_1 \rightarrow a_2 \rightarrow \ldots \rightarrow a_t$ and 
$d_1 \rightarrow d_2 \rightarrow \ldots \rightarrow d_t$;
\item either the directed path $c_1 \rightarrow c_2 \rightarrow \ldots \rightarrow c_t \rightarrow b_1 \rightarrow b_2 \rightarrow \ldots \rightarrow b_t$, or any of the directed paths $b_1 \rightarrow b_2 \rightarrow \ldots \rightarrow b_t$ and $c_1 \rightarrow c_2 \rightarrow \ldots \rightarrow c_t$;
\item an $S$-source for exactly one $S \in \{A,B,C,D,A \cup C,B \cup D,B \cup C,A \cup D\}$;
\item an $S$-sink for exactly one $S \in \{A,B,C,D,A \cup C,B \cup D,B \cup C,A \cup d\}$;
\item a vertex $r_{AD}$ such that $N^-(r_{AD}) = A$ and $N^+(r_{AD}) = D$;
\item a vertex $r_{CB}$ such that $N^-(r_{CB}) = Y$ and $N^+(r_{CB}) = B$.
\end{enumerate}
In particular, there are $5 \cdot 5 \cdot 9 \cdot 9 \cdot 2 \cdot 2$ types of $t$-hard matching patterns: we let $\mathcal{C}_5, \ldots, \mathcal{C}_{8104}$ be the 8100 classes that each contain all the $t$-hard matching patterns of a specific type for every~$t$.
\end{definition}
\addtocounter{theorem}{-1}
\endgroup

Formally, we aim to prove the following.

\begin{lemma}
\label{lem:cleanedtoughpair}
For every $\ell \in [5,8104]$, 
{\sc Planar $\mathcal{C}_\ell$-Steiner Network} is $\mathsf{W}[1]$-hard parameterized by the number $k$ of terminals and does not admit a $f(k) \cdot n^{o(k)}$ algorithm for any computable function $f$,
unless $\mathsf{ETH}$ fails.
\end{lemma}

To prove the lemma, we use the construction described above and further add vertices and edges so as to take into account the specific combination of the items in \Cref{def:cleanedorderedtoughpair} for each fixed $\ell \in [5,8104]$.\\

\noindent
\textbf{Reduction.} Consider $\ell \in [7,2031]$. Given an instance $(k,n,\mathcal{S}=\{S_{i,j}~|~i,j \in [k]\})$ of {\sc $k \times k$-Grid Tiling}, we construct an equivalent instance $(G(\mathcal{S}),T,D)$ of (edge-weighted) {\sc Planar $\mathcal{C}_\ell$-Steiner Network} as follows. Let $MG(\mathcal{S})$ be the main gadget for $\mathcal{S}$ as constructed above.
Then $T$ contains $\{a_i,b_i,c_i,d_i~|~i \in [k]\}$ (and possibly, as described below, some additional vertices depending on the class $\mathcal{C}\ell$) where $\{(a_i,b_i)~|~i \in [k]\}$ and $\{(c_i,d_i)~|~i \in [k]\}$ are the two perfect matchings contained in $D$. The graph $G(\mathcal{S})$ is then obtained from $MG(\mathcal{S})$ by introducing the directed paths
\[
a_1 \rightarrow \ldots \rightarrow a_k, d_1 \rightarrow \ldots \rightarrow d_k, c_1 \rightarrow \ldots \rightarrow c_k 
\text{ and } b_1 \rightarrow \ldots \rightarrow b_k
\]
where the weight of each newly added edge is set to 0. Furthermore, if the patterns in the class $C_\ell$ contain
\begin{itemize}
\item a source (that is, a vertex of item 3 in \Cref{def:cleanedorderedtoughpair}) then we add a vertex $s$ to $G(\mathcal{S})$ and $T$, and the edges $(s,a_1)$ and $(s,c_1)$, both of weight 0;
\item a sink (that is, a vertex of item 4 in \Cref{def:cleanedorderedtoughpair}) then we add a vertex $t$ to $G(\mathcal{S})$ and $T$, and the edges $(b_k,t)$ and $(d_k,t)$, both of weight 0;
\item the vertex of item 5 in \Cref{def:cleanedorderedtoughpair}, then we add a vertex $r_{AD}$ to $G(\mathcal{S})$ and $T$, and the edges $(a_k,r_{AD})$ and $(r_{AD},d_1)$, both of weight 0;
\item the vertex of item 6 in \Cref{def:cleanedorderedtoughpair}, then we add a vertex $r_{CB}$ to $G(\mathcal{S})$ and $T$, and the edges $(c_k,r_{CB})$ and $(r_{CB},b_1)$, both of weight 0.
\end{itemize}
This concludes the construction of $G(\mathcal{S})$. We let $D$ be the corresponding $k$-hard matching pattern of $\mathcal{C}_\ell$ on vertex set $T$. \Cref{lem:cleanedtoughpair} then follows from \Cref{lem:grid-dsn} and the lemma below.

\begin{lemma}
The (edge-weighted) {\sc Planar $\mathcal{C}_\ell$-Steiner Network} instance $(G(\mathcal{S}),T,D)$ has a solution of weight at most $B^*-k^2$ if and only if the (edge-weighted) {\sc Planar $M$-Steiner Network} instance $(MG(\mathcal{S}),I,M)$ has a solution of weight at most $B^*-k^2$.
\end{lemma}

\begin{proof}
If $E$ is a solution of $(MG(\mathcal{S}),I,M)$ of weight at most $B^*-k^2$ then it is easy to see that $E \cup (E(G(\mathcal{S})) \setminus E(MG(\mathcal{S})))$ is a solution of $(G(\mathcal{S}),T,D)$ of weight at most $B^*-k^2$.
Conversely, if $E$ is a solution of $(G(\mathcal{S}),T,D)$ of weight at most $B^*-k^2$ then the restriction of $E$ to $MG(\mathcal{S})$ is readily seen to be a solution of $(MG(\mathcal{S}),I,M)$ of weight at most $B^*-k^2$.
\end{proof}

\subsubsection{Hard biclique patterns}
\label{sec:biclique}

The aim of this section is to prove hardness for the hard biclique patterns. We restate here the definition of these graphs for the reader's convenience.

\begingroup
\def\thetheorem{\ref{def:cleanedbiclique}}
\begin{definition}[$t$-hard biclique pattern]
A \emph{$t$-hard biclique pattern} is an (acyclic) digraph $D$ constructed the following way. We start with two disjoint sets $A$ and $B$ with  $|A| = |B| = t$ and introduce every edge from $A$ to $B$. Furthermore, we introduce into $D$ any combination of the following items (see \Cref{fig:cleanedbiclique}):
\begin{enumerate}
\item an $A$-source;
\item a $B$-sink.
\end{enumerate}
In particular, there are $2 \cdot 2$ types of $t$-hard biclique patterns: we let $\mathcal{C}_1,\ldots,\mathcal{C}_4$ be the 4 classes that each contain all the $t$-hard biclique patterns of a specific type for every $t$.
\end{definition}
\addtocounter{theorem}{-1}
\endgroup

Formally, we aim to prove the following.

\begin{lemma}
\label{lem:cleanedbiclique}
For every $\ell \in [4]$,
{\sc Planar $\mathcal{C}_\ell$-Steiner Network} is $\mathsf{W}[1]$-hard parameterized by the number $k$ of terminals and does not admit a $f(k) \cdot n^{o(k)}$ algorithm for any computable function $f$,
unless $\mathsf{ETH}$ fails.
\end{lemma}

We only formally prove the statement for the class of all hard biclique patterns containing no further vertices 
as it will become clear from the proof that to handle the classes of all hard biclique patterns containing
\begin{itemize}
\item the sink vertex, it suffices to add a vertex $t$ and an edge $(t_i,t)$ for each $i \in [2k+1]$ in the construction below; or
\item the source vertex, it suffices to add a vertex $s$ and an edge $(s,s_i)$ for each $i \in [2k+1]$ in the construction below.
\end{itemize} 
In the following, we assume that the class of all hard biclique patterns with no further vertices is $\mathcal{C}_1$.\\

\newcommand{\biclique}{{\sc Planar $\mathcal{C}_1$-Steiner Network}}

\noindent
\textbf{Reduction.} Given an instance $(k,n,\mathcal{S}=\{S_{i,j}~|~i,j \in [k]\})$ of {\sc $k \times k$-Grid Tiling}, we construct an equivalent instance $(G(\mathcal{S}),T,D)$ of (edge-weighted) \biclique\ as follows. We start by constructing an auxiliary planar digraph $H$ consisting of $2(2k+1)$ distinguished vertices $s_1,\ldots,s_{2k+1},t_1,\ldots,t_{2k+1}$, and $2(2k+1)$ edge-disjoint directed paths $P_{i,j}$ with $i \in [2k+1]$ and $j \in \{i,i+1\}$ (where indices are taken modulo $2k+1$ henceforth), defined as follows.
\begin{itemize}
\item For every $i \in [2k+1]$, $P_{i,i}= s_iu_i^1\ldots u_i^{2k-1}t_i$ is a directed path from $s_i$ to $t_i$ of length $2k$.
\item For every $i \in [2k+1]$, $P_{i,i+1}$ is the directed path $s_i u^1_{i-1}u^2_{i-2}\ldots u^j_{i-j}\ldots u^{2k-1}_{i+2}t_{i+1}$ from $s_i$ to $t_{i+1}$ of length $2k$.
\end{itemize}
(For $k=3$, the graph $H$ can be obtained from the graph depicted in \Cref{fig:bicliquered} by ignoring the blue vertices and contracting every grey box into a single vertex.) Note that by construction, the following holds.

\begin{observation}
\label{obs:uniqupath}
For every $i \in [2k+1]$, there is a unique path from $s_i$ to $t_i$ ($t_{i+1}$, respectively) in $H$, namely $P_{i,i}$ ($P_{i,i+1}$, respectively). 
\end{observation}

\noindent
The graph $G(\mathcal{S})$ is then obtained from $H$ as follows (see \Cref{fig:bicliquered}).
\begin{itemize}
\item We subdivide the edge $(s_{k+1},u^{k+1}_1)$ by adding the vertex $a_1$ and the edges $(s_{k+1},a_1)$ and $(a_1,u^{k+1}_1)$; and for every $2 \leq i \leq k$, we subdivide the edge $(u^i_{k+1-i},u^i_{k+2-i})$ by adding the vertex $a_{k+2-i}$ and the edges $(u^i_{k+1-i},a_{k+2-i})$ and $(a_{k+2-i},u^i_{k-i})$.
\item We subdivide the edge $(u^2_{2k-1},t_2)$ by adding the vertex $b_k$ and the edges $(u^2_{2k-1},b_k)$ and $(b_k,t_2)$; and for every $3 \leq i \leq k+1$, we subdivide the edge $(u^i_{2k+1-i},u^i_{2k+2-i})$ by adding the vertex $b_{k+2-i}$ and the edges $(u^i_{2k+1-i},b_{k+2-i})$ and $(b_{k+2-i},u^i_{2k+2-i})$.
\item We subdivide the edge $(s_{k+2},u^{k+1}_1)$ by adding the vertex $c_1$ and the edges $(s_{k+2},c_1)$ and $(c_1,u^{k+1}_1)$; and for every $k+3 \leq i \leq 2k+1$, we subdivide the edge $(u^{k+2}_{i-k-2},u^{k+1}_{i-k-2})$ by adding the vertex $c_{i-k-1}$ and the edges $(u^{k+2}_{i-k-2},c_{i-k-1})$ and $(c_{i-k-1},u^{k+1}_{i-k-2})$.
\item For every $k+2 \leq i \leq 2k+1$, we subdivide the edge $(u^2_{i-k-2},u^1_{i-k-2})$ by adding the vertex $d_{i-k-1}$ and the edges $(u^2_{i-k-2},d_{i-k-1})$ and $(d_{i-k-1},u^1_{i-k-2})$.
\item For every $1\leq i,j \leq k$, we replace the vertex $u^{j+1}_{k-j+i}$ with a copy of the graph $G(S_{i,k+1-j})$ and add the necessary edges so that the subgraph of $G$ induced by $\{a_i,b_i,c_i,d_i~|~ 1 \leq i \leq k\} \cup \bigcup_{1 \leq i,j \leq k} V(G(S_{i,j}))$ is isomorphic to $MG(\mathcal{S})$ (and has the same edge-weight function).
\end{itemize}
The weight of each edge outside the copy of the main gadget $MG(\mathcal{S})$ is set to 0. This concludes the construction of $G(\mathcal{S})$. We set $T = \{s_i,t_i~|~i\in [2k+1]\}$ and let the demand graph $D$ be the corresponding $(\{s_i~|~i \in [2k+1]\},\{t_i~|~i \in [2k+1]\})$-biclique. \Cref{lem:cleanedbiclique} for $\ell=1$ then follows from \Cref{lem:grid-dsn} and the lemma below.

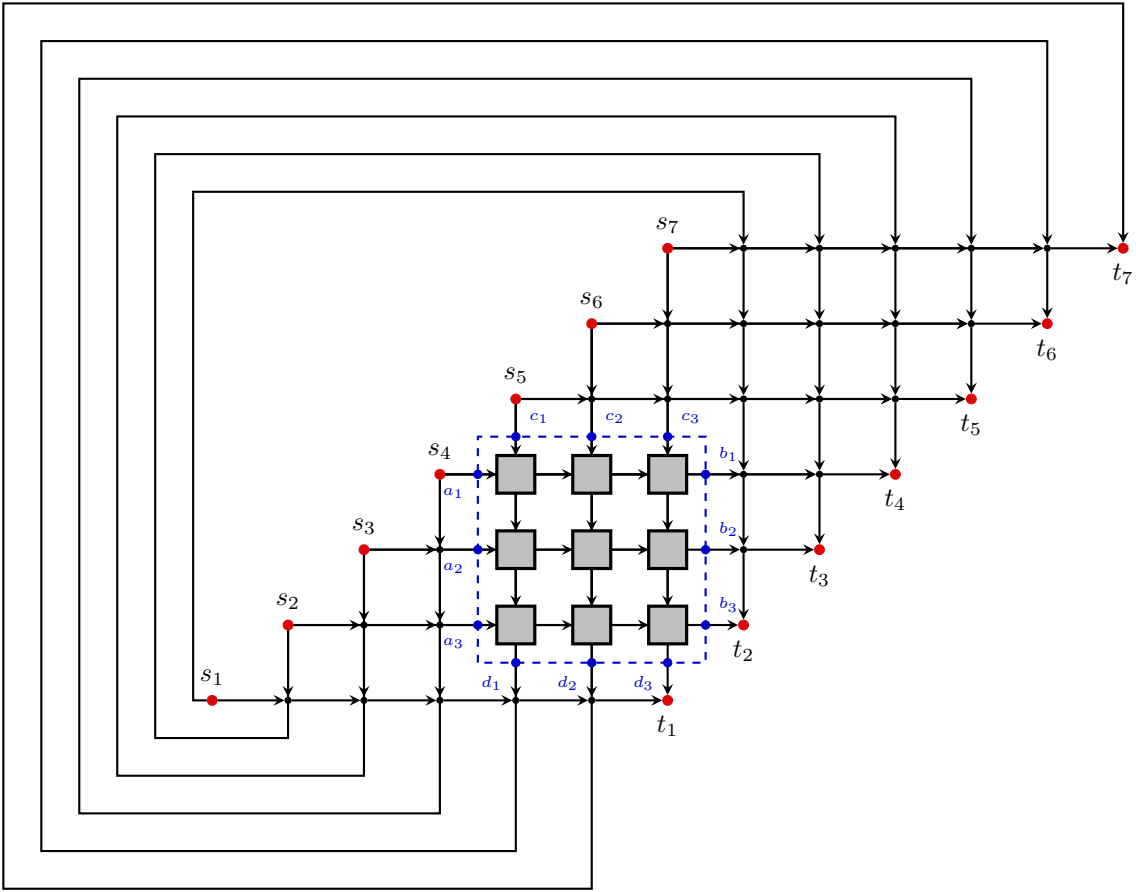
\begin{figure}
\centering
\begin{tikzpicture}
\foreach \i in {1,...,7}
{\pgfmathsetmacro{\x}{\i-.75}
\node[circR,label={above:\small $s_\i$}] (a\i) at (\x,\x) {};
}

\foreach \i in {1,...,7}
{\pgfmathsetmacro{\y}{\i-.75}
\pgfmathsetmacro{\x}{\y+6}
\node[circR,label={below:\small $t_\i$}] (b\i) at (\x,\y) {};
}

\foreach \i in {2,...,6}
{\pgfmathsetmacro{\x}{\i+5.25}
\node[scirc] (\i7) at (\x,6.25) {};
}

\foreach \i in {1,...,6}
{\pgfmathtruncatemacro{\j}{\i+1}
\pgfmathsetmacro{\a}{\i-1}
\pgfmathsetmacro{\1}{\a*.5}
\pgfmathsetmacro{\b}{-\1}
\pgfmathsetmacro{\c}{\b+.25}
\pgfmathsetmacro{\2}{\1+7}
\pgfmathsetmacro{\3}{\a+.25}
\pgfmathsetmacro{\4}{\i+6.25}
\ifthenelse{\i=1 \OR \i=6}{\ifthenelse{\i=1}{\draw[->,thick,>=stealth] (a\i)-- (\b,\c) -- (\b,\2) -- (\4,\2) -- (\j7);}{\draw[->,thick,>=stealth] (a\i) -- (\3,\c) -- (\b,\c) -- (\b,\2) -- (\4,\2) -- (b\j);}}{\draw[->,thick,>=stealth] (a\i) -- (\3,\c) -- (\b,\c) -- (\b,\2) -- (\4,\2) -- (\j7);}
}

\foreach \i in {2,...,6}
{\ifthenelse{\i=6}{\draw[->,thick,>=stealth] (67) -- (b6);}{
\pgfmathtruncatemacro{\k}{\i+1}
\foreach \j in {6,...,\k}
{\pgfmathsetmacro{\x}{\i+5.25}
\pgfmathsetmacro{\2}{\j-.75}
\pgfmathtruncatemacro{\l}{\j+1}
\node[scirc] (\i\j) at (\x,\2) {};
\ifthenelse{\j=\k}{\draw[->,thick,>=stealth] (\i\l) -- (\i\j); \draw[->,thick,>=stealth] (\i\j) -- (b\i);}{\draw[->,thick,>=stealth] (\i\l) -- (\i\j);}
}
}
}

\node[scirc] (16) at (6.25,5.25) {};
\draw[->,thick,>=stealth] (a7) -- (16);
\node[scirc] (15) at (6.25,4.25) {};
\draw[->,thick,>=stealth] (16) -- (15);
\node[scirc] (05) at (5.25,4.25) {};
\draw[->,thick,>=stealth] (a6) -- (05);

\node[scirc] (-23) at (3.25,2.25) {};
\draw[->,thick,>=stealth] (a4) -- (-23);
\node[scirc] (-22) at (3.25,1.25) {};
\draw[->,thick,>=stealth] (-23) -- (-22);
\node[scirc] (-32) at (2.25,1.25) {};
\draw[->,thick,>=stealth] (a3) -- (-32);

\foreach \i in {1,...,5}
{\pgfmathsetmacro{\x}{\i+.25}
\pgfmathtruncatemacro{\j}{5-\i}
\node[scirc] (-\j1) at (\x,.25) {};
}

\draw[->,thick,>=stealth] (-22) -- (-21);
\draw[->,thick,>=stealth] (-32) -- (-31);
\draw[->,thick,>=stealth] (a2) -- (-41);

\foreach \i in {4,...,1}
{\pgfmathtruncatemacro{\j}{\i-1}
\draw[->,thick,>=stealth] (-\i1) -- (-\j1);
}
\draw[->,thick,>=stealth] (a1) -- (-41);
\draw[->,thick,>=stealth] (-01) -- (b1);

\draw[->,thick,>=stealth] (a2) -- (-32);
\draw[->,thick,>=stealth] (-32) -- (-22);

\draw[->,thick,>=stealth] (a3) -- (-23);

\draw[->,thick,>=stealth] (24) -- (34);

\foreach \i in {5,6,7}
{\pgfmathtruncatemacro{\k}{\i-5}
\pgfmathtruncatemacro{\l}{\i-2}
\foreach \j in {\k,...,\l}
{\pgfmathtruncatemacro{\p}{\j+1}
\ifthenelse{\j=\k}{\draw[->,thick,>=stealth] (a\i) -- (\j\i);\draw[->,thick,>=stealth] (\j\i) -- (\p\i);}{
\draw[->,thick,>=stealth] (\j\i) -- (\p\i);}
}
}

\draw[->,thick,>=stealth] (a7) -- (b1);
\draw[->,thick,>=stealth] (a2) -- (b2);

\foreach \i in {3,...,7}
{\pgfmathtruncatemacro{\j}{\i-1}
\draw[->,thick,>=stealth] (a\i) -- (\j\i);
\draw[->,thick,>=stealth] (\j\i) -- (b\i);
}

\draw[->,thick,>=stealth] (a4) -- (24);
\draw[->,thick,>=stealth] (a5) -- (-11);
\draw[->,thick,>=stealth] (a6) -- (-01);

\draw[->,thick,>=stealth] (a5) -- (4.25,3.5);
\draw[->,thick,>=stealth] (05) -- (5.25,3.5);
\draw[->,thick,>=stealth] (15) -- (6.25,3.5);

\draw[->,thick,>=stealth] (a4) -- (4,3.25); 
\draw[->,thick,>=stealth] (-23) -- (4,2.25);
\draw[->,thick,>=stealth] (-22) -- (4,1.25);

\foreach \i in {4,5,6}
\foreach \j in {1,2,3}
{\pgfmathsetmacro{\x}{\i+.5}
\pgfmathsetmacro{\y}{\j+.5}
\filldraw[fill=lightgray,very thick] (\i,\j) rectangle (\x,\y);
}

\foreach \i in {4,5}
\foreach \j in {1,2,3}
{\pgfmathsetmacro{\x}{\i+.5}
\pgfmathsetmacro{\u}{\i+1}
\pgfmathsetmacro{\y}{\j+.25}
\draw[->,thick,>=stealth] (\x,\y) -- (\u,\y);
}

\foreach \j in {3,2}
\foreach \i in {4,5,6}
{\pgfmathsetmacro{\x}{\i+.25}
\pgfmathsetmacro{\u}{\j-.5}
\draw[->,thick,>=stealth] (\x,\j) -- (\x,\u);
}

\foreach \i in {2,3,4}
{\pgfmathsetmacro{\y}{\i-.75}
\pgfmathtruncatemacro{\j}{5-\i}
\node[circb,label={[blue] below left:\tiny $a_\j$}] at (3.75,\y) {};
}

\foreach \i in {2,3,4}
{\pgfmathsetmacro{\y}{\i-.75}
\pgfmathtruncatemacro{\j}{5-\i}
\node[circb,label={[blue] above right:\tiny $b_\j$}] at (6.75,\y) {};
}

\foreach \i in {4,5,6}
{\pgfmathsetmacro{\x}{\i+.25}
\pgfmathtruncatemacro{\j}{\i-3}
\node[circb,label={[blue] below left:\tiny $d_\j$}] at (\x,.75) {};
}

\foreach \i in {4,5,6}
{\pgfmathsetmacro{\x}{\i+.25}
\pgfmathtruncatemacro{\j}{\i-3}
\node[circb,label={[blue] above right:\tiny $c_\j$}] at (\x,3.75) {};
}

\draw[dashed,blue,thick] (3.75,.75) rectangle (6.75,3.75);
\end{tikzpicture}
\caption{An illustration of the reduction from {\sc $k \times k$-Grid Tiling} to \biclique\ with $k=3$ (the red vertices are the terminals and the dashed blue square together with the grey boxes represent the main gadget).}
\label{fig:bicliquered}
\end{figure}

\begin{lemma}
The (edge-weighted) \biclique\ instance $(G(\mathcal{S}),T,D)$ has a solution of weight at most $B^*-k^2$ if and only if the (edge-weighted) {\sc Planar $M$-Steiner Network} instance $(MG(\mathcal{S}),I,M)$ has a solution of weight at most $B^*-k^2$.
\end{lemma}

\begin{proof}
If $E$ is a solution of $(MG(\mathcal{S}),I,M)$ of weight at most $B^*-k^2$ then it is not difficult to see that $E \cup (E(G(\mathcal{S})) \setminus E(MG(\mathcal{S})))$ is a solution of $(G(\mathcal{S}),T,D)$ of weight at most $B^*-k^2$. Conversely, let $E$ be a solution of $(G(\mathcal{S}),T,D)$ of weight at most $B^*-k^2$. Since for every $2 \leq i \leq k+1$, there is a unique path in $H$ from $s_i$ to $t_i$, namely $P_{i,i}$, and for every $k+2 \leq i \leq 2k+1$, there is a unique path in $H$ from $s_i$ to $t_{i+1}$, namely $P_{i,i+1}$, the restriction $E'$ of $E$ to $MG(\mathcal{S})$ is a solution of $(MG(\mathcal{S}),I,M)$; and since every edge in $E(G(\mathcal{S})) \setminus E(MG(\mathcal{S}))$ has weight 0, the weight of $E'$ is that of $E$, that is, $E'$ has weight at most $B^*-k^2$.
\end{proof}

\bibliographystyle{siam}
\bibliography{references}

\end{document}